\UseRawInputEncoding
\documentclass[superscriptaddress,aps,prx,nofootinbib,notitlepage,longbibliography,10pt]{revtex4-1}
\usepackage[colorlinks]{hyperref}
\usepackage{amssymb,amsmath,amsthm,amsfonts}
\usepackage{thm-restate}
\usepackage{url}
\usepackage{graphicx}
\usepackage{float}
\usepackage{footnote}
\usepackage{textcomp}
\usepackage{gensymb}
\usepackage[margin=1in]{geometry}
\usepackage[dvipsnames]{xcolor}
\usepackage[section]{placeins}
\usepackage{cleveref}
\usepackage{array}
\usepackage{enumitem}
\usepackage{colortbl}
\newcommand{\ketbra}[2]{\ket{#1}\!\bra{#2}}
\definecolor{mycolor1}{rgb}{0.00000,0.44700,0.74100}%

\usepackage{mathtools}
\mathtoolsset{centercolon}
\DeclarePairedDelimiter\bra{\langle}{\rvert}
\DeclarePairedDelimiter\ket{\lvert}{\rangle}
\DeclarePairedDelimiterX\braket[2]{\langle}{\rangle}{#1 \delimsize\vert #2}

\newcommand{\gline}{\arrayrulecolor{gray}\hline\arrayrulecolor{black}}

\newcommand{\ceil}[1]{\lceil{#1}\rceil}

\newtheorem{theorem}{Theorem}

\newtheorem{lemma}[theorem]{Lemma}

\newcommand{\eq}[1]{(\ref{eq:#1})}
\newcommand{\thm}[1]{\hyperref[thm:#1]{Theorem~\ref*{thm:#1}}}
\newcommand{\defn}[1]{\hyperref[defn:#1]{Definition~\ref*{defn:#1}}}
\newcommand{\lem}[1]{\hyperref[lem:#1]{Lemma~\ref*{lem:#1}}}
\newcommand{\prop}[1]{\hyperref[prop:#1]{Proposition~\ref*{prop:#1}}}
\newcommand{\fig}[1]{\hyperref[fig:#1]{Figure~\ref*{fig:#1}}}
\newcommand{\tab}[1]{\hyperref[tab:#1]{Table~\ref*{tab:#1}}}
\renewcommand{\sec}[1]{\hyperref[sec:#1]{Section~\ref*{sec:#1}}}
\newcommand{\app}[1]{\hyperref[app:#1]{Appendix~\ref*{app:#1}}}
\newcommand{\cor}[1]{\hyperref[cor:#1]{Corollary~\ref*{cor:#1}}}
\newcommand{\obs}[1]{\hyperref[obs:#1]{Observation~\ref*{obs:#1}}}
\newcommand{\nn}{\nonumber \\}
\newcommand{\append}[1]{\hyperref[append:#1]{Appendix~\ref*{append:#1}}}

\newcommand{\amam}{a} % A variable which is 1 if we don't do amplitude amplification or 13 if we do.
\newcommand{\erase}[1]{{\rm Er}\left(#1\right)} % A function that is defined to be the cost of erasing a QROM.
\newcommand{\srt}[1]{{\rm St}\left( #1 \right)} % The number of comparators for the sort.
\newcommand{\reps}{\mathcal{N}} % The number of repetitions used in the phase estimation.
\newcommand{\psuc}{p_\nu} % The probability of success for the preparation of the state with amplitudes $1/\|\nu\|$.
\newcommand{\eqprep}[1]{{\rm Ps}\left(#1\right)} % The probability of success for preparing an equal superporition state.
\newcommand{\mi}{\mathrm{i}} % The complex number i=sqrt(-1).
\newcommand{\bgr}{b_{\rm grad}}

\newcommand{\norm}[1]{\left\lVert#1\right\rVert}

\newcommand{\Toffoli}{\mathrm{Toffoli}}

\newcommand{\PREP}{\mathrm{PREP}}
\newcommand{\SEL}{\mathrm{SEL}}
\newcommand{\CNOT}{\mathrm{CNOT}}

\newcommand{\sig}{\Sigma}
\newcommand{\alp}{\sigma}

\usepackage{tikz}
\usetikzlibrary{quantikz}
\usepackage{rotating}
\usepackage{multirow}

\makeatletter
\renewcommand*\env@matrix[1][\arraystretch]{%
	\edef\arraystretch{#1}%
	\hskip -\arraycolsep
	\let\@ifnextchar\new@ifnextchar
	\array{*\c@MaxMatrixCols c}}
\makeatother

\newcommand{\MQ}{\affiliation{%
Department of Physics and Astronomy,
Macquarie University, Sydney, NSW 2109, Australia}}

\newcommand{\Toronto}{\affiliation{%
Department of Computer Science, University of Toronto, ON M5S 2E4, Canada}}

\newcommand{\PNNL}{\affiliation{%
Pacific Northwest National Laboratory, Richland, WA 99354, United States}}

\newcommand{\Google}{\affiliation{%
Google Quantum AI, Venice, CA 90291, United States}}

\newcommand{\Caltech}{\affiliation{%
Institute for Quantum Information and Matter, Caltech, Pasadena, CA 91125, United States}}

\begin{document}
	
	\title{Fault-Tolerant Quantum Simulations of Chemistry in First Quantization}
	
	\date{\today}
	
	\author{Yuan Su}
	\email{corresponding author: buptsuyuan@gmail.com}
	\Google
	\Caltech
	
	\author{Dominic W. Berry}
	\email{corresponding author: dominic.berry@mq.edu.au}
	\MQ
	
	\author{Nathan Wiebe}
	\email{corresponding author: nathanwiebe@gmail.com}
	\Toronto \PNNL \Google
	
	\author{Nicholas Rubin}
	\Google
	
	\author{Ryan Babbush}
	\email{corresponding author: ryanbabbush@gmail.com}
	\Google
	
	\begin{abstract}
	Quantum simulations of chemistry in first quantization offer some important advantages over approaches in second quantization including faster convergence to the continuum limit and the opportunity for practical simulations outside of the Born-Oppenheimer approximation. However, since all prior work on quantum simulation of chemistry in first quantization has been limited to asymptotic analysis, it has been impossible to directly compare the resources required for these approaches to those required for the more commonly studied algorithms in second quantization. Here, we compile, optimize and analyze the finite resources required to implement two first quantized quantum algorithms for chemistry from Babbush \emph{et al.}~\cite{BabbushContinuum_b} that realize block encodings for the qubitization and interaction picture frameworks of Low \emph{et al.}~\cite{Low2016,Low2018}. The two algorithms we study enable simulation with gate complexities of $\widetilde{\cal O}(\eta^{8/3} N^{1/3} t + \eta^{4/3} N^{2/3} t)$ and
	$\widetilde{\cal O}(\eta^{8/3} N^{1/3} t)$ where $\eta$ is the number of electrons, $N$ is the number of plane wave basis functions, and $t$ is the duration of time-evolution ($t$ is linearly inverse to target precision when the goal is to estimate energies). In addition to providing the first explicit circuits and constant factors for any first quantized simulation, and then introducing improvements which reduce circuit complexity by about a thousandfold over naive implementations for modest sized systems, we also describe new algorithms that asymptotically achieve the same scaling in a real space representation. Finally, we assess the resources required to simulate various molecules and materials and conclude that the qubitized algorithm will often be more practical than the interaction picture algorithm. We demonstrate that our qubitization algorithm often requires much less surface code spacetime volume for simulating millions of plane waves than the best second quantized algorithms require for simulating hundreds of Gaussian orbitals.
	\end{abstract}
	
	\maketitle

	\tableofcontents
	
	%%%%%%%%%%%%%%%%%%%%%%%%%%%%%%%%%%%%%%%%%%%%%%%%%%%%%%%%%%%%%%%%%%%%%%%%%%%%%%

	\section{Introduction}
		\label{sec:intro}
		
	The last several years have seen monumental developments in quantum algorithms, resulting in significant progress towards the goal of showing that a modest sized quantum computer can provide a decisive advantage for scientifically relevant problems.  The expectation that quantum simulation might enable an exponential speedup for physical problems stretches back to Feynman's original proposal for quantum computers~\cite{Feynman1982,Lloyd1996} and was first concretely proposed for solving the the central problem of quantum chemistry by Aspuru-Guzik \emph{et al}.~\cite{Aspuru-Guzik2005}. Since then, progress has accelerated on multiple fronts: improvements in general purpose simulation frameworks~\cite{Berry2015,Wiebe2015b,Low2016,Low2018,Poulin2017,Haah2018b,CampbellPRL2018,Childs2021,motta2020determining,SS20,Tran2021}, more efficient algorithms for quantum chemistry \cite{Whitfield2010,Wecker2014_b,Hastings2015,Poulin2014,BabbushSparse1,BGBWMPFN18,Berry2018,Berry2019b,Kivlichan2017,Kivlichan2020improvedfault,mcardle2020quantum,vonBurg2020Catalysis,Lee2020,Su2020_b}, more advantageous-to-quantum-simulate representations of molecular Hamiltonians~\cite{BabbushLow,Seeley2012,Bravyi2017,Setia2017,Motta2018,bauman2019downfolding,Takeshita2019b,McClean2020Galerkin_b,Motta2020_b} and an improved understanding of the problems that genuinely require a quantum computer versus those that can be classically computed~\cite{Reiher2017,Li2019_b}.
	
	However, a problem lurks within this discussion that has not received as much attention as it deserves.  Specifically, most existing methods for simulating quantum chemistry (especially those leveraging simple basis functions like plane waves) become impractical as we scale to the continuum limit.  This is because the second quantized simulations that are deployed in the vast majority of quantum computing algorithms for chemistry require a number of qubits that scales linearly with the number of spin-orbital basis functions. Dating back to early work by Kassal \emph{et al.}~\cite{Kassal2008}, first quantized simulations of chemistry have been proposed as a method to ameliorate this problem. The central idea behind such simulations is to track the basis state that each particle is in rather than storing the occupancy of each basis state.  This has two significant benefits: the number of qubits needed to represent the state scales logarithmically with the number of basis functions (although linearly with the number of particles) and also, such simulations are easy to adapt to cases where the entanglement between the electronic and nuclear subsystems is non-negligible.  Despite the promise of these methods, no prior work investigating first quantized approaches for chemistry has managed to rigorously estimate all of the constant factors associated with the realization of these algorithms\footnote{In the first paper to discuss error-correcting quantum algorithms for chemistry, Jones \emph{et al.}~\cite{Jones2012} estimate the number of Toffoli gates required to implement subroutines (Trotter steps) of the Kassal \emph{et al.}~\cite{Kassal2008} algorithm. However, their work stops short of the error analysis required to predict the constant factors associated with performing fixed precision quantum simulation. For example, they do not estimate how many Trotter steps would be required to achieve target precision}.
	
	Our aim in this paper is to address this shortcoming by providing a full costing of the two leading quantum simulation algorithms for chemistry in a first quantized representation, both first described by Babbush \emph{et al.}~\cite{BabbushContinuum_b}. These algorithms comprise asymptotically efficient schemes for realizing the oracles required by the qubitization and interaction picture frameworks of Low \emph{et al.}~\cite{Low2016,Low2018}. Here, we compile and optimize these algorithms within an error-correctable gate set. We develop new techniques to reduce the complexity, resulting in about three orders of magnitude improvement over a naive implementation for even modest sized systems. Our detailed cost analysis reveals that despite the improved scaling of the interaction picture based method, the qubitization based method often proves to be more practical. We conduct numerical experiments to characterize the overhead of both approaches and find that they are highly competitive with second quantized approaches, often requiring significantly fewer resources to reach comparable accuracy.
	
	The layout of this paper is as follows.  In the remainder of this section, we provide the formalism and context needed to understand the two most efficient schemes for simulating the molecular Hamiltonian in first quantization. \sec{qubitization} develops and analyzes our optimized implementation of the qubitization based algorithm in first quantization. \sec{interaction_picture} contains a similar analysis for the interaction picture based algorithm. Finally, \sec{results} numerically compares these algorithms (to each other and to other methods in the literature), estimating the resources required to simulate real systems in a first quantized basis.
	
	\subsection{Background on representing molecular Hamiltonians in a first quantized plane wave basis}
\label{sec:problem_definition}
	
	The accurate simulation of quantum systems was Feynman's original vision for quantum computing \cite{Feynman1982}. There is perhaps no more natural application of this idea than the simulation of atoms and molecules, governed by the interactions of electrons and nuclei. Such systems give rise to the properties of most matter, determining everything from the rates of chemical reactions to the conductivity of materials. In the non-relativistic case, the dynamics of these systems are governed by the Coulomb Hamiltonian:
\begin{align}
H =  \underbrace{-\sum_{i=1}^{\eta} \frac{\nabla^2_i}{2}}_{T} 
\underbrace{- \sum_{i=1}^\eta\sum_{\ell =1}^{L}\frac{\zeta_\ell}{\left\|R_\ell - r_i\right\|}}_{U} 
\underbrace{+ \sum_{i\neq j=1}^{\eta} \frac{1}{2\left\|r_i - r_j\right\|}}_{V} \underbrace{-\sum_{\ell=1}^{L} \frac{\nabla^2_\ell}{2 m_\ell}}_{T_{\rm nuclei}}
\underbrace{+ \sum_{\ell \neq \kappa=1}^L\frac{\zeta_\ell \zeta_\kappa}{2\left\|R_\ell - R_\kappa\right\|}}_{V_\textrm{nuclei}} \, .
\label{eq:non-bo}
\end{align}
Here we have used atomic units where $\hbar$ as well as the mass and charge of the electron are unity, and $\|\cdot\|$ denotes the two-norm. In the above expression $i,j$ represent electronic degrees of freedom and $\ell,\kappa$ represent nuclear degrees of freedom; thus, $r_i$ represent the positions of electrons whereas $R_\ell$ represent the positions of nuclei, $m_\ell$ the atomic masses of nuclei, and $\zeta_\ell$ the atomic numbers of nuclei. Throughout this work we will use $\eta$ to denote the number of electrons in our simulation and $L$ to denote the number of nuclei in our simulation. We keep a catalogue of the different symbols used throughout this paper in \app{names}.

The methods we analyze and develop in this work are useful for simulations of the dynamics of \eq{non-bo} (treating both nuclear and electronic degrees of freedom as fully quantum) with only minor modifications. But in order to simplify and reduce our analysis, we will primarily focus on simulating \eq{non-bo} under the Born-Oppenheimer approximation: i.e., $H( R) = T + U( R) + V + C( R)$ where $C( R)$ is a constant given by the Coulomb repulsion of the point charges of ``classical'' nuclei with locations $R_\ell$. Thus, the Born-Oppenheimer approximation assumes that one can decouple nuclear and electronic degrees of freedom, treating the former as essentially ``classical''\cite{polanyi1972bconcepts, born1927quantentheorie}. This is appropriate for many systems near room temperature because the electronic degrees of freedom are thousands of times faster than the nuclear degrees of freedom (due to the disparity in mass between electrons and nuclei) and hence, electrons often relax nearly instantaneously to their quantum ground states whereas the nuclei are often heavy enough to be modelled as classical objects responding to forces arising from the electronic interactions. However, the Born-Oppenheimer approximation is known to break down in some circumstances, including at low temperatures or when chemical bonds involve low atomic mass nuclei (e.g., due to the tunneling of hydrogen atoms); hence, methods for performing simulations of \eq{non-bo}, referred to as ``non-Born-Oppenheimer dynamics'' are also of great interest~\cite{yarkony2012nonadiabatic, butler1998chemical, curchod2018ab, zhu2016non}.

The frequently studied ``electronic structure problem'' is to solve for the ground state energies of electrons interacting in the external potential of nuclei under the Born-Oppenheimer approximation (i.e., to solve for the ground state energies of $H( R)$). As a function of the nuclear coordinates $ R$, these energies define an energy surface and those energy surfaces can be extremely helpful for understanding the mechanisms and dynamics of chemical reactions as well as material properties. For example, energy surfaces are commonly sought after so that one may simulate the ``classical'' dynamics of nuclei moving on these ``quantum'' energy surfaces by integrating Newton's equations of motion for the nuclei. When the trajectories of these classical nuclei responding to quantum electronic energy surfaces are coupled to a finite temperature bath so that the system may sample the canonical or grand canonical ensemble (rather than just Hamiltonian dynamics), these simulations are referred to as ``\emph{ab initio} molecular dynamics''~\cite{Car1985, marx2009, hutter2012car}. While the electronic energy surfaces (referred to in the molecular dynamics community as ``force-fields'') are often determined empirically or by crude calculations that introduce large errors into the end result, molecular dynamics simulations are among the most common calculations in scientific computing and can provide invaluable insights into the thermodynamic properties of real materials and chemical reactions. 

In order to simulate molecular systems on a computer one must discretize them in some fashion. For reasons discussed in more depth in \app{background}, most molecular modeling employs a Galerkin representation which involves projecting the system onto some well behaved set of basis functions $\{\phi_p(r)\}$. The matrix elements of the Hamiltonian operators are then given by the following integrals over these basis functions:
\begin{align}
\label{eq:integrals}
T_{pq}^{(m)} & = \int \! dr \, \phi_p^* \left(r\right) \left(-\frac{\nabla^2}{2 m} \right)  \phi_q \left(r\right) \\
	U_{pq} & = \sum_{\ell=1}^{L} \int \! dr \, \phi_p^* \left(r\right) \left(\frac{\zeta_\ell}{\left\|r - R_\ell\right\|}\right)  \phi_q \left(r\right) \\
	V_{pqrs}^{(\alpha,\beta)} &= \int \! dr_1 \,  dr_2 \, \phi_p^*\left(r_1\right) \phi_q^* \left(r_2\right) \left(\frac{\alpha \beta}{\left\|r_1-r_2\right\|}\right) \phi_r(r_2) \phi_s \left(r_1\right) \, .
	\end{align}
	We note that simple grid representations (as opposed to discrete-variable-representation-like grids, e.g., those discussed for quantum simulation in \cite{BabbushLow}) are incompatible with Galerkin discretizations. For example, if one associates a grid with compact basis functions having disjoint support (e.g., delta functions or step functions) then evaluation of the above integrals would give both potential operators and kinetic operators that are simultaneously diagonal. For more discussion on real space representations see \app{real_space}.

A key consideration for Galerkin discretizations is the compactness of the basis set for the states of interest (often the ground state); i.e., how many basis functions are required to converge the discretization of states of interest to within some threshold $\epsilon$ of the continuum limit. Throughout this paper we will use $N$ to refer to the number of basis functions used to discretize the system. For most reasonable choices of basis sets, the asymptotic behavior is that $\epsilon = {\cal O}(1/N)$ \cite{gruneis2013explicitly,hattig2011explicitly,Harl2008,shepherd2012convergence,helgaker1997basis_b,klopper1995ab2,Halkier1998b}. However, the constant factors in this scaling can differ considerably depending on the context of the simulation and the choice of basis. The primary factor limiting the convergence of these discretizations is the resolution of ``cusps'' that are known to be features of eigenstates of molecular Hamiltonians \cite{Kato1957}. Wavefunction cusps appear at all points in space where particles overlap. Away from these cusps, the wavefunction is generally smooth and easier to converge. When working within the Born-Oppenheimer approximation, the nuclei have fixed position and there is a clear cusp in the electronic charge density at the location of those nuclei. The electronic density tends to be sharply peaked at these nuclei and then falls off exponentially away from the nuclei. To capture this structure, one of the most common types of basis sets used for electronic structure simulations are nuclei centered Gaussian orbitals. Often, one uses these ``primitive'' orbitals to construct even more compact numerical orbitals; e.g., by using the primitive orbitals to discretize the diagonalization of an approximate single-particle description of the physics. When the single-particle description is a mean-field model of molecule (usually Hartree-Fock) these numerically optimized orbitals are referred to as ``molecular orbitals''.

Aside from Gaussian orbitals, the other most commonly used class of basis functions are plane waves. Plane waves are the type of basis function that we will focus on in this paper. Plane wave basis functions are eigenstates of the linear momentum operator (and thus, also the operator $T$), expressed as
\begin{equation}
\phi_p\left(r\right) = \sqrt{\frac{1}{\Omega}} e^{-\mi \, k_{p} \cdot r} \, ,
\end{equation}
where $r$ is a position vector in real space, $\Omega$ is the computational cell volume and $k_p$ is a reciprocal lattice vector in three dimensions. In this paper we will focus on methods to perform simulations with plane waves defined over a cubic reciprocal lattice so that
\begin{equation}
\label{eq:G}
k_p = \frac{2 \pi p}{\Omega^{1/3}} \qquad \qquad
p \in G \qquad \qquad G = \left[-\frac{N^{1/3}-1}{2},\frac{N^{1/3}-1}{2}\right]^3 \subset \mathbb{Z}^3 \, .
\end{equation}
We note that while this reciprocal lattice is appropriate for simulating non-periodic systems or those with cubic periodicity, to simulate systems with other crystal symmetries one may need to generalize the above equation to the case of non-orthogonal Bravais vectors (we leave this to future work) \cite{MartinES2004}.

Plane waves are a proper basis (unlike a grid) but, like a grid, provide an unbiased discretization of space that is especially well suited to representing dynamics. Due to their regularity, plane waves also give very systematic convergence to the continuum limit, allowing one to accurately estimate properties in the continuum by extrapolating from a series of calculations with progressively more plane waves (the extrapolation results in higher accuracy predictions than even the largest simulation used in the extrapolation). But from a quantum algorithms perspective, perhaps the most useful property of plane waves is that the integrals defining their Galerkin representation have a convenient closed form (arising from the Fourier transform of the Coulomb kernel) \cite{MartinES2004}:
\begin{equation}
\label{eq:plane_wave_integrals}
    T_{pq}^{(m)} = \delta_{pq} \frac{\left\|k_p \right\|^2}{2 m} \qquad \qquad
    U_{pq} = \frac{4\pi}{\Omega}\sum_{\ell=1}^{L}\zeta_\ell \frac{e^{i k_{q-p}\cdot R_\ell}}{\norm{k_{p-q}}^2}
    \qquad \qquad
    V_{pqrs}^{(\alpha,\beta)} = \delta_{p-s,r-q} \frac{4\pi\,\alpha \beta}{\Omega \left\| k_{\nu}\right\|^2} \, ,
\end{equation}
where $\nu = p-s = r-q \neq 0$ in the last expression. As we will see, this closed form enables some especially efficient quantum algorithms.

The structure of the plane wave basis can be used to give quantum algorithms with lower asymptotic gate complexity compared to algorithms with Gaussian orbitals. One reason for this is that most methods for simulating second quantized Hamiltonians tend to have some cost that scales with the number of terms in the Hamiltonian (or specifically, the number of terms in the two-body operator since that is the difficult part). A second quantized plane wave Hamiltonian has ${\cal O}(N^3)$ terms in the two-body operator (less than the ${\cal O}(N^4)$ terms of a molecular orbital representation by a factor of $N$ due to conservation of linear momentum) but only ${\cal O}(N^2)$ terms in the dual basis obtained by Fourier transforming the basis. This argument was first made in the context of second quantized quantum simulations in \cite{BabbushLow}. Today, the most efficient second quantized algorithm for simulating plane wave electronic structure Hamiltonians with $N$ plane waves and $\eta$ electrons has Toffoli complexity $(N^{5/3} / \eta^{2/3} + N^{4/3} \eta^{2/3})N^{o(1)}$ with space complexity ${\cal O}(N \log N)$ \cite{Su2020_b}, where this complexity assumes that particle density is held fixed as the system size grows. We compare the scaling of all prior quantum algorithms for plane wave electronic structure in more detail in \tab{plane_wave_comparison}. By contrast, the best second quantized algorithms for simulating arbitrary basis (e.g., molecular orbital) electronic structure Hamiltonians with $N$ orbitals have Toffoli complexity that is roughly $\widetilde{\cal O}(N^3)$ with space complexity ${\cal O}(N \log N)$ \cite{Lee2020}\footnote{Here and throughout the paper we use $\widetilde{\cal O}(\cdot)$ to indicate an asymptotic upper bound suppressing polylogarithmic factors and $o(1)$ to represent a positive number that approaches zero as some parameter grows.}. Thus, the best plane wave algorithms have better asymptotic scaling than the best arbitrary basis algorithms. However, while the basis set discretization error of both plane waves and Gaussians is asymptotically refined as $\epsilon = {\cal O}(1/N)$, in practice one often requires far fewer Gaussian orbitals than plane waves in order to achieve ``chemically accurate'' electronic structure simulations\footnote{Conventionally, ``chemical accuracy'' is defined as having precision in the energy that is better than 1 kcal/mol, or equivalently, 0.0016 Hartree. An error in the energy on the order of this quantity corresponds to an error in the predicted rates of chemical reactions by roughly an order of magnitude at room temperature \cite{Helgaker2000}.}. For some systems, one could easily require thousands of times more plane waves to reach chemical accuracy.

Plane waves are especially well suited to treat regular solid-state systems (like crystals) since the basis is naturally periodic. Furthermore, there are some special systems, such as the uniform electron gas, for which plane waves are essentially the most compact basis. Thus, several papers \cite{BGBWMPFN18,Kivlichan2020improvedfault,MCS21} suggest that using a plane wave basis in second quantization one can encounter interesting and challenging instances of electronic structure that require fewer resources to solve than any second quantized Gaussian orbital representation. However, for most chemical calculations on a quantum computer, a second quantized plane wave representation is not realistic to use and this is why the vast majority of work on quantum computing quantum chemistry focuses on using Gaussian orbitals. Even though second quantized plane wave representations have asymptotically lower gate complexity than second quantized Gaussian basis algorithms, the space complexity required for accurate calculations using second quantized plane waves is daunting. For example, if a calculation requires one-hundred Gaussian orbitals but one-hundred thousand plane waves that means in second quantization the former would require hundreds of logical system qubits (reasonable) whereas the latter would require hundreds of thousands of logical system qubits (an extravagant cost). Thus, for chemistry calculations in second quantization, Gaussian orbitals are more practical than plane waves.

However, in first quantization the story is rather different. Because the space complexity of first quantized representations is only ${\cal O}(\eta \log N)$ one can easily simulate systems with a very large number of plane waves without the number of qubits becoming too extravagant. As we will later discuss, the gate complexity of first quantized plane wave algorithms can also have sublinear scaling in $N$; for example, the algorithm of \cite{BabbushContinuum_b} has gate complexity $\widetilde{\cal O}(\eta^{8/3} N^{1/3})$. Sublinear scaling in $N$ would be impossible in second quantization as each of $N$ qubits must be acted upon at least once. In principle one can also perform simulations of Gaussian orbitals in first quantization but due to the lack of structure in the integrals this would likely give a relatively inefficient gate complexity. The work of \cite{BabbushSparse2} does something that is very close to this; in that work, the authors simulate a fixed particle number manifold of the second quantized molecular orbital Hamiltonian (this is known in chemistry as the ``configuration interaction'' representation). Their approach has space complexity ${\cal O}(\eta \log N)$ but it is technically a second quantized approach in a fixed particle manifold rather than a first quantized approach because the symmetries are still encoded by how the Hamiltonian operator acts on the state rather than the algorithm requiring anti-symmetrized initial states. However, the gate complexity of that approach is $\widetilde{\cal O}(\eta^2 N^3)$ which is still relatively high in $N$, meaning that one cannot refine the basis to extremely large sizes the way that one can with first quantized plane wave approaches.

\begin{table*}[t]
\begin{tabular}{|c|c|c|c|c|}
\hline
Year
& Reference
& Primary innovation
& Logical qubits
& Toffoli/T complexity\\
\hline\hline
2017
& Babbush \emph{et al}.~\cite{BabbushLow}
& Using plane waves with Trotter
& $N + {\cal O}(\log N)$
& $\widetilde{\cal O}(\eta^2 N^{17/6} \sqrt{1 + \eta \Omega^{1/3}/N^{1/3}} /(\Omega^{5/6} \epsilon^{3/2}))$\\
2017
& Babbush \emph{et al}.~\cite{BabbushLow}
& Using plane waves with LCU
& $N + {\cal O}(\log N)$
& $\widetilde{\cal O}((N^4/\Omega^{1/3} + N^{11/3}/\Omega^{2/3} )/\epsilon)$\\
2018
& Babbush \emph{et al}.~\cite{BGBWMPFN18}
& Linear scaling quantum walks
& $N + {\cal O}(\log N)$
& $\widetilde{\cal O}((N^{10/3}/\Omega^{1/3} + N^{8/3}/\Omega^{2/3})/\epsilon)$\\
2018
& Low \emph{et al}.~\cite{Low2018}
& Using the interaction pic.
& ${\cal O}(N \log N)$
& $\widetilde{\cal O}(N^{8/3} / (\Omega^{2/3} \epsilon))$\\
2018
& Babbush \emph{et al}.~\cite{BabbushContinuum_b}
& First quantized qubitization
& ${\cal O}(\eta \log N)$
& $\widetilde{\cal O}((\eta^{3} N^{1/3} / \Omega^{1/3} + \eta^{2} N^{2/3} / \Omega^{2/3}) / \epsilon)$\\
2018
& Babbush \emph{et al}.~\cite{BabbushContinuum_b}
& First quantized interaction pic.
& ${\cal O}(\eta \log N)$
& $\widetilde{\cal O}(\eta^{3} N^{1/3} / ( \Omega^{1/3} \epsilon) )$\\
2019
& Kivlichan \emph{et al}.~\cite{Kivlichan2020improvedfault}
& Better Trotter steps
& $N + {\cal O}(\log N)$
& $\widetilde{\cal O}(N^{3} / (\Omega^{2/3} \epsilon^{3/2}))$\\
2019
& Childs \emph{et al}.~\cite{Childs2021}
& Tighter Trotter bounds
& ${\cal O}(N \log N)$
& $N^{7/3 + o(1)} / (\Omega^{1/3} \epsilon^{1 + o(1)})$\\
2020
& Su \emph{et al}.~\cite{Su2020_b}
& Yet tighter Trotter bounds
& ${\cal O}(N \log N)$
& $(\eta / \Omega^{1/3}+ N^{1/3} / \Omega^{2/3})N^{4/3 + o(1)} /\epsilon^{1 + o(1)}$\\
\hline
\end{tabular}
\caption{\label{tab:plane_wave_comparison} Best quantum algorithms for phase estimating chemistry in a plane wave basis. $N$ is number of basis functions, $\eta < N$ is number of electrons, $\Omega$ is the computational cell volume, and $\epsilon$ is target precision. In some of these papers it is assumed that $\Omega \propto N$ or that $\Omega \propto \eta$, but here we report the complexities without any such assumptions. Unlike other entries here, the scaling of the Kivlichan \emph{et al.}~\cite{Kivlichan2020improvedfault} algorithm is empirically observed for specific systems (with some systems scaling somewhat better and some systems scaling somewhat worse) as opposed to rigorous upper-bounds. The ${\cal O}(N \log N)$ space complexity of \cite{Low2018,Childs2021,Su2020_b} results from a Fourier transform based method of computing the potential operator discussed in \cite{Low2018}. If those papers were instead to use the Trotter steps introduced in \cite{Kivlichan2017} or \cite{Kivlichan2020improvedfault}, the algorithms would have $N + {\cal O}(\log N)$ space complexity but gate complexity that is worse by approximately a factor of $N$. Algorithms with ${\cal O}(N)$ space complexity are likely impractical for any real problems in chemistry when using plane waves due to the large number of plane waves required to reach chemical accuracy; thus, ${\cal O}(N \log N)$ space complexity is even less viable. We note that \app{real_space} introduces new algorithms defined in real space that match the scaling of those from Babbush \emph{et al.}~\cite{BabbushContinuum_b}.}
\end{table*}

We have thus argued that plane waves are especially well suited for use within first quantized quantum simulations whereas Gaussian orbitals are especially well suited for use within second quantized quantum simulations. The relative merits of these two representations is a fairly subtle issue because they have different properties. One requires more plane waves than Gaussian orbitals to converge a calculation but the first quantized plane wave algorithm asymptotically scales better for essentially all values of $\eta$ and $N$.
That is because $\eta < N$, so the $\widetilde{\cal O}(\eta^{8/3}N^{1/3})$ Toffoli complexity of the best first quantized plane algorithm \cite{BabbushContinuum_b} reduces to the roughly $\widetilde{\cal O}(N^3)$ Toffoli complexity of the best second quantized arbitrary basis algorithm \cite{Lee2020} in the worst case that $\eta = \Theta(N)$. There is also a substantial difference in how the algorithms approach the continuum limit. For example, the Toffoli complexity to achieve basis error $\epsilon$ scales as $\widetilde{\cal O}(1/\epsilon^3)$ with $\widetilde{\cal O}(1/\epsilon)$ space complexity for the arbitrary basis algorithm but the Toffoli complexity is $\widetilde{\cal O}(1/\epsilon^{1/3})$ with space complexity ${\cal O}(\log (1/\epsilon))$ for the plane wave algorithm. This enormous difference -- a polynomial scaling difference by a power of nine -- suggests a strong preference for the first quantized plane wave algorithm when targeting high precision\footnote{The first quantized plane wave algorithm of \cite{BabbushContinuum_b} has a Clifford complexity matching the Toffoli complexity whereas the second quantized arbitrary basis algorithm of \cite{Lee2020} has Clifford complexity that is worse than the Toffoli complexity by a factor of $\widetilde{\cal O}(N)$. Hence, in terms of overall gate complexity the scaling to the continuum limit is actually $\widetilde{\cal O}(1/\epsilon^4)$ versus $\widetilde{\cal O}(1/\epsilon^{1/3})$ -- a difference by a power of twelve!}. Furthermore, the plane wave representation is the clear preference when performing simulations beyond the Born-Oppenheimer approximation. The reason for this is because the advantage of Gaussian basis sets lies in their ability to resolve cusps in the wavefunction centered on the nuclei. But outside of the Born-Oppenheimer approximation the nuclei are treated explicitly and thus one cannot center Gaussians on the nuclei because the nuclei do not have a fixed position that is known in advance. Thus, for simulations of nuclear quantum dynamics plane waves are going to be the better choice.

By plugging the plane wave integrals given in \eq{plane_wave_integrals} into the definitions of the Born-Oppenheimer and non-Born-Oppenheimer first quantized Galerkin discretizations discussed in \app{background}, we obtain
\begin{align}
\label{eq:first_quant_ham}
    H_{\rm BO} & = T + U + V + \frac{1}{2}\sum_{\ell \neq \kappa=1}^{L} \frac{\zeta_\ell \zeta_\kappa}{\left\|R_\ell - R_\kappa\right\|} \\
        \label{eq:non_bo_ham}
    H_{\rm non-BO} & = T + T_{\rm nuc} + U_{\rm non-BO} + V + V_{\rm nuc}\\
    T &=  \sum_{i=1}^{\eta}  \sum_{p\in G} \frac{\left \| k_p\right\|^2}{2} \ket{p}\!\bra{p}_{i} 
    \qquad \qquad
    T_{\rm nuc}  =\sum_{\ell=\eta+1}^{L + \eta}  \sum_{p\in G} \frac{\left \| k_p\right\|^2}{2 \, m_\ell} \ket{p}\!\bra{p}_\ell\\
	U & =-\frac{4\pi}{\Omega}\sum_{\ell=1}^{L}\sum_{i=1}^{\eta}\sum_{\substack{p,q\in G\\ p\neq q}}\bigg(\zeta_\ell\frac{e^{ik_{q-p}\cdot R_\ell}}{\norm{k_{p-q}}^2}\bigg)\ket{p}\!\bra{q}_i,\\
		U_{\rm non-BO} &=-\frac{4 \pi}{\Omega} \sum_{i=1}^{\eta}\sum_{\ell=\eta+1}^{L + \eta}\sum_{p,q\in G} \sum_{\substack{\nu\in G_0\\(p+\nu)\in G\\(q-\nu)\in G}}\frac{1}{\left\| k_{\nu}\right\|^2} \ket{p + \nu}\!\bra{p}_i \ket{q-\nu}\!\bra{q}_\ell\\
	V &=\frac{2 \pi}{\Omega} \sum_{i\neq j=1}^{\eta}\sum_{p,q\in G} \sum_{\substack{\nu\in G_0\\(p+\nu)\in G\\(q-\nu)\in G}}\frac{1}{\left\| k_{\nu}\right\|^2} \ket{p + \nu}\!\bra{p}_i \ket{q-\nu}\!\bra{q}_j \\
	V_{\rm nuc} &= \frac{2 \pi}{\Omega} \sum_{\ell\neq\kappa=\eta+1}^{L+\eta}\sum_{p,q\in G} \sum_{\substack{\nu\in G_0\\(p+\nu)\in G\\(q-\nu)\in G}}\frac{\zeta_\ell \zeta_\kappa}{\left\| k_{\nu}\right\|^2} \ket{p + \nu}\!\bra{p}_\ell \ket{q-\nu}\!\bra{q}_\kappa \, ,
\end{align}
where $G_0 = [-N^{1/3}, N^{1/3}]^3 \subset \mathbb{Z}^3 \backslash\{(0,0,0)\}$ is the set of the difference of frequencies from $G$ excluding the singular zero mode, and $\ket{p}\!\bra{q}_j$ is a shorthand for $I_1\otimes\cdots\otimes\ket{p}\!\bra{q}_j\otimes\cdots\otimes I_\eta$ where $I_i$ is the identity operator on the $\log N$ bit register associated with electron $i$. We will store our wavefunction by having a computational basis that encodes configurations of the electrons in $N$ basis functions such that a configuration is specified as $\ket{\phi_1 \phi_2 \cdots \phi_{\eta}}$ where each $\phi$ encodes the index of an occupied basis function.
	
	The representation of the Hamiltonian we show discretized above is in ``first quantization''. The distinction between ``first'' and ``second'' quantization is fundamentally about how one deals with the symmetries of electrons and nuclei. Electrons (and some nuclei) are fermions. Some nuclei are bosons. Whereas identical fermions must be antisymmetric with respect to particle label exchange, identical bosons must be symmetric with respect to particle label exchange. For example, if our state corresponds to a valid first quantized wavefunction for identical particles then the system must have the property that
	\begin{equation}
	\ket{\psi} = \sum_{\phi_p \in \{\phi\}} a_{\phi_1 \cdots \phi_\eta} \ket{\phi_1 \cdots \phi_i \cdots \phi_j \cdots \phi_\eta} = \left(-1\right)^\pi \sum_{\phi_p \in \{\phi\}} a_{\phi_1 \cdots \phi_\eta} \ket{\phi_1 \cdots \phi_j \cdots \phi_i \cdots \phi_\eta}
	\end{equation}
	where $\pi=1$ if fermions and $\pi=0$ if bosons and $a_{\phi_1 \cdots \phi_\eta}$ are the amplitudes of the associated computational basis state. No special symmetries need to exist between distinguishable particles. The work of \cite{Berry2018} shows that one can symmetrize (or antisymmetrize) any set of $\eta$ registers which each index $N$ basis functions with a quantum algorithm that has gate complexity ${\cal O}(\eta \log \eta \log N)$. Once symmetrized (or antisymmetrized) any simulation under the Hamiltonian will maintain that property as a consequence of \eq{non-bo} commuting with the particle permutation operator between any identical particles. Since gate complexity of ${\cal O}(\eta \log \eta \log N)$ is significantly less than the cost of time-evolving these systems, it can be regarded as a negligible additive cost to the overall simulation. 
	
	Finally, one might notice that the spin degree of freedom does not appear in the Hamiltonians written above, nor does it appear in \eq{non-bo}. A spin is associated with each of the $\eta$ registers but because there are not magnetic field terms in \eq{non-bo} that would explicitly interact with these spins, it is not necessary to include the spin label and so we avoid doing so for simplicity. However, because particles of different spins are distinguishable, this should be accounted for during the antisymmetrization of the initial wavefunction.
	
	\subsection{Overview of the Hamiltonian simulation frameworks deployed in this work}
	\label{sec:alg_overview}
	
	In the prior section we described the full non-Born-Oppenheimer Hamiltonian in a plane wave basis and suggested that plane waves would be an ideal representation in which to perform non-adiabatic dynamics. However, the algorithms we analyze in detail in this paper focus on the Born-Oppenheimer Hamiltonian. Still, we emphasize that the non-Born-Oppenheimer application helps to motivate our work since the algorithms used to simulate the Born-Oppenheimer Hamiltonian apply to the non-Born-Oppenheimer Hamiltonian with very minor modifications; the only additional terms in the non-Born-Oppenheimer Hamiltonian are $U_{\rm non-BO}$ (identical to $V$ up to different charges) and $T_{\rm nuc}$ (identical to $T$ up to different masses).
	
	Because we will be simulating the Born-Oppenheimer Hamiltonian, which is typically of interest in electronic structure, our focus will be preparing eigenstates rather than effecting dynamics. Specifically, we will use the phase estimation algorithm to sample in the eigenbasis of the Hamiltonian \cite{KitaevARX1995,Abrams1999}. Among several possible uses of this technique is the refinement of molecular ground states from initial guess states with non-vanishing support on the exact ground state \cite{,Aspuru-Guzik2005}. The requirement for such algorithms is that we are able to synthesize a unitary circuit encoding eigenvalues of the Born-Oppenheimer Hamiltonian as a simple function of their eigenphases. In this paper, we will study two different algorithms for this task based on the qubitization \cite{Low2016} and interaction picture simulation \cite{Low2018} algorithms of Low \emph{et al.}, with subroutines adapted for simulating first quantized chemistry by Babbush \emph{et al.}~\cite{BabbushContinuum_b}. We will now outline these approaches at a high level and establish conventions useful for later (more technical) sections.
	
	Both of the algorithms we will explore are based (at least in part) on block encoding the Hamiltonian using qubitization operators \cite{Low2016}. This method begins from the observation that we express our target Hamiltonian as a linear combination of unitaries (often referred to here as ``LCU'') \cite{Childs2012}, i.e.
	\begin{equation}
	    H=\sum_{\ell}\alpha_{\ell} H_\ell,
	\end{equation}
	where $H_\ell$ are unitary operators and $\alpha_\ell$ are positive coefficients (by convention we choose to incorporate any phases into the $H_\ell$). Using this notation we then define the quantum operators $\PREP_H$ and $\SEL_H$ as
	\begin{equation}
	\label{eq:oracles}
	    \PREP_{H}\ket{0}=\sum_{\ell}\sqrt{\frac{\alpha_\ell}{\lambda}}\ket{\ell},\qquad
	    \SEL_{H}\ket{\ell}=\ket{\ell}\otimes H_\ell,\qquad
	    \lambda=\sum_{\ell} \alpha_{\ell} \, 
	\end{equation}
	and we notice that $\bra{0}\PREP_{H}^\dagger\cdot\SEL_{H}\cdot\PREP_{H}\ket{0}=H/\lambda$ \cite{Low2018}. Because of this property, if we can realize $\PREP_H$ and $\SEL_H$ as quantum circuits then by calling them in a certain sequence it is possible to ``block encode'' the Hamiltonian in a particular subspace of a qubitization operator that we define as 
	\begin{equation}
	    Q=\left(2\ket{0}\!\bra{0}-I\right)\cdot\PREP_{H}^\dagger\cdot\SEL_{H}\cdot\PREP_{H} \, .
	\end{equation}
	What this means is that $Q$ acts as a quantum walk \cite{Szegedy2004} with eigenvalues $e^{\pm i\arccos(E_k/\lambda)}$~\cite{Low2016}, related to the energy $E_k$ of the target Hamiltonian $H$. The work of \cite{Berry2018,Poulin2017} showed that we can use these quantum walks to project to an eigenstate of the Hamiltonian and estimate its eigenvalue by performing phase estimation on $Q$ and applying the cosine function on the resulting phase. To estimate with accuracy $\epsilon$, we need $\mathcal{O}(\lambda/\epsilon)$ uses of controlled-$Q$ operations. Alternatively, the resource requirements of performing phase estimation can be reduced by using the multiplexed qubitization operation \cite{BGBWMPFN18}:
	\begin{equation}
	    \ket{0}\!\bra{0}\otimes Q^\dagger+\ket{1}\!\bra{1}\otimes Q \, ,
	\end{equation}
	instead of the controlled-$Q$ operation. 
	This amplifies the difference between the eigenvalues by a factor of 2, reducing the complexity by a factor of 2.
	Provided $\SEL_{H}$ is self-inverse, this control can be achieved simply by making the reflection controlled by a qubit from the control register for phase estimation, as per Figure 2 of \cite{Lee2020}.
	This control only adds a cost of one Toffoli per step.
	
	We note that in addition to being useful for phase estimation, one can also realize time-evolution under the Born-Oppenheimer Hamiltonian for duration $t$ by invoking these quantum walks $Q$ within the quantum signal processing framework of Low \emph{et al.}~\cite{Low2017} a number of times scaling as
	\begin{equation}
	    \label{eq:signals}
	    {\cal O}\left(\lambda t + \frac{\log(1/\epsilon)}{\log \log(1/\epsilon)}\right) \, .
	\end{equation}
	Thus, while our explicit focus will be on phase estimation of the Born-Oppenheimer Hamiltonian, the same algorithms can be used to perform time-evolution under the non-Born-Oppenheimer Hamiltonian, with trivial modifications and marginal additional cost. The similarity of the algorithms enabling these disparate applications results from our choice to work in a first quantized representation with basis functions that are independent of nuclear coordinates; if we were instead to work in second quantization or use Gaussian orbitals, there would be a substantial gap between the algorithms required for these two applications. 
	Getting back to the implementation of the qubitization algorithm, in order to express the Hamiltonian of \eq{first_quant_ham} as a linear combination of unitaries we rewrite the $T$, $U$ and $V$ operators as follows:
	\begin{equation}
	\begin{aligned}
	T&=
	\frac {2\pi^2}{\Omega^{2/3}} \sum_{j=1}^{\eta}\sum_{p\in G}\sum_{w\in\{x,y,z\}} \sum_{r=0}^{n_p-2}\sum_{s=0}^{n_p-2} 2^{r+s} p_{w,r} p_{w,s} \ket{p}\!\bra{p}_j, \\
	U&=\sum_{\nu\in G_0}\sum_{\ell=1}^{L}\frac{2\pi\zeta_\ell}{\Omega\norm{k_\nu}^2}\sum_{j=1}^{\eta}\sum_{b\in\{0,1\}}\left(-e^{-\mi k_\nu\cdot R_\ell}\sum_{q\in G}(-1)^{b[(q-\nu)\notin G]}\ket{q-\nu}\!\bra{q}_j\right),\label{eq:opdef}\\
	V&=\sum_{\nu\in G_0}\frac{\pi}{\Omega\norm{k_\nu}^2}\sum_{i\neq j=1}^{\eta}\sum_{b\in\{0,1\}}\left(\sum_{p,q\in G}(-1)^{b([(p+\nu)\notin G]\lor[(q-\nu)\notin G])}\ket{p+\nu}\bra{p}_i\cdot\ket{q-\nu}\!\bra{q}_j\right)\, ,\\
	\end{aligned}
	\end{equation}
	where
	\begin{equation}
	\label{eq:n_p}
	    n_p=\left\lceil\log \left(N^{1/3}+1\right)\right\rceil \,
	\end{equation}
	represents the number of qubits required to store a signed binary representation of one component of the momentum of a single electron.
	We can see that the potential operators $U$ and $V$ are now expressed as linear combinations of unitaries. To express the $T$ operator as a linear combination of unitaries we note that the quantity $p_{w,r}$ is bit $r$ of component $w$ of $p$, and the sum of these products of bits over $r$ and $s$ gives $\|p\|^2$; we apply the identity
	\begin{equation}
	    p_{w,r} p_{w,s} = \frac{1- (-1)^{p_{w,r} p_{w,s}}}{2},
	\end{equation}
	which gives
	\begin{equation}
	T=\frac {\pi^2}{\Omega^{2/3}} \sum_{j=1}^{\eta}\sum_{w\in\{x,y,z\}} \sum_{r=0}^{n_p-2}\sum_{s=0}^{n_p-2} 2^{r+s} \sum_{b\in\{0,1\}}\left(\sum_{p\in G} (-1)^{b(p_{w,r} p_{w,s}\oplus 1)} \ket{p}\!\bra{p}_j\right).\label{eq:Tlcu}
	\end{equation}
	
	With the Hamiltonian components expressed as a linear combination of unitaries we can now discuss the normalizing parameter $\lambda$ appearing in \eq{oracles}, which will have significant ramifications for the complexity of both algorithms studied in this work. The parameter $\lambda$ is an induced one-norm of the Hamiltonian expressed as a linear combination of unitaries. For the entire Hamiltonian we have $\lambda = \lambda_T + \lambda_U + \lambda_V$ where
		\begin{align}
	\label{eq:lambdas}
	    \lambda_T &= \frac{6\eta\pi^2}{\Omega^{2/3}} \left( 2^{n_p-1}\!-1 \right)^2 = {\cal O}\left(\frac{\eta N^{2/3}}{\Omega^{2/3}}\right),\qquad 
	    \qquad
	    \lambda_U = \frac{\eta \sum_\ell \zeta_\ell}{\pi \Omega^{1/3}}  \lambda_\nu = {\cal O}\left(\frac{\eta^2 N^{1/3}}{\Omega^{1/3}}\right)\\
	    \lambda_V &= \frac {\eta(\eta-1)}{2\pi \Omega^{1/3}} \lambda_\nu = {\cal O}\left(\frac{\eta^2 N^{1/3}}{\Omega^{1/3}}\right),\qquad\lambda_\nu = \sum_{\nu\in G_0} \frac 1{\norm{\nu}^2} \leq \int_0^{2\pi}\!\! \!\! \int_{0}^\pi\! \! \!\int_0^{N^{1/3}} \!\!\!\!\!\! {\rm d}r \, {\rm d}\phi \,{\rm d}\theta \, \left(\frac{1}{r^2}\right)r^2 \sin \theta = {\cal O}\left(N^{1/3}\right) \, , \nonumber
	\end{align}
     where in the asymptotic scaling for $\lambda_U$ we have used the fact that when including all electrons explicitly in a charge neutral simulation, $\sum_\ell \zeta_\ell = \eta$. Further, we note that it is often appropriate to take $\Omega \propto \eta$ (e.g., whenever growing condensed phase systems towards their thermodynamic limit), which simplifies the asymptotic complexities to $\lambda_T = {\cal O}(\eta^{1/3} N^{2/3})$ and $\lambda_U \propto \lambda_V = {\cal O}(\eta^{5/3}N^{1/3})$.
     
    The main challenge of realizing an algorithm like qubitization with low complexity is coming up with efficient circuits for the operators defined in \eq{oracles}. Describing such algorithms asymptotically for first quantized chemistry is one of the main contributions of the work of Babbush \emph{et al.}~\cite{BabbushContinuum_b} and this work will loosely follow the methods outlined therein. However, the fact that this paper is an order of magnitude longer than the former is a testament to the fact that bounding asymptotic complexity is significantly less work than coming up with, improving, and compiling, a practical implementation of that algorithm with low constant factors in a fault-tolerant cost model. Thus, the main technical contribution of \sec{qubitization} is to describe a detailed algorithm for realizing the qubitization operators corresponding to the first quantized plane wave electronic structure Hamiltonian. We introduce various new techniques, including a strategy for controlled swapping the momentum registers (\sec{block_encode_t} and \sec{SelectUV}) and a proposal for performing the kinetic term when the state preparation for the potential fails (\app{selTUV}), which lead to a significant reduction in the algorithmic complexity. Ultimately, we show that the operator $Q$ can be realized with a gate complexity of $\widetilde{\cal O}(\eta)$ and that one can implement phase estimation to estimate eigenvalues of \eq{first_quant_ham} to within precision $\epsilon$ with a leading order Toffoli complexity given by \thm{qubitization}. 
 
    By combining the $\widetilde{\cal O}(\eta)$ gate complexity of $Q$ with the ${\cal O}(\lambda / \epsilon)$ times that $Q$ must be invoked within phase estimation, we see that the asymptotic complexity of the qubitized algorithm for phase estimation is
    \begin{equation}
    \label{eq:asymptotic_qubitization}
        \widetilde{\cal O} \left(\frac{\left(\lambda_T + \lambda_U + \lambda_V\right) \eta}{\epsilon}\right) 
        = \widetilde{\cal O}\left(\frac{\eta^2 N^{2/3}}{\epsilon \, \Omega^{2/3}} + \frac{ \eta^{3}N^{1/3}}{\epsilon \, \Omega^{1/3}}\right)
        = \widetilde{\cal O}\left(\frac{\eta^{4/3} N^{2/3} + \eta^{8/3}N^{1/3}}{\epsilon}\right)\,
    \end{equation}
	as reported for the qubitized algorithm complexity in \cite{BabbushContinuum_b}. While the scaling of this algorithm in terms of $\eta$ is worse than the scaling in terms of $N$, for practical applications we may be more concerned about the scaling with $N$. The reason is because in order to accurately solve the electronic structure problem in a plane wave basis it will be necessary to use a very large number of plane waves. For example, some enticing applications may require roughly $\eta = 100$ and $N = 10^6$. 	While such simulations would be impossible within second quantization because that would (in this case) require at least $N = 10^6$ logical qubits, within first quantization the space complexity required for the system register is only $\eta \log N \approx 2 \times 10^3$ qubits.
	Here and throughout the manuscript we use the convention that $\log$ indicates a logarithm base 2.
	
	The vast majority of this space complexity arises from representing the simulated system, which has basis states of the form
	\begin{equation}
	\ket{p_1}\cdots\ket{p_j}\cdots\ket{p_\eta}.
	\end{equation}
	Representing coordinates of the orbitals explicitly, we have
	\begin{equation}
	\ket{p_{1x},p_{1y},p_{1z}}\cdots\ket{p_{jx},p_{jy},p_{jz}}\cdots\ket{p_{\eta x},p_{\eta y},p_{\eta z}}.
	\end{equation}
	As in \eq{G}, these registers hold values in the range $[-(N^{1/3}-1)/2,(N^{1/3}-1)/2]$. Thus, we will use $n_p$ qubits with $N^{1/3}=2^{n_p}-1$ to encode each register. Furthermore, for reasons of algorithmic efficiency described later, we will use signed binary numbers to represent each register of length $n_p$, so
	\begin{equation}
	\ket{q}=\ket{q_{\text{sign}},q_{n_p-2},\ldots,q_0}.
	\end{equation}
	The total number of qubits in the simulated system is thus
	\begin{equation}
	\label{eq:n_qubits}
	n_s=3 \, \eta \, n_p.
	\end{equation}
	
	So while the space complexity is likely not a problem for performing simulations at large $N$, the gate complexity scaling as $\widetilde{\cal O}(N^{2/3})$ is somewhat concerning and we might wonder if it is possible to reduce this scaling. This is the motivation for the second algorithm we will explore (also introduced in \cite{BabbushContinuum_b}), which provides an algorithm that scales linearly in $\lambda_U + \lambda_V$ but avoids scaling polynomially in $\lambda_T$ (which is where the $N^{2/3}$ scaling was entering previously). The approach is based on the interaction picture simulation technique introduced by Low \emph{et al.}~\cite{Low2018} which begins by supposing that we want to simulate a target Hamiltonian containing two terms $A$ and $B$, where $\norm{A}$ is significantly larger than $\norm{B}$ but time-evolution of $A$ can be fast-forwarded. Then, we can represent the evolution in the interaction picture of the rotating frame of $A$,
\begin{equation}
    e^{-\mi\tau (A+B)}=e^{-\mi\tau A}\mathcal{T}\exp\left(-\mi\int_{0}^{\tau}ds\ e^{\mi sA}Be^{-\mi sA}\right),
\end{equation}
where $\mathcal{T}$ is the time ordering operator. The potential benefit of using this representation is that the larger term $A$ only appears in the exponent of matrix exponentials and its implementation can be efficient even when $\norm{A}$ is large. For the first quantized plane wave Hamiltonian in the regime that $N \gg \eta$, we have the norm relation $\norm{T}\geq\norm{U+V}$ so we choose $A=T$ and $B=U+V$ \cite{BabbushContinuum_b}.

To perform quantum simulation in the interaction picture, we truncate and discretize the Dyson series of the time-ordered exponential, obtaining
\begin{align}
e^{-\mi(A+B)\tau} &= e^{-\mi\tau A} \mathcal{T} \exp \left( -\mi \int_0^\tau ds \, e^{\mi sA} B e^{-\mi sA}\right) \\
&= \sum_{k=0}^\infty \frac{(-\mi)^k}{k!} \int_{0}^\tau d\tau_1 \int_{0}^{\tau} d\tau_2 \cdots \int_{0}^\tau d\tau_k \,  e^{-\mi(\tau-\tau'_k)A} B e^{-\mi(\tau'_k-\tau'_{k-1})A} B \ldots B e^{-\mi(\tau'_2-\tau'_1)A} B e^{-\mi\tau_1 A} \nn
&\approx \sum_{k=0}^K \frac{(-\mi\tau)^k}{M^k k!} \sum_{m_1=0}^{M-1} \sum_{m_2=0}^{M-1} \cdots \sum_{m_k=0}^{M-1} e^{-\mi\tau(M-1/2-m'_k)A/M} B e^{-\mi\tau(m'_k-m'_{k-1})A/M} B \ldots \nn& \quad \ldots B e^{-\mi\tau(m'_2-m'_1)A/M} B e^{-\mi\tau(m'_1+1/2) A/M} \nonumber
\end{align}
where $0\leq\tau'_1\leq\ldots\leq\tau'_k\leq\tau$ are sorted times from $\tau_1,\ldots,\tau_k$, and $0\leq m'_1\leq\ldots\leq m'_k\leq M-1$ are sorted integers from $m_1,\ldots,m_k$. The resulting expression is only an approximation of the ideal evolution, but the truncation and discretization error can be made arbitrarily small by choosing $K$ and $M$ sufficiently large. These parameters give a contribution to the complexity which scales as a polynomial in $\log(1/\epsilon)$, which is omitted when giving complexity using $\widetilde O$.

As discussed in \sec{interaction_picture}, the bottleneck requirements for realizing this Dyson series are that one implements $e^{-\mi T}$ and also block encodes $U + V$ a number of times scaling as $(\lambda_U + \lambda_V) / \epsilon$. For the most part, one can reuse parts of the circuits in \sec{qubitization} to perform the block encoding of $U + V$ with complexity $\widetilde{\cal O}(\eta)$. Furthermore, it is relatively straightforward to provide circuits that fast-forward the operator $e^{-\mi T}$ with gate complexity $\widetilde{\cal O}(\eta)$. Thus, this suggests that the overall asymptotic complexity will scale as
	\begin{equation}
	\label{eq:asymptotic_interaction}
	  \widetilde{\cal O}\left(\frac{\left(\lambda_U + \lambda_V\right)\eta}{\epsilon}\right) 
	  =\widetilde{\cal O}\left(\frac{\eta^{3}N^{1/3}}{\epsilon \, \Omega^{1/3}}\right) 
	  =  \widetilde{\cal O}\left(\frac{\eta^{8/3}N^{1/3}}{\epsilon}\right) \, .
	\end{equation}
However, the most challenging part of the interaction picture algorithm is to efficiently tie together these primitives within the Dyson series simulation framework and propagate errors from that simulation through the phase estimation circuits. While past work has discussed quantum simulation of the Dyson series \cite{Kieferova18,Low2018} in asymptotic terms, no prior work has provided concrete circuits for the task and worked out the constant factors in the scaling. One of the primary contributions of \sec{interaction_picture} is to work out those circuits in the context of chemistry, optimize their implementations, and bound the overall Toffoli cost, as reported in \thm{interaction_cost}. In doing so, we make several technical contributions, including a novel qubitization of the pre-amplified Dyson series (\sec{choice}), an efficient strategy to update the kinetic energy (\sec{kinetic_exp}), and a higher-order discretization of the time integrals (\app{timedisc}), which greatly improves over the naive approaches. We note that we also work out the constant factor scalings for the interaction picture simulation in a more general context (not related to chemistry), which we report in \app{interaction_picture}.

	Finally, in \sec{results} our paper concludes with numerical comparisons between these two algorithms assessing the resources required to deploy them to real molecular systems. In addition to comparing the two algorithms, we also compare to prior work analyzing the fault-tolerant viability of second quantized Gaussian orbital approaches. Although the comparison is somewhat nuanced, the first quantized approaches developed here appear to give lower Toffoli complexities than any past work for realistic simulations of both material Hamiltonians and non-periodic molecules.
	We note that throughout this work we focus on estimating the constant factor Toffoli complexity and number of ancilla required. We focus on Toffoli complexity (as opposed to the complexity of other gates) because these are the operations that would be most challenging to implement within most practical error-correcting codes such as the surface code \cite{Fowler2012}. Furthermore, the algorithms described here have the same asymptotic Clifford complexities as the Toffoli complexities that we report. This makes it more likely that the implementation would actually be bottlenecked by Toffoli distillation. This stands in contrast to other recent chemistry algorithms compiled to fault-tolerance with extensive use of  quantum read-only memory (QROM) \cite{BGBWMPFN18,Low2018a}. Approaches relying heavily on QROM such as \cite{BGBWMPFN18,Berry2019b,vonBurg2020Catalysis,Lee2020} tend to have a worse asymptotic Clifford complexity than Toffoli complexity, thus necessitating a careful layout of how the algorithm might be realized using lattice surgery or other surface code fault-tolerant gates, in order to assess whether state distillation is actually the dominant bottleneck.
	
There are a number of new techniques which we introduce which can be used in many other applications.
One is to note that in some cases amplitude amplification can be eliminated from the state preparation procedure for the linear combination of unitaries.
The reason for this is that the Hamiltonian can be composed of a number of parts, where that state preparation is only needed for one part of the Hamiltonian.
When the state preparation is regarded as having ``failed'', we can simply flag that as a part of the state which will control the select for another part of the Hamiltonian.
It is useful to eliminate the amplitude amplification, which would otherwise triple the cost (or more if further amplification steps are needed).

Another new technique is to eliminate the arithmetic for some parts of the Hamiltonian.
If a component of the Hamiltonian gives a multiplication of the state by a function of the basis state, then it can be broken up into a sum of individual bit products.
The bit product only requires a single Toffoli, and the sum is implemented via the linear combination of unitaries.
In our case the component of the Hamiltonian is the kinetic energy, and we break up the kinetic energy into a sum of products of bits of the momenta.
This dramatically reduces the Toffoli complexity.

In the case of the interaction picture we can greatly reduce the complexity needed for dealing with the kinetic energy (though not completely eliminate arithmetic like we can for qubitisation).
The idea is to use a kinetic energy register.
This register can be updated with relatively small complexity and used for the phasing.
Another approach we use for the interaction picture is to use a novel qubitisation of the Dyson series.
This again eliminates the need for an amplitude amplification which would otherwise triple the complexity.
This method works generally for applications where the goal is to measure the eigenvalue, rather than to generate the time evolution under the Hamiltonian.

Our implementation of these algorithms allows us to report the first constant factors for the scaling of any first quantized chemistry algorithm (or any Dyson series based simulation), and also optimizes those constant factors more than a thousandfold. This allows us to compare the viability of these approaches to other methods based on second quantization, and to each other. Ultimately, we find that the qubitization algorithm is usually more practical and that first quantized methods often require much less surface code spacetime volume for simulating millions of plane waves than the best second quantized algorithms require for simulating hundreds of Gaussian orbitals.

	%%%%%%%%%%%%%%%%%%%%%%%%%%%%%%%%%%%%%%%%%%%
%%%%%%%%%%%%%%%%%%%%%%%%%%%%%%%%%%%
	
		\section{The qubitization based algorithm}
	\label{sec:qubitization}

We now give an overview of the qubitization algorithm and its circuit implementation for simulating the first-quantized quantum chemistry. For this purpose, we seek to represent the target Hamiltonian as a linear combination of unitaries $H=\sum_{\ell=1}^{L}\alpha_\ell H_\ell$, where $H_\ell$ are unitary operators and $\alpha_\ell$ are positive coefficients. Recall from \sec{alg_overview} that we can construct a state preparation subroutine $\PREP_H$ and a selection subroutine $\SEL_H$, such that the block encoding $\bra{0}\PREP_H^\dagger\cdot\SEL_H\cdot\PREP_H\ket{0}=H/\lambda$ holds with normalization factor $\lambda=\sum_{\ell}\alpha_\ell$. Then, the operator $Q=\left(2\ket{0}\bra{0}-I\right)\cdot\PREP_H^\dagger\cdot\SEL_H\cdot\PREP_H$ has eigenvalues $e^{\pm i\arccos(E_k/\lambda)}$ corresponding to the energy $E_k$ of $H$. We can thus estimate the spectrum of $H$ by performing phase estimation on $Q$ with classical postprocessing.

For the chemistry Hamiltonian $H=T+U+V$, we see from \sec{alg_overview} that the terms $T$, $U$, and $V$ can indeed be expressed as linear combinations of unitaries. This means that we have $\PREP_T$, $\PREP_{U+V}$ and $\SEL_T$, $\SEL_{U+V}$, such that the block encodings
\begin{equation}
    \bra{0}\PREP_T^\dagger\cdot\SEL_T\cdot\PREP_T\ket{0}=\frac{T}{\lambda_T},\qquad
    \bra{0}\PREP_{U+V}^\dagger\cdot\SEL_{U+V}\cdot\PREP_{U+V}\ket{0}=\frac{U+V}{\lambda_U+\lambda_V}
\end{equation}
hold with normalization factors $\lambda_T$ and $\lambda_U+\lambda_V$, respectively. As a result, we find that the state preparation subroutine
\begin{equation}
    \left(\sqrt{\frac{\lambda_T}{\lambda_T+\lambda_U+\lambda_V}}\ket{0}+\sqrt{\frac{\lambda_U+\lambda_V}{\lambda_T+\lambda_U+\lambda_V}}\ket{1}\right)\otimes\PREP_T\ket{0}\otimes\PREP_{U+V}\ket{0}
\end{equation}
and the selection subroutine
\begin{equation}
    \ket{0}\!\bra{0}\otimes\SEL_T\otimes I+\ket{1}\!\bra{1}\otimes I\otimes\SEL_{U+V}
\end{equation}
block-encode the target Hamiltonian $H$ with normalization factor $\lambda_T+\lambda_U+\lambda_V$, giving a circuit implementation of the qubitization algorithm.

There is a subtle issue with the above naive analysis. While the kinetic operator $T$ can be encoded with a normalization factor close to $\lambda_T$, in the most efficient approach for preparation the potential operators $U+V$ are only encoded with a normalization factor close to $4(\lambda_U+\lambda_V)$. To be more specific, we can construct a state preparation
\begin{equation}
    \widetilde{\PREP}_{U+V}\ket{0}\otimes I
    \approx\frac{1}{2}\ket{0}\otimes\PREP_{U+V}+\frac{\sqrt{3}}{2}\ket{1}\otimes\PREP_{U+V}^{\bot}
\end{equation}
that only succeeds with a probability close to $1/4$, where $\ket{0}$ flags that the preparation is successful and $\ket{1}$ indicates a failure. If we were to use such a probabilistic subroutine in our block encoding, we would implement qubitization with a normalization factor close to $\lambda_T+4\left(\lambda_U+\lambda_V\right)$. Alternatively, we could perform a single step of oblivious amplitude amplification to boost the success probability close to unity, but this triples the cost of the state preparation and can introduce overhead to the overall gate complexity.

Our new approach improves the above implementations by proceeding even when the state preparation of $U+V$ fails. To elaborate, first consider the special case where $\lambda_T=3(\lambda_U+\lambda_V)$. Then, we propose to perform the selection subroutine of $T$ following the failure of the preparation of $U+V$, i.e.,
\begin{equation}
    I\otimes \ket{0}\!\bra{0}\otimes\SEL_{U+V}
    +\SEL_T\otimes \ket{1}\!\bra{1}\otimes I.
\end{equation}
This new selection operator, combined with the state preparation
\begin{equation}
    \PREP_T\ket{0}\otimes\widetilde{\PREP}_{U+V}\ket{0,0},
\end{equation}
gives the block encoding 
\begin{equation}
    \frac{3T}{4\lambda_T}+\frac{U+V}{4(\lambda_U+\lambda_V)}
    =\frac{T}{4(\lambda_U+\lambda_V)}+\frac{U+V}{4(\lambda_U+\lambda_V)}
    =\frac{T+U+V}{\lambda_T+\lambda_U+\lambda_V}
\end{equation}
with the desired normalization factor $\lambda_T+\lambda_U+\lambda_V$.

The general case can be handled by using an ancilla qubit to adjust the relative amplitudes between $T$ and $U+V$. If $\lambda_T<3(\lambda_U+\lambda_V)$, then we introduce an ancilla qubit and only perform the selection of $T$ conditioned on both the failure of the preparation of $U+V$ AND the ancilla qubit is in state $\ket{0}$, i.e.,
\begin{equation}
    I\otimes I\otimes \ket{0}\!\bra{0}\otimes\SEL_{U+V}
    +\ket{0}\!\bra{0}\otimes\SEL_T\otimes\ket{1}\!\bra{1}\otimes I
    +\ket{1}\!\bra{1}\otimes I\otimes\ket{1}\!\bra{1}\otimes I.
\end{equation}
This selection operator, combined with the state preparation
\begin{equation}
    \left(\cos\theta\ket{0}+\sin\theta\ket{1}\right)
    \otimes\PREP_T\ket{0}\otimes\widetilde{\PREP}_{U+V}\ket{0,0},
\end{equation}
gives the block encoding
\begin{equation}
    \frac{U+V}{4(\lambda_U+\lambda_V)}+\frac{3T\cos^2\theta}{4\lambda_T}
\end{equation}
up to a shifting by multiples of the identity operator.
For this to be an accurate encoding of the Hamiltonian, this expression must be proportional to $U+V+T$.
That means we need $1/[4(\lambda_U+\lambda_V)] = (3\cos^2\theta)/(4\lambda_T)$; solving for $\theta$, we find that $\theta=\arccos\sqrt{\lambda_T/3(\lambda_U+\lambda_V)}$, which gives the target Hamiltonian encoded with normalization factor $4(\lambda_U+\lambda_V)$.

In the case where $\lambda_T>3(\lambda_U+\lambda_V)$, we would still introduce an ancilla qubit to adjust the relative amplitudes. However, we now perform the selection of $T$ conditioned on failure of the preparation subroutine OR the ancilla qubit being in state $\ket{0}$ (and perform $\SEL_T$ in the remaining situations), i.e.,
\begin{equation}
    \ket{1}\!\bra{1}\otimes I\otimes \ket{0}\!\bra{0}\otimes\SEL_{U+V}
    +\ket{1}\!\bra{1}\otimes \SEL_T\otimes\ket{1}\!\bra{1}\otimes I
    +\ket{0}\!\bra{0}\otimes\SEL_T\otimes I\otimes I.
\end{equation}
Combined with the state preparation
\begin{equation}
    \left(\cos\theta\ket{0}+\sin\theta\ket{1}\right)
    \otimes\PREP_T\ket{0}\otimes\widetilde{\PREP}_{U+V}\ket{0,0},
\end{equation}
this gives the block encoding
\begin{equation}
    \frac{(U+V)\sin^2\theta}{4(\lambda_U+\lambda_V)}+\frac{3T\sin^2\theta}{4\lambda_T}+\frac{T\cos^2\theta}{\lambda_T}.
\end{equation}
For this to be proportional to $U+V+T$ we need
\begin{equation}
    \frac{\sin^2\theta}{4(\lambda_U+\lambda_V)}=\frac{3\sin^2\theta}{4\lambda_T}+\frac{\cos^2\theta}{\lambda_T}.
\end{equation}
Solving for $\theta$, we find that $\theta=\arcsin(2\sqrt{(\lambda_U+\lambda_V)/(\lambda_T+\lambda_U+\lambda_V)})$, which gives the target Hamiltonian encoded with the normalization factor $\lambda_T+\lambda_U+\lambda_V$.

In the following, we sketch the high-level ideas behind the circuit implementation of the state preparation and selection subroutines, and leave a discussion of the details to subsequent sections, including a more detailed description of the implementation in \sec{qubitization_overview}. For simplicity, we focus on the case where $\lambda_T<3(\lambda_U+\lambda_V)$ although similar analysis applies to the other cases as well. Note that $\lambda_T$ is not significantly larger than $\lambda_U+\lambda_V$ in this parameter regime and this is when the qubitization algorithm is expected to be advantageous.

For state preparation, we aim to prepare the quantum state
\begin{align}\label{eq:prepsta}
    &(\cos\theta\ket{0}+\sin\theta\ket{1})_a\ket{+}_b\frac{1}{\sqrt{\eta}}\!\left(\! \sqrt{\eta-1} \ket{0}_c\!\!\sum_{i\neq j=1}^{\eta}\ket{i}_d\ket{j}_e+\ket{1}_c\!\sum_{j=1}^{\eta}\ket{j}_d\ket{j}_e\! \right)\!\left(\! \frac 1{\sqrt{3}}\!\sum_{w=0}^2 \ket{w}_f\!\right)\! \nn
    &\otimes\left(\! \frac 1{2^{n_p-1}-1}\!\sum_{r,s=0}^{n_p-2} 2^{(r+s)/2} \ket{r}_g \ket{s}_h\! \right)
    \left(\sqrt{\frac{\lambda_U}{\lambda_U+\lambda_V}}\ket{0}+\sqrt{\frac{\lambda_V}{\lambda_U+\lambda_V}}\ket{1}\right)_i\nn
    &\otimes\left( \sqrt{\frac{\psuc}{\lambda_\nu}} \ket{0}_j\sum_{\nu\in G_0}\frac 1{\norm{\nu}}\ket{\nu}_k+\sqrt{1-\psuc}\ket{1}_j\ket{\nu^\bot}_k \right)\left( \frac 1{\sqrt{\sum_\ell \zeta_\ell}}\sum_{\ell=1}^{L} \sqrt{\zeta_\ell}\ket{\ell}\right)_l,
\end{align}
with $\psuc$ the probability of successfully preparing the momentum state, where the first ancilla register is used to adjust the relative amplitudes between $T$ and $U+V$ (as explained above), and the $i$ register (holding state with relative amplitudes $\sqrt{\lambda_U}$ and $\sqrt{\lambda_V}$) is used to select between $U$ and $V$. Here, the preparation subroutine $\PREP_T$ acts on the registers $b,f,g,h$
and outputs the corresponding states, whereas $\widetilde{\PREP}_{U+V}$ acts on registers $c,d,e,j,k,l$.  Note that as some of these registers are used for multiple terms in the Hamiltonian, the assignment of common registers such as $\ket{\cdot}_b$ to any particular preparation circuit is arbitrary.
Quantum circuits for these preparation subroutines are detailed in \sec{block_encode_t} and \sec{momentum}, respectively. 
The normalization factor $\lambda_\nu$ is defined in \eq{lambdas}, and $\ket{\nu^\bot}$ is a state obtained in the failure of initializing the momentum register. We will see that the probability of successfully preparing the momentum state $(1/ \sqrt{\lambda_\nu}) \sum_{\nu\in G_0} (1/ \norm{\nu})\ket{\nu}$ is close to $1/4$, which is consistent with our discussion above. 

Note that several important details, including the flagging of $i\neq j$ in register $c$ and the precision in preparing the ancilla state $\cos\theta\ket{0}+\sin\theta\ket{1}$, are omitted in our analysis for exposition purposes. We provide full details on how the encodings between $T$, $U$, and $V$ are performed and what precision the ancilla state is prepared with in \app{selTUV}.

We now describe the key steps in the implementation of the selection operations. For $\SEL_T$, we can compute the product $p_{w,r}p_{w,s}$ in an ancilla register and use that to perform a controlled $Z$ gate on $\ket{+}$, followed by a final uncomputation. The net effect of this transformation is
\begin{equation}
    \SEL_T:\ket{b}_b\ket{j}_e\ket{w}_f\ket{r}_g\ket{s}_h\ket{p_j}
    \mapsto(-1)^{b(p_{w,r}p_{w,s}\oplus 1)}\ket{b}_b\ket{j}_e\ket{w}_f\ket{r}_g\ket{s}_h\ket{p_j},\label{eq:Timp}
\end{equation}
where $\ket{b}$ is either the component $\ket{0}$ or $\ket{1}$ of the state $\ket{+}$ and $\ket{p_j}$ is the $j^{\rm th}$ particle's momentum. To implement $\SEL_{U+V}$, we first add $\nu$ into the $i$th momentum register. Since this should only be done for $V$, we perform the addition controlled by the register selecting between $U$ and $V$. Next, we subtract $\nu$ from the $j$th momentum register for both $U$ and $V$, and the subtraction does not need to be controlled by the corresponding selection register. We now check if $(p-\nu)\notin G$ (for $U$) or $[(p+\nu)\notin G]\lor [(q-\nu)\notin G]$ (for $V$) and use the result to perform controlled $Z$ gate on $\ket{+}$. For the case of $U$, we also perform the phase rotation $-e^{-\mi k_\nu\cdot R_\ell}$. Overall, this implements
\begin{align}\label{eq:Uimp}
    \SEL_U: \ket{b}_b\ket{j}_e\ket{0}_i\ket{\nu}_k\ket{\ell}_l\ket{q_j}    \mapsto &\ \ket{b}_b\ket{j}_e\ket{0}_i\ket{\nu}_k\ket{\ell}_l\ket{q_j-\nu}\nn
    \mapsto&\ (-1)^{b[(p-\nu)\notin G]} \ket{b}_b\ket{j}_e\ket{0}_i\ket{\nu}_k\ket{\ell}_l\ket{q_j-\nu}\nn
    \mapsto&\ -e^{-\mi k_\nu\cdot R_\ell}(-1)^{b[(p-\nu)\notin G]}\ket{b}_b\ket{j}_e\ket{0}_i\ket{\nu}_k\ket{\ell}_l\ket{q_j-\nu}
\end{align}
and
\begin{align}\label{eq:Vimp}
    \SEL_V: \ket{b}_b\ket{i}_d\ket{j}_e\ket{1}_i\ket{\nu}_k\ket{p_i}\ket{q_j}\mapsto&\ \ket{b}_b\ket{i}_d\ket{j}_e\ket{1}_i\ket{\nu}_k\ket{p_i+\nu}\ket{q_j-\nu}\nn
    \mapsto&\ (-1)^{b\left([(p+\nu)\notin G]\lor [(q-\nu)\notin G]\right)}\ket{b}\ket{i}\ket{j}\ket{1}\ket{\nu}\ket{p_i+\nu}\ket{q_j-\nu}
\end{align}
for $U$ and $V$ respectively, as controlled by the register selecting them.
See \sec{block_encode_t} and \sec{SelectUV} for a detailed description of the selection operations and their circuit implementations.

We bound the error of our implementation and provide a complete costing of the qubitization algorithm in \sec{erroranly}, with our main result summarized in \thm{qubitization}.
	
\subsection{Circuit implementation of qubitization}
\label{sec:qubitization_overview}

Our aim in this section is to provide a high-level overview of how qubitization can be realized to sample from the energy eigenvalues of a first quantized Hamiltonian. The general strategy is to combine the circuits used for state preparation and selection between $T$, $U$, and $V$ in order to minimize the overall cost.
	In particular, we will be preparing superpositions over $i$ and $j$, which are used independent of which of $T$, $U$, and $V$ is being block encoded.
	Then we will need to perform operations on the momentum given in registers $i$ and $j$, so we will use $i$ and $j$ to control the swap of the momenta into temporary ancillae in order to perform the needed operations.
	For the case of $T$ we are just computing a phase factor $\pm 1$ dependent on the value of $p$.
	In contrast, for $U$ we are performing arithmetic on the momentum from register $j$, and for $V$ we are performing arithmetic on the momenta from both registers $i$ and $j$.
	Therefore, the operations $p+\nu$ and $p-\nu$ will be controlled on registers that select between $T$, $U$, and $V$.
	
	Our protocol for preparing the state \eq{prepsta} consists of the following steps.
	\begin{enumerate}
	\item Rotate the qubit for selecting between $T$ and $U+V$.
	\item Prepare an equal superposition state over $\eta$ values of $i$ and $j$, and check that they are not equal.
	\item Prepare superpositions over $w$, $r$ and $s$ that will be used for the $T$ part of the select operation.
	These preparations do not need to be controlled, because it is only the $\SEL$ part that is controlled on the qubit for $T$. 	The costing for this will be shown to be $2n_p+2b_r-7$ in the subsection below for $T$.
	\item Prepare an appropriate superposition state for selecting between $U$ and $V$.
    We will show that the complexity is approximately $4\left\lceil\log\left(\eta+\lambda_\zeta\right)\right\rceil+2b_r-9$ Toffoli gates. 
	\item Prepare a superposition over $\nu$ with weights proportional to $1/\|\nu\|$.
	The costing for this will be described below for the $U$ and $V$ parts of the Hamiltonian.
	\item Prepare a state with amplitudes proportional to $\sqrt{\zeta_\ell}$, which may be done with complexity $\lambda_\zeta=\sum_\ell\zeta_\ell$.
	\end{enumerate}
	Because we have written the state to be prepared as a tensor product of the individual states, the overall cost of the state preparation is just the cost of the preparation of the states on the individual subsystems.
    
    In order to perform the $\SEL$ operation, the key parts are as follows.
    \begin{enumerate}
        \item Add $\nu$ into the momentum register at location $i$, controlled on the first and sixth registers which select $V$.
        \item Subtract $\nu$ from the momentum register at location $j$, controlled on the first qubit that is $\ket{1}$ for $U$ or $V$.
        \item Perform a $Z$ gate on the $\ket{+}$ state if $(q-\nu)\notin G$ (for $U$) or $[(p+\nu)\notin G]\lor [(q-\nu)\notin G]$ (for $V$).
        \item Perform a $Z$ gate on $\ket{+}$ state controlled on $p_{w,r} p_{w,s}$ for the case of $T$.
        \item For the case of $U$ perform the phase rotation $-e^{-\mi k_\nu\cdot R_\ell}$.
    \end{enumerate}
    In addition, for our phase measurement procedure the $\SEL$ operation should be self-inverse.
    It is easily seen that the component of the $\SEL$ for $T$ is self-inverse.
    To account for the $\SEL$ for $U$ and $V$ we can introduce an ancilla qubit that selects whether we add or subtract $\nu$.
    Controlling addition versus subtraction is done with no additional Toffoli complexity.
    This ancilla qubit is flipped each step, to produce a self-inverse operation in the same way as in Eq.~(8) of \cite{Berry2018}.
    
    The following discussion provides a statement of the block-encoding procedure for the Born-Oppenheimer Hamiltonian.
 
	\begin{lemma}
	The operator $(\PREP^\dagger \otimes I)\SEL(\PREP \otimes I)$ forms, for $N_a = 10+2n_\eta +3n_p+2\lceil \log(n_p-1) \rceil + \lceil \log(L) \rceil$, an $(\lambda,N_a,0)$-block encoding of $H_{\rm BO}$, meaning that $$(\bra{0}^{N_a}\otimes I)(\PREP^\dagger \otimes I)\SEL(\PREP \otimes I) (\ket{0}^{N_a}\otimes I)= \frac{\ketbra{0}{0}^{\otimes N_a} \otimes H_{\rm BO}}{\lambda}$$
	\end{lemma}
\begin{proof}
    The proof of this lemma follows by noting that the number of qubits needed to encode the state yielded by $\PREP$ as given in \eq{prepsta}.    
    That coincides with the ancillary space for block-encoding using the LCU lemma as in~\cite{Berry2015, Childs2012}.
    Further, the LCU lemma specifically shows that if $\PREP' \ket{0}^{N_a} = \sum_\ell \alpha_\ell \ket{j}$ and $\SEL' \ket{\ell}\ket{\psi}=\ket{\ell} U_\ell \ket{\psi}$ for $\alpha_\ell \ge 0$ then 
    \begin{equation}
        (\ketbra{0}{0}^{\otimes N_a} \otimes I) (\PREP'^\dagger \SEL' \PREP') \ket{0}^{N_a}\ket{\psi} = \frac{ \ket{0}^{N_a}\sum_\ell \alpha_\ell H_\ell \ket{\psi}}{\sum_\ell \alpha_\ell}\, .\label{eq:LCU2}
    \end{equation}
    The unitary operations $U_\ell$ are given implicitly in the above list for $\SEL$.
    In this section we do not consider the complexity of implementing the $\SEL$ operator and so only use the expression as a definition of the action of the operator on input states.
    
    Specifically, recall that up to a constant energy offset due to the nuclear positions the Hamiltonian is $H = T + U + V$ where
    \begin{equation}
	\begin{aligned}
	T&=\frac {\pi^2}{\Omega^{2/3}} \sum_{j=1}^{\eta}\sum_{w\in\{x,y,z\}} \sum_{r=0}^{n_p-2}\sum_{s=0}^{n_p-2} 2^{r+s} \sum_{b\in\{0,1\}}\left(\sum_{p\in G} (-1)^{b(p_{w,r} p_{w,s}\oplus 1)} \ket{p}\!\bra{p}_j\right), \\
	U&=\sum_{\nu\in G_0}\sum_{\ell=1}^{L}\frac{2\pi\zeta_\ell}{\Omega\norm{k_\nu}^2}\sum_{j=1}^{\eta}\sum_{b\in\{0,1\}}\left(-e^{-\mi k_\nu\cdot R_\ell}\sum_{q\in G}(-1)^{b[(q-\nu)\notin G]}\ket{q-\nu}\!\bra{q}_j\right),\label{eq:opdef2}\\
	V&=\sum_{\nu\in G_0}\frac{\pi}{\Omega\norm{k_\nu}^2}\sum_{i\neq j=1}^{\eta}\sum_{b\in\{0,1\}}\left(\sum_{p,q\in G}(-1)^{b([(p+\nu)\notin G]\lor[(q-\nu)\notin G])}\ket{p+\nu}\bra{p}_i\cdot\ket{q-\nu}\!\bra{q}_j\right)\, ,\\
	\end{aligned}
	\end{equation}
	here $p_{w,r}$, for example, denotes the $r^{\rm th}$ bit of the $w^{\rm th}$ component of the momentum.
  
  In this form, it is clear that the Hamiltonian operators in the LCU decomposition satisfy for each $\ell$, $H_\ell \in \{S_T,S_U,S_V\}$ where each of these unitary operators are defined in~\eqref{eq:Timp},~\eqref{eq:Uimp} and~\eqref{eq:Vimp}.  Thus we need to demonstrate two facts.  First, that the operations in $\SEL$ selectively apply these transformations and that the coefficients of the expansion match the square roots of those in the preceeding expression.

    Consider steps $1$ and $2$ in the definition of $\SEL$.  These steps can be implemented using a controlled unitary adder.
    A unitary subtractor subtracting the vector $\nu$ from a momentum register $\ket{q}_j$ has the action
    \begin{equation}
        \ket{\nu} \ket{q}_j \mapsto \ket{\nu} \ket{q - \nu}_j, 
    \end{equation}
    This therefore has the correct action for implementing the subtraction needed to implement $U$ and $V$.
    
    There are two cases we need to consider.  First, assume $q-\nu \in G$.  In this case the operator describing the transformation has the same action as $\ketbra{\nu}{\nu} \otimes \ketbra{p-\nu}{p}_j$.  Now let us assume that $p - \nu \not \in G$.  In this case we have an overflow.
    However, since $\sum_{b\in \{0,1\}} (-1)^{b[(q-\nu)\not \in G]}=0$ in~\eqref{eq:opdef} the action of the unitary adder need not be defined in this case and hence an additional qubit is not needed in the register to address potential overflow problems.
    (We will use that later to simplify the implementation.)
    With this fact in mind, it is then straightforward to show that steps 1 and 2 block encode the $V$ and $U$ terms in the Hamiltonian.
    
       To explain how this gives the required linear combination of unitaries, note that the general procedure is to implement an operator of the form $\sum_\ell \alpha_\ell H_\ell$, prepare an ancilla with amplitudes proportional to $\sqrt{\alpha_\ell}$, then perform a controlled $H_\ell$ operation.
    The first step in $\SEL$ gives the addition $p+\nu$ in momentum register $i$ required for $V$.
    Step 2 gives the subtraction $q-\nu$ required for $V$, the together with step 3 the complete operation needed for $V$ is performed as described in \eq{Vimp}.
    Step 2 also gives $q-\nu$ as required for $U$, step 3 also gives the sign flip for $U$, and step 5 gives the phase factor needed for $U$.
    Together they give the required operations for $U$ as shown in \eq{Uimp}.
    The correct amplitudes are obtained by the amplitudes $1/\|\nu\|$ in \eq{prepsta}, which gives the correct weights $1/\|\nu\|^2$ in the Hamiltonian.
    The $\SEL$ operation for $T$ is obtained entirely by step 4, which gives the sign flip.
    The appropriate weights in the sum are obtained by using the superpositions over $j,w,r,s$ in \eq{prepsta}. With these definitions in place we can formally state the following.
    
    The block-encoding of $T$ is a little less obvious, but it ultimately follows similarly from the definitions of the involved terms.
    Recalling from~\eqref{eq:lambdas} that $\lambda_T = 6\eta\pi^2 ( 2^{n_p-1}\!-1 )^2 / \Omega^{2/3}$ and noting the definition of $T$ given in~\eqref{eq:Tlcu} we see immediately that the theorem holds after substituting the terms in the LCU decomposition in~\eqref{eq:Tlcu} into the $U_j$ in~\eqref{eq:LCU2}.
\end{proof}

    For the rotation of the qubit to select between $T$ and $U+V$, we find that error in the rotation angle of $\Delta\theta$ results in error in the energy no more than $\lambda\Delta\theta$ (see \app{selTUV}).
    As a result, if we take the number of bits for this rotation to be $n_T$, then the error from this source is bounded as
    \begin{equation}\label{eq:selTUV}
        \frac {\pi\lambda}{2^{n_T}} \le \epsilon_T.
    \end{equation}
    The Toffoli complexity of performing this rotation is $n_T-3$.
    In some applications it may be possible to reduce the complexity by selecting $\Omega$ appropriately.
    That is, the ratio $\lambda_T/(\lambda_U+\lambda_V)$ will depend on $\Omega$, so one can choose $\Omega$ to that the rotation may be performed exactly with a relatively small number of bits.

In order to quantify the complexity of performing the state preparation for selecting between $U$ and $V$, we can take account of the relative amplitudes between the different parts.
    We find that
    \begin{equation}
        \frac{\lambda_U}{\lambda_V} = \frac{2\lambda_\zeta}{\eta-1} \qquad \qquad \lambda_\zeta = \sum_{\ell=1}^L \zeta_\ell \, 
    \end{equation}
    where we remind the reader that $\lambda_\zeta = \eta$ for charge neutral systems. Even if that is not the case, we see that $\lambda_U / \lambda_V$ is simply a ratio of integers.
    Rather than performing a rotation, one can instead achieve an exact state preparation as follows.
    First, prepare an equal preparation over $\eta-1+2\lambda_\zeta$ numbers.
    Then, perform an inequality test between this register and $\eta$ to give the correct relative weightings $\eta-1$ and $2\lambda_\zeta$.
    There is a subtlety in that we will initially prepare superpositions over all numbers $i$ and $j$ including the case $i=j$.
    This effectively increases the value of $\lambda_V$, so the ratio we want is $2\lambda_\zeta/\eta$.
    Then we prepare an equal superposition over $\eta+2\lambda_\zeta$ numbers (rather than $\eta-1+2\lambda_\zeta$).
    
    After preparing the superposition state and applying the inequality test, use a Toffoli between the result of the inequality test and the first qubit selecting between $T$ and $U+V$, to give a qubit flagging that we are applying $V$ (so will perform the addition $p+\nu$).
    Writing the number of qubits needed to store $\eta+2\lambda_\zeta$ as $n_{\eta\zeta}$, the complexity of preparing the equal superposition state is \cite{Lee2020} (see Appendix A.2, page 39 of version 2):
    \begin{equation}
        3n_{\eta\zeta}+2b_r-9,
    \end{equation}
    and the cost of the inequality test is $n_{\eta\zeta}-1$.
    Here $b_r$ is a number of bits of precision used for rotation on an ancilla.
    Together with the extra Toffoli, the cost is
    \begin{equation}\label{eq:UVprepcost}
    4n_{\eta\zeta}+2b_r-9.
    \end{equation}
    The convenience of this is that it introduces little error due to the state preparation.
    
    The probability of success of the state preparation is usually above $0.999$ for $b_r=7$ (see Figure 10 of \cite{Sanders2020_b}).
    The formula for the probability of success is given in Eq.~(205) of \cite{Sanders2020_b}, and is
    \begin{equation}\label{eq:eqprep}
        \eqprep{n,b_r} = \frac{n}{2^{\lceil\log n\rceil}} \left[ \left( 1+ \left(2-\frac{4n}{2^{\lceil\log n\rceil}}\right)\sin^2\theta(n,b_r)\right)^2 + \sin^2(2\theta(n,b_r)) \right],
    \end{equation}
    where $n$ is the number of basis states to create an equal superposition over, and $\theta(u,b_r)$ is the angle of rotation.
    Since we are using $b_r$ bits for the rotation angle, the angle is of the form $2\pi n/2^{b_r}$ for some $n$, and 
    \begin{equation}
        \theta(n,b_r) = \frac {2\pi}{2^{b_r}} {\rm round}\left( \frac {2^{b_r}}{2\pi}\arcsin\left(\sqrt{\frac{2^{\lceil\log n\rceil}}{4n}}\right) \right).
    \end{equation}
    More generally, we have proven that the probability of error is no more than $(3/2)^2(\pi/2^{b_r})^2$ with a slight adjustment to the formula for $\theta$ (see \app{sucbnd}). A simpler case is that where $\lambda_\zeta=\eta$, which occurs when the total charge is zero.
    In that case, the ratio we want is just 2.
    Instead of needing to prepare an equal superposition over $\eta+2\lambda_\zeta$ values, we only need a superposition over 3.
    Including $n_T-3$ Toffolis to perform the rotation on the first qubit, the total complexity for the preparation on the qubits selecting $T$, $U$ and $V$ would be
    \begin{equation}\label{eq:TUVprepcost}
        n_T+4n_{\eta\zeta}+2b_r-12.
    \end{equation}

    Next we consider the cost of the preparation of the equal superposition of $i$ and $j$ in unary.
    The steps needed are as follows.
    \begin{enumerate}
        \item Prepare an equal superposition in binary, using the approach in Appendix A.2 of \cite{Lee2020}, with cost $3n_\eta +2b_r - 9$ for each, where $n_\eta=\lceil\log\eta\rceil$ and $b_r$ is the angle of rotation on an ancilla qubit.
        \item Test if $i$ and $j$ are equal in binary, flagging the result in a qubit. This qubit will need to be combined with another success qubit for applying the potential, so the cost is $n_\eta$.
        \item The testing $i\ne j$ can be inverted with $n_\eta$ Toffolis.
        \item The inversion of the equal superposition of $i$ and $j$ can be inverted with cost $3n_\eta +2b_r - 9$ for each.
    \end{enumerate}
    Adding all these costs together gives a total complexity
    \begin{equation}\label{eq:ijprepcost}
        14n_\eta+8b_r-36.
    \end{equation}
    For the failure case where $i=j$ \emph{and} we would otherwise be selecting $V$ (rather than $U$), we can simply combine this with the failure of the state preparation for $\ket{\nu}$, and perform $\SEL$ for $T$ in that case.
    
    For convenience of the rest of the procedure, we produce three qubits for selecting $T$, $U$ and $V$.
    These are as follows.
    \begin{enumerate}
        \item A qubit which is 0 if we are to perform $T$.
        \item A qubit which is 0 if we are performing the subtraction $q-\nu$.
        This will be the case for $U$ or $V$.
        \item A qubit which is 0 for $V$ only, which controls the addition $p+\nu$.
    \end{enumerate}
    First, there are four preparations of equal superposition states, and checking that all of those succeeded takes 3 Toffolis.
    There are also another four qubits needed to select $T,U,V$.
    These are the rotated qubit (register 1 in \eq{prepsta}),
    the qubit resulting from the inequality test to select between $U$ and $V$,
    the qubit flagging $i\ne j$,
    and the qubit flagging failure of the state preparation for $\nu$.
    For simplicity of the explanation, let us call these logic values $A,B,C,D$, respectively, with 0 corresponding to True.
    
    To check that the state preparation for $U+V$ succeeded, we can use $(B \vee C) \wedge D$, which takes two Toffolis (one for each logic operation).
    In one scenario (OR), $T$ was selected as $A\vee \neg ((B \vee C) \wedge D)$, and in the other (AND)
    $T$ was selected as $A\wedge \neg ((B \vee C) \wedge D)$.
    Either takes one more Toffoli, but an extra Toffoli is needed to ensure that qubit is set only for success of the preparation of the equal superposition states.
    That is 7 Toffolis so far.
    
    Now, in the OR case, we set the second qubit (selecting $U$ or $V$) if $\neg (A\vee \neg ((B \vee C) \wedge D))$ and there were success of the equal superposition preparations.
    That takes one more Toffoli.
    In the AND case we would set that second qubit according to $(B \vee C) \wedge D$ and success of the preparations, which is again one Toffoli.
    The third qubit can then be set by taking the AND of $\neg B$ and the result in the second qubit.
    That is one more Toffoli for a total of 9.
    
    For the preparations of equal superposition states, the four needed are the equal superposition over 3 basis states for $w$, one with $\eta+2\lambda_\zeta$ basis states, and two with $\eta$.
    For the superposition over three basis states required for $w$, it is highly efficient to take $b_r=8$, so we use that rather than allowing arbitrary $b_r$ there.
    As a result, the block encoded state is multiplied by the probability
    \begin{equation}
        P_{\rm eq} := \eqprep{3,8}\eqprep{\eta+2\lambda_\zeta,b_r}\eqprep{\eta,b_r}^2 .
    \end{equation}
    This corresponds to giving an effective $\lambda$ divided by $P_{\rm eq}$.
    
	%%%%%%%%%%%%%%%%%%%%%%%%%%%%%%%%%%%%%%%%%%%%%%%%%%%%%%%%%%%%%%%%%%%%%%%%%%%%%%
	\subsection{Block-encoding the kinetic operator}
	\label{sec:block_encode_t}

	Recall that the usual form of the kinetic operator is given by
	\begin{equation}
	T=\sum_{j=1}^{\eta}\sum_{p\in G}\frac{\norm{k_p}^2}{2}\ket{p}\!\bra{p}_j.
	\end{equation}
	We have rewritten this expression to express the square explicitly as a sum of products of bits, as
	\begin{equation}
	T=\frac {2\pi^2}{\Omega^{2/3}} \sum_{j=1}^{\eta}\sum_{p\in G}\sum_{w\in\{x,y,z\}} \sum_{r=0}^{n_p-2}\sum_{s=0}^{n_p-2} 2^{r+s} p_{w,r} p_{w,s} \ket{p}\!\bra{p}_j,
	\end{equation}
	where $p_{w,r}$ is indicating bit $r$ of the $w$ component of $p$.
	To block-encode this, we need to prepare the superposition state proportional to
	\begin{equation}
\ket{+}\sum_{j=1}^{\eta}\ket{j} \sum_{w=0}^2 \ket{w} \sum_{r=0}^{n_p-2} 2^{r/2} \ket{r} \sum_{s=0}^{n_p-2} 2^{s/2} \ket{s}.
	\end{equation}
	The preparation of the equal superposition over $j$ was addressed above.
	For the equal superposition over $w$, we take $b_r=8$ and $n=2$, so the formula $3n+2b_r-9$ for the complexity gives a cost of only $13$ Toffolis.
	
	The superposition over $r$ and $s$ can be prepared with high chance of success by using the following approach.
	First, we aim to prepare the following state in unary
	\begin{equation}
	    2^{-n_p/2}\ket{n_p-1} + \sum_{r=0}^{n_p-2} 2^{-(r+1)/2} \ket{r}.
	\end{equation}
	Half the amplitude is on $\ket{0}$, then of the remaining states half the amplitude is on $\ket{1}$, and so forth.
	This means that to prepare the state in unary, we perform a Hadamard on the first qubit (labelled 0), then controlled on that we perform a Hadamard on the second qubit,
	and so forth.
	There are in total $n_p-2$ controlled Hadamards, and each controlled Hadamard can be performed with a single Toffoli and a catalytically used $\ket{T}$ state as in Figure 17 of \cite{Lee2020}.
	The qubit for $n_p-1$ is being used to flag success of the state preparation (it flags success with $\ket{0}$).
	Omitting the state $\ket{n_p-1}$, the state flagged for success is (up to normalization)
	\begin{equation}
	    \sum_{r=0}^{n_p-2} 2^{-(r+1)/2} \ket{r}.
	\end{equation}	
	We then flip all these remaining bits so that the state becomes
	\begin{equation}
	    \sum_{r=0}^{n_p-2} 2^{-(r+1)/2} \ket{n_p-2-r} = 2^{-(n_p-1)/2} \sum_{r=0}^{n_p-2} 2^{r/2} \ket{r}.
	\end{equation}	
	We will want the unary to be one-hot, but that conversion may be performed with CNOTs so the complexity is omitted in the Toffoli count. Note that to account for the case $r=0$, one can add one qubit in the state $\ket{1}$, so $r=0$ is encoded as $\ket{10000\ldots}$, $r=1$ is encoded as $\ket{11000\ldots}$, and so forth.
	Then, the sequence of CNOTs gives the desired one-hot unary. We therefore get a complexity of $n_p-2$ for this state preparation.
	Together with the $13$ Toffolis used to prepare the superposition over $w$, the overall complexity for the preparation over the $w$, $r$, and $s$ registers is
	\begin{equation}\label{eq:wrsprep}
	    2n_p+9.
	\end{equation}
	The effect of the preparation having a small probability of failure is that the effective value of $\lambda$ is increased to
	\begin{equation}\label{eq:lamT}
	    \lambda'_T = \frac{6\eta\pi^2}{\Omega^{2/3}} 2^{2(n_p-1)},
	\end{equation}
	where we have replaced $2^{n_p-1}-1$ with $2^{n_p-1}$.
	
	In order to perform the $\SEL$ operation, we need to use the value of $p$ in momentum register $j$.
	Because we need to perform arithmetic on the momentum in register $j$ for $U$ and $V$ as well, the most efficient procedure is to simply control a swap of the momentum in register $j$ into an ancilla register.
	Similarly, the momentum in register $i$ will be swapped into an ancilla register.
	The operations we perform on the momenta in these ancilla registers will then be controlled by the first two qubits selecting between $T,U,V$.
	That is more efficient than performing the operations in a controlled way on every momentum register.
	The complexity of controlling the swaps of the momenta out of and back into momentum registers $i$ and $j$ is $12\eta n_p$.
	For the control, we perform unary iteration on $i$ and $j$.
	That requires $\eta-2$ Toffolis for each of the four times (swapping the register out and back for $i$, then again for $j$), for a total of $4\eta-8$.
	Together with $12\eta n_p$ that gives a total of
	\begin{equation}\label{eq:swapscost}
	    12\eta n_p +4\eta-8.
	\end{equation}
	We will consider that cost separately from the costs for implementing $\SEL$ specifically for $T,U,V$.
	In the case of $T$, the other operations needed to perform $\SEL$ are as follows.
	\begin{enumerate}
	\item Use $w$ to control copying of component $w$ of $p$ into another ancilla.
	This has cost $3(n_p-1)$.
	\item Copy bit $r$ of component $w$ of $p$ to an ancilla qubit, with cost $n_p-1$.
	\item Copy bit $s$ of component $w$ of $p$ to an ancilla qubit, with cost $n_p-1$.
	\item Controlled on those two bits, perform a phase flip on the $\ket{+}$ state unless both bits are one.
	This gives an orthogonal $\ket{-}$ state unless both bits are one.
	The orthogonal state gives no contribution to the block-encoding.
	Only one Toffoli is needed here.
	\item Erase the two qubits containing $p_{w,r}$ and $p_{w,s}$, then the register containing component $w$ of $p$.
	These erasures can be done with measurements and Clifford gates.
	\end{enumerate}
	The overall Toffoli complexity is only about as much as is needed to run through the qubits representing $p$.
	It is $5(n_p-1)+1$.
	There is also some complexity due to the fact that we also have a control qubit selecting $T$, but control on one extra qubit needs only one more Toffoli.
	That gives a total $\SEL$ cost for $T$
    \begin{equation}\label{eq:tsel}
        5(n_p-1)+2,
    \end{equation}
	in addition to the $\SEL$ cost $12\eta n_p$ that is common between $T,U,V$.

	%%%%%%%%%%%%%%%%%%%%%%%%%%%%%%%%%%%%%%%%%%%%%%%%%%%%%%%%%%%%%%%%%%%%%%%%%%%%%%
	\subsection{State preparation for \texorpdfstring{$U$}{U} and \texorpdfstring{$V$}{V}}
	\label{sec:momentum}
	The $U$ and $V$ potential operators both have a $1/\|k_\nu\|^2$ weighting, which means that in order to do the state preparation for the block encoding we need to prepare a state proportional to
	\begin{equation}
	\sum_{\nu\in G_0}\frac{1}{\norm{\nu}}\ket{\nu},
	\end{equation}
	where $\nu\in G_0=\big[-(N^{1/3}-1),N^{1/3}-1\big]^3\backslash\{(0,0,0)\}$ is a three-component integer vector.
	In our implementation, we use eight registers $\ket{\mu}\ket{\nu_x}\ket{\nu_y}\ket{\nu_z}\ket{m}\ket{0}$ to hold this state.
	These subsystems are used as follows.
	\begin{itemize}
	    \item $\ket{\mu}$ flags which box was used in the series of nested boxes for the state preparation.
	    \item $\ket{\nu_x}\ket{\nu_y}\ket{\nu_z}$ are the components of $\nu$ given as signed integers.
	    \item $\ket{m}$ is an ancilla in an equal superposition used to give the correct amplitude via an inequality test.
	    \item $\ket{0}$ flags that the state preparation was successful.
	\end{itemize}
	There will be four parts of the preparation where there is some amplitude for failure.
	First, there is the preparation of $\mu$.
	Second, there is the negative zero in the signed integer representation of $\nu$, which is not allowed.
	Third, there is testing whether the position is inside the inner box in the state preparation.
	Fourth, there is the success of the inequality test used in the state preparation.
	Success for all four is indicated by $\ket{0}$ on the flag qubit above,
	which can then be used in amplitude amplification.
	
	The state that we will prepare has the form
	\begin{equation}
	\frac{1}{\sqrt{\mathcal{M}2^{n_p+2}}}\sum_{\mu=2}^{n_p+1}\sum_{\nu\in B_\mu}\sum_{m=0}^{\ceil{\mathcal{M}(2^{\mu-2}/\norm{\nu})^2}-1}\frac{1}{2^\mu}\ket{\mu}\ket{\nu_x}\ket{\nu_y}\ket{\nu_z}\ket{m}\ket{0}
	+\ket{\Psi^\bot},
	\end{equation}
	where $B_\mu$ is a subset of $\nu$ to be defined later and $\big(I\otimes\bra{111}\big)\ket{\Psi^\bot}=0$. Conditioned on the qubit flagging success with $\ket{0}$, the amplitude for each $\nu$ will then be
	\begin{equation}
	\sqrt{\frac{\ceil{\mathcal{M}(2^{\mu-2}/\norm{\nu})^2}}{\mathcal{M}2^{2\mu}2^{n_p+2}}}\approx
	\frac{1}{8\sqrt{2^{n_p}}}\frac{1}{\norm{\nu}}.
	\end{equation}
	The probability of the flag qubit in state $\ket{0}$ is around $1/4$, and a single step of amplitude amplification boosts the success probability to close to $1$ \cite[Fig.\ 1]{BabbushContinuum_b}.
	Note that the probability is slightly smaller here than in \cite{BabbushContinuum_b}, because we use a preparation on $\mu$ with nonzero amplitude for failure.
	
	The approach is to use a hierarchy of nested boxes in $\big[{-}(N^{1/3}-1),N^{1/3}-1\big]$ indexed by $\mu$ and to prepare a set of $\nu$ belonging to each box. We initially prepare a unary-encoded superposition state
	\begin{equation}\label{eq:nestedcube}
	\frac{1}{\sqrt{2^{n_p+2}}}\sum_{\mu=2}^{n_p+1}\sqrt{2^\mu}\ket{\mu}
	=\frac{1}{\sqrt{2^{n_p+2}}}\sum_{\mu=2}^{n_p+1}\sqrt{2^\mu}\ket{0\cdots 0\underbrace{1\cdots 1}_{\mu}}.
	\end{equation}
	This state can be created using Hadamard and controlled Hadamard gates in a similar way as we performed the state preparation for $r$ and $s$ in order to implement the squaring of $p$.
	Here we have omitted the failure part of that preparation, so this state is subnormalized, with the norm indicating the probability of success.
	There is just one more bit than in the case where we were preparing $r$ and $s$, so the cost is $n_p-1$ Toffolis.

	In the next stage, we prepare
	\begin{equation}
	\frac{1}{\sqrt{2^{n_p+2}}}\sum_{\mu=2}^{n_p+1}\sum_{\nu_x,\nu_y,\nu_z=-(2^{\mu-1}-1),-0}^{2^{\mu-1}-1}\frac{1}{2^{\mu}}\ket{\mu}\ket{\nu_x}\ket{\nu_y}\ket{\nu_z},
	\end{equation}
	where the basis state is explicitly represented as
	\begin{equation}
	\ket{0\cdots 0\underbrace{1\cdots 1}_{\mu}}\ket{\nu_{x,\text{sign}}0\cdots 0\nu_{x,\mu-2}\cdots\nu_{x,0}}\ket{\nu_{y,\text{sign}}0\cdots 0\nu_{y,\mu-2}\cdots\nu_{y,0}}\ket{\nu_{z,\text{sign}}0\cdots 0\nu_{z,\mu-2}\cdots\nu_{z,0}}.
	\end{equation}
	Note that in this signed integer representation, the number zero appears twice, with both a plus sign and a minus sign.
	The sum with ``$,-0$'' indicates that the negative zero is included at this point.
	This state preparation can be implemented using six Hadamards and $3(n_p-1)$ controlled Hadamards. Specifically, we create the uniform superposition state in registers $\ket{\nu_{x,\text{sign}},\nu_{x,0}}\ket{\nu_{y,\text{sign}},\nu_{y,0}}\ket{\nu_{z,\text{sign}},\nu_{z,0}}$ using six Hadamards.
	We always have an equal superposition on $\nu_{x,0}$, $\nu_{y,0}$ and $\nu_{z,0}$ because $\mu$ has a minimum value of $2$.
	Then we generate uniform superpositions in $\ket{\nu_{x,j}}$, $\ket{\nu_{y,j}}$, and $\ket{\nu_{z,j}}$ controlled on the ($j+1$)th qubit in $\ket{\mu}$ from the right for $j=2,\ldots,n_p$.
	As before, the controlled Hadamards can be implemented with a Toffoli and a catalytic state, so the Toffoli cost is $3(n_p-1)$.
	
	Now, to remove the minus zero in this representation, we flag it as a failure. This can be done by checking each $\ket{\nu_x}$, $\ket{\nu_y}$, and $\ket{\nu_z}$ whether the sign bit is one and the remaining bits are zero, which requires controlled $\Toffoli$s each having $n_p+1$ controls. We then use another two $\Toffoli$ gates to flag if any of these checkings returned true. Altogether, the Toffoli cost is
	\begin{equation}
	3n_p+2.
	\end{equation}
	
	In the next stage, we test whether all of $\nu_x$, $\nu_y$, and $\nu_z$ are smaller in absolute value than $2^{\mu-2}$.
	First, we convert the $\mu$ register to one-hot unary using $\CNOT$s.
	Then, conditioned on the $\mu$th qubit being one, we flag $\ket{f_{\text{box}}}$ depending on whether the ($\mu-1$)th qubit $\ket{\nu_{x,\mu-2}}$, $\ket{\nu_{y,\mu-2}}$, and $\ket{\nu_{z,\mu-2}}$ are all zero, using a controlled $\Toffoli$ gate with four controls, which has a cost of three Toffolis. 
	Because this must be done $n_p$ times, this gives a cost of
	\begin{equation}
	3n_p .
	\end{equation}
	The costs so far are then $4(n_p-1)$ controlled Hadamards and $6n_p+2$ Toffolis.
	Although the complexity of the controlled Hadamard is one Toffoli, we will need to pay this cost again when inverting the preparation, but the Toffolis can be erased with Clifford gates.
	The state excluding the failures then becomes
	\begin{equation}
	\frac{1}{\sqrt{2^{n_p+2}}}\sum_{\mu=2}^{n_p+1}\sum_{\nu\in B_\mu}\frac{1}{2^{\mu}}\ket{\mu}\ket{\nu_x}\ket{\nu_y}\ket{\nu_z},
	\end{equation}
	where $B_\mu$ is the set of $\nu$ such that $|\nu_x|$, $|\nu_y|$, $|\nu_z|$ are less than $2^{\mu-1}$, excluding the case that they are all less than $2^{\mu-2}$.
	
	We now prepare an ancilla state in an equal superposition of $\ket{m}$ for $m=0$ to $\mathcal{M}-1$. The parameter $\mathcal{M}$ is a power of two and is chosen to be sufficiently large to obtain an accurate block-encoding. This state can be prepared entirely using Hadamards, so no non-Clifford gates are required. The resulting state excluding the failures is now
	\begin{equation}
	\frac{1}{\sqrt{\mathcal{M}2^{n_p+2}}}\sum_{\mu=2}^{n_p+1}\sum_{\nu\in B_\mu}\sum_{m=0}^{\mathcal{M}-1}\frac{1}{2^{\mu}}\ket{\mu}\ket{\nu_x}\ket{\nu_y}\ket{\nu_z}\ket{m}.
	\end{equation}
	
	In the final stage, we test the inequality
	\begin{equation}
	\big(2^{\mu-2}\big)^2\mathcal{M}>m\big(\nu_x^2+\nu_y^2+\nu_z^2\big).
	\end{equation}
	We then need to use the four flag qubits for success to produce an overall success flag.
	The final state will then be
	\begin{equation}\label{eq:invprep}
	\frac{1}{\sqrt{\mathcal{M}2^{n_p+2}}}\sum_{\mu=2}^{n_\mu}\sum_{\nu\in B_\mu}\sum_{m=0}^{\ceil{\mathcal{M}(2^{\mu-2}/\norm{\nu})^2}-1}\frac{1}{2^\mu}\ket{\mu}\ket{\nu_x}\ket{\nu_y}\ket{\nu_z}\ket{m}\ket{0}
	+\ket{\Psi^\bot}
	\end{equation}
	as desired. To achieve this, we compute $\big(2^{\mu-2}\big)^2\mathcal{M}$ and $m\big(\nu_x^2+\nu_y^2+\nu_z^2\big)$ separately in two ancilla registers and compare the results.
	The unsigned number $m\big(\nu_x^2+\nu_y^2+\nu_z^2\big)$ is bounded by $3\mathcal{M}2^{2n_\mu-2}$, so we need at most $2n_\mu+n_{\mathcal{M}}$ qubits to represent it,
	where we use
	\begin{equation}
	    n_{\mathcal{M}} = \lceil\log(\mathcal{M})\rceil.
	\end{equation}
	The value $\big(2^{\mu-2}\big)^2\mathcal{M}$ is upper bounded by $2^{2n_\mu-4}\mathcal{M}$, so we could represent it with $2n_\mu-4+n_{\mathcal{M}}$ qubits.
	
	Because $\mathcal{M}$ is a classically chosen value, and $\mu$ is given in unary, $\big(2^{\mu-2}\big)^2\mathcal{M}$ can be computed in place with \emph{no} gates, and simply a relabelling of qubits.
	We can place zeroed qubits in between the qubits of the (one-hot) unary representation of $\mu$, to give $2\mu$.
	For example, $\mu=1$ is changed from $000001$ to $0\underline{0}0\underline{0}0\underline{0}0\underline{0}0\underline{0}1\underline{0}$, $\mu=2$ is changed from $000010$ to $0\underline{0}0\underline{0}0\underline{0}0\underline{0}1\underline{0}0\underline{0}$, and so forth, where the underlined zeros are those added. Then this unary representation of $2\mu$ is also a binary representation of $2^{2\mu+1}$.
	Lastly, to multiply by $2^{-5}\mathcal{M}$, we simply pad with $\log \mathcal{M} -5$ zeros.
	In practice these zeroed ancillae do not even need to be included, and we simply relabel the qubits for $\mu$.
		
	To compute $m\big(\nu_x^2+\nu_y^2+\nu_z^2\big)$, we first use the result \lem{triple} given in \app{squaring} to cost the sum of three squares, which is $3n^2-n-1$ for $n$-bit numbers.
	We can regard $\nu_x$, $\nu_y$, and $\nu_z$ as unsigned numbers of $n_p$ bits, so squaring them costs $3n_p^2-n_p-1$ $\Toffoli$s.
	Finally, we multiply two binary numbers of length $\log(\mathcal{M})$ and $2n_p+2$ using a similar approach (see \lem{products}) with
	\begin{equation}
	    2n_{\mathcal{M}}\big(2n_p+2\big)-n_{\mathcal{M}}
	\end{equation}
	$\Toffoli$ gates. 
	To test the inequality $\big(2^{\mu-2}\big)^2\mathcal{M}>m\big(\nu_x^2+\nu_y^2+\nu_z^2\big)$, we use a ($2n_p+n_{\mathcal{M}}+2$)-bit comparator, which has cost \cite{SLSB19}
	\begin{equation}
	2n_p+n_{\mathcal{M}}+2.
	\end{equation}
	The net cost of the squaring, multiplication and inequality test is then
	\begin{equation}
	3n_p^2+n_p+1+4n_{\mathcal{M}}\big(n_p+1\big)
	\end{equation}
	$\Toffoli$ gates.
	There are a further 3 Toffolis needed to produce the overall success flag qubit.
	
	Now for the overall block encoding, the state preparation needs to be inverted.
	Rather than paying this cost again, we can keep the ancillae used to perform the arithmetic and inequality test.
	As a result, the inequality test and arithmetic can be inverted with Clifford gates.
	The 3 Toffolis used for the overall success flag qubit can be inverted with Clifford gates too.
	Similarly, the $6n_p+2$ Toffolis used for the testing on $\nu$ can be inverted with zero non-Clifford cost.
	The only operations in this state preparation that have a non-Clifford cost in the inverse preparation are the $4(n_p-1)$ controlled Hadamards used for the preparation of the $\mu$ register and superpositions on the $\nu$ register.
	As a result, the total Toffoli cost for the preparation and inverse preparation is
	\begin{equation}\label{eq:nuprepcost}
	    3n_p^2+15n_p-7+4n_{\mathcal{M}}\big(n_p+1\big).
	\end{equation}
	In the case where amplitude amplification is performed, the cost will be three times this, due to the need to invert the inequality test and calculation, then perform it again.
	
	The other key parts of the state preparation for $U$ and $V$ are to prepare the superpositions over $i$, $j$ and $\ell$.
	The preparation over $i$ and $j$ are just simple preparations of equal superposition states which we have addressed already.
	The final preparation we need to consider here is therefore the preparation over $\ell$ with weightings $\sqrt{\zeta_\ell}$.
	In the preparation over the second qubit selecting between $U$ and $V$, we have already prepared a superposition, which will be an equal superposition over $2\lambda_\zeta$ values for the case of $U$.
	That is the case where we need to prepare this superposition state, so we can just prepare it in this case, rather than preparing it for all cases as described above.
	
	For the preparation, we can just use a QROM on this register, outputting the value of $\ell$ in an output register.
	When there are $\zeta_\ell$ values in the input register that give the same value of $\ell$ in the output register, the result is obtaining an amplitude exactly proportional to $\sqrt{\zeta_\ell}$.
	Moreover, we will need to compute $k_\nu \cdot R_\ell$ in performing the $\SEL$ operation.
	The value of $R_\ell$ can be output in another ancilla register at the same time.
	The complexity of the QROM, using the simple form of QROM, is just $\lambda_\zeta -1$.
	The QROM should be controlled on the qubit which $U+V$, so $R_\ell$ are not output there.
	The QROM automatically does not output $R_\ell$ for $V$, because this input register is outside the range for $V$.

    Although we have aimed to prepare a state with amplitudes $\sqrt{\zeta_\ell}$, note that the only way that $\ell$ is used is in outputting $R_\ell$.
    That means, if we are using QROM to output both $\ell$ and $R_\ell$ at once, then the $\ell$ register is not actually used.
    That means we can simply omit the $\ell$ register, and output $R_\ell$.
    That saves ancilla qubits, though does not change the Toffoli count.
	
	In inverting this preparation, erasing the output of the QROM has a Toffoli complexity of
	$2^k+\left\lceil 2^{-k} \lambda_\zeta \right\rceil$
	where $k$ is taken to be a power of 2.
	The resulting complexity is approximately the square root of the initial QROM.
	For convenience of describing the complexities, we will define
	\begin{equation}\label{eq:erse}
	    \erase{x} := \min_{k\in\mathbb{N}}\left( 2^k+\left\lceil 2^{-k} x \right\rceil \right),
	\end{equation}
	so this cost is $\erase{\lambda_\zeta}$.
	In terms of this function, the cost of outputting the values of $R_\ell$ via QROM and erasing them again can be given as
	\begin{equation}\label{eq:Rlqrcost}
	    \lambda_\zeta + \erase{\lambda_\zeta}.
	\end{equation}
	
	%%%%%%%%%%%%%%%%%%%%%%%%%%%%%%%%%%%%%%%%%%%%%%%%%%%%%%%%%%%%%%%%%%%%%%%%%%%%%%
	\subsection{The \texorpdfstring{$\SEL$}{SEL} operations for \texorpdfstring{$U$}{U} and \texorpdfstring{$V$}{V}}
	\label{sec:SelectUV}
	We have described the block encoding for $T$, and the state preparation for $U$ and $V$, and now give the method to perform the $\SEL$ operations for $U$ and $V$.
	The general principle is as follows.
	\begin{enumerate}
	    \item Use the $i$ and $j$ registers to control a swap of the momentum registers $p$ and $q$ into ancillae.
	    \item Controlled on the qubit that is $\ket{0}$ only for $V$, add $\nu$ (in place) into the ancilla for $p$.
	    \item Controlled on the qubit that is $\ket{0}$ for $U$ or $V$, subtract $\nu$ into the ancilla for $q$.
	    \item Check if either $p+\nu$ or $p-\nu$ is outside the box $G$, and if so perform a phase gate on the $\ket{+}$ state.
	    \item Control swap the ancillae back into the momentum registers.
	    \item Apply the phase $e^{-\mi k_\nu\cdot R_\ell}$ in the case of $U$ only.
	\end{enumerate}
	The cost of the controlled swaps was discussed above, and is $12\eta n_p$.
	The key cost to discuss now is the addition and subtraction of $\nu$ to momenta.
	The difficulty is that the momentum and $\nu$ are encoded as signed integers, whereas the usual methods for addition/subtraction are for integers encoded in two's complement.
	
    In order to achieve this, just one of the numbers needs to be converted into two's complement.
    For example, since $p+\nu$ is computed in-place, we can convert $p$ to two's complement, then use the sign bit of $\nu$ to control adding or subtracting $|\nu|$ into $p$.
    This control of addition versus subtraction can be done with no more complexity than regular addition and subtraction.
    This is done for each component, so
	the overall complexity of converting between signed and two's complement is approximately $12$ times $n_p$ for the conversions.
	The addition and subtraction also need to be controlled by the qubits selecting between $T,U,V$, so have complexity twice uncontrolled additions and subtractions.
	That gives an overall complexity about
	\begin{equation}\label{eq:addcosts}
	    24n_p.
	\end{equation}
	
	More specifically, to convert an $n$ digit number from signed to two's complement, we take the sign qubit and use it as a control on CNOTs with the remaining qubits.
	Then we add it to the remaining $n-1$ bits with complexity $n-2$ Toffolis.
    We can just invert this process for converting in the opposite direction.
    Because $p$ and $q$ have $n_p$ bits in each of their three components, the complexity of converting them to two's complement is $6(n_p-2)$.
    
    For the cost of the controlled addition (or subtraction), there is cost $n_p+1$ to make each component controlled, which can be achieved by (for example) controlling copying out of $\nu$ to an ancilla.
    (That can be erased without Toffolis by measurements and phase fixups.)
    Then the addition has cost $n_p+1$, giving a total cost of $2(n_p+1)$ for the controlled addition/subtraction of each component.
    The total cost is $12(n_p+1)$ for these controlled additions and subtractions.
    
    Lastly we need to convert $p+\nu$ and $q-\nu$ back to signed form.
    Adding/subtracting $\nu$ expands these numbers by 2 bits, so the cost is now $6n_p$.
    The total cost is then
    \begin{equation}
        6(n_p-2)+12(n_p+1)+6n_p=24n_p
    \end{equation}
    Toffolis.
    Thus we get exactly the same cost as approximately estimated above.
    
    Another major part of the $\SEL$ operations is to check whether $p+\nu$ or $q-\nu$ is outside the bounding box $G$.
    Because the complexity only depends on $N$ through the number of bits used in the arithmetic $n_p$, it is best to take $N$ as large as possible given $n_p$, which means that the bounds of the box $\pm(N^{1/3}-1)/2$ are taken to be of the form $2^{n_p-1}-1$.
    That means that, if $p+\nu$ or $q-\nu$ is outside the bounding box, then at least one of the extra qubits introduced in the arithmetic is set to $\ket{1}$.
    These qubits are additional, so are not swapped from the ancilla back to the main momentum registers.
    
    If we perform the checking as presented above, then we use these qubits to control a $Z$ on the $\ket{+}$ state, in order to remove parts of the Hamiltonian where $p+\nu$ or $q-\nu$ are outside the bounding box, which are excluded.
    However, note that in the block encoding we would select on $\ket{0}$ for these ancillae, which means that these parts of the Hamiltonian are automatically eliminated.
    It is therefore unnecessary to perform this controlled phase in the $\SEL$ operation.
    
    The last thing we need to do for implementing $\SEL$ is to apply the phase factor $-e^{-\mi k_\nu\cdot R_\ell}$.
    To achieve this we need to compute $k_\nu\cdot R_\ell$ then add it into a phase gradient register to apply the phase.
    We do not need to make the addition into the phase gradient register controlled on the qubits that select $U$, because the register with $R_\ell$ is nonzero only for $U$.
    The minus sign in this phase does need to be performed in a controlled way, but that is just a controlled phase on two qubits which is a Clifford gate.
    
    In order to perform the phase shift, a subtlety is that it is given in terms of $k_\nu$ rather than $\nu$.
    However, the conversion factor between $k_\nu$ and $\nu$ can simply be absorbed into $R_\ell$, rather than explicitly multiplied.
    This multiplication is most conveniently performed with signed integers.
    The product of each pair of components will then have a sign that is the product of the signs of the individual components.
    That product of signs can be computed with a CNOT.
    However, if we were to then add the components to compute $k_\nu\cdot R_\ell$, we would need to deal with the problem of adding signed integers again.
    Instead, it is convenient to add the products of each pair of components into the phase gradient state separately.
    They can be added or subtracted into the phase gradient state controlled on the sign with no more complexity than a regular addition.

    The maximum value of a component of $k_\nu$ is $2\pi/\Omega^{1/3}$ times $N^{1/3}-1$, which is the maximum value of a component of $\nu$, and the maximum value of a component of $R_\ell$ is $\Omega^{1/3}/2$, since $\Omega^{1/3}$ is the width of the spatial region.
    Therefore, a component of $k_\nu$ times $R_\ell$ can be rewritten as a component of $\nu$ times $2\pi R_\ell/\Omega^{1/3}$, which has a maximum value of $\pi$.
    There is one initial CNOT used to find the product of the signs of the component of $\nu$ and $R_\ell$.
    That will be used to control addition or subtraction of numbers into the phase gradient state.
    For the multiplication of the least significant bit of $\nu$ times $2\pi R_\ell/\Omega^{1/3}$, we have cost $n_R-1$ to compute the bit products, because we now exclude the sign bit.
    The most significant bit of that product being added into the phase gradient state corresponds to a phase shift of $\pi/2$, and the least significant bit corresponds to a phase shift of $2\pi/2^{n_R}$.
    As a result, the cost of adding these bits into the phase gradient state is $n_R-2$.
    For the next most significant bit of $\nu$, the most significant bit of the product would give phases of $\pi$, so we still need to compute the $n_R-1$ bit products, but since all phases are multiplied by 2, the cost of adding into the phase gradient state is reduced to $n_R-3$.
    After that, there is cost $n_R-2$ for the bit products, and $n_R-4$ for addition into the phase gradient state.
    For the $k$'th bit of $\nu$, we have cost $n_R-k+1$ for the bit products and $n_R-k-1$ for the addition for a total of $2(n_R-k)$.
    There is a subtlety in that if $n_R\le n_p$ then $n_R-k-1$ can be negative.
    Ignoring that case for the moment, because we would typically expect $n_R> n_p$, the Toffoli cost is
    \begin{equation}
        -1+\sum_{k=1}^{n_p} 2(n_R-k) = 2n_p n_R - n_p(n_p+1) -1 .
    \end{equation}
    In the case where $n_R=n_p$, then we get one more Toffoli, because we simply have $n_R-k$ instead of $n_R-k-1$ for the maximum value of $k$.
    Then the formula for the complexity can be given as $2n_p n_R - n_p(n_p+1)=n_R(n_R-1)$.
    In the case where $n_R<n_p$, then at $k=n_R+1$ we have both costs equal to zero, and so the sum can be taken to $k=n_R$, to give
    \begin{equation}
        \sum_{k=1}^{n_R} 2(n_R-k) = n_R(n_R-1) .
    \end{equation}
    Here we have removed the $-1$ at the beginning because the expression in the sum overcounts the Toffolis by 1 at $k=n_R$.
    This cost is multiplied by 3, because it is done once for each component.
    Therefore the Toffoli cost may be given as
    \begin{equation}\label{eq:Rphacost1}
        \begin{cases}
            3[2n_p n_R - n_p(n_p+1) -1], & n_R> n_p\\
            3n_R(n_R-1), & n_R\le n_p.
        \end{cases}
    \end{equation}
    For convenience of the final expression, we can give the Toffoli cost as just $6n_p n_R$, but the cost can be computed more accurately with the expression \eq{Rphacost1}.
    The maximum number of ancilla qubits used is $n_R-1$ for the bit products, plus $n_R-3$ for the addition into the phase gradient register, for a total of $2n_R-4$.
    
    The total Toffoli costs for the block encoding of the Hamiltonian are now as in \tab{qubcosts}.
    Here $b_r$ is a number of bits used for ancilla qubit rotations in preparations of equal superposition states, and can be taken to be $7$ for a probability of success no less than $0.999$.
    
\begin{table}
\begin{tabular}{| m{11cm} | m{5cm} |}
\hline
\centering Procedure &  \hspace{16mm} Toffoli cost \\
 \hline \hline
Preparing the superposition on the registers selecting between $T,U,V$, where $n_{\eta\zeta}$ is the number of bits needed to represent $\eta+\lambda_\zeta$; see \eq{TUVprepcost}. & $2(n_T+4n_{\eta\zeta}+2b_r-12)$ \\
\gline
Preparing equal superpositions over $\eta$ values of $i$ and $j$ in unary; see \eq{ijprepcost}. &
$14n_\eta+8b_r-36$ \\
\gline
The state preparation cost for the $w,r,s$ registers used for $T$, including the factor of 2 for inversion; see \eq{wrsprep}. & $2(2n_p+9)$ \\
\gline
Controlled swaps of the $p$ and $q$ registers into and out of ancillae (which is used for $T$, $U$, and $V$); see \eq{swapscost}. & $12\eta n_p+4\eta-8$ \\
\gline
The $\SEL$ cost for $T$; see \eq{tsel}. & $5(n_p-1)+2$ \\
\gline
The preparation of the state with amplitude $1/{\norm{\nu}}$ and inversion; see \eq{nuprepcost}. & $3n_p^2+15n_p-7+4n_{\mathcal{M}}\big(n_p+1\big)$\\
\gline
The QROM for $R_\ell$, combined with the state preparation with amplitudes $\sqrt{\zeta_\ell}$ (which may be implicit); see \eq{Rlqrcost}. & $\lambda_\zeta+\erase{\lambda_\zeta}$\\
\gline
Additions and subtractions of $\nu$ into the momentum registers; see \eq{addcosts}. & $24n_p$ \\
\gline
Phasing by $-e^{-\mi k_\nu\cdot R_\ell}$; see \eq{Rphacost1}. & $6n_p n_R$.\\
\hline
\end{tabular}
\caption{The costs involved in the block encoding of the Hamiltonian for qubitization.}
\label{tab:qubcosts}
\end{table}
    
    The complete step needed for the qubitization, to give the step that one would perform phase estimation on, needs a reflection on the control ancilla.
    The cost of this reflection corresponds to the numbers of qubits used in the state preparation.
    The total qubit cost, and the qubits needed to be reflected upon, are given in \app{qubcost}.
    There are a total of
    \begin{equation}\label{eq:refcost1}
        n_{\eta\zeta} + 2n_\eta + 6n_p + n_{\mathcal{M}} + 16
    \end{equation}
    qubits to reflect on.
    The corresponding Toffoli cost is two less.
    However, this reflection needs to be controlled on the qubits used as control for the phase measurement.
    Also, there is a single Toffoli cost for iterating those control qubits.
    We bundle those extra two Toffolis together with this cost, giving the cost as in \eq{refcost1}.

	\subsection{Total cost of constructing the qubitization operator}
	\label{sec:erroranly}
    This is a complete accounting for the cost of a single step of the block encoding, except that we need to estimate the error due to the finite values of $\cal{M}$ and $n_R$ in order to determine the appropriate values to use for these quantities.
    The finite value of $n_R$ corresponds to replacing $R_\ell$ with an $n_R$-digit approximation $\widetilde{R}_\ell$, so that the potential term $U$ that the algorithm actually implements is close to
	\begin{equation}
	U=-\frac{4\pi}{\Omega}\sum_{\ell=1}^{L}\sum_{j=1}^{\eta}\sum_{\substack{p,q\in G\\ p\neq q}}\bigg(\zeta_\ell\frac{e^{ik_{q-p}\cdot \tilde{R}_\ell}}{\norm{k_{p-q}}^2}\bigg)\ket{p}\bra{q}_j . \label{eq:Uapprox}
	\end{equation}
	By choosing $n_R$ sufficiently large, this error can be at most an appropriate fraction of the overall error budget $\epsilon$.
	To bound the error due to this source and choose the value of $n_R$, we use the following Lemma.
    
    	\begin{lemma}\label{lem:epsR}
		Given $\|\widetilde{R}_\ell - R_\ell\| \le \delta_R$ for all $\ell$, if $\lambda_j$ is the $j^{\rm th}$ eigenvalue of the true Hamiltonian and \eqref{eq:Uapprox} is used to approximate the nuclear potential, then there exists an eigenvalue of the approximated Hamiltonian $\widetilde{\lambda}_j$ such that
		\begin{equation}
		|\lambda_j - \widetilde{\lambda_j}| \le \frac{2\delta_R\eta\lambda_\zeta}{\Omega^{2/3}}\sum_{\nu\in G_0}\frac{1}{\norm{\nu}} .
		\end{equation}
	\end{lemma}
	\begin{proof}
		First, note from~\cite[Corollary 6.3.4]{horn2012matrix} that for any normal matrices $M$ and $M'$ that $\|M-M'\|\le \epsilon$ implies that for every eigenvalue $\lambda_j$ of $M$ that there exists an eigenvalue of $M'$ within distance $\epsilon$. Thus it suffices to bound the spectral norm of the difference between the true potential operator and the approximated potential.  Our claim then follows from the following bound on the difference between the two Hamiltonians.
		\begin{align}
		    \norm{U -\widetilde{U}}&=\norm{\sum_{\nu\in G_0}\sum_{\ell=1}^{L}\frac{4\pi\zeta_\ell}{\Omega\norm{k_\nu}^2}\sum_{j=1}^{\eta}\big(e^{-\mi k_\nu\cdot \tilde{R}_\ell}-e^{-\mi k_\nu\cdot R_\ell}\big)\bigg(\sum_{p\in G}\ket{p-\nu}\!\bra{p}_j\bigg)}\\
        	&\leq\eta\sum_{\nu\in G_0}\sum_{\ell=1}^{L}\frac{4\pi\zeta_\ell}{\Omega\norm{k_\nu}^2}\big|e^{-\mi k_\nu\cdot \tilde{R}_\ell}-e^{-\mi k_\nu\cdot R_\ell}\big|
        	\leq\eta\sum_{\nu\in G_0}\sum_{\ell=1}^{L}\frac{4\pi\zeta_\ell}{\Omega\norm{k_\nu}^2}\big|k_\nu\cdot (\tilde{R}_\ell-R_\ell)\big|\\
        	&\leq\eta\sum_{\nu\in G_0}\sum_{\ell=1}^{L}\frac{4\pi\zeta_\ell}{\Omega\norm{k_\nu}^2}\norm{k_\nu}\norm{\tilde{R}_\ell-R_\ell}
        	\leq\frac{2\delta_R\eta\lambda_\zeta}{\Omega^{2/3}}\sum_{\nu\in G_0}\frac{1}{\norm{\nu}}.
		\end{align}
		Note that $\sum_{\nu\in G_0}\frac{1}{\norm{\nu}}$ can be further upper bounded by a function of $N$. However, we use the current bound as it is tighter and more versatile for numerical implementation.
	\end{proof}

    Next we consider the error due to the finite value of $\cal{M}$.
    As can be seen from \eq{invprep}, the amplitude for each $\nu$ is proportional to
	\begin{equation}
	\sqrt{\frac{\ceil{\mathcal{M}(2^{\mu-2}/\norm{\nu})^2}}{\mathcal{M}2^{2\mu}2^{n_p+2}}},
	\end{equation}
	where $\mu$ is chosen such that $\nu\in B_\mu$.
	From the definition of $B_\mu$, this means that $\mu$ is the minimum integer such that $|\nu_x|$, $|\nu_y|$ and $|\nu_z|$ are all less than $2^{\mu-1}$.
	That is,
	\begin{equation}
	    \mu = \lfloor \log\max(|\nu_x|, |\nu_y|,|\nu_z|) \rfloor +2.
	\end{equation}
	As a result, the success probability of this state preparation is given by
	\begin{equation}\label{eq:psuc}
	\psuc =\sum_{\mu=2}^{n_p+1} \sum_{\nu\in B_\mu}\frac{\ceil{\mathcal{M}(2^{\mu-2}/\norm{\nu})^2}}{\mathcal{M}2^{2\mu}2^{n_\mu+1}}.
	\end{equation}
	As discussed in \cite[Section II D]{BabbushContinuum_b}, this is a constant around $1/4$ provided the ceiling function is not used.
	In the case where the ceiling function is used, then the probability of success is only increased.
		
	To bound the error due to the ceiling function, we can write each weight that is prepared as $1/\norm{\nu'}^2$, for a modified $\nu'$.
	Now it would appear that all coefficients are larger, because of the ceiling function.
	However, when accounting for the normalization of the state, that will be larger, so each coefficient is multiplied by a constant that is smaller than 1, resulting in approximate coefficients that may be larger or smaller than the exact coefficients.
	In fact, because we are choosing the amplitude on the potential terms, we have freedom to adjust it to minimise the error.
	Let us therefore call the constant factor $\alpha$.
	The squared amplitudes in the state we have prepared correspond to approximately
	\begin{equation}
	    \frac 1{2^{n_p+6}} \frac 1{\norm{\nu}^2}.
	\end{equation}
	Therefore, let us put the values of $\nu'$ as
	\begin{equation}
	    \frac 1{2^{n_p+6}}\frac 1{\norm{\nu'}^2} = \alpha \frac{\ceil{\mathcal{M}2^{\mu-2}/\norm{\nu})^2}}{\mathcal{M}2^{2\mu}2^{n_p+2}},
	\end{equation}
	so
	\begin{equation}
	    \frac 1{\norm{\nu'}^2} = 16\alpha \frac{\ceil{\mathcal{M}2^{2\mu}/(16\norm{\nu}^2)}}{\mathcal{M}2^{2\mu}}.
	\end{equation}
	Using this approach the error in the eigenvalues due to the approximation of the state preparation may be bounded as follows.
    	\begin{lemma}\label{lem:epsM}
    	Given that the state preparation for $\nu$ is performed as described above with a given $\mathcal{M}$,
		if $\lambda_j$ is the $j^{\rm th}$ eigenvalue of the true Hamiltonian then there exists an eigenvalue of the approximated Hamiltonian $\widetilde{\lambda}_j$ such that
		\begin{align*}
		|\lambda_j - \widetilde{\lambda}_j| &\le \frac {2\eta}{\pi\mathcal{M}\Omega^{1/3}}\left( \eta-1 + 2\lambda_\zeta\right) \left( 7 \times 2^{n_p+1} -9n_p -11 -3\times 2^{-n_p} \right).\\
		&=\frac{2^{n_p}(28\eta)\left( \eta-1 + 2\lambda_\zeta\right) }{\pi\mathcal{M}\Omega^{1/3}}\left(1+ \mathcal{O}(2^{-n_p}) \right)
		\end{align*}
	\end{lemma}
	
	\begin{proof}
	Again we use the fact that the norm of the difference in the operators can be used to upper bound the difference in the eigenvalues.
	The error in the approximation of $U$ can be bounded as
	\begin{align}\label{eq:Udif}
		    \norm{U -\widetilde{U}}&\le \norm{\sum_{\nu\in G_0}\sum_{\ell=1}^{L}\frac{\zeta_\ell}{\pi\Omega^{1/3}}\left( \frac 1{\norm{\nu}^2} - \frac 1{\norm{\nu'}^2} \right)\sum_{j=1}^{\eta}e^{-\mi k_\nu\cdot R_\ell}\bigg(\sum_{p\in G} \ket{p-\nu}\bra{p}_j\bigg)} \nn
		    &\le \sum_{\nu\in G_0}\frac{\eta\lambda_\zeta}{\pi\Omega^{1/3}}\left| \frac 1{\norm{\nu'}^2} - \frac 1{\norm{\nu}^2} \right| \nn
		    &= \frac {\eta\lambda_\zeta}{\pi\Omega^{1/3}}\sum_{\mu=2}^{n_p+1}\sum_{\nu\in B_\mu}\left| \frac 1{\norm{\nu'}^2} - \frac 1{\norm{\nu}^2} \right|.
	\end{align}
		Similarly, the error in the approximation of $V$ can be bounded as
    \begin{align}
    \norm{V -\widetilde{V}}&\le
        \norm{\sum_{\nu\in G_0}\frac{1}{2\pi\Omega^{1/3}}\left(\frac 1{\norm{\nu}^2} - \frac 1{\norm{\nu'}^2}\right)\sum_{\substack{i,j=1\\i\neq j}}^{\eta}\left(\sum_{p,q\in G}\ket{p+\nu}\bra{p}_i\cdot\ket{q-\nu}\bra{q}_j\right)}\nn
        &\le \sum_{\nu\in G_0}\frac{\eta(\eta-1)}{2\pi\Omega^{1/3}}\left| \frac 1{\norm{\nu}^2} - \frac 1{\norm{\nu'}^2} \right| \nn
        &= \frac {\eta(\eta-1)}{2\pi\Omega^{1/3}}\sum_{\mu=2}^{n_p+1}\sum_{\nu\in B_\mu}\left| \frac 1{\norm{\nu'}^2} - \frac 1{\norm{\nu}^2} \right|.
    \end{align}
    Thus the total error in the potential component of the Hamiltonian can be upper bounded as
    \begin{equation}
        \norm{U+V -(\widetilde{U}+\widetilde{V})} \le \frac {\eta}{2\pi\Omega^{1/3}}\left( \eta-1 + 2\lambda_\zeta\right) \sum_{\mu=2}^{n_\mu}\sum_{\nu\in B_\mu}\left| \frac 1{\norm{\nu'}^2} - \frac 1{\norm{\nu}^2} \right|.
    \end{equation}
    For given parameters, it is possible to bound the error by numerically performing this sum.
    To give an analytic result, we take $\alpha=1$ and place an upper bound on the sum.
    We obtain
    \begin{align}
        \sum_{\mu=2}^{n_\mu}\sum_{\nu\in B_\mu}\left| \frac 1{\norm{\nu'}^2} - \frac 1{\norm{\nu}^2} \right| &\le \sum_{\mu=2}^{n_p+1}\sum_{\nu\in B_\mu}\frac{16}{\mathcal{M}2^{2\mu}} \nn
        &= \sum_{\mu=2}^{n_p+1} \left[ (2^\mu-1)^3 - (2^{\mu-1}-1)^3 \right]\frac{16}{\mathcal{M}2^{2\mu}} \nn
        &= \frac 4{\mathcal{M}}\left( 7 \times 2^{n_p+1} -9n_p -11 -3\times 2^{-n_p} \right).
    \end{align}
    Thus the total error in the approximation of the Hamiltonian due to finite $\mathcal{M}$ can be upper bounded by
    \begin{equation}
        \norm{U+V -(\widetilde{U}+\widetilde{V})} \le \frac {2\eta}{\pi\mathcal{M}\Omega^{1/3}}\left( \eta-1 + 2\sum_{\ell=1}^L \zeta_\ell\right) \left( 7 \times 2^{n_p+1} -9n_p -11 -3\times 2^{-n_p} \right).
    \end{equation}
    This expression then provides the upper bound on the difference in eigenvalues as required.
    \end{proof}
    
    As a result, if we are allowing an error due to the state preparation of $\epsilon_\mathcal{M}$, then the value of $\mathcal{M}$ that should be chosen is then about
	\begin{equation}\label{eq:Mval}
	    \mathcal{M} \approx  \frac{28\eta 2^{n_p}}{\epsilon_\mathcal{M}\pi\Omega^{1/3}} \left( \eta + 2 \lambda_\zeta \right).
	\end{equation}
	Note that $\mathcal{M}$ is required to be a power of 2, so $2^{n_{\mathcal{M}}}$.
	In numerical testing, we find that computing the actual sum gives a result about $50\%$ of the upper bound.
	This is what would be expected, because the ceiling will add $1/2$ on average, which is $50\%$ of the upper bound of 1.
	If $\alpha\ne 1$ is used, then the error can be reduced further to about $30\%$ of the upper bound, using $\alpha$ in the range $1-3/2\mathcal{M}$ to $1-1/\mathcal{M}$.
	
	The errors due to approximation of $\nu$ and approximation of $R_\ell$ can be upper bounded by the sum of the two sources of error.
	Calling the $U$ with both approximations $\overline{U}$, and that with just $R_\ell$ approximated $\widetilde{U}$, we have $\| \overline{U} - U \|\le \| \overline{U} - \widetilde{U} \|+\| \widetilde{U} - U \|$.
	Here $\| \widetilde{U} - U \|$ is upper bounded as in \lem{epsR}.
	For $\| \overline{U} - \widetilde{U} \|$ exactly the same reasoning as in \eq{Udif} holds, except with $R_\ell$ replaced with $\widetilde{R}_\ell$.
	This means that the same error bound holds and so the bounds on the error due to the two sources can be added.
	
	The other main source of error is the error in the phase estimation.
	With $m$ steps of phase estimation, the root-mean-square error in the estimated energy due to the phase estimation is $\pi\lambda/(2m)$.
	It is possible to use $m$ that is not a power of 2 simply by using unary iteration on the control register for the phase estimation.
	The effective value of $\lambda$ will vary dependent on the relative sizes of the $T$, $U$, and $V$ operators.
	
	As discussed above, if $\psuc\lambda_T = (1-\psuc)(\lambda_U+\lambda_V)$, then in the case of the failure of the inequality test for the preparation of $\nu$, we apply the $\SEL$ operation for $T$, and the net value of $\lambda$ is $\lambda_T+\lambda_U+\lambda_V$.
	If $\psuc\lambda_T > (1-\psuc)(\lambda_U+\lambda_V)$, we apply $\SEL$ for $T$ in the case of failure of the inequality test, \emph{or} $\ket{0}$ on an ancilla qubit, and again $\lambda=\lambda_T+\lambda_U+\lambda_V$.
	In the case where $\psuc\lambda_T < (1-\psuc)(\lambda_U+\lambda_V)$, we apply $\SEL$ for $T$ in the case of failure of the inequality test, \emph{and} $\ket{0}$ on an ancilla qubit.
	In that case, the value of $\lambda$ is $(\lambda_U+\lambda_V)/\psuc$.
	To account for the three cases, we have
	\begin{equation}
	    \lambda = \max\left[ \lambda_T+ \lambda_U+\lambda_V, (\lambda_U+\lambda_V)/\psuc \right].
	\end{equation}
	We can alternatively perform amplitude amplification in the state preparation for $\ket\nu$, which triples that cost, but boosts the probability of success to
	\begin{equation}\label{eq:psucamp}
	    \psuc^{\rm amp} = \sin^2 \left( 3\arcsin\left( \sqrt{\psuc} \right) \right),
	\end{equation}
	which is close to 1.
	As a result the value of $\lambda$ can be given as
	\begin{equation}
	    \lambda = \max\left[ \lambda_T+ \lambda_U+\lambda_V, (\lambda_U+\lambda_V)/\psuc^{\rm amp} \right].
	\end{equation}
	
	A further subtlety is that we are testing $i\ne j$ in preparing the equal superposition over $i$ and $j$ for $V$.
	In that case, if we apply $T$ in the case where this test fails, then the effective $\lambda$ becomes
	\begin{equation}
	    \lambda = \max\left[ \lambda_T+ \lambda_U+\lambda_V, [\lambda_U+\lambda_V/(1-1/\eta)]/\psuc \right].
	\end{equation}
    The reason for this expression is that the preparation over all values of $i$ and $j$ gives an effective $\lambda_V$ which is proportional to $\eta^2$ rather than $\eta(\eta-1)$.
    This means the effective $\lambda_V$ is $\lambda_V\eta^2/[\eta(\eta-1)]=\lambda_V/(1-1/\eta)$.
	Then, in the case of amplitude amplification on $\ket\nu$, the effective $\lambda$ would be
	\begin{equation}
	    \lambda = \max\left[ \lambda_T+ \lambda_U+\lambda_V, [\lambda_U+\lambda_V/(1-1/\eta)]/\psuc^{\rm amp} \right].
	\end{equation}
	See \app{selTUV} for a more in-depth derivation of this expression.
	
	Another adjustment to the value of $\lambda$ is that needed to account for the probabilities of failure in the preparation of the equal superposition states for $w$, $i$ and $j$.
	These equal superpositions are required for both parts of the Hamiltonian (kinetic and potential), and so the value of $\lambda$ is divided by $P_{\rm eq}=\eqprep{3,8}\eqprep{\eta+2\lambda_\zeta,b_r}\eqprep{\eta,b_r}^2$.
	That gives the resulting values of $\lambda$ without amplitude amplification
	\begin{equation}\label{eq:lamnoaa}
	    \lambda = \max\left[ \lambda_T+ \lambda_U+\lambda_V, [\lambda_U+\lambda_V/(1-1/\eta)]/\psuc \right]/P_{\rm eq} ,
	\end{equation}
	and with amplitude amplification
	\begin{equation}\label{eq:lamaa}
	    \lambda = \max\left[ \lambda_T+ \lambda_U+\lambda_V, [\lambda_U+\lambda_V/(1-1/\eta)]/\psuc^{\rm amp} \right]/P_{\rm eq}.
	\end{equation}
	Lastly we note that the $\lambda$-values are modified by our block encoding method.
	We use $\lambda'_T$ as given by \eq{lamT}, because of the state preparation used for $T$.
	Moreover, because we use an imperfect state preparation for $\nu$, the values of $\lambda_U,\lambda_V$ are modified.
	We define the modified $\lambda_\nu$ value
	\begin{equation}
	    \lambda_\nu^\alpha := \alpha \sum_{\nu\in G_0} \frac{\ceil{\mathcal{M}(2^{\mu-2}/\norm{\nu})^2}}{\mathcal{M}2^{2\mu-4}}.
	\end{equation}
	Then the $\lambda$-values for $U$ and $V$ are scaled in this way to give
	\begin{equation}
	    \lambda_U^\alpha := \lambda_\nu^\alpha\lambda_U / \lambda_\nu, \qquad \lambda_V^\alpha := \lambda_\nu^\alpha\lambda_V/\lambda_\nu.
	\end{equation}
	
	As a result of these approaches, we have the following theorem for the complexity.
	
	\begin{theorem}
	    \label{thm:qubitization}
	    The eigenvalue of the Hamiltonian given by \eq{first_quant_ham} can be estimated to within RMS error $\epsilon$ using a number of Toffoli gates
	    \begin{align}\label{eq:totqubcost}
	        &\left\lceil\frac{\pi\lambda}{2\epsilon_{\rm pha}}\right\rceil \big\{ 2(n_T+4n_{\eta\zeta}+2b_r-12)  % Cost of preparing the superposition on the qubits selecting T,U,V.
	        +14n_\eta+8b_r-36 % Preparing the superposition over i and j.
	        +\amam [3n_p^2+15 n_p -7 + 4n_{\mathcal{M}}(n_p+1)] % Cost of preparing $1/\|\nu\|$ state.
	        \nn & \quad
	        +\lambda_\zeta + \erase{\lambda_\zeta} % The cost of outputting the values of $R_\ell$.
	        +2(2n_p+2b_r-7) % Cost of preparing w,r,s registers.
	        +12\eta n_p % Controlled swaps of p and q registers.
	        +5(n_p-1)+2 % The SEL cost for T.
	        +24 n_p % The additions and subtractions of $\nu$ into momentum registers.
	        +6n_p n_R % The cost of phasing by $R_\ell$.
	        +18 % The extra Toffolis for producing the two qubits selecting between T,U,V and indicating success of the preparations.
	        \nn
	        & \quad + n_{\eta\zeta}+2n_\eta+6n_p+n_{\mathcal{M}}+16 % The cost of the reflection in the qubitization.
	        + \mathcal{O}(\log(1/\epsilon))
	        \big\}
	    \end{align}
	    where ``$\amam$'' is equal to 3 if we perform amplitude amplification and 1 otherwise.
	    For this complexity we require that a phase gradient state and a $\ket{T}$ state are given as a resource.
	    The value of $\lambda$ is given by
	\begin{equation}
	    \lambda = \max\left[ \lambda'_T+ \lambda^1_U+\lambda^1_V, [\lambda^1_U+\lambda^1_V/(1-1/\eta)]/\psuc^{\rm amp} \right]/P_{\rm eq}
	\end{equation}
	if we perform amplitude amplification, and
	\begin{equation}
	    \lambda = \max\left[ \lambda'_T+ \lambda^1_U+\lambda^1_V, [\lambda^1_U+\lambda^1_V/(1-1/\eta)]/\psuc \right]/P_{\rm eq}
	\end{equation}
	otherwise, where
	\begin{align}
	\psuc &=\sum_{\mu=2}^{n_p+1} \sum_{\nu\in B_\mu}\frac{\ceil{\mathcal{M}(2^{\mu-2}/\norm{\nu})^2}}{\mathcal{M}2^{2\mu}2^{n_\mu+1}}, \\
	\psuc^{\rm amp} &= \sin^2 \left( 3\arcsin\left( \sqrt{\psuc} \right) \right), \\
	P_{\rm eq}&=\eqprep{3,8}\eqprep{\eta+2\lambda_\zeta,b_r}\eqprep{\eta,b_r}^2.
	\end{align}
	The value of $\epsilon_{\rm pha}$ is the allowable RMS error in phase estimation.
	We require that $\epsilon_{\rm pha}>0$, $n_R,n_{\mathcal{M}},n_T\in\mathbb{N}$ are chosen such that
	\begin{align}\label{eq:phaer}
	    \epsilon^2 &\ge \epsilon_{\rm pha}^2+(\epsilon_{\mathcal{M}}+\epsilon_R+\epsilon_T)^2  \\
	    \label{eq:Mer}
	    \epsilon_\mathcal{M} &= \frac{2\eta}{2^{n_\mathcal{M}}\pi\Omega^{1/3}} \left( \eta-1 + 2 \lambda_\zeta \right)\left( 7 \times 2^{n_p+1} -9n_p -11 -3\times 2^{-n_p} \right)  \\
	    \label{eq:Rer}
	    \epsilon_R &=\frac{\eta\lambda_\zeta}{2^{n_R}\Omega^{1/3}}\sum_{\nu\in G_0}\frac{1}{\norm{\nu}}  \\
	    \epsilon_T &=\frac{\pi\lambda}{2^{n_T}} .
	\end{align}
	\end{theorem}

    \begin{proof}
        The complexity given is the sum of the costs for the block encoding given in the list in \tab{qubcosts},
        plus the number of qubits that need to be reflected upon given in \eq{refcost1}.
        The term $18$ is to account for the Toffolis used to produce the qubits selecting between $T,U,V$, and the one flagging success of the state preparation.
        That takes 9 Toffolis both in the preparation and the inverse preparation.
        The term $\mathcal{O}(\log(1/\epsilon)$ is to account for complexity of preparing the initial state to use as control for the phase measurement, and the processing of this control at the end to obtain the phase estimate.
        
        The appropriate values of $\lambda$ in the cases with and without amplitude amplification were given above in \eq{lamaa} and \eq{lamnoaa}.
        In those expressions we replace $\lambda_T,\lambda_U,\lambda_V$ with $\lambda'_T,\lambda^1_U,\lambda^1_V$ to account for our block encoding, where we have taken $\alpha=1$ in order to use \lem{epsM}.
        The error due to the phase estimation is independent from all other sources of error, so we can upper bound the RMS error due to all sources as the sum of squares of the error due to the phase estimation and the error from the other sources.
        The bound $\epsilon_\mathcal{M}$ on the error due to finite $\mathcal{M}$ is obtained from \eq{Mval}, which is based on \lem{epsM}.
        The bound $\epsilon_R$ on the error due to the nuclear positions is obtained from \lem{epsR}.
        In that expression we have substituted $\delta_R=\Omega^{1/3}/2^{n_R+1}$, which comes from the cell length of a side being $\Omega^{1/3}$, and the maximum error in the position being half of $1/2^{n_R}$ times the length of the region.
        The bound $\epsilon_T$ on the error due to the rotation of the qubit selecting between $T$ and $U+V$ is given in \eq{selTUV}.
        
        The phase gradient state and $\ket{T}$ state are used catalytically throughout the procedure, with the phase gradient state being used to implement phase rotations and the $\ket{T}$ state used to implement controlled Hadamards.
        These states can be prepared at the beginning of the algorithm with negligible complexity compared to the rest of the algorithm.
    \end{proof}
    
    In order to choose parameters for the algorithm, we first take some maximum allowable values for $\epsilon_\mathcal{M}$, $\epsilon_R$ and $\epsilon_T$ as, for example, $\epsilon/10$.
    Then we substitute those values into \eq{Mer}, \eq{Rer}, and \eq{selTUV} to find minimum values of $n_\mathcal{M}$, $n_R$ and $n_T$.
    We can then use those values on the right-hand side (RHS) of \eq{Mer} and \eq{Rer} to compute updated (smaller) values of $\epsilon_\mathcal{M}$ and $\epsilon_R$.
    Those values can then be substituted into \eq{phaer} to solve for the maximum $\epsilon_{\rm pha}$ that satisfies the inequality.
    These values of the parameters can then be substituted into \eq{totqubcost} to determine the total Toffoli cost.   In order to optimise the parameters, instead of using just the values of $n_\mathcal{M}$, $n_R$ and $n_T$ as given by this method, we can compute the cost for an array of values of $n_\mathcal{M}$, $n_R$ and $n_T$, for example testing 4 below and 4 above the numbers of bits given by this approach. We can then select from this array the parameters that give the minimum Toffoli complexity.
    
    Another improvement is to compute the exact bound on the state preparation error, rather than just using the bound given in \lem{epsM}.
    That is, we use
    \begin{equation}
        \epsilon_{\mathcal{M}} = \frac {\eta}{2\pi\Omega^{1/3}}\left( \eta-1 + 2\lambda_\zeta\right) \sum_{\mu=2}^{n_\mu}\sum_{\nu\in B_\mu}\left| \frac 1{\norm{\nu'}^2} - \frac 1{\norm{\nu}^2} \right|,
    \end{equation}
    with
	\begin{equation}
	    \frac 1{\norm{\nu'}^2} = \frac{16\alpha}{\mathcal{M}2^{2\mu}} \left\lceil{\frac{\mathcal{M}2^{2\mu}}{16\norm{\nu}^2}}\right\rceil .
	\end{equation}
	Choosing a finite number of bits for $n_T$ corresponds to giving a value of $\alpha\ne 1$.
	Since the lowest $\epsilon_{\mathcal{M}}$ is obtained in the region $\alpha=1-3/2\mathcal{M}$ to $1-1/\mathcal{M}$, we can choose $n_T$ in order to give $\alpha$ in that range and recompute $\epsilon_{\mathcal{M}}$.
	The value of $n_T$ needed to obtain $\alpha$ in that range is numerically found to be typically around $n_\mathcal{M}$ plus 4 to 8.
	With this choice of $n_T$ and $\alpha$, $\epsilon_{\mathcal{M}}$ is about 30\% of the bound proven in \lem{epsM}, and there is no additional error due to finite $n_T$.

	%%%%%%%%%%%%%%%%%%%%%%%%%%%%%%%%%%%

	\section{The interaction picture based algorithm}
	\label{sec:interaction_picture}

In the interaction picture approach, one has a time-independent Hamiltonian $H=A+B$, where the norm of $B$ is much smaller than that of $A$.
The time evolution under $A+B$ for time $\tau$ can then be given as (with $\mathcal{T}$ the time ordering operator)
\begin{align}
e^{-\mi(A+B)\tau} &= e^{-\mi\tau A} \mathcal{T} \exp \left( -\mi \int_0^\tau ds \, e^{\mi sA} B e^{-\mi sA}\right) \\
&= \sum_{k=0}^\infty \frac{(-\mi)^k}{k!} \int_{0}^\tau d\tau_1 \int_{0}^{\tau} d\tau_2 \cdots \int_{0}^\tau d\tau_k \,  e^{-\mi(\tau-\tau'_k)A} B e^{-\mi(\tau'_k-\tau'_{k-1})A} B \ldots B e^{-\mi(\tau'_2-\tau'_1)A} B e^{-\mi\tau_1 A} \nn
&= \lim_{\substack{K\rightarrow\infty\\ M\rightarrow \infty}}\sum_{k=0}^K \frac{(-\mi\tau)^k}{M^k k!} \sum_{m_1=0}^{M-1} \sum_{m_2=0}^{M-1} \cdots \sum_{m_k=0}^{M-1} e^{-\mi\tau(M-1/2-m'_k)A/M} B e^{-\mi\tau(m'_k-m'_{k-1})A/M} B \ldots \nn& \quad \ldots B e^{-\mi\tau(m'_2-m'_1)A/M} B e^{-\mi\tau(m'_1+1/2) A/M} \nonumber
\end{align}
where $\tau'_1,\ldots,\tau'_k$ are sorted times from $\tau_1,\ldots,\tau_k$, and $m'_1,\ldots,m'_k$ are sorted integers from $m_1,\ldots,m_k$.
The approximation comes from both truncating the sum at $K$ and discretising the integrals.
The discretization of the integral here is taken by using $\tau_j = (m_j+1/2)/M$, so the midpoints of the intervals are used.
This does not affect the final expression, except that there are shifts by 1/2 on the beginning and ending rotations.

Note that we do not immediately get a linear combination of unitaries from the above Dyson-series expansion, as the operator $B$ from the Hamiltonian is not unitary in general. Instead, we now describe the Hamiltonian input model for the interaction picture algorithm in which $B$ is specified. We assume that $B$ is given by a block encoding as
\begin{equation}
    \bra{0}\PREP_B^\dagger\cdot \SEL_B\cdot\PREP_B\ket{0}=\frac{B}{\lambda_B} \, .
\end{equation}
Here, $\PREP_B$ can be viewed as some state preparation procedure acting nontrivially only on an ancilla register initialized in state $\ket{0}$, and $\SEL_B$ is a selection unitary acting jointly on the ancilla register and the register of the target system. By preparing the ancilla state, performing the selection operation, and unpreparing the state, we can implement operator $B$ on the target system with normalization factor $\lambda_B\geq\norm{B}$. We describe a quantum circuit in \sec{sel_prep} that block-encodes $B$ for the first quantized quantum chemistry simulation.

For the operator $A$, we assume that the matrix exponential $e^{-\mi sA}$ can be efficiently implemented for an arbitrary real number $s$. This is the case for instance when $A$ is already diagonalized and contains pairwise commuting terms, each of which can be directly exponentiated. In particular, this holds for the kinetic operator $T$ when the first quantized Hamiltonian is represented in the plane wave basis. See \sec{kinetic_exp} for a detailed discussion of the circuit implementation.

Using the block encoding of $B$, we now rewrite the Dyson-series expansion as
\begin{align}
e^{-\mi(A+B)\tau} &\approx\sum_{k=0}^K \frac{(-\mi\lambda_B\tau)^k}{M^k k!}\bra{0}^{\otimes k}\PREP_B^{\dagger\otimes k} \sum_{m_1,...,m_k=0}^{M-1} \bigg(e^{-\mi\tau(M-1/2-m'_k)A/M} \SEL_B e^{-\mi\tau(m'_k-m'_{k-1})A/M} \SEL_B \ldots \nn& \quad \ldots \SEL_B e^{-\mi\tau(m'_2-m'_1)A/M} \SEL_B e^{-\mi\tau(m'_1+1/2) A/M}\bigg)\PREP_B^{\otimes k}\ket{0}^{\otimes k}\\
&=\left(\bra{0}\PREP_B^\dagger\right)^{\otimes K}\!\sum_{k=0}^K \frac{(-\mi\lambda_B\tau)^k}{M^k k!}\!\!\! \sum_{m_1,...,m_k=0}^{M-1} \! \bigg(e^{-\mi\tau(M-1/2-m'_k)A/M} \SEL_B e^{-\mi\tau(m'_k-m'_{k-1})A/M} \SEL_B \ldots \nn& \quad \ldots \SEL_B e^{-\mi\tau(m'_2-m'_1)A/M} \SEL_B e^{-\mi\tau(m'_1+1/2) A/M}\bigg)\left(\PREP_B\ket{0}\right)^{\otimes K}.\nonumber
\end{align}
Here for each fixed value of $k$, the $k$ selection operators $\SEL_B$ act on $k$ different ancilla registers and the same target register; the last equality follows from the fact that the remaining $K-k$ ancilla states $\PREP_B\ket{0}$ cancel with $\bra{0}\PREP_B^\dagger$ in pairs and so the number of their occurrences can be made independent of $k$. Note that the resulting expression is expressed as a linear combination of unitaries and can thus in principle be implemented using quantum circuits.

Although we have written this as performing all the preparations at the beginning and inverting them at the end, in practice we will perform the preparations and inverse preparations for each $\SEL_B$ before the next preparation.
That avoids needing to use all qubits for all preparations at once, which would have a much higher qubit cost.
A further simplification is that we can separate out $e^{-\mi\tau A/2M}$ at the beginning and the end.
Then $e^{-\mi\tau(M-1/2-m'_k)A/M}$ is replaced with $e^{-\mi\tau(M-1-m'_k)A/M}$.
The advantage here is that $M-1-m'_k$ can be computed simply by performing NOT gates on all qubits, with zero Toffoli cost.

To implement this linear combination of unitaries, we need to prepare quantum states encoding the coefficients from the combination. For readability, we describe a naive circuit implementation here, and leave a detailed discussion of the subtleties and improvements in subsequent subsections. We first prepare
\begin{equation}
    \frac{1}{\sqrt{\beta}}\sum_{k=0}^{K}\sqrt{\frac{(-\mi\lambda_B\tau)^k}{k!}}\ket{k},
\end{equation}
where $\beta$ is the normalization factor and $\ket{k}$ is encoded in unary as $\ket{1^k0^{K-k}}$. This state may be prepared using the naive approach with one rotation gate and $K$ controlled rotation gates. However, we propose an improved approach based on the inequality test which can significantly reduce the implementation cost; see \sec{dyscost} for details.

Next, we prepare the correct time indices used in the interaction picture simulation. One possible method is to prepare an equal superposition over $k$ time indices $m_1,...,m_k$ and quantumly sort them \cite{Kieferova18}, giving
\begin{equation}
    \frac{1}{\sqrt{\beta}}\sum_{k=0}^{K}\sqrt{\frac{(-\mi\lambda_B\tau)^k}{M^kk!}}\ket{k}\sum_{m_1,...,m_k=0}^{M-1}\ket{m_1',...,m_k'},
\end{equation}
where again $0\leq m_1'\leq...\leq m_k'\leq M-1$ are sorted integers from $m_1,..,m_k$. We then compute the time differences as
\begin{equation}
    \frac{1}{\sqrt{\beta}}\sum_{k=0}^{K}\sqrt{\frac{(-\mi\lambda_B\tau)^k}{M^kk!}}\ket{k}\sum_{m_1,...,m_k=0}^{M-1}\ket{m_1',m_2'-m_1',...,m_k'-m_{k-1}',M-1-m_k'}.
\end{equation}
These difference values are further used to control the matrix exponentials in the Dyson-series expansion. Specifically, we have the decomposition
\begin{equation}
    \sum_{m=0}^{M-1}\ket{m}\bra{m}\otimes e^{-\mi\tau mA/M}
    =\prod_{j=0}^{\log (M)}\left(\ket{0}\bra{0}_j\otimes I+\ket{1}\bra{1}_j\otimes e^{-\mi\tau 2^jA/M}\right),
\end{equation}
where each controlled exponential of $A$ can be directly implemented.
In practice, in the implementation it will convenient to apply the controlled exponential jointly controlled by all bits of $m$, rather than controlling on the bits of $m$ independently.
On the other hand, we apply the selection unitaries $\SEL_B$ controlled by the unary-encoded $k$ values. These operations, together with the preparation and unpreparation procedures, implement the Dyson-series expansion up to the normalization factor $\beta$. Beyond the sorting method, previous work also suggested alternative methods based on a ``compressed encoding'' \cite{Kieferova18} and ``inequality tests'' \cite{Low2018}. We compare these methods and justify our choice of implementation in \sec{choice}.

For $K$ sufficiently large, the normalization factor $\beta$ is approximately given by
\begin{equation}
    \beta=\sum_{k=0}^{K}\frac{(\lambda_B\tau)^k}{k!}
    \approx \exp(\lambda_B \tau).
\end{equation}
In the original interaction picture algorithm, we choose $\tau\approx \ln(2)/\lambda_B$ so that $\beta\approx 1/2$. Then, a single step of oblivious amplitude amplification boosts the success amplitude close to $1$. The ideal evolution can be simulated for time $t$ by repeating this $\lambda_B t/\ln(2)$ times. However, for performing phase estimation, we develop a new approach by applying qubitization on the pre-amplified interaction picture subroutine, which significantly reduces the implementation cost. We discuss our approach in detail in \sec{choice}.

\subsection{Implementation choices for the interaction picture algorithm}
\label{sec:choice}
There are two main choices for how to provide the ordered integers $m'_j$.
One is to generate a sequence of differences according to a negative exponential weighting.
Then one directly has the differences $m'_{j+1}-m'_j$, but it is necessary to add these differences together to determine which difference gives a total over $M$.
That is used to determine $k$, rather than generating a register for $k$ directly.
This method is called the ``compressed encoding'' in \cite{Kieferova18}.
The preparation of the state with the negative exponential weighting can be performed as described in \cite{BerryQIC14} with a sequence of controlled rotations.
That results in state preparation complexity corresponding to a factor of $\log(1/\epsilon)$ times the number of qubits in the time registers.

The other choice for preparing the state with the ordered integers $m'_j$ is to prepare a set of unordered integers, then sort them. The type of sort needed is a quantum version of a sorting network, as described in \cite{Qsort1,Qsort2}.
There are algorithms for generating sorting networks, but in practice for small numbers (as we have here) it is best to choose optimized sorting networks.
The number of comparator operations required for the best known sorting networks for up to 16 items are as shown in \tab{sorttab} \cite{SortingNetworks}.
Each comparator needed corresponds to an inequality test and a controlled swap.
The inequality test uses approximately a Toffoli for each qubit of the registers, and the controlled swap uses a Toffoli for each qubit of the registers.
The number of Toffolis therefore corresponds to twice the number of comparators multiplied by the quits used for each time register.
For the case of $n=16$, that is $7.5$ times the number of qubits in the time registers, as compared to $\log(1/\epsilon)$ for the compressed encoding approach.
This will be a smaller multiplying factor, which is why we will use the sorting approach instead of the compressed encoding.

\begin{table}[tbh]
    \begin{tabular}{|c|c|c|c|c|c|c|c|c|c|c|c|c|c|c|c|}\hline
        $K$ & 2 & 3 & 4 & 5 & 6 & 7 & 8 & 9 & 10 & 11 & 12 & 13 & 14 & 15 & 16 \\ \hline
        comparators & 1 & 3 & 5 & 9 & 12 & 16 & 19 & 25 & 29 & 35 & 39 & 45 & 51 & 56 & 60 \\
         \hline
    \end{tabular}
    \caption{The number of comparators needed for a quantum sort of $K$ items.
    Each comparator involves an inequality test and a controlled swap.
    The number of comparators for a given $K$ is denoted $\srt{K}$.}
    \label{tab:sorttab}
\end{table}

For the implementation of $e^{-\mi(A+B)\tau}$ as a linear combination of unitaries, the general method is to prepare a superposition state with weights that are the square roots of the amplitudes in the linear combination of unitaries, then apply controlled unitary operations.
Oblivious amplitude amplification is used in order to implement the step deterministically, and in order to only need a single step of oblivious amplitude amplification it is necessary to choose the parameters so the amplitude of success is $1/2$.

In the case where the sorting is performed, the sum of the amplitudes needed for the approximation of $e^{-\mi(A+B)\tau}$ as a linear combination of unitaries is
\begin{equation}
    \sum_{k=0}^K \frac{\tau^k}{k!} \lambda_B^k \le \exp(\tau\lambda_B).
\end{equation}
An inequality is written here due to the truncation at finite $K$, but the two sides of this equation will be very close (with error corresponding to error in the evolution), so it is sufficient to assume that they are equal for the purpose of costing the algorithm.
The amplitude for success in the linear combination of unitaries is the inverse of this value, so in order to obtain amplitude $1/2$, we need $\exp(\tau\lambda_B)\approx 2$, which corresponds to choosing $\tau \approx \ln 2/\lambda_B$. In order to perform an accurate phase estimation we need long time evolution, so greater efficiency is obtained by the time being broken into intervals $\tau$ that are as large as possible.

An alternative approach is given by \cite{Low2018}, where instead of performing a sort, inequality tests are performed in the time registers.
When performing the inequality test, the factor of $k!$ is removed in the denominator, because that corresponds to the number of permutations of times.
As a result, the sum of amplitudes in the method of \cite{Low2018} is
\begin{equation}
    \sum_{k=0}^K \tau^k \lambda_B^k \le \frac 1{1-\tau\lambda_B}.
\end{equation}
Again the two sides of the expression are approximately equal.
For the amplitude amplification, we need this amplitude to be $2$ (because the amplitude for success is the inverse of the sum of amplitudes).
That means that in the approach of \cite{Low2018} one needs $\tau\approx 1/(2\lambda_B)$.
This value of $\tau$ is about $72\%$ of the value of $\tau$ in the method of \cite{Kieferova18}, which means that the overall complexity is about $39\%$ larger.
The method of \cite{Low2018} saves the complexity of the sort, but that complexity is trivial as compared to other complexities in the algorithm.

So far we have been assuming that phase estimation is performed on the time evolution, which requires performing amplitude amplification.  That triples the cost, because even a single step of oblivious amplitude amplification requires a forward step, a reverse step, then a forward step again.
An alternative method can be obtained by simply performing phase estimation on the single step of the evolution, as per the discussion in \cite{Berry2018,Poulin2017}.
The difference here is that we are block encoding the evolution $e^{-\mi(A+B)\tau}$ rather than the Hamiltonian, and the evolution is not Hermitian so the analysis does not hold.
However, we can easily block encode $\sin[(A+B)\tau]=\left[e^{\mi(A+B)\tau}-e^{-\mi(A+B)\tau}\right]/(2\mi)$.
All we need to do is use a single qubit that chooses between the block encoding of $e^{\mi(A+B)\tau}$ and $e^{-\mi(A+B)\tau}$.
All that qubit needs to do is control the sign of the exponential in terms like $e^{-\mi\tau(m'_2-m'_1)A/M}$, which can be achieved with negligible complexity.

If the operation $\sin[(A+B)\tau]$ has eigenvalues $\mu_j=\sin(E_j\tau)$, where $E_j$ are eigenvalues of $A+B$, then as discussed in \cite{Berry2018} the qubitization procedure gives an operator with eigenvalues $\pm e^{ \pm\mi\arcsin(\mu_j/\lambda)}$, where $\lambda$ is the sum of the amplitudes in the block encoding of $\sin[(A+B)\tau]$. (In \cite{Berry2018} $H$ is equivalent to $\sin[(A+B)\tau]$ here, and $E_k$ there is equivalent to $\mu_j$ here.)
Now the value of $\lambda$ in the block encoding of $\sin[(A+B)\tau]$ is the same as that of $e^{-\mi(A+B)\tau}$, and so with the method of \cite{Kieferova18} is $\exp(\tau\lambda_B)$.
The magnitude of the eigenvalues in the exponential are therefore $\arcsin[\sin(E_j\tau)/\exp(\tau\lambda_B)]$.
In the case of small $E_j$ (as we expect here), the linearization is approximately $E_j\tau/\exp(\tau\lambda_B)$, so this is equivalent to using an effective evolution time $\tau_{\rm eff}=\tau/\exp(\tau\lambda_B)$. If we were to choose $\tau\approx \ln 2/\lambda_B$, then we would have $\tau_{\rm eff}\approx \ln 2/(2\lambda_B)$.
That is, the evolution time is halved from $\tau$, which corresponds to doubling the cost for the phase estimation.
On the other hand, the oblivious amplitude amplification triples the cost, so this qubitization method gives $2/3$ the cost for performing phase estimation on the time evolution.
However, we can further reduce the cost by increasing $\tau$ to $1/\lambda_B$.
That gives an effective time $\tau_{\rm eff}\approx 1/(e\lambda_B)$ where $e$ is Euler's number.
That means the complexity is then $63\%$ of the method via time evolution from \cite{Kieferova18}, or $45\%$ of the complexity using the method of \cite{Low2018}.
Therefore, we see that even though these individual improvements are only minor, there is total improvement by more than a factor of 2 over the method of \cite{Low2018}.

For this effective time, we are assuming that $E_j$ is near zero, so $\arcsin[\sin(E_j\tau)/e]$ (where we have taken $\tau=1/\lambda_B$) is approximately $E_j\tau/e$.
In cases where $E_j$ is not near zero, then the uncertainty in estimating $\arcsin[\sin(E_j\tau)/e]$ will correspond to uncertainty in estimating $E_j$ with a different constant factor.
However, we have freedom in adding a phase shift to the operator we are implementing.
Instead of implementing $e^{\pm \mi(A+B)\tau}$, we can implement $e^{\pm \mi(A+B-E_0)\tau}$, simply by performing an additional phase shift.
That means we are then estimating eigenvalues $\arcsin[\sin((E_j-E_0)\tau)/e]$.
Provided we have an initial estimate of $E_j$, we can ensure that $E_j-E_0$ is near zero so we are in the region where the sine and arcsine can be linearized.

To be more specific, the slope at $E_j-E_0$ is increased to
\begin{equation}
    \frac \tau e \left( 1 + \frac{(e^2-1)(E_j-E_0)^2\tau^2}{2e^2} + \mathcal{O}((E_j-E_0)^4\tau^4)\right).
\end{equation}
This means that there is a relative increase in the number of steps needed to obtain a given accuracy of $\mathcal{O}((E_j-E_0)^2\tau^2)$.
There is also a slight increase in the estimation error due to the nonlinearity of this function.
This relative increase turns out to be on the order of the standard deviation of the phase estimation, so scales as $\mathcal{O}(1/\reps)$ for phase estimation on $\reps$ steps.
This higher-order increase in the error can be accounted for by increasing $\reps$ by $\mathcal{O}(1)$.
We do not need to take account of these factors when considering phase estimation for the qubitization, because there the effect of the nonlinearity is to reduce the error.

To summarise, there are two main choices for the algorithm.
First, there is the choice of how to prepare the time registers, for which there are three main alternatives.
\begin{enumerate}
    \item The ``compressed encoding'' method from \cite{Kieferova18}, which requires controlled rotations by arbitrary angles on the qubits for the time registers, and therefore has a complexity scaling like the logarithm of the inverse error for each qubit of the time registers.
    \item The sorting method from \cite{Kieferova18}, which requires a number of Toffolis for each qubit that can be determined from \tab{sorttab}, and is expected to be around 8 (so less than for compressed encoding).
    \item The method of \cite{Low2018}, which involves postselecting on success of inequality tests, rather than a sort. That reduces the complexity of the state preparation for the time registers, but means the number of steps is increased by a factor of $2\ln 2\approx 1.38$ over the methods of \cite{Kieferova18}.
    This is expected to result in an increased overall cost, because the state preparation for the time registers is a small part of the overall cost in each step.
\end{enumerate}
The other choice is between the following two alternatives.
\begin{enumerate}
    \item One can simulate the time evolution under the Hamiltonian, and perform phase estimation on the time evolution. Then the number of steps scales as $\lambda_B T/\ln 2$ for the methods of \cite{Kieferova18}, or  $2\lambda_B T$ for the method of \cite{Low2018}.
    \item One can instead perform a qubitization of $\sin(H\tau)$, and perform phase estimation on that. That reduces the complexity because it removes the need for amplitude amplification for time evolution (and the corresponding error). The number of steps would be $e\lambda_B T$ for the methods of \cite{Kieferova18} or $2e\ln 2\lambda_B T$ for the method of \cite{Low2018}.
    The number of steps is increased, but the complexity of each step is reduced by a factor of 3 as compared to the time evolution approach.
\end{enumerate}
	
We will focus on the case where the time registers are prepared using a sort, and we perform qubitization of $\sin(H\tau)$, as this gives the best performance.
As in the case of the qubitization approach, we need to ensure that our $\SEL$ procedure is self-inverse.
We will use an ancilla to control $e^{-\mi(A+B)\tau}$ versus $e^{\mi(A+B)\tau}$.
To ensure that it is self-inverse, it will be controlled on the ancilla qubit as
\begin{equation}
    \ket{0}\bra{1} \otimes e^{-\mi(A+B)\tau} + \ket{1}\bra{0} \otimes e^{\mi(A+B)\tau}.
\end{equation}
This ancilla qubit is prepared as $\ket{+}$, ensuring we obtain the sum of $e^{-\mi(A+B)\tau}$ and $e^{\mi(A+B)\tau}$ in the block encoding.
However, we do need to ensure that the implementation of $e^{-\mi(A+B)\tau}$ is the inverse of $e^{\mi(A+B)\tau}$.
To achieve that, we apply the following method.
\begin{itemize}
    \item Use the qubit as control for a phase gate on each qubit in the unary encoding of $k$ to control between $(-\mi)^k$ and $(\mi)^k$ in the Dyson series. There is no extra Toffoli cost.
    \item Use the qubit to control whether the sort of the time registers sorts them in ascending or descending order.
    One can simply use the qubit as a control with the qubits flagging the result of the inequality test as target.
    There is again no Toffoli cost.
    \item Use the qubit to control taking the absolute values of the time differences (see \app{kstate}).
    This is the only part with a Toffoli cost.
    \item Use this qubit to control between adding and subtracting $\nu$, which has no extra Toffoli cost.
    \item Use the qubit to control positive versus negative phasing by the kinetic energy, which can be performed without Toffolis by controlling between adding and subtracting the phase into the phase gradient state.
\end{itemize}

	\subsection{Exponentiating the kinetic operator}
	\label{sec:kinetic_exp}
	It is estimated in \cite{BabbushContinuum_b} that the cost of simulating the first quantized electronic structure Hamiltonian is dominated by the exponentiation $e^{-\mi\tau T}$. 
	The reason for that is the kinetic energy needs to be computed in the computational basis, which requires squaring $3\eta$ numbers (for the 3 components of the momenta of each electron) and adding them together.
	Squaring an $n_p$-bit number can be performed with approximately $n_p^2$ Toffolis, giving an overall complexity of approximately $3\eta n_p^2$.
	
	However, it is possible to greatly reduce the complexity, by keeping the sum of squares in an ancilla, instead of recomputing it every time.
	Then when we do the block encoding of the $U+V$, we can update this sum of squares.
	The cost of updating the sum of squares is only logarithmically dependent on $\eta$.
	To be more specific, for $V$ we have $p$ replaced with $p+\nu$ and $q$ replaced with $q-\nu$.
	For the first the change in the sum of squares is $\|p+\nu\|^2-\|p\|^2=2p\cdot \nu+\|\nu\|^2$, and for the second the difference in squares is $\|q-\nu\|^2-\|q\|^2=-2q\cdot \nu+\|\nu\|^2$.
	The quantity $\|\nu\|^2$ has already been computed in preparing the state with amplitudes $1/\|\nu\|$, so does not need to be recomputed.
	
	This method also enables us to take account of the energy offset $E_0$ we are using.
	We aim to subtract all energies by $E_0$ to put us in the linear region of the function we are estimating.
	We can simply subtract $E_0$ from the kinetic energy register at the beginning of the procedure.
	That means the energy of every step will be shifted by $E_0$ with no more gate complexity (other than the initial subtraction).
	
	Therefore the majority of the cost is in the computation of $p\cdot\nu$ and $q\cdot\nu$.
	The cost for computing the product of the \emph{absolute value} of each component is $2n_p(n_p-1)-n_p=2n_p^2-3n_p$, and the sign of the product can be computed with a CNOT.
	Rather than adding all three components then adding into the energy register, it is convenient to add them one at a time.
	This is because we have signed integers, and adding the three signed integers would incur an additional cost for converting to two's complement.
	When adding the signed integer into the energy register, we can instead use the sign bit to control addition or subtraction into the energy register, which incurs no additional cost.
	
	The number of bits needed to represent the sum of momenta is $n_\eta+2n_p$.
	We are guaranteed that the net result of all the arithmetic will result in no overflow qubit, but because we are adding terms in succession there may be overflow in intermediate steps.
	Therefore the cost of addition of $\|\nu\|$ will be $n_\eta+2n_p$.
	On the other hand, because we have twice $p\cdot\nu$ and $q\cdot\nu$, the least significant bit is 0, and the cost of the addition of each component is reduced by 1 to $n_\eta+2n_p-1$.

    We will have qubits that control whether $U$ and $V$ are implemented, and these will need to control whether each of the additions are performed.
    For $\|\nu\|^2$ we have a cost of the control each time of $2n_p$, for a total of $4n_p$.
    To eliminate the cost of controlling the addition of $p\cdot\nu$ and $q\cdot\nu$, we can instead make the swaps of these momenta into the ancilla controlled on the qubits which control whether $U$ and $V$ are implemented.
    That control on the unary iteration on $i$ and $j$ takes only one extra Toffoli each time, for a total cost of 4 Toffolis there.
    However, if the swap was not performed, then $p\cdot\nu$ or $q\cdot\nu$ was computed with a momentum of zero, giving zero.
    Therefore these additions do not need to be made controlled.
	
	There is an important nuance in that the phase needed will not be exactly a multiple of $2\pi/2^{b_{\rm grad}}$ for some integer $b_{\rm grad}$.
    The kinetic energy is given by $2\pi^2/\Omega^{2/3}$ times the sum of squares of the integers stored in the momentum registers, and this energy is multiplied by $\tau=1/(\lambda_U+\lambda_V)$ to give the phase required.
    The ideal situation is that $2\tau\pi^2/\Omega^{2/3}$ (the smallest phase shift) is exactly $2\pi/2^{b_{\rm grad}}$.
    Then one can simply take the sum of squares of momenta and add it into the phase gradient register.
    
    On the other hand, say that for example $2\tau\pi^2/\Omega^{2/3}$ is instead $3\times 2\pi/2^{b_{\rm grad}}$.
    One can then instead use the modified phase gradient state
    \begin{equation}
        \frac 1{\sqrt{2^{b_{\rm grad}}}} \sum_{\ell=0}^{2^{b_{\rm grad}}-1} e^{-6\pi \mi \ell/2^{b_{\rm grad}}}\ket{\ell} .
    \end{equation}
    Then it is easy to see that adding an integer to the register will give a phase shift that is a multiple of $3\times 2\pi/2^{b_{\rm grad}}$.
    Typically, $2\tau\pi^2/\Omega^{2/3}$ would not be of the form of an integer multiple of $2\pi/2^{b_{\rm grad}}$, but we have freedom to adjust the time $\tau$.
    Say for example that we are using a multiple with 10 bits.
    We will use $b_T$ for the number of bits in the multiple.
    Then the worst case is when $2\tau\pi^2/\Omega^{2/3}$ with $\tau=1/(\lambda_U+\lambda_V)$ is about $512.5$ times $2\pi/2^{b_{\rm grad}}$.
    Then rounding to 513 gives the largest relative error, whereas rounding numbers like, for example, $599.5$ to $600$ would result in smaller relative error.
    (Numbers below 512 could be represented with fewer bits.)
    Adjusting the value of $\tau$ to make the multiple exactly 513, means that the value of $\tau/\exp(\tau\lambda_B)$ (which governs the number of steps needed in the phase estimation) would be changed by less than 1 part in a million.
    In general, as discussed in \app{mgrad}, one can take
    \begin{equation}\label{eq:bgrad}
        b_{\rm grad} = b_T - \left\lceil \log \left( \frac{\pi}{(\lambda_U+\lambda_V)\Omega^{2/3}} \right) \right\rceil ,
    \end{equation}
    and the increase in the $\lambda$ will be no more than $1 + 1/2^{2b_T+1}$.
	
	To summarise, the costs associated with the phasing by $T$ are as follows.
	\begin{enumerate}
	    \item We need to compute time differences from two times of $n_t$ qubits with no carry bit, which has cost $n_t-1$.
	    \item There is cost $2n_t(n_\eta+2n_p)-n_t$ for the multiplication of the sum of momenta by the difference of times.
	    \item The addition into the phase gradient register has cost $b_{\rm grad}-2$ Toffolis.
	    \item The cost of computing all the products of components of $p$ or $q$ with $\nu$ is $6\times (2n_p^2-3n_p)=12n_p^2-18n_p$ altogether.
	    \item The cost of the addition of all these components into the energy register is $6(n_\eta+2n_p-1)$.
	    \item The cost of two additions of $\|\nu\|^2$ into the energy is $2(n_\eta+2n_p)$.
	    \item The cost of making each addition of $\|\nu\|^2$ controlled is $4n_p$.
	    \item For the kinetic energy phasing for time $\tau/2M$ at the beginning and the end we simply need to add the (appropriately bit-shifted) kinetic energy into the phase gradient state, with cost $b_{\rm grad}-2$ Toffolis each time.
	\end{enumerate}
	These costs are needed different numbers of times.
	For the time difference, there is only Toffoli cost where we compute differences of times given in quantum registers, which is done $K-1$ times.
	This is because $M-1-m'_k$ can be computed with NOT gates (see \app{kstate} for more detail).
	The multiplication in step 2 and the phasing in step 3 are both needed $K+1$ times for all of the phasings.
	The remaining costs are updating the energy register, which only needs to be done between phasings, so we have $K$ of these costs.
	
	The arithmetic for the time differences and multiplications can be performed in place with ancillae, so the arithmetic can be inverted without Toffolis.
	This gives a complexity for each time difference $n_t-1$. There is also cost $n_t-1$ to take the absolute value of the time, as controlled by the qubit selecting between forward and reverse time evolution (see \app{kstate}).
	With $K-1$ differences, the total cost is
    \begin{equation}\label{eq:timesubcost}
        2(K-1)(n_t-1) .
    \end{equation}
	For the phase shift the total Toffoli cost is
    \begin{equation}\label{eq:kinexpcost}
        (K+1)[2n_t(n_\eta+2n_p)-n_t+b_{\rm grad}-2]+2(b_{\rm grad}-2),
    \end{equation}	
    where we have included $2(b_{\rm grad}-2)$ for the initial and final phasings.
	For updating the sum of momenta
	\begin{align}\label{eq:momentasum}
	   &K\{ [12n_p^2-18n_p]  +  [6(n_\eta+2n_p-1)]  + [2(n_\eta+2n_p)] + [4n_p] \} \nn
	   &= K (12n_p^2 +2n_p + 8n_\eta ) .
	\end{align}
    Note that the complexity is now only logarithmically dependent on $\eta$.
    There is still linear dependence on $\eta$ for selecting momentum registers in block encoding $U+V$.

	%%%%%%%%%%%%%%%%%%%%%%%%%%%%%%%%%%%%%%%%%%%%%%%%%%%%%%%%%%%%%%%%%%%%%%%%%%%%%%
	\subsection{Implementing \texorpdfstring{$\SEL$}{SEL} and \texorpdfstring{$\PREP$}{PREP}}
	\label{sec:sel_prep}
	In this section, we discuss the costing of the block encoding of $U+V$ for the interaction picture.
    As part of this block encoding, we will update the sum of squares of momenta, as described above.
    We have previously described the complexity for block encoding $T+U+V$.
    The complexity of block encoding $U+V$ is very similar, except the parts needed to block encode $T$ are omitted, and we need to update the register storing the sum of momenta.
    Therefore our previous costing of the block encoding from \tab{qubcosts} can be modified as in \tab{intcosts1}.
    We will make the controlled swaps of momenta into the ancillas controlled on qubits flagging whether $U$ and $V$ are to be performed.
    That overall control takes one extra Toffoli for each of the controlled swaps, for a total of 4.
    The final item in this table is the new cost for updating the sum of momenta.
    
\begin{table}
\begin{tabular}{| m{11cm} | m{5cm} |}
\hline
\centering Procedure &  \hspace{16mm} Toffoli cost \\
 \hline \hline
Preparing the superposition on the qubit selecting between $U$ and $V$, but not $T$; see \eq{UVprepcost}. & $2(4n_{\eta\zeta}+2b_r-9)$ \\
\gline
Preparing equal superpositions over $\eta$ values of $i$ and $j$ in unary; see \eq{ijprepcost}. &
$14n_\eta+8b_r-36$ \\
\gline
The state preparation cost for the $w,r,s$ registers used for $T$ is eliminated. & $0$ \\
\gline
Controlled swaps of the $p$ and $q$ registers into and out of ancillae (which is used for $U$ and $V$), with another 4 Toffolis to give overall control; see \eq{swapscost}. & $12\eta n_p+4\eta-4$ \\
\gline
The $\SEL$ cost for $T$ is eliminated. & $0$ \\
\gline
The preparation of the state with amplitude $1/{\norm{\nu}}$ and inversion; see \eq{nuprepcost}. & $3n_p^2+15n_p-7+4n_{\mathcal{M}}\big(n_p+1\big)$\\
\gline
The QROM for $R_\ell$, combined with the state preparation with amplitudes $\sqrt{\zeta_\ell}$ (which may be implicit); see \eq{Rlqrcost}. & $\lambda_\zeta+\erase{\lambda_\zeta}$ \\
\gline
The additions and subtractions of $\nu$ into the momentum registers; see \eq{addcosts}. & $24n_p$ \\
\gline
Phasing by $-e^{-\mi k_\nu\cdot R_\ell}$; see \eq{Rphacost1}. & $6n_p n_R$ \\
\gline
Updating the sum of momenta; see \eq{momentasum}. & $12n_p^2 +2n_p + 8n_\eta$\\
\hline
\end{tabular}
\caption{The costs involved in the block encoding of the Hamiltonian for the interaction picture, following the same entries as in \tab{qubcosts}.}
\label{tab:intcosts1}
\end{table}
    
    For the total complexity, in the Dyson series the most efficient method is to qubitize a single step of the Dyson series, and perform phase estimation on it, rather than to use amplitude amplification to obtain the explicit time evolution.
    As a result, the effective time for a single step is approximately $1/[e\lambda_B]$.
    The RMS error in the energy due to phase estimation is then $\pi e\lambda_B/(2\reps)$, where $\reps$ is the number of steps, so one should take
    \begin{equation}
        \reps = \left\lceil \frac{\pi e \lambda_B}{2\epsilon_{\rm pha}} \right\rceil,
    \end{equation}
    where $\epsilon_{\rm pha}$ is the allowable error due to phase estimation.
    For the application to the electronic structure problem, we take $\lambda_B=\lambda_U+\lambda_V$.
    As described above, we will not be using exactly the optimal size of time step, but the increase in the cost is upper bounded by $1+1/2^{2b_T+1}$, giving
    \begin{equation}
        \reps = \left\lceil (1+1/2^{2b_T+1}) \frac{\pi e \lambda_B}{2\epsilon_{\rm pha}} \right\rceil .
    \end{equation}

	%%%%%%%%%%%%%%%%%%%%%%%%%%%%%%%%%%%%%%%%%%%%%%%%%%%%%%%%%%%%%%%%%%%%%%%%%%%%%%
	\subsection{Total cost of the interaction picture simulation}
	\label{sec:dyscost}
	
    In the block encoding of the Dyson series, there are a number of parts that need to be costed.
    \begin{enumerate}
        \item Preparing the superposition over $k$ in unary.
        \item The superpositions on the time registers will be prepared with controlled Hadamards.
        \item The $K$ registers can be sorted with a number of comparators as given in \tab{sorttab}.
    \end{enumerate}
    The contributions to the error are as follows.
    \begin{enumerate}
        \item The finite cutoff on the Dyson series, $K$.
        \item The finite angles of the rotations for preparing the superposition over $k$.
        \item The finite number of bits used to represent the times, $n_t$.
    \end{enumerate}
    
    For the time steps we are using for the qubitized Dyson series approach, the time interval is $\tau = 1/\lambda_B$, which means that the size of term $k$ in the Dyson series is simply $1/k!$, and the error is approximately $1/(K+1)!$.
    More precisely, the error is
    \begin{equation}
        e - \sum_{k=0}^K \frac 1{k!}.
    \end{equation}
    This is error in block-encoding $\sin((A+B)\tau)=\sin((A+B)/\lambda_B)$, and therefore corresponds to error in the estimation of the energy of
    \begin{equation}\label{eq:errK}
        \lambda_B \left( e - \sum_{k=0}^K \frac 1{k!} \right) \approx \frac{\lambda_B}{(K+1)!} .
    \end{equation}
    Note that here we are assuming that the eigenvalue is near zero, so the sine function can be linearized in estimating the error.
    In the case where $E_j$ (the energy eigenvalue to be measured) is not near zero, it is possible to just apply a phase shift so we are block encoding $\sin((A+B-E_0)\tau)$ with $E_0$ near $E_j$.
    
    For the error in preparing the superposition over $k$, we note that because the amplitudes we need are precisely $1/\sqrt{k!}$, the state can be written as proportional to
    \begin{equation}
        \sum_{k=0}^K \sqrt{\frac{K!}{k!}} \ket{k},
    \end{equation}
    where the amplitudes are the square roots of integers. The $-\mi$ factors in the Dyson series can be implemented just with $S$ gates on the qubits with the unary encoding of $k$, so do not require non-Clifford gates.
    We will also be controlling between forward and reverse time evolution, which means that we will be controlling between $\pm \mi$.
    That can be achieved simply with a controlled phase gate, and again there is no non-Clifford cost.
    We can therefore prepare an equal superposition over a number of integers $\sig (0)\approx K! e$, where we define
    \begin{equation}\label{eq:sigdef}
        \sig (k) = \sum_{\ell=k}^K \frac{K!}{\ell!} ,
    \end{equation}    
    and use inequality testing to exactly obtain the required state.
    The equal superposition over the integers is not performed exactly, but the failure, which has small probability, is flagged so it only corresponds to slightly reducing the size of the block encoded operation, rather than introducing any error into that block encoded operation.
    
    For the complexity, the contributions are as follows.
    \begin{enumerate}
        \item First there is the preparation of an equal superposition over a number of integers that is upper bounded by $K! e$.
        The number of bits can be taken to be $n_k =\lceil \log( \sig (0) ) \rceil$, so the cost of preparing the equal superposition is (according to the usual formula)
        \begin{equation}
            3n_k +2b_r -9.
        \end{equation}
        \item To select $k=0$ or $k=1$, we can perform an inequality test between the equal superposition register and $\sig (2)$.
        The Toffoli cost of this inequality test is $n_k-1$.
        To distinguish $k=0$ and $k=1$, we simply use a Hadamard on an ancilla qubit and use a Toffoli to set the qubit.
        \item For $1<k<K$, we perform an inequality test with the number $\sig (k+1)$.
        To save complexity, we can only perform an inequality test with a number of bits necessary to store $\sig (k)$, conditioned on already passing the inequality test with $\sig (k)$ in the previous step.
        The Toffoli complexity may then be given as $\lceil \log( \sig (k) ) \rceil-1$ for the inequality test, and another Toffoli is needed for the conditioning.
    \end{enumerate}
    The first contribution here is a complexity that must be paid again in the inverse preparation, giving a total cost
        \begin{equation}\label{eq:nkprep}
            2 ( 3n_k +2b_r -9 ) .
        \end{equation}
    For contributions 2 and 3, the Toffolis and inequality tests may be inverted with Clifford gates, giving a Toffoli cost
    \begin{equation}\label{eq:nkprep2}
        n_k + \sum_{k=2}^{K-1} \lceil \log( \sig (k) ) \rceil .
    \end{equation}
    In order to ensure that these inequality tests and Toffolis can be inverted with Cliffords, the same number of qubits are kept throughout the calculation.
    For a more detailed explanation of the costing, see \app{kstate}.

    The other contributor to the error listed above is the finite number of bits used to represent the time, $n_t$.
    In prior work, this error was estimated based on a bound on the derivative of the Hamiltonian.
    Here we use a tighter bound on the error based on the second derivative.
    This is a lengthy derivation, so is relegated to \app{timedisc}.
    The result is that, for time interval $\tau$, the error in the time-ordered exponential can be bounded by
\begin{equation}
\left(\frac{2\|A\|}{\|B\|} + \frac{\|A\|^2}{\|B\|^2}\right) \left[ 2^{n_t} \left( e^{\|B\|\tau/2^{n_t}}-1 \right) - \|B\|\tau (1+\|B\|\tau/2^{n_t+1}) \right] .
\end{equation}
    Here we take $\tau=1/(\lambda_U+\lambda_V)$, and replace $\|B\|$ with $\lambda_U+\lambda_V$ and $\|A\|$ with $\lambda_T$.
    That gives the bound
\begin{equation}
\left(\frac{2\lambda_T}{\lambda_U+\lambda_V} + \frac{\lambda_T^2}{(\lambda_U+\lambda_V)^2}\right) \left[ 2^{n_t} \left( e^{2^{-n_t}}-1 \right) - (1+2^{-(n_t+1)}) \right] .
\end{equation}   
    Again, this is error in block encoding of $\sin((A+B)/\lambda_B)$, and corresponds to error in the estimation of the energy that is multiplied by $\lambda_B=\lambda_U+\lambda_V$, so the error in the energy is
\begin{equation}\label{eq:epsnt}
\left(2\lambda_T + \frac{\lambda_T^2}{\lambda_U+\lambda_V}\right) \left[ 2^{n_t} \left( e^{2^{-n_t}}-1 \right) - (1+2^{-(n_t+1)}) \right] .
\end{equation}

    The complexity of the controlled Hadamards to generate the superpositions over times is therefore $Kn_t$. The same complexity is needed in inverting the state preparation, giving total cost
    \begin{equation}\label{eq:ctrlhadcosts}
        2Kn_t.
    \end{equation}
    Sorting the times with $n_t$ bits gives complexity $2n_t$ times the numbers of comparators given in \tab{sorttab}.
    These numbers of comparators are denoted $\srt{K}$, and the same cost is needed to invert the sort, giving a total sorting cost
    \begin{equation}\label{eq:srtcosts}
        4n_t\srt{K}.
    \end{equation}
    
    For the times, we need to compute the differences.
    There is a subtlety here, in that we need to compute the difference with the maximum time for the final time difference for a given $k$.
    That can be achieved by taking all times where we have \emph{not} performed the Hadamards to be set to the maximum time (all ones).
    One can simply perform $X$ gates after the controlled Hadamards.
    The total cost of the differences and absolute values was then given before as $2(K-1)(n_t-1)$.

    That accounts for the need to make the exponentiation of the kinetic operator controlled on the $k$ register, but we also need to do this for the block encoding of $U+V$.
    The block encoding is a combination of the $\PREP$ operation, a $\SEL$ operation, and an inverse preparation $\PREP^\dagger$.
    Because the preparation and inverse preparation cancel, only the $\SEL$ operation need be made controlled.
    The cost of this control may be quantified as follows.
    
    First, we need to control on success of the three equal superposition state preparations and the preparation of $\nu$, as well as the qubit in the $k$ register.
    That takes 4 Toffolis to produce a qubit that is the AND of all five bits.
    We will apply another 4 Toffolis in erasing this qubit (though we could also keep the ancillas to erase it without Toffolis).
    
    For the phasing by $-e^{-\mi k_\nu\cdot R_\ell}$, the $R_\ell$ can be made zero in the case where we do not want to perform $\SEL$ by simply controlling the QROM for outputting of $R_\ell$ on that qubit.
    That QROM is already controlled on the qubit selecting between $U$ and $V$, and controlling on another qubit merely needs one more Toffoli.
    There also needs to be control of the minus sign on that qubit as well.
    However, in controlling the output of the QROM, we are generating an ancilla that can be used to control that sign flip, and no additional Toffolis are required.
    
    The addition $p+\nu$ is already controlled by the qubit selecting between $U$ and $V$.
    Again, making it controlled on one more qubit only takes one more Toffoli.
    We have costed the subtraction $q-\nu$ as being controlled, but for the interaction picture that subtraction only needs to be controlled on this one qubit.
    We have already given the cost for this subtraction to be controlled, so there is \emph{no} additional Toffoli cost over what we have already presented.    
    As a result, the total additional complexity to make the block encoding of $U+V$ controlled on a qubit is just \emph{two} Toffolis, which is completely trivial compared to the other costs in the circuit.
    
    In the block encoding, we also need to reflect on all qubits used that are not guaranteed to be returned to zero.
    The detailed analysis is given in \app{qubcost}, and yields a Toffoli cost of
    \begin{equation}\label{eq:refqbtsint}
         K ( n_{\eta\zeta}+2n_\eta+4n_p+n_{\mathcal{M}}+n_t + 12) + n_k+3 .
    \end{equation}
    In this costing, we use intermediate checks between the block encodings of $U+V$ to ensure that the control qubits are rezeroed, and we can reuse them.
 Lastly we provide the total cost of a single step of the interaction picture including inversion in \tab{intcosts2}.
 There, 10 Toffolis are added to the costs from \tab{intcosts1} to account for the complexity of making the operations for $U$ and $V$ controlled.

\begin{table}
\begin{tabular}{| m{11cm} | m{5cm} |}
\hline
\centering Procedure & \hspace{15mm} Toffoli cost \\
 \hline \hline
The preparation of the equal superposition over $\sig (0)$ numbers; see \eq{nkprep}. & $2(3n_k +2b_r - 9)$ \\
\gline
Inequality tests used for the preparation of the superposition over $k$; see \eq{nkprep2}. & $n_k+ \sum_{k=2}^{K-1} \lceil \log( \sig (k) ) \rceil $ \\
\gline
Controlled Hadamards for preparing the superpositions over times; see \eq{ctrlhadcosts}. & $2Kn_t$ \\
\gline
The sorting of the time registers; see \eq{srtcosts}. & $4n_t\srt{K}$ \\
\gline
The differences of the times; see \eq{timesubcost}. & $2(K-1)(n_t-1)$ \\
\gline
The kinetic operator needs to be exponentiated $K+1$ times; see \eq{kinexpcost}. & $(K+1)[2n_t(n_\eta+2n_p)-n_t+b_{\rm grad}-2]+2(b_{\rm grad}-2)$ \\
\gline
Perform the controlled block encoding of the potential operator $K$ times. & $K(10+{\rm costs~in~\tab{intcosts1}})$ \\
\gline
The reflection cost for qubitization. & \eq{refqbtsint} \\
\hline
\end{tabular}
\caption{The costs involved in the block encoding of the interaction picture representation of the Hamiltonian.}
\label{tab:intcosts2}
\end{table}
    
    To determine the effective $\lambda$, we need to also account for the probabilities of success of the amplitude amplification for the various state preparations we have performed.
    For the block encoding of $U+V$, these are as follows.
    \begin{enumerate}
        \item The preparation of the state with amplitudes $1/\|\nu\|$, which has probability of success denoted $\psuc^{\rm amp}$.
        \item The preparation of the equal superposition state for selecting between $U$ and $V$, which has probability $\eqprep{\eta+2\lambda_\zeta,b_r}$.
        \item Preparing the equal superpositions over $i$ and $j$ have probabilities of success $\eqprep{\eta,b_r}$ each.
        \item The probability of obtaining $i\ne j$ is $1-1/\eta$.
    \end{enumerate}
    The value of $\lambda_U+\lambda_V$ needs to be divided by all the probabilities in parts 1 to 3, and just $\lambda_V$ needs to be divided by $1-1/\eta$.
    It is also possible to perform a joint state preparation over $i$ and $j$, eliminating the need to divide by $1-1/\eta$.
    We do not explain that here as the procedure is already quite complicated.
    
    In performing the block encoding of $\sin((U+V)\tau)$ there is a further preparation of an equal superposition state in order to prepare the superposition over $k$.
    That has probability of success $\eqprep{\sig(0),b_r}$.
    Because we are considering a region where the sine can be linearized, the effect is the same as multiplying $U+V$ by this probability, and we should therefore multiply all these probabilities.
    Also taking into account the factor $(1+1/2^{2b_T+1})$ due to adjusting the length of time, we get
 \begin{equation}\label{eq:Peq}
     P_{\rm eq} = \psuc^{\rm amp} \eqprep{\sig(0),b_r}\eqprep{\eta+2\lambda_\zeta,b_r}\eqprep{\eta,b_r}^2 /(1+1/2^{2b_T+1}).
 \end{equation}
    The value of $\lambda_U+\lambda_V$ then needs to be divided by $P_{\rm eq}$ to obtain the effective $\lambda$.
    Therefore, to get the total cost, the costs for a single step of the interaction picture need to be multiplied by a number of steps
    \begin{equation}
        \reps = \left \lceil \frac{\pi e (\lambda_U+\lambda_V/(1-1/\eta))}{2\epsilon_{\rm pha}P_{\rm eq}} \right\rceil .
    \end{equation}   
    
    Now we take account of the nonlinearity in using an operator with eigenvalue $\arcsin[\sin((E_j-E_0)\tau)/e]$.
    As discussed above, the error of the estimation is increased by a relative amount $\mathcal{O}((\lambda_U+\lambda_V)^2\Delta E^2)$, so we need to introduce this factor to the number of steps used.
    Also, the nonlinearity gives an $\mathcal{O}(1)$ increase to the number of steps needed.
    If we have an initial estimate of the energy with precision $\Delta E$, we can choose $E_0$ such that the number of steps needed is
    \begin{equation}\label{eq:intreps}
        \reps = \frac{\pi e (\lambda_U+\lambda_V/(1-1/\eta))}{2\epsilon_{\rm pha}P_{\rm eq}}  \left[ 1+ \mathcal{O}((\lambda_U+\lambda_V)^2\Delta E^2) \right]+ \mathcal{O}(1).
    \end{equation}
    In this expression, we should also use the values $\lambda^\alpha_U,\lambda^\alpha_V$ given by the block encoding to be precise, though in practice the difference should be negligible.
    
    In computing these costs, we need to ensure that the total error is no larger than $\epsilon$, where the sources of error are as follows.
    \begin{enumerate}
        \item The error due to the phase estimation $\epsilon_{\rm pha}$.
        \item The error due to the truncation of the Dyson series, $\epsilon_K\approx (\lambda_U+\lambda_V)/(K+1)!$.
        \item The error $\epsilon_R$ due to the finite number of bits used to represent the nuclear positions, which is upper bounded as in \lem{epsR}.
        \item The error $\epsilon_{\mathcal{M}}$ due to the preparation of the state with amplitudes $1/\|\nu\|$, which is upper bounded as in \lem{epsM}.
        \item The error due to the finite number of time steps, $\epsilon_t$, which is upper bounded as in \eq{epsnt}.
    \end{enumerate}
    Because the error $\epsilon_{\rm pha}$ is an RMS value, we combine it with the other errors as a square root of a sum of squares.
    We therefore require that
    \begin{equation}
        \epsilon_{\rm pha}^2 + ( \epsilon_K+\epsilon_R+\epsilon_{\mathcal{M}}+\epsilon_t)^2 \le \epsilon^2.
    \end{equation}
	
    To summarise, the result can be given as in the following theorem.
    \begin{theorem}
    \label{thm:interaction_cost}
 Using the interaction picture, assuming an initial estimate with accuracy $\Delta E$, the eigenvalue of the Hamiltonian given by \eq{first_quant_ham} can be estimated to within RMS error $\epsilon$, using a number of Toffoli gates
 \begin{align}
     \reps & \left( 
     2(3n_k +2b_r -9) % Preparing superposition over $\sig(0)$ numbers.
     + n_k+ \sum_{k=2}^{K-1} \lceil \log( \sig (k) ) \rceil  % Inequality tests for preparation of superposition over $k$.
     +2Kn_t % Controlled Hadamards for the preparation of the superposition over times.
     +4n_t\srt{K}\right. % Cost for sorting.
     \nn
     & \quad +2(K-1)(n_t-1) % Computing differences in times.
     +(K+1)[2n_t(n_\eta+2n_p)-n_t+b_{\rm grad}-2]+2(b_{\rm grad}-2) % Exponentiating the kinetic operator.
     \nn
     &\quad +K \left\{10+ 2(4 n_{\eta\zeta}+2b_r-9) + (14 n_\eta + 8b_r-36) + 3[3n_p^2+15n_p-7+4n_{\mathcal{M}}(n_p+1)]\right. \nn
            &\quad\left. +\lambda_\zeta+\erase{\lambda_\zeta}
            + (12\eta n_p +4\eta-4)+ 24 n_p + 6n_p n_R + (12n_p^2 +2 n_p + 8n_\eta)\right\} % Costs from the block encoding of the SEL operation.
    \nn
    &\quad \left. \vphantom{\sum_{k=2}^{K-1}} +K ( n_{\eta\zeta}+2n_\eta+2+4n_p+n_{\mathcal{M}}+n_t + 12) + n_k+3 +\mathcal{O}(\log(1/\epsilon))
     \right),
 \end{align}
 where
 \begin{align}
 \reps &= \frac{ \pi e (\lambda^1_U+\lambda^1_V/(1-1/\eta))}{2\epsilon_{\rm pha}P_{\rm eq}} \left[ 1+ \mathcal{O}((\lambda_U+\lambda_V)^2\Delta E^2) \right]+ \mathcal{O}(1), \\
     P_{\rm eq} &= \psuc^{\rm amp} \eqprep{\sig(0),b_r}\eqprep{\eta+2\lambda_\zeta,b_r}\eqprep{\eta,b_r}^2 /(1+1/2^{2b_T+1}),\\    
        b_{\rm grad} &= b_T - \left\lceil \log \left( \frac{\pi}{(\lambda_U+\lambda_V)\Omega^{2/3}} \right) \right\rceil ,
\end{align}
with $b_T\in\mathbb{N}$.
For this complexity we require appropriately chosen phase gradient states to be given as a resource.
The quantities $\epsilon_{\rm pha}>0$ and $K,n_R,n_{\mathcal{M}},n_T\in\mathbb{N}$ are chosen such that
     \begin{align}\label{eq:totinterr}
        \epsilon^2 & \ge \epsilon_{\rm pha}^2 + ( \epsilon_K+\epsilon_R+\epsilon_{\mathcal{M}}+\epsilon_t)^2 ,\\
        \epsilon_K &= (\lambda_U+\lambda_V) \left( e - \sum_{k=0}^K \frac 1{k!} \right) ,\\
	    \epsilon_\mathcal{M} &= \frac{2\eta}{2^{n_{\mathcal{M}}}\pi\Omega^{1/3}} \left( \eta + 2 \sum_{\ell=1}^L \zeta_\ell \right)\left( 7 \times 2^{n_p+1} -9n_p -11 -3\times 2^{-n_p} \right) \\
		\epsilon_R &= \frac{\eta\lambda_\zeta}{2^{n_R}\Omega^{2/3}}\sum_{\nu\in G_0}\frac{1}{\norm{\nu}} , \\
\epsilon_t &= \left(2\lambda_T + \frac{\lambda_T^2}{\lambda_U+\lambda_V}\right) \left[ 2^{n_t} \left( e^{2^{-n_t}}-1 \right) - (1+2^{-(n_t+1)}) \right] .
		\end{align}
    \end{theorem}
    
\begin{proof}
    The Toffoli costs given in this theorem are simply the sum of the costs found above in \tab{intcosts2}.
    We take $\epsilon_{\rm pha}$, $\epsilon_K$, $\epsilon_R$, $\epsilon_\mathcal{M}$, and $\epsilon_t$ to be the allowable errors from phase estimation, truncation of the Dyson series at order $K$, the finite approximation of the positions of the nuclei, preparation of the state with amplitudes $1/\|\nu\|$, and finite number of bits for the time, respectively.
    Because the error in the phase estimate is independent of all other errors, the sum of the squares of this error and all other errors needs to be smaller than $\epsilon^2$.
    The expression for $\epsilon_K$ is from \eq{errK}, the expression for $\epsilon_R$ is from \lem{epsR}, the expression for $\epsilon_{\mathcal{M}}$ is from \lem{epsM}, 
    and the expression for $\epsilon_t$ is obtained from the upper bound in \eq{epsnt}.
    The probability of success of the state preparations is given by \eq{Peq}, and the number of repetitions needed for the state preparation is given by \eq{intreps}.
    We have used $\lambda^\alpha_U,\lambda^\alpha_V$ with $\alpha=1$ so we can use \lem{epsM}.
    The term $\mathcal{O}(\log(1/\epsilon))$ is to account for complexity of preparing the initial state to use as control for the phase measurement, and the processing of this control at the end to obtain the phase estimate.
\end{proof}

As discussed previously for the case of qubitization, the general way to use this theorem is to choose some allowable maxima for the errors, compute the appropriate values of the parameters required to obtain those errors, then compute the actual error to obtain $\epsilon_K$, $\epsilon_{\mathcal{M}}$, $\epsilon_R$, and $\epsilon_t$.
Those are then used in \eq{totinterr} to obtain $\epsilon_{\rm pha}$, which is then used to obtain the total complexity.
The errors $\epsilon_{\mathcal{M}}$ and $\epsilon_R$ were discussed above for \thm{qubitization}.
For the case of $\epsilon_K$, we can choose $K$ proportional to $\ln[(\lambda_U+\lambda_V)/\epsilon_V]/\ln\ln[(\lambda_U+\lambda_V)/\epsilon_V]$.
For $\epsilon_t$ we can use the approximation that the expression in the square brackets is about $2^{-2n_t}/6$ to estimate the appropriate value of $n_t$.
The value of $b_T$ is chosen to be an integer, usually around 8.
We can then adjust the values used for $n_{\mathcal{M}}$, $n_R$, $K$, $n_t$, and $b_T$ to minimise the complexity.

		%%%%%%%%%%%%%%%%%%%%%%%%%%%%%%%%%%%%%%%%%%%%%%%%%%%%%%%%%%%%%%%%%%%%%%%%%%%%%%
\section{Results and discussion}
\label{sec:results}

	\FloatBarrier

\subsection{The resources required to simulate real systems}

In this section we numerically compute and compare the costs of the central algorithms of the preceding two sections to each other, and to other methods in the literature. The constant factor Toffoli complexity is given for our qubitization approach by \thm{qubitization} and given for our interaction picture approach by \thm{interaction_cost}. While those expressions are rather complex, the simpler asymptotic complexities are given in \eq{asymptotic_qubitization} and \eq{asymptotic_interaction}, respectively. We see that we can characterize the complexity of these approaches in terms of $\eta$, $N$, $\Omega$ and $\epsilon$. The number of nuclear charges $L$ also appears in the scaling but only inside of subdominant terms that do not significantly impact the complexity. Thus, our algorithm has no leading order dependence on the location, magnitude, or number of nuclear charges. This is a stark contrast to algorithms based on Gaussian orbitals where the basis itself, and hence the $\lambda$ values as well as the cost of the quantum walk steps, are greatly influenced by the particular details of the system being simulated. Because our costs do not depend strongly on such properties, we will be able to analyze the scaling of our algorithms much more broadly than is possible with the molecular orbital based algorithms.

	\begin{figure*}[h]
		\centering
	  \begin{tabular}{c @{\qquad} c }
	  \includegraphics[width=0.45\textwidth]{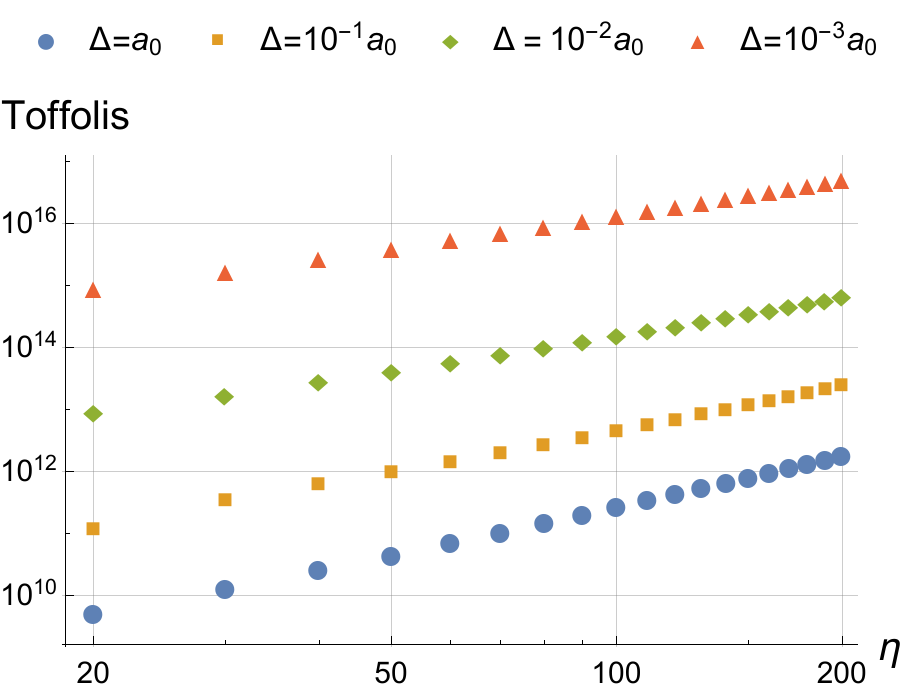}& \includegraphics[width=0.45\textwidth]{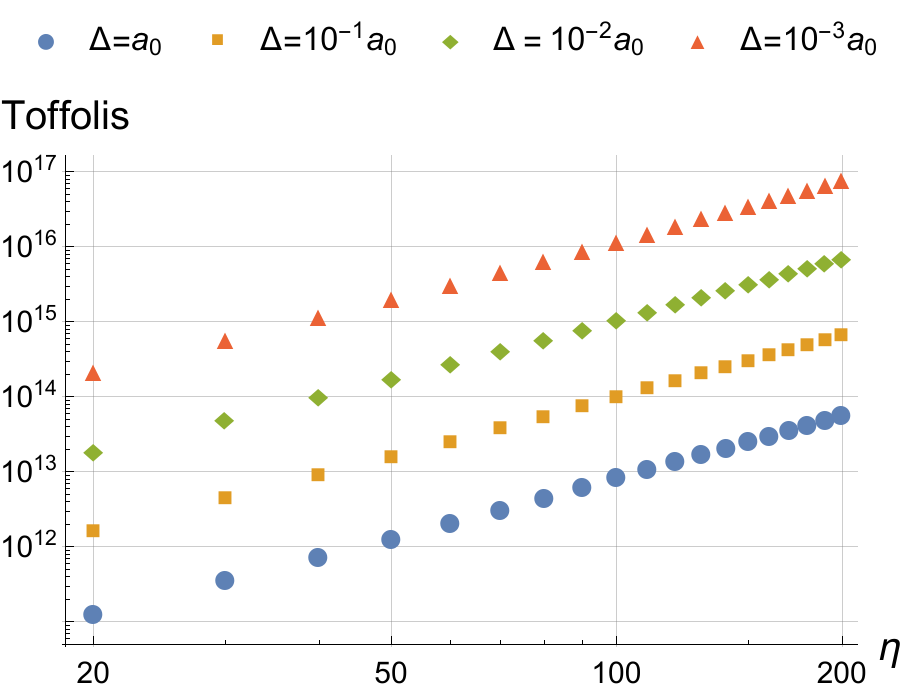}\\
	  \small (a) Qubitization Toffoli count & \small (b) Interaction picture Toffoli count \\
	  \includegraphics[width=0.45\textwidth]{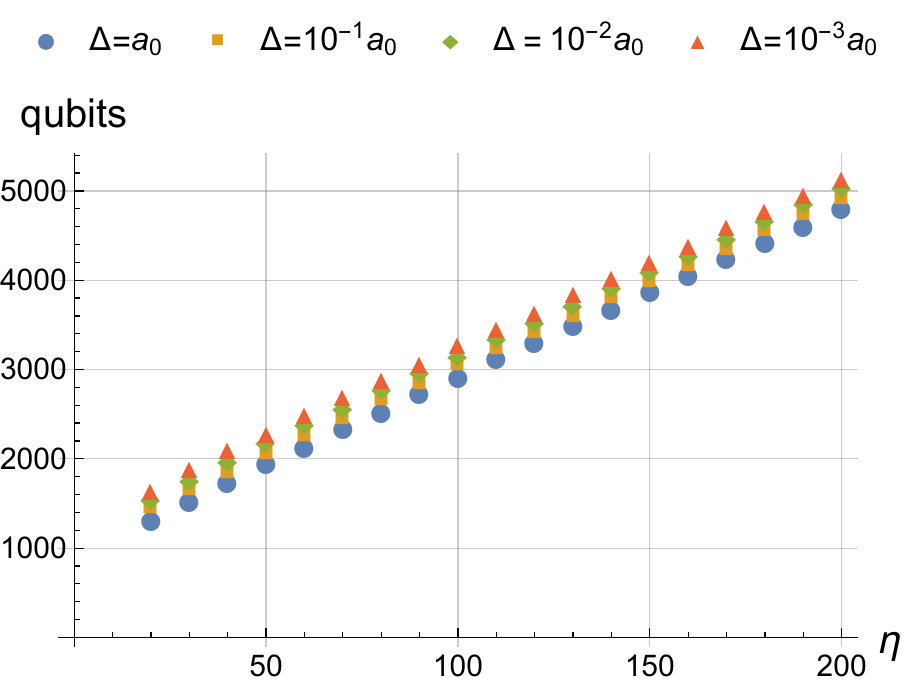}& \includegraphics[width=0.45\textwidth]{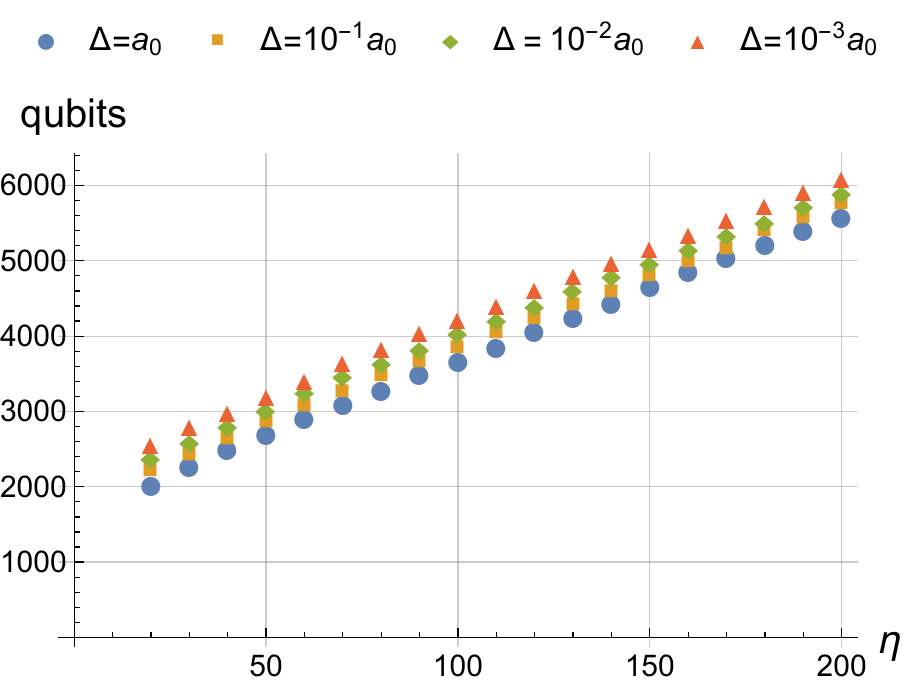}\\
	  \small (c) Qubitization logical qubit count & \small (d) Interaction picture logical qubit count
	  \end{tabular}
		\caption{Toffoli (top) and logical qubit (bottom) counts of the qubitization algorithm (left) and interaction picture algorithm (right) as a function of $\eta$. Here we take $\epsilon=0.0016$ Hartree (corresponding to ``chemical accuracy'') and $N=2^{18} \approx 2.6 \times 10^5$ but we note that outside of the dependence on $\Delta$, algorithm costs depend only logarithmically on $N$. For simplicity these plots correspond to the uniform electron gas (i.e., $U=0$) but block encoding the external potential adds only marginal overhead so we expect these costs to be broadly representative of the resources required to simulate arbitrary molecular systems as a function of $\eta$ and $\Delta$. We emphasize that $\Delta=10^{-3} a_0$ is an excessively accurate discretization that would correspond to a grid of one billion points per each Bohr radius cubed of system volume. In most cases, we expect that a resolution of $\Delta = 10^{-2} a_0$ (which is one million grid points per Bohr radius cubed) would correspond to more accurate energies than large correlation consistent Gaussian basis sets. In the bottom panels we see that the qubit counts of the two algorithms are similar.}
		\label{fig:toffolis_sigma}
	\end{figure*}
	
	\begin{figure*}[h]
		\centering
	  \begin{tabular}{c @{\qquad} c }
	  \includegraphics[width=0.45\textwidth]{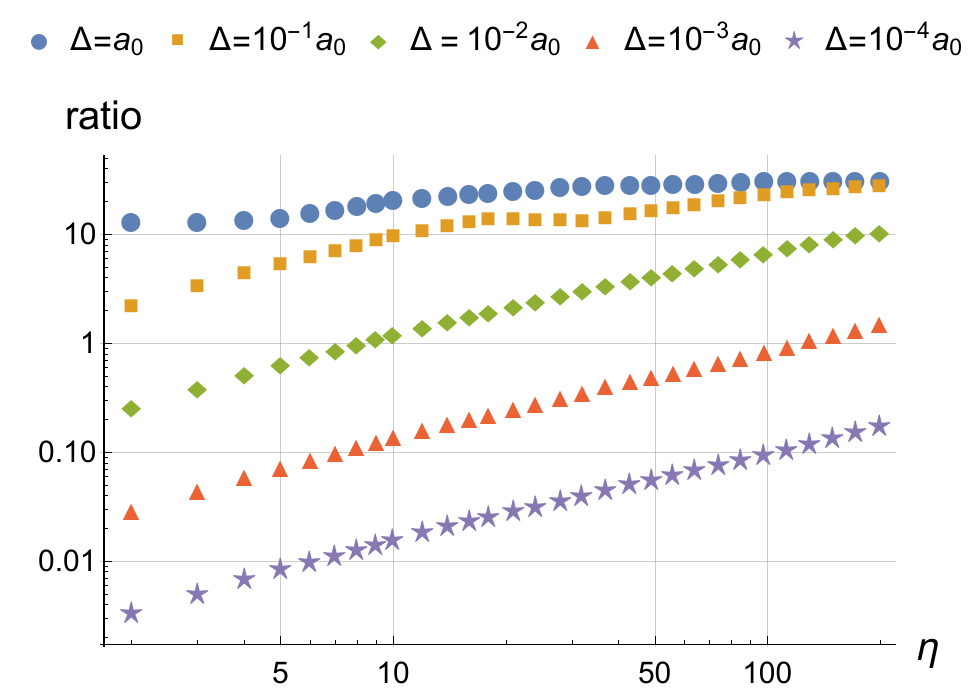} & 
	  \includegraphics[width=0.45\textwidth]{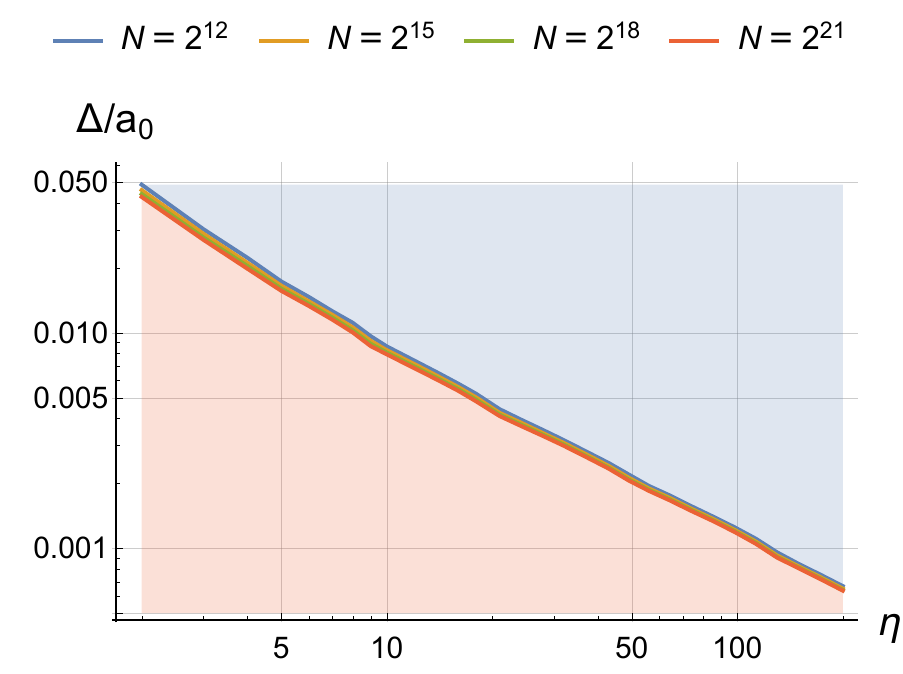} \\
	  \small (a) Ratio of Toffoli counts & \small (b) Parameter regions of advantage \\
	  \end{tabular}
		\caption{(a) The ratio of the number of Toffolis for the interaction picture approach to the number of Toffolis for the qubitization approach. We take $\epsilon=0.0016$ Hartree (corresponding to ``chemical accuracy'') and $N=2^{18} \approx 2.6 \times 10^5$. We again note that outside of the dependence on $\Delta$, algorithm costs depend only logarithmically on $N$. For simplicity these plots correspond to the uniform electron gas (i.e., $U=0$) but block encoding the external potential adds only marginal overhead so we expect these costs to be broadly representative of the resources required to simulate arbitrary molecular systems as a function of $\eta$ and $\Delta$. (b) The parameter region of $\Delta$ versus $\eta$ where there is an advantage for the Toffoli complexity for the interaction picture approach versus qubitization. The region where the interaction picture has the advantage is shown in red, and the region where qubitization has the advantage is shown in blue.
		That is, the lower left of (b) is showing the values of $\Delta$ for which the cost ratio plotted in panel (a) is less than one, as a function of $\eta$.
		In this figure we show the crossover for four different values of $N$, demonstrating that it barely changes as a function of $N$.
		In both figures we see that the interaction picture approach only has an advantage for low particle numbers and/or when simulations target unusually high precision.}
		\label{fig:toffolis_ratio}
	\end{figure*}

While the logarithmic factors in the complexity sometimes depend on $N$ without $\Omega$ (a result of the $\log N$ sized ancilla and momenta registers), in the leading order complexity we see that $N$ and $\Omega$ always appear together as a ratio. For example, assuming fixed precision (in these numerics we will target chemical accuracy, $\epsilon = 0.0016$ Hartree) we can express the leading order algorithmic gate complexities as
\begin{equation}
{\rm gates}\left({\rm qubitization}\right) = \widetilde{\cal O}\left(\frac{\eta^2}{\Delta^2}  + \frac{\eta^3}{\Delta}\right) 
\qquad
{\rm gates}\left({\rm interaction \,\, pic}\right) =\widetilde{\cal O}\left( \frac{\eta^3}{\Delta}\right) 
\qquad
\Delta = \left(\frac{\Omega}{N}\right)^{1/3}
\label{eq:resolution_complexity}
\end{equation}
where we have defined $\Delta$ as the spacing of the plane wave reciprocal lattice. Intuitively, $\Delta$ can be thought of as the length scale of resolution within our basis (e.g., wavefunction features smaller than $\Delta$ cannot be resolved but those larger than $\Delta$ can be resolved). It is easy to see that the $1/\Delta$ dependence arises from the maximum energy of the Coulomb operator between grid points separated by $\Delta$ whereas the $1/\Delta^2$ dependence arises due to the maximum energy of the kinetic operator which is proportional to the square of velocities on such a grid. If one desires to model a non-periodic system using plane waves then the typical strategy is to choose a sufficiently large box size $\Omega$ such that the periodic images of the Coulomb operator do not interact\footnote{A more efficient procedure for modeling non-periodic systems using plane waves is to use a truncated Coulomb operator with a slightly larger cell size~\cite{rozzi2006exact,sundararaman2013regularization,ismail2006truncation} or to subtract terms from the molecular multipole \cite{makov1995periodic2} or Martyna-Tuckerman corrections \cite{Tuckerman1999}. This can likely be incorporated into our approach in a straightforward fashion but would require minor modifications to the potential block encodings which we leave for future work.}. At first glace, the larger one makes the box, the cheaper our algorithms due to the inverse dependence on $\Omega$. However, as $\Omega$ gets larger the resolution $\Delta$ decreases, and thus to maintain reasonable precision one must also increase $N$. By expressing the complexity as in \eq{resolution_complexity} we see that if the goal is perform a calculation with fixed spatial resolution $\Delta$, the exact size of the box will not substantially impact the complexity.

In \fig{toffolis_sigma} we compare the complexities of the two algorithms of this paper by plotting the total Toffolis required against $\eta$ for several different values of $\Delta$. While the interaction picture algorithm has slightly better scaling in terms of $\Delta$, the constant factors are meaningfully worse.
To make the comparison between the complexities more directly, in \fig{toffolis_ratio} we show both the ratio of the Toffoli complexities and the parameter regimes where each approach is more efficient.
It would appear that the qubitization algorithm has lower cost in almost all regimes of interest. The exception would be for \emph{extremely} high resolution simulations where both $\eta$ and $\Delta$ are on the smaller side of what we would imagine requiring (there is little dependence on $N$).
For example, the interaction picture algorithm is cheaper when $\eta = 20$ and $\Delta = 10^{-3} a_0$. While such precisions might be appropriate for performing dynamics, we not anticipate that those precisions will typically be required for preparing and sampling observables from molecular eigenstates. In terms of spatial resolution, $\Delta = 10^{-3} a_0$ is quite high since it corresponds to a reciprocal lattice grid with a billion lattice points per Bohr radius cubed.
Of course, at higher values of $\eta$ or $\Delta$ the term with $\eta^3 / \Delta$ scaling dominates in the qubitization cost, which gives a scaling that is similar to that of the interaction picture algorithm but with lower constant factors. 
A priori, it was not obvious to us that the qubitization algorithm would be more efficient than the interaction picture algorithm for modest values of $\eta$. Thus, the fact that qubitization is usually more practical for chemistry problems, is an important finding of our work.

\begin{table*}[t]
\begin{tabular}{|c|c|c|c|c|c|c|}
\hline
system
& conventional cell type
& $\Omega / a_0^3$
& $\eta$ (cell total)
& $r_s / a_0$ (total)
& $\eta$ (cell valence)
& $r_s / a_0$ (valence)\\
\hline\hline
metallic lithium
& body-centered-cubic
& 284.94
& 6
& 2.25
& 2
& 3.24\\
metallic potassium
& body-centered-cubic
& 961.67
& 38
& 1.82
& 2
& 4.86\\
diamond
& diamond cubic
& 307.04
& 48
& 1.15
& 32
& 1.32\\
crystalline silicon
& diamond cubic
& 1,080.43
& 112
& 1.32
& 32
& 2.01\\
iron (II) oxide
& NaCl structure
& 539.84
& 136
& 0.98
& 52
& 1.35\\
cobalt oxide
& NaCl structure
& 522.81
& 140
& 0.96
& 60
& 1.28\\
aluminium arsenide
& zincblende
& 1,197.86
& 184
& 1.16
& 32
& 2.08\\
indium phosphide
& zincblende
& 1,364.93
& 256
& 1.08
& 32
& 2.17\\
\hline
\end{tabular}
\caption{\label{tab:wigner_seitz} Values of $\eta$, $\Omega$ and $r_s$ for a variety of common materials. Above, $a_0$ is the Bohr radius of the hydrogen atom and thus $r_s$ values are in units of Bohr and $\Omega$ values are units of ${\rm Bohr}^3$. Often one would want to perform a simulation with multiple unit cells in order to resolve correlations in between the units. Doing this would not change $r_s$ but it would increase $\eta$ proportional to the number of unit cells included. The valence $\eta$ and $r_s$ values correspond to simulations with core electrons removed. Doing this would require the use of pseudopotentials, which is beyond the scope of the current paper but likely an important next step. Note further that the unit cell volumes tabulated here correspond to conventional unit cells, rather than primitive unit cells. Using primitive unit cells would necessitate adapting our approach to reciprocal spaces beyond the cubic Bravais lattices considered here.}
\end{table*}

	\begin{figure*}[h]
		\centering
		  \begin{tabular}{c @{\qquad} c }
		  \includegraphics[width=0.45\textwidth]{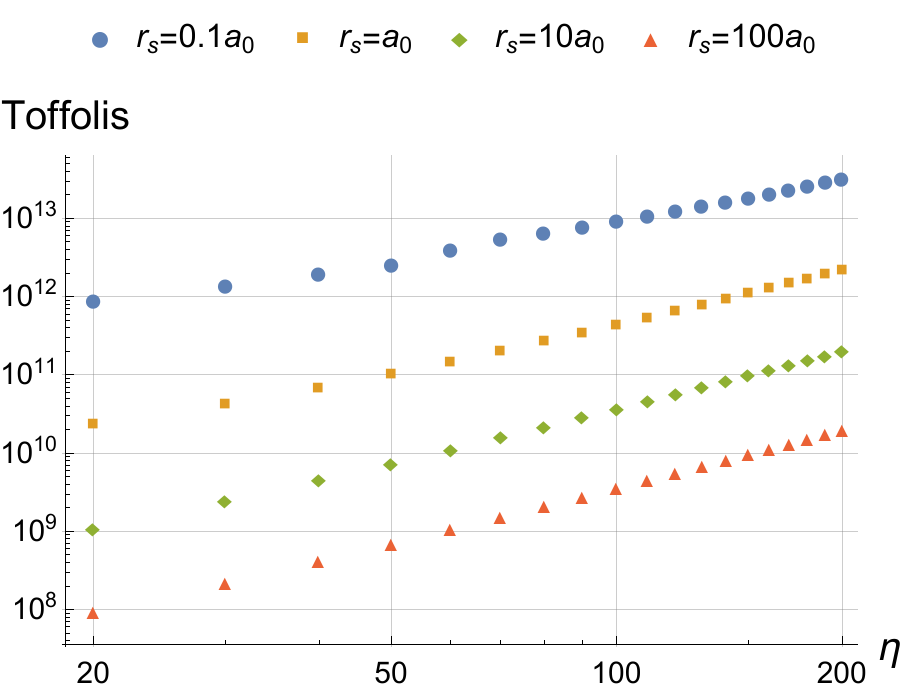} & \includegraphics[width=0.45\textwidth]{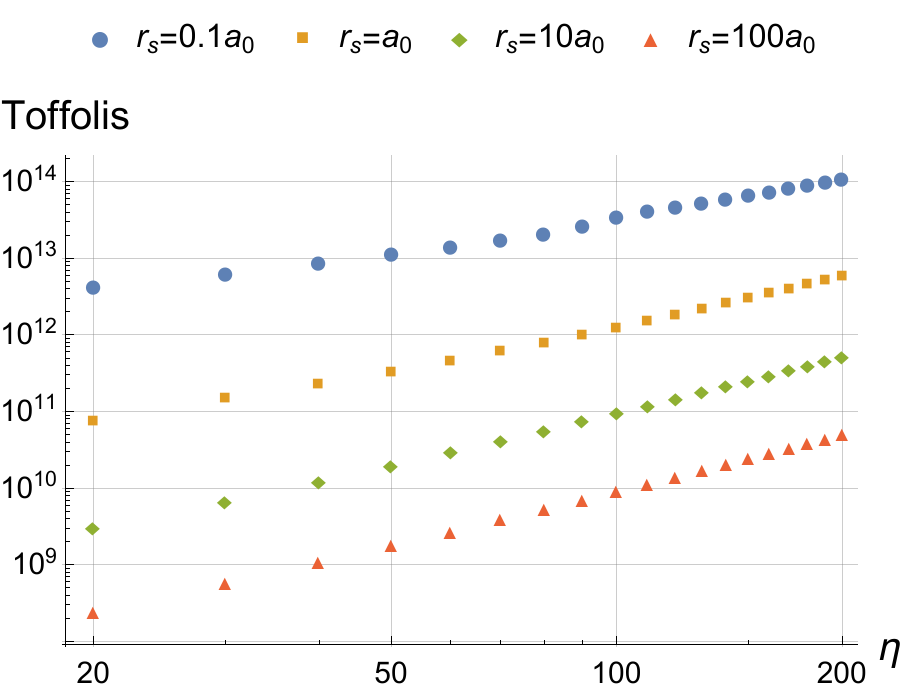} \\
		  \small (a) $N=2^{12} \approx 4.1 \times 10^3$ & \small (b) $N=2^{15} \approx 3.3 \times 10^4$ \\
		  \includegraphics[width=0.45\textwidth]{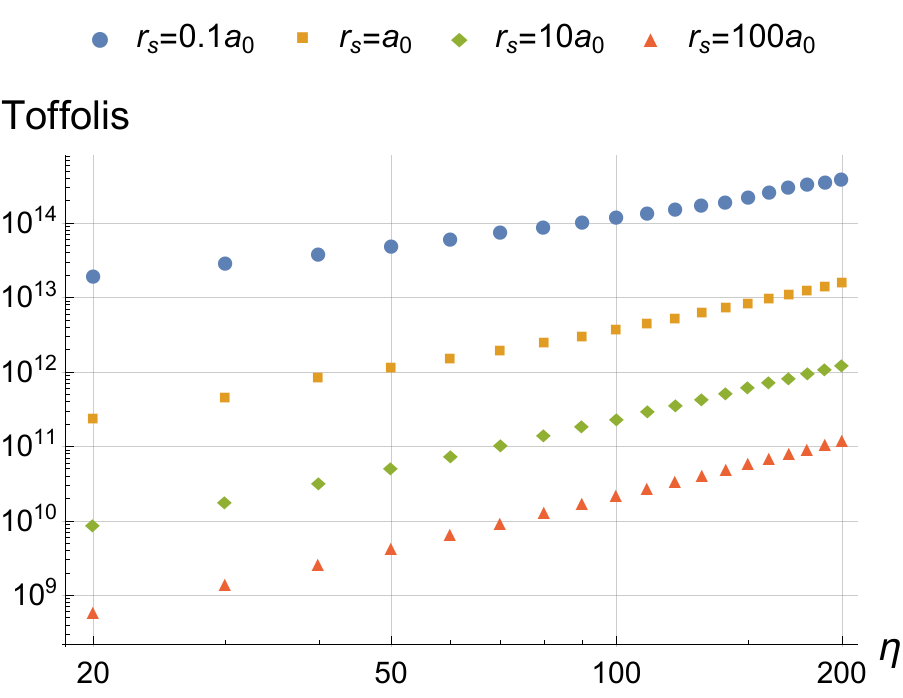} & \includegraphics[width=0.45\textwidth]{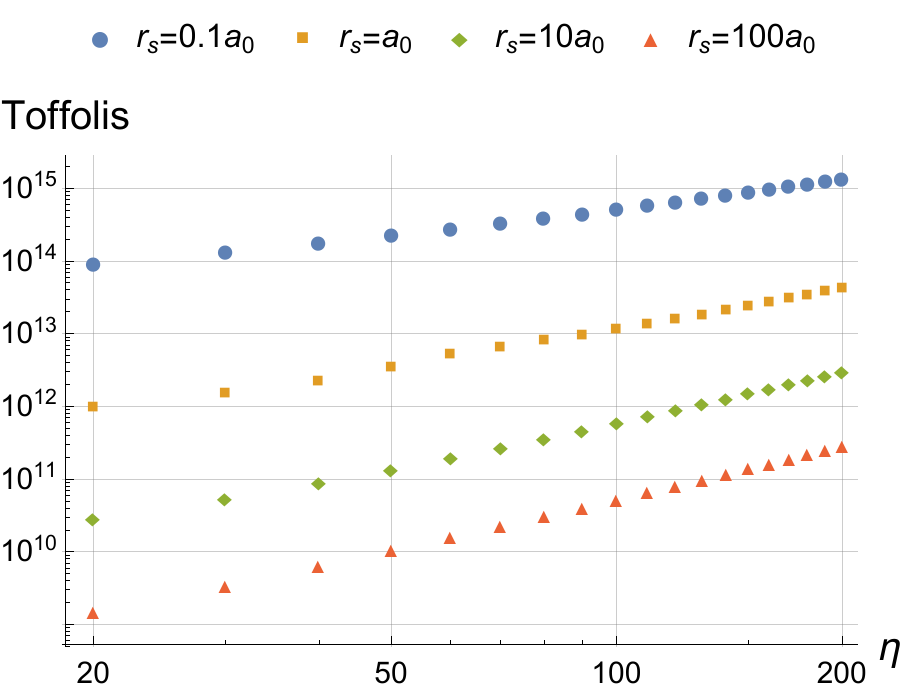} \\
		  \small (c) $N=2^{18} \approx 2.6 \times 10^5$ & \small (d) $N=2^{21} \approx 2.1 \times 10^6$
		  \end{tabular}
		\caption{$\Toffoli$ count of qubitization as a function of $\eta$ using the improved techniques of this paper. In contrast to \fig{toffolis_sigma}, we plot these costs for fixed values of the Wigner-Seitz radius $r_s$ to give a sense of the scaling when $\eta$ increases proportionally to $\Omega$ (as would be the case when growing condensed phase simulations towards their thermodynamic limit). Thus, we expect these plots to be representative of costs when simulating periodic materials. We see that the qubitization based algorithm performs better across all regimes studied to compared to the interaction picture based algorithm (\fig{toffoli_interaction}).}
		\label{fig:toffoli_qubitization}
	\end{figure*}

	\begin{figure*}[h]
		\centering
		  \begin{tabular}{c @{\qquad} c }
		  \includegraphics[width=0.45\textwidth]{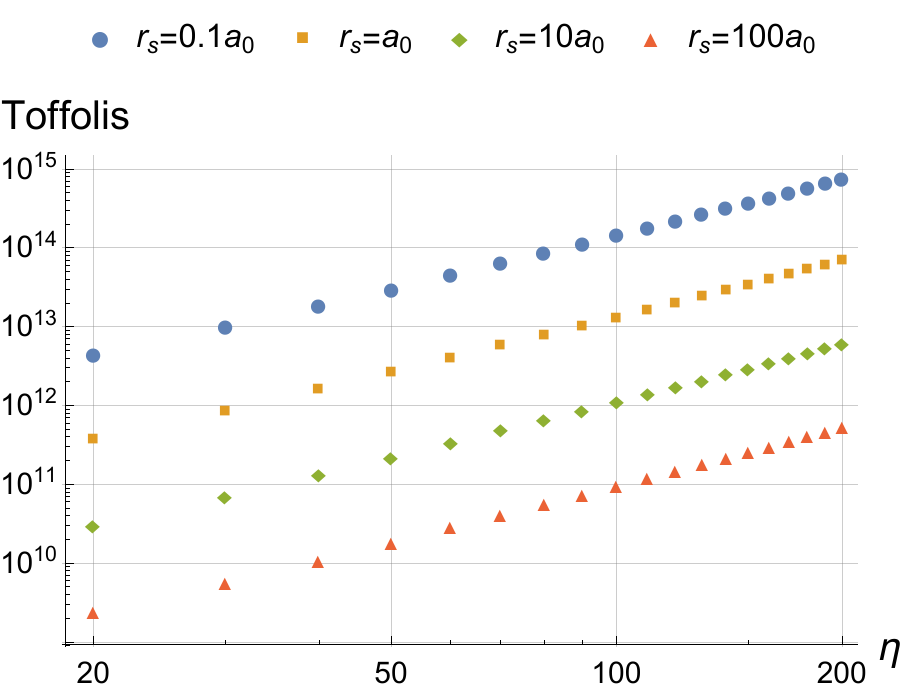} & \includegraphics[width=0.45\textwidth]{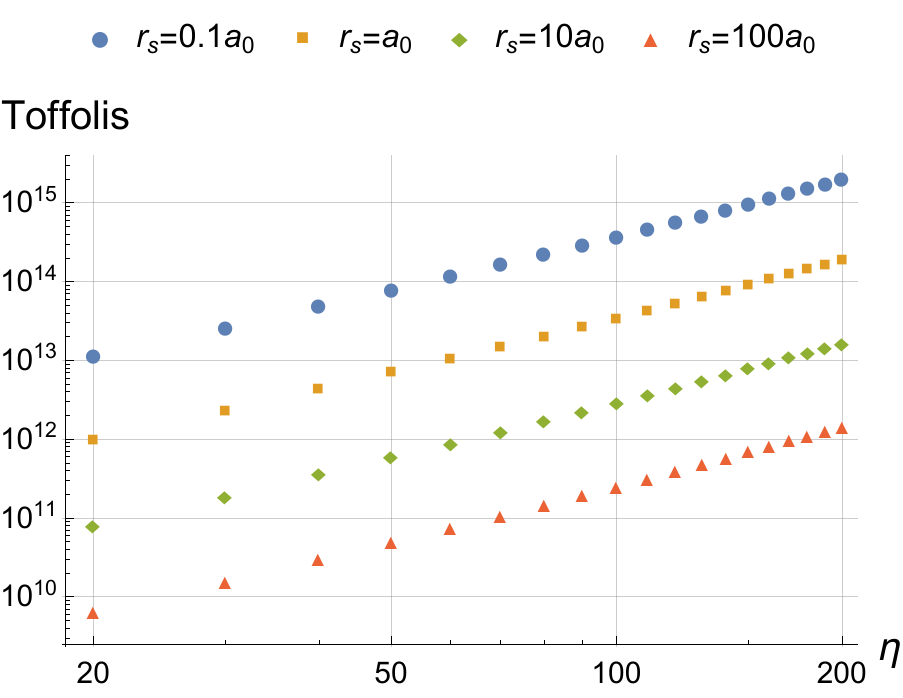} \\
		  \small (a) $N=2^{12} \approx 4.1 \times 10^3$ & \small (b) $N=2^{15} \approx 3.3 \times 10^4$ \\
		  \includegraphics[width=0.45\textwidth]{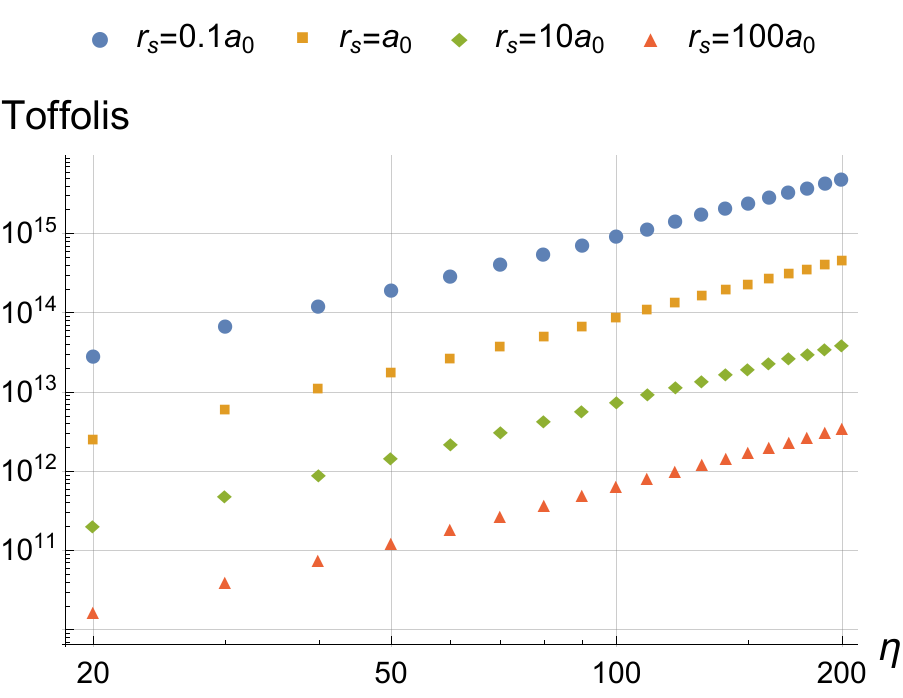} & \includegraphics[width=0.45\textwidth]{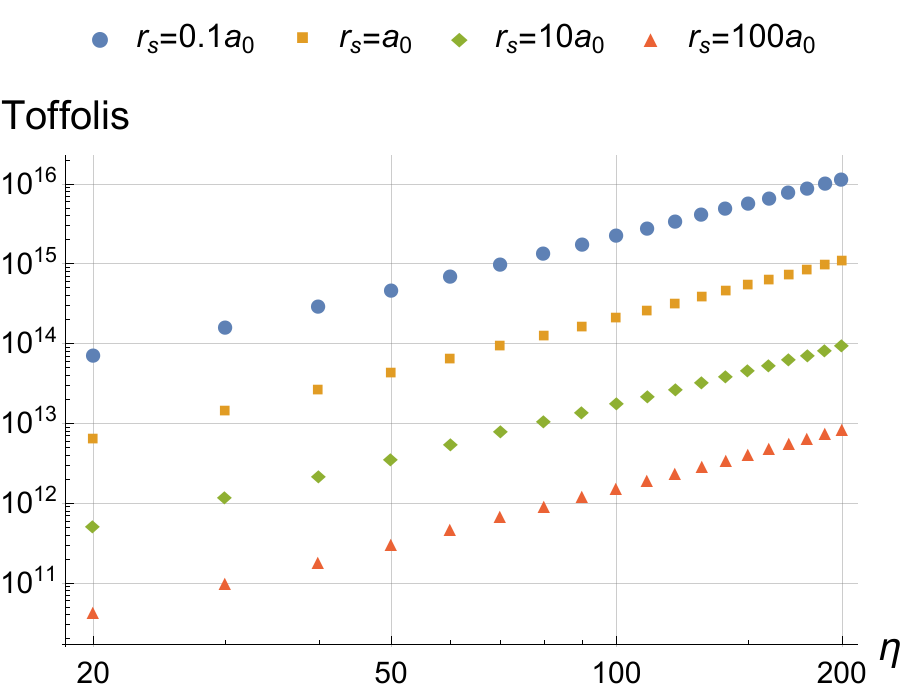} \\
		  \small (c) $N=2^{18} \approx 2.6 \times 10^5$ & \small (d) $N=2^{21} \approx 2.1 \times 10^6$
		  \end{tabular}
		\caption{$\Toffoli$ count of interaction picture algorithm as a function of $\eta$ using improved techniques. In contrast to \fig{toffolis_sigma}, we plot these costs for fixed values of the Wigner-Seitz radius $r_s$ to give a sense of the scaling when $\eta$ increases proportionally to $\Omega$ (as would be the case when growing condensed phase simulations towards their thermodynamic limit). Thus, we expect these plots to be representative of costs when simulating periodic materials. We see that the scaling depends rather weakly on the number of plane waves, $N$.}
		\label{fig:toffoli_interaction}
	\end{figure*}

Next, we turn to the topic of how our algorithm scales for the important application of simulating periodic systems such as crystals, metals, surfaces and polymers. Such systems are ubiquitous in nature and are a major focus of electronic structure methods deployed to materials science. While it is important to understand how our algorithm scales as one approaches the continuum basis set limit where $\Delta \rightarrow 0$ (as discussed in the context of \fig{toffolis_sigma}) for condensed phase applications it is also important to understand how our algorithm scales towards the thermodynamic limit. That is, the computational cell size should be larger than the correlation length in order for the simulation to faithfully capture the behavior of the infinite system.  The computational cell will be composed of multiple unit cells (the simplest repeating unit in the crystalline lattice). For such simulations, one can imagine adding more unit cells to the computational cell but keeping the overall density of particles (i.e., $\eta / \Omega$) fixed. Often, behavior in the thermodynamic limit is extrapolated from a series of progressively larger calculations. In order to discuss simulations with fixed $\eta / \Omega$ in more intuitive terms, we will analyze scaling in terms of the Wigner-Seitz radius, $r_s$ (a common measure of a system's average radius between electrons), defined through the relation
\begin{equation}
\label{eq:wigner_seitz}
\frac{4}{3} \pi r_s^3 = \frac{\Omega}{\eta}, \qquad \qquad r_s = \left(\frac{3 \Omega}{4 \pi \eta}\right)^{1/3} \, .
\end{equation}
We report Wigner-Seitz radii for several common materials in \tab{wigner_seitz} and note that they are mostly between one and a few Bohr radii. In terms of Wigner-Seitz radii we can express the overall Toffoli complexities as
\begin{equation}
    {\rm Toffolis}\left({\rm qubitization}\right) \! = \widetilde{\cal O}\!\left(\frac{\eta^{4/3} N^{2/3}}{r_s^2}  + \frac{\eta^{8/3} N^{1/3}}{r_s}\right) 
\qquad 
{\rm Toffolis}\left({\rm interaction \,\, pic}\right) \! =\widetilde{\cal O}\!\left(\frac{\eta^{8/3} N^{1/3}}{r_s}\right) \, .
\qquad 
\end{equation}

In \fig{toffoli_qubitization} and \fig{toffoli_interaction} we plot the estimated resources required to perform condensed phase simulations for the qubitization and interaction picture algorithms, respectively. We make these plots between $\eta=20$ and $\eta=200$, between $N=2^{12}$ and $N=2^{21}$ and between $r_s = 0.1$ and $r_s=100$. We observe that the qubitization algorithm is most practical due to lower constant factors across all regimes plotted. As before, in the limit of large $\eta$, the qubitization algorithm also recovers the better scaling of the interaction picture algorithm (but with better constant factors). By comparing \fig{toffoli_qubitization} with the requirements to simulate various materials in \tab{wigner_seitz} we see that our method can simulate classically intractable real material Hamiltonians with between $10^{10}$ and $10^{12}$ Toffoli gates. 
These numbers of Toffoli gates and qubits are comparable to those (for a different system) in \cite{Lee2020}, where estimates are given of the resources needed to implement the algorithm on a real device.
If we adopt the same assumptions here, including a surface code cycle time of 1 \micro s, our proposed quantum algorithms would require on the order of a few million physical qubits and a few days of runtime.
No prior quantum algorithm for chemistry has been shown to require anywhere near this few resources to simulate realistic descriptions of condensed phase molecular systems.

We now address the question of how the more practical of the two algorithms described in this work (the qubitization one) compares to other methods for quantum simulating chemistry. To date, the only prior methods that have been assessed for viability within fault-tolerance have used second quantization. These papers fall into two categories: those that focus on plane waves (see \tab{plane_wave_comparison}) and those that focus on Gaussian orbitals. The plane wave algorithms assessed for fault-tolerance are by Babbush \emph{et al.}~\cite{BGBWMPFN18} and Kivlichan \emph{et al.}~\cite{Kivlichan2020improvedfault}. Because the present work also uses plane waves it is straightforward to make a comparison. As discussed in \sec{problem_definition}, for simulating molecules these prior papers using plane waves are not particularly practical options because at least tens or hundreds of thousands of plane waves would be required for reasonably accurate simulations and those works have a space complexity scaling linearly in $N$. Those algorithms also have significantly worse Toffoli complexity, scaling roughly like $\widetilde{\cal O}(N^3)$. 

The main motivation for those past papers using plane waves in second quantization is essentially a different regime than that of realistic simulations of materials. Those papers emphasize their resource estimates in terms of the first classically intractable applications of very small fault-tolerant processors (perhaps only a few hundred thousand physical qubits). In this sense, the goal is to articulate example problems on which ``quantum supremacy'' \cite{Preskill2012b,Boixo2016} might be first demonstrated for quantum chemistry in an error-corrected context, rather than to highlight applications corresponding to realistic simulations of nature. For example, at small $N$ and high filling fraction $\eta / N$, there are a variety of model systems that are not converged to the continuum limit but provide examples of electronic structure that would be classically intractable to simulate within the basis. The value of such simulations for the jellium model (the homogeneous electron gas that occurs when $U=0$) near $r_s = 10 a_0$, is argued in Section III of \cite{BabbushLow}. As a result, this model has been used as a benchmark in several subsequent papers. For example, the largest jellium simulation reported in Table III of \cite{BGBWMPFN18} would require $2.1 \times 10^{10}$ Toffoli gates\footnote{Many of the prior works we compare to, including Kim \emph{et al.}~\cite{Kim2021b}, Kivlichan \emph{et al.}~\cite{Kivlichan2020improvedfault} and Babbush \emph{et al.}~\cite{BGBWMPFN18}, report the number of T gates required rather than the number of Toffoli gates required. In order to facilitate a clear comparison we assume that the cost of a Toffoli gate is equal to the cost of 2 T gates. This equivalence is justified by the surface code spacetime volume required for a Toffoli gate compared to a T gate using the catalyzed T gate distillation scheme of \cite{Gidney2019c}.} and $1,\!136$ logical qubits when $r_s = 10 a_0$. This simulation would involve just $1,\!024$ plane wave spin-orbitals (equivalent here to $N=512$). The largest $r_s = 10 a_0$ jellium system analyzed in Table II of Kivlichan \emph{et al.}~\cite{Kivlichan2020improvedfault} involves only $N=125$ plane waves but would require $378$ logical qubits and $2.8 \times 10^9$ Toffolis. Because those simulations occur at half-filling ($\eta = N$), they are likely more efficient than the methods of this paper at those system sizes. However, if we instead target a more physically accurate simulation (with less basis set discretization error) where $N \geq 10^3$ and $\eta / N \leq 1/10$, we would expect the first quantized approach of this paper to require far fewer error-corrected resources than these past algorithms utilizing second quantization.

To compare the the complexities reported here, the smallest value of $N$ we consider here is for $N=2^{12}$, which is $8$ times larger than in \cite{BGBWMPFN18}.
The $N^3$ scaling of the second quantization approach would mean that the complexity is about $10^{13}$ Toffolis, nearly 100 times larger than the largest complexity we have for $r_s=10 a_0$ for the first quantized qubitization approach.
For the largest value of $N$ considered of $2^{21}$, the $N^3$ scaling would indicate a complexity of over $10^{21}$ Toffolis, around a billion times larger than the complexity of the first quantized approach here.

Unlike second quantized plane wave methods, second quantized algorithms using Gaussian basis functions seem promising options for practical quantum simulations of realistic molecular Hamiltonians. There are now a handful of papers presenting in depth analysis of new fault-tolerant algorithms \cite{Reiher2017,Berry2019b,vonBurg2020Catalysis,Lee2020} in this representation. Dating back to the first of these papers by Reiher \emph{et al.}~\cite{Reiher2017}, these papers have demonstrated algorithmic improvements by estimating the resources required to simulate a few benchmark molecules such as the FeMoCo molecule relevant to biological nitrogen fixation. The work by Lee \emph{et al.}~\cite{Lee2020} holds the record for the Gaussian orbital based algorithms with both the lowest constant factor and lowest asymptotic Toffoli scaling. We emphasize that all of these papers have focused on simulating molecular systems in active spaces, and none of them have developed methods for simulating periodic materials. As a result it is difficult to directly compare the algorithms developed here to those prior papers.

Active space Hamiltonians are coarse grained versions of larger molecular Hamiltonians, designed to capture the physics of a subset of the total electronic degrees of freedom. Often the active space electrons are the valence electrons involved in the breaking or formation of chemical bonds. For example, the FeMoCo molecule (${\rm Fe}_7 {\rm Mo} {\rm S}_9 {\rm C}$) is a large system with 374 electrons; however, the active spaces studied by these papers require only 54 \cite{Reiher2017} or 113 \cite{Li2019_b} electrons. A variety of methods are available for computing active space Hamiltonians~\cite{stein2019autocas, sayfutyarova2017automated, keller2015selection, bao2019automatic, hermes2020variational} that are compatible with molecular orbital based quantum algorithms. Most such approaches involve removing from the basis molecular orbitals that are expected to be completely occupied in the ground state and adding the potential arising from the charge density of those removed electrons to $U$. Unfortunately, such methods are not immediately compatible with the algorithms developed in this paper. This is because the molecular orbitals occupied by the core electrons are not orthogonal to the rest of the plane wave basis and thus cannot be easily removed from the basis.

However, there is a way to make similar approximations for plane wave methods. In particular, pseudopotentials \cite{tosoni2007comparison,booth2016plane} provide a well studied methodology for removing core electrons from plane wave calculations by modifying the external potential to account for the core. In addition to reducing the effective system size $\eta$, pseudopotentials dramatically accelerate the convergence of the basis error. The reason for this is because the cusps in electronic wavefunction which occur at the nuclei are smoothed out by pseudpotentials. Because the Fourier transform of smooth functions is exponentially convergent, the plane wave basis can then converge dramatically faster. Thus, we suggest that the use of pseudopotentials~\cite{Kleinam1982, dal2014pseudopotentials} is likely to be essential for the practical deployment of these methods. However, as the present paper is already rather long, we leave the topic of removing core electrons via pseudopotentials to future work. Instead, we will compare our algorithms to Gaussian orbital resource estimates that do not use active spaces.

\begin{table*}[h]
\begin{tabular}{|c|c|c|c|c|c|c|c|}
\hline
reference
& system
& basis
& $\Omega$
& $\eta$
& $N$
& logical qubits
& Toffoli count\\
\hline\hline
Kim \emph{et al.}~\cite{Kim2021b}
& ethylene carbonate
& STO-3G
& N/A
& 46
& 34
& 2,685
& $3.2 \times 10^{10}$\\
Kim \emph{et al.}~\cite{Kim2021b}
& ethylene carbonate
& 6-311G
& N/A
& 46
& 90
& 14,492
& $8.5 \times 10^{11}$\\
Kim \emph{et al.}~\cite{Kim2021b}
& ethylene carbonate
& cc-pVDZ
& N/A
& 46
& 104
& 16,698
& $1.3 \times 10^{12}$\\
Kim \emph{et al.}~\cite{Kim2021b}
& ethylene carbonate
& cc-pVTZ
& N/A
& 46
& 236
& 81,958
& $3.1 \times 10^{13}$\\
This work
& ethylene carbonate
& plane waves
& $10^5 a_0^3$
& 46
& $4.2 \times 10^3$
& 1,395
& $2.5\times 10^{10}$ \\
This work
& ethylene carbonate
& plane waves
& $10^5 a_0^3$
& 46
& $3.3 \times 10^4$
& 1,701
& $6.6\times 10^{10}$ \\
This work
& ethylene carbonate
& plane waves
& $10^5 a_0^3$
& 46
& $2.6 \times 10^5$
& 2,021
& $1.7\times 10^{11}$ \\
This work
& ethylene carbonate
& plane waves
& $10^5 a_0^3$
& 46
& $2.1 \times 10^6$
& 2,355
& $4.2\times 10^{11}$ \\
\hline
Kim \emph{et al.}~\cite{Kim2021b}
& ${\rm LiPF}_6$
& STO-6G
& N/A
& 72
& 44
& 3,507
& $1.0 \times 10^{11}$\\
Kim \emph{et al.}~\cite{Kim2021b}
& ${\rm LiPF}_6$
& 6-311G
& N/A
& 72
& 112
& 18,423
& $1.9 \times 10^{12}$\\
Kim \emph{et al.}~\cite{Kim2021b}
& ${\rm LiPF}_6$
& cc-pVDZ
& N/A
& 72
& 116
& 18,600
& $1.7 \times 10^{12}$\\
Kim \emph{et al.}~\cite{Kim2021b}
& ${\rm LiPF}_6$
& cc-pVTZ
& N/A
& 72
& 244
& 84,721
& $3.3 \times 10^{13}$\\
This work
& ${\rm LiPF}_6$
& plane waves
& $10^5 a_0^3$
& 72
& $4.1 \times 10^3$
& 1,758
& $8.0\times 10^{10}$ \\
This work
& ${\rm LiPF}_6$
& plane waves
& $10^5 a_0^3$
& 72
& $3.3 \times 10^4$
& 2,150
& $2.1\times 10^{11}$ \\
This work
& ${\rm LiPF}_6$
& plane waves
& $10^5 a_0^3$
& 72
& $2.6 \times 10^5$
& 2,556
& $5.1\times 10^{11}$ \\
This work
& ${\rm LiPF}_6$
& plane waves
& $10^5 a_0^3$
& 72
& $2.1 \times 10^6$
& 2,976
& $1.3\times 10^{12}$ \\
\hline
\end{tabular}
\caption{\label{tab:kim_comparison} Table comparing the cost of the qubitization algorithm of this paper to the cost of a state-of-the-art second quantized Gaussian basis algorithm recently analyzed by Kim \emph{et al.}~\cite{Kim2021b}. The work of Kim \emph{et al.}~estimates the resources required to simulate systems, which despite being small single molecules, are presented as relevant to the chemistry of lithium batteries. Their estimates are a natural point of comparison because they do not employ active spaces, unlike the focus of most past work using molecular orbitals \cite{Reiher2017,Berry2019b,vonBurg2020Catalysis,Lee2020,Elfving2020b}. We focus on two representative molecules from that work: ethylene carbonate and ${\rm LiPF}_6$, discretized using four different Gaussian orbital basis sets. We note that the algorithm of Lee \emph{et al.}~\cite{Lee2020} has been shown to scale better than the algorithm of von Burg \emph{et al.}~\cite{vonBurg2020Catalysis}, on which the work by Kim \emph{et al.}~is based, but all resource estimates using that approach involve active spaces. We expect that in most contexts, if one uses a number of plane waves that is more than a thousand times the number of orbitals in a given Gaussian basis set, then one can expect similar or higher precision from the larger plane wave calculation, even for small non-periodic systems like the ones analyzed here. This would suggest that for $N=2.6 \times 10^5$ plane waves, our algorithms may achieve higher precision than even the correlation consistent triple zeta (cc-pVTZ) basis set calculations studied by Kim \emph{et al.}~\cite{Kim2021b}. Thus, the results in this table may suggest that our algorithms requires significantly less surface code spacetime volume to reach similar levels of precision compared to some of the most performant Gaussian based calculations.}
\end{table*}

A recent paper by Kim \emph{et al.}~\cite{Kim2021b} provides resource estimates based on the algorithm of von Burg \emph{et al.}~\cite{vonBurg2020Catalysis} (with modifications from Lee \emph{et al.}~\cite{Lee2020}) that happens to not use active spaces, thus enabling a comparison to the present work which we provide in \tab{kim_comparison}. However, it is still difficult to know precisely how many plane waves would be required to match the accuracy of a given Gaussian orbital calculation. While both suppress basis error as ${\cal O}(1/N)$, even if a given Gaussian orbital calculation is more accurate than a given plane wave calculation at fixed $N$, plane waves tend to converge more systematically than Gaussian orbitals to the continuum limit. As a result, one can more reliably predict the complete basis set limit energy from extrapolations made from a series of calculations in progressively larger basis sets when using plane waves.

Appendix E of \cite{BabbushLow} attempts to give an idea of the relative prefactors between the number of plane waves and Gaussians required to obtain similar levels of precision. While the precise ratio of plane waves to Gaussians will depend on system specific details, there it is suggested that periodic calculations with roughly $20\, N$ plane waves will often give higher precision than the same calculations with $N$ Gaussian orbitals. It is more difficult to simulate non-periodic systems like single molecules with plane waves and reducing this periodicity requires yet more plane waves by about a factor of $8$ \cite{fusti2002accurateb,BabbushLow}. Achieving this convergence with plane waves may require softening nuclear cusps with pseudopotentials. Moreover, when targeting extremely high precision, the ratio of the required number of plane waves to Gaussians should decrease because the challenge becomes resolving the electron-electron cusp. Since the electron-electron cusp occurs at all points in space, Gaussians centered on nuclei have no special advantages for resolving it. The fact that Gaussians derive their advantage from being centered on nuclei also means that as the number of atoms in a fixed volume increases, the number of Gaussian orbitals to represent the system to fixed accuracy must also increase roughly proportionally. On the other hand, plane waves are not biased towards representing any particular nuclear locations; thus, as the number of atoms per unit volume is increased, the ratio of the number of plane waves to the number of Gaussians required for the same accuracy should decrease.

Considering all of these factors, \tab{kim_comparison} suggests that our qubitized first quantized plane wave algorithm might be more performant than the best Gaussian orbital basis algorithms applied to even small non-periodic molecules. For example, when comparing our approach with $N=2.6 \times 10^5$ plane waves to the approach of Kim \emph{et al.}~\cite{Kim2021b} using a triple zeta basis set for which $N$ is between 236 and 244 for the two systems included in \tab{kim_comparison}, see that we are using more than a thousand times the number of plane waves as Gaussian orbitals but that the first quantized plane wave approach needs several dozen times fewer logical qubits and about a hundred times fewer Toffoli gates (amounting to approximately three orders of magnitude less surface code spacetime volume). Furthermore, it is clear that if the objective is significantly lower basis set error (e.g., if spectroscopic precision of $10^{-6}$ Hartree is desired) or if the goal is to simulate periodic systems such as crystals, metals, polymers, or surfaces, then our first quantized qubitization approach is dramatically more practical than other fault-tolerant algorithms described in the literature. Thus, our analysis has shown that these methods are an important and state-of-the-art tool for the quantum simulation of chemistry.

\subsection{Conclusion and outlook}

    We have completed the first constant factor analysis of the resources required to implement first quantized quantum simulations of quantum chemistry. Unsurprisingly, these methods appear to be the most promising approach for simulating realistic materials Hamiltonians within quantum error-correction. More surprisingly, they are also highly competitive with the best second quantized approaches for simulating single molecules (often requiring fewer surface code resources by multiple orders of magnitude).  We focused on optimizing and compiling the first quantized algorithms using a plane wave basis first described by Babbush \emph{et al.}~\cite{BabbushContinuum_b}. Despite the interaction picture of that work having better scaling, we found that the qubitization algorithm of that work tends to require fewer resources for the regimes of interest to us.
    
    Our work goes significantly beyond just translating the prior asymptotic algorithm descriptions of Babbush \emph{et al.}~\cite{BabbushContinuum_b} to concrete circuits, even though just doing that would have still required a similarly detailed manuscript. We also introduce many improvements to the implementation of both the qubitized form and the interaction picture simulation that reduce the complexity. The innovations we introduce that lead to the most significant reduction in algorithmic complexity are summarized below.
    \begin{itemize}[leftmargin=*]
        \item In the case of qubitizing the Hamiltonian, we entirely eliminate the cost of the arithmetic for the kinetic energy term $T$ (see \sec{block_encode_t}).
        Instead of computing the sum of squares of momenta, control registers are used to select the momenta as well as the individual bits to multiply.
        The only arithmetic is the multiplication of bits, which is performed with a single Toffoli.
        This is a significant improvement, because multiplication would otherwise be the main cost.
        \item For the interaction picture simulations we greatly reduce the cost of computing the total kinetic energy (see \sec{kinetic_exp}). 
        Instead of computing it at each step, we compute the sum of squares of momenta at the beginning.
        Then, each time we block encode the potential operator, which involves updating just two momenta, we also update the register with the total momentum by adding the change in momentum.
        \item In both cases it is necessary to prepare a state with amplitudes $\propto 1/\|\nu\|$, where $\nu$ is a 3-dimensional vector of integers.
        This preparation normally involves amplitude amplification, which triples the cost of state preparation~\cite{Brassard2002}.
        For qubitizing the Hamiltonian, rather than using amplitude amplification, we can use the insight that this state preparation is only needed for the potential term, but not for the kinetic term.
        As a result, instead of using amplitude amplification, we can implement the kinetic term in branches of the quantum state that would otherwise correspond to failure (see \app{selTUV}).
        That reduces the cost of this term by a factor of 3.
        \item In the case of the interaction picture one would normally require amplitude amplification on the Dyson series in order to implement a step of the time evolution deterministically, which again triples the cost.
        Instead, we can qubitize a single step of the Dyson series (see \sec{choice}).
        This does not deterministically implement time evolution, but gives a step with eigenvalues related to those of the Hamiltonian.
        That eliminates one part where amplitude amplification was required, but in that case we still need to use amplitude amplification on the preparation of the state with amplitudes $\propto 1/\|\nu\|$.
        \item Rather than performing arithmetic in place on the momentum registers, we swap the correct momentum registers into ancilla registers controlled on ancillae that select the registers (see \sec{SelectUV} and \sec{sel_prep}).
        That means the arithmetic only needs to be performed once in those ancillae, rather than for every momentum register.
        \item For the discretization of the integrals, we use times in the centers of the intervals, which results in the error being higher order (see \app{timedisc}).
        With a careful analysis of the error due to the time discretization, the number of bits needed to represent the time is about half of that using the methods from prior work. 
    \end{itemize}
    As a result of these improvements, in cases where there are a large number of electrons the majority of the cost in each step is just selecting the momentum registers.
    There is an overhead in the interaction picture approach corresponding to the order of the Dyson series.    That means that to provide an improvement via the interaction picture, there should be a ratio between the kinetic and potential energy that is at least as large as the order of the Dyson series.
    That ratio is large in cases where the number of electrons is small, but in those cases the cost of preparing the state with amplitudes $1/\|\nu\|$ tends to be dominant.
    In those cases the direct qubitization approach has the additional advantage that one can avoid the amplitude amplification.
    That means the ratio needs to be even larger before there is an advantage due to using the interaction picture.
    We find that there is a minor improvement via the interaction picture with only $\eta=20$ electrons and $N=2^{21}\approx 2 \times 10^6$ plane waves, but as the number of electrons grows this advantage diminishes.
    
    We have demonstrated that our first quantized qubitization algorithm compares favorably to all prior fault-tolerant algorithms for chemistry for both periodic, and non-periodic systems. However, there are several obvious next steps that should be pursued in order to further refine and generalize our algorithms. While we have alluded to all of these ideas in the main text, we also summarize them below.
    \begin{itemize}[leftmargin=*]
        \item We should modify the algorithm for block encoding $U$ to include various types of accurate pseudopotentials. Using pseudopotentials will allow us to remove core electrons from the simulation in order to only treat electronic degrees of freedom that are involved in chemical reactivity and bonding, thus effectively reducing $\eta$ while increasing $r_s$. Furthermore, the screened potentials resulting from removing core electrons (even if only 1s electrons are removed) are smooth functions. With sharp nuclear cusps removed from the Coulomb potential (and thus, the wavefunction), the plane wave basis will more rapidly converge to the continuum limit. It will be straightforward to make this modification for spherically symmetric pseudopotentials such as those that remove just the 1s electrons, but it will be more involved to adapt the block encodings to more precise pseudopotentials for higher angular momenta electrons (e.g., using non-local pseudopotentials). Including pseudopotentials will also facilitate a more straightforward comparison to Gaussian basis approaches using active spaces.
        
        \item In \eq{G} of this work we define the reciprocal lattice associated with our plane wave mesh on a cubic lattice. However, for systems that lack a cubic conventional unit cell (e.g., those with hexagonal symmetry), such lattices cannot be used. Instead, we would need to define the reciprocal lattice using non-orthogonal Bravais vectors. This would change the definition of $k_p$ from what is given in \eq{G} to
\begin{equation}
\label{eq:non_orthogonal}
k_p = \frac{2 \pi}{\Omega^{1/3}}\left(p_1 g^{(1)} + p_2 g^{(2)} + p_3 g^{(3)}\right) \qquad \qquad p  = \left(p_1, p_2, p_3\right) \in G \, ,
\end{equation}
where $G$ is defined in \eq{G} and $g^{(1)}$, $g^{(2)}$ and $g^{(3)}$ are the three dimensional non-orthogonal unit vectors of the Bravais lattice defined by the \emph{primitive} unit cell geometry. Even for systems with cubic conventional unit cells, there are likely convergence advantages associated with being able to represent non-cubic primitive unit cells. However, using this more general Bravais lattice will also come at a cost since it will make computing $\|k_p\|^2$ (a bottleneck of our approach) more costly. Further work should explore optimizations of these dot products in order to generalize our methods while retaining high efficiency.
\item Related to the above point, there are different ways to converge the system physics towards its thermodynamic limit. The approach explored here involves growing the system in real space (i.e., by adding more unit cells to the computational cell). However, one can also grow the system towards the thermodynamic limit by evaluating an integral over the Brillouin zone in a fashion referred to as ``k-point sampling'' \cite{MartinES2004}. These schemes make use of generalized Bravais lattices such as the Monkhorst-Pack mesh, which allows one to represent systems of arbitrary symmetry while distributing the reciprocal lattice vector points homogeneously in the Brillouin zone \cite{MartinES2004}. Thus far we have only described how to perform so-called ``$\Gamma$-point'' calculations. Future work should investigate the benefits of k-point sampling in the context of these algorithms and, if advantageous to do so, make the relatively simple modification of the algorithms required to support sampling outside of the $\Gamma$-point. 
        \item We have discussed that the methods of this paper can be used to simulate non-periodic systems despite the underlying periodicity of plane waves. In particular, this can be accomplish by using a sufficiently large cell volume $\Omega$ so that periodic images do not interact. However, there are more efficient ways of preventing these periodic images from interacting. For example, one can use a range truncated Coulomb operator~\cite{rozzi2006exact,sundararaman2013regularization,ismail2006truncation}. For an isolated molecule, the total electronic density decays exponentially quickly away from its center, and thus the molecule can be represented in a box of volume $\Omega = D^3$ with only exponentially small parts of its true wavefunction density outside of the box. By using a Coulomb kernel truncated at distance $D$~\cite{fusti2002accurateb}, such that
\begin{equation}
V\left(r, r'\right) = \begin{cases} \frac{1}{\left | r - r' \right |} & \left | r - r' \right | \leq D \\
0  & \left | r - r' \right | > D,
\end{cases}
\end{equation}
and by carrying out the simulation in a box of size $8\, \Omega = (2\,D)^3$, we ensure that there is no Coulomb interaction at all between the repeated images of the molecule, up to exponentially small terms in $\Omega$ arising from the density of the molecule outside of the box. While the Fourier amplitudes of the normal Coulomb kernel are $4 \pi / k^2$, the Fourier amplitudes of the truncated Coulomb interaction shown above are $4 \pi (1 - \cos(|k|D)) / k^2$. This would modify the potential operator coefficients of \eq{plane_wave_integrals} to be
\begin{equation}
    U_{pq} \!= \sum_{\ell=1}^{L}  \frac{4\pi \left(1\! -\! \cos\left(\left \|k_{p-q} \right \|D \right)\right) \zeta_\ell  e^{i k_{q-p}\cdot R_\ell}}{\Omega \norm{k_{p-q}}^2}
    \qquad 
    V_{pqrs}^{(\alpha,\beta)} \! = \delta_{p-q,r-s}  \frac{4\pi \left(1\! - \!\cos\left(\left \|k_{p-q} \right \|D \right)\right) \zeta_\alpha \zeta_\beta}{\Omega \left\| k_{p-q}\right\|^2} \, .
\end{equation}
Thus, future work focused on simulating systems of reduced periodicity should adapt the algorithms presented here to block encode operators with the above modification to the coefficients.
        \item This paper has provided first quantized algorithms for preparing molecular eigenstates given an initial state with sufficient overlap on an eigenstate of interest. However, methods for state preparation in first quantization are under-developed compared to their second quantized counterparts. The work of \cite{Berry2018} provided an algorithm for symmetrizing or antisymmetrizing initial states with complexity ${\cal O}(\eta \log N \log \eta)$, but another common step in state preparation would be to initialize an arbitrary Slater determinant. Such states are described by ${\cal O}(\eta N)$ coefficients, and in second quantization can be prepared with gate complexity of $\widetilde{\cal O}(\eta N)$ \cite{Kivlichan2017}. We would expect that this is also a lower bound on the gate complexity required in first quantization. It would not be ideal for such state preparations to need to scale linearly in $N$. However, one can choose to involve fewer plane waves in the initial state preparation and so long as the number of plane waves taken is less than $M = \widetilde{\cal O}(\eta^{1/3}N^{2/3}/\epsilon + \eta^{5/3} N^{1/3}/\epsilon)$ then state preparation will require a number of gates that is less than the ${\cal O}(\eta M)$ of our qubitization algorithm, and thus be a negligible additive cost (we also expect lower constant factors). But the most straightforward translation of this second quantized algorithm into first quantization gives an approach scaling as $\widetilde{\cal O}(\eta^2 M)$, which would require taking $M = \widetilde{\cal O}(N^{2/3}/(\eta^{2/3} \epsilon) + \eta^{2/3} N^{1/3}/\epsilon)$ in order to remain an additive negligible cost. This is likely fine for some applications but future work should investigate whether this scaling can be improved. Likewise, the topic of state preparation beyond single Slater determinants (even in second quantization) has often been ignored by papers providing error-corrected cost estimates for chemistry. However, when  complex initial states are required (e.g., for challenging systems like FeMoCo), preparing these states may introduce a non-negligible cost that should be taken into account more carefully.
        \item The qubitization algorithm of this work can be used to perform time-evolution, as opposed to used for preparing eigenstates for phase estimation, with the very minimal additional overhead (see \eq{signals}) of quantum signal processing \cite{Low2016}. We expect that time-evolution of the non-Born-Oppenheimer Hamiltonian would be especially useful for computing chemical reaction rates \cite{Lidar1999,Kassal2008}. There would also need to be minor modifications to the block encoding algorithms of this paper to simulate the nuclear kinetic operators and explicit nuclei interaction terms appearing in \eq{non_bo_ham}. Future work should analyze the constant factors of such modifications in the context of a specific chemical dynamics problem and estimate the resources required to solve that problem within fault-tolerance. This would constitute a fundamentally new type of application assessed for viability on an error-corrected quantum processor. Furthermore, we note that our interaction picture technique is trivially compatible with state preparation methods based on adiabatic state preparation \cite{BabbushSR14} as we can easily use that approach to perform time-dependent Hamiltonian evolutions. Such methods might prove important for preparing complex initial states when single Slater determinants fail to have sufficient overlap with the eigenstates of interest. However, more work is required to figure out how to best use the qubitization based approach in the context of adiabatic state preparation.
        \item Finally, in \app{real_space} we asymptotically describe how one can instead simulate a real space version of these algorithms using the grid-like representation employed by Kassal \emph{et al.}~\cite{Kassal2008}, but employing more modern quantum simulation methods. We introduce algorithms for block encoding $U$ and $V$ in this representation that lead to qubitization and interaction picture algorithms with asymptotically the same complexity as the main algorithms explored in this work. However, the constant factors are different and the non-periodic nature of that representation may have advantages when simulating non-periodic systems. Thus, we propose that future work should more thoroughly explore this simulation strategy and compare the constant factors associated with it to those of the techniques in this paper.
    \end{itemize}

    In conclusion, there are still many aspects of these algorithms (and their applications) to explore. Nonetheless, this first study of fault-tolerant overheads of first quantized quantum simulations of chemistry suggests many potential advantages over their second quantized counterparts. While Gaussian orbital based algorithms in second quantization will continue to compete with first quantized algorithms when simulating small non-periodic molecules, we believe that the longer term future of quantum computing for chemistry will focus much more on first quantized simulations. We expect this is the case due to the ease of removing the Born-Oppenheimer approximation, the fact that one can more efficiently refine the accuracy of simulations in first quantization, the fact that one does not require complicated and costly classical precomputation of active spaces and molecular orbitals, and the fact that, on a whole, these methods scale better than their second quantized counterparts, while still having reasonable constant factors (as our work has shown).

	\subsection*{Acknowledgements}
	
	The authors thank Garnet Kin-Lic Chan, Matthias Degroote, Craig Gidney, Cody Jones, Joonho Lee, Jarrod McClean, Yuval Sanders, Norm Tubman and the quantum computing team at BASF for helpful discussions. DWB worked on this project under a sponsored research agreement with Google Quantum AI. DWB is also supported by Australian Research Council Discovery Projects DP190102633 and DP210101367. NW was funded by a grant from Google Quantum AI as well as the US Department of Energy, Office of Science, National Quantum Information Science Research Centers, Co-Design Center for Quantum Advantage under contract number DE-SC0012704 which supported his work on the asymptotic analysis of the algorithms. YS was partly supported by the National Science Foundation RAISE-TAQS 1839204. The Institute for Quantum Information and Matter is an NSF Physics Frontiers Center PHY-1733907.
	
	%%%%%%%%%%%%%%%%%%%%%%%%%%%%%%%%%%%%%%%%%%%%%%%%%%%%%%%%%%%%%%%%%%%%%%%%%%%%%%
	\FloatBarrier
	\bibliography{interaction_picture,references,chem_refs}

\begin{thebibliography}{113}%
\makeatletter
\providecommand \@ifxundefined [1]{%
 \@ifx{#1\undefined}
}%
\providecommand \@ifnum [1]{%
 \ifnum #1\expandafter \@firstoftwo
 \else \expandafter \@secondoftwo
 \fi
}%
\providecommand \@ifx [1]{%
 \ifx #1\expandafter \@firstoftwo
 \else \expandafter \@secondoftwo
 \fi
}%
\providecommand \natexlab [1]{#1}%
\providecommand \enquote  [1]{``#1''}%
\providecommand \bibnamefont  [1]{#1}%
\providecommand \bibfnamefont [1]{#1}%
\providecommand \citenamefont [1]{#1}%
\providecommand \href@noop [0]{\@secondoftwo}%
\providecommand \href [0]{\begingroup \@sanitize@url \@href}%
\providecommand \@href[1]{\@@startlink{#1}\@@href}%
\providecommand \@@href[1]{\endgroup#1\@@endlink}%
\providecommand \@sanitize@url [0]{\catcode `\\12\catcode `\$12\catcode
  `\&12\catcode `\#12\catcode `\^12\catcode `\_12\catcode `\%12\relax}%
\providecommand \@@startlink[1]{}%
\providecommand \@@endlink[0]{}%
\providecommand \url  [0]{\begingroup\@sanitize@url \@url }%
\providecommand \@url [1]{\endgroup\@href {#1}{\urlprefix }}%
\providecommand \urlprefix  [0]{URL }%
\providecommand \Eprint [0]{\href }%
\providecommand \doibase [0]{http://dx.doi.org/}%
\providecommand \selectlanguage [0]{\@gobble}%
\providecommand \bibinfo  [0]{\@secondoftwo}%
\providecommand \bibfield  [0]{\@secondoftwo}%
\providecommand \translation [1]{[#1]}%
\providecommand \BibitemOpen [0]{}%
\providecommand \bibitemStop [0]{}%
\providecommand \bibitemNoStop [0]{.\EOS\space}%
\providecommand \EOS [0]{\spacefactor3000\relax}%
\providecommand \BibitemShut  [1]{\csname bibitem#1\endcsname}%
\let\auto@bib@innerbib\@empty
%</preamble>
\bibitem [{\citenamefont {Babbush}\ \emph {et~al.}(2019)\citenamefont
  {Babbush}, \citenamefont {Berry}, \citenamefont {McClean},\ and\
  \citenamefont {Neven}}]{BabbushContinuum_b}%
  \BibitemOpen
  \bibfield  {author} {\bibinfo {author} {\bibfnamefont {Ryan}\ \bibnamefont
  {Babbush}}, \bibinfo {author} {\bibfnamefont {Dominic~W}\ \bibnamefont
  {Berry}}, \bibinfo {author} {\bibfnamefont {Jarrod~R}\ \bibnamefont
  {McClean}}, \ and\ \bibinfo {author} {\bibfnamefont {Hartmut}\ \bibnamefont
  {Neven}},\ }\bibfield  {title} {\enquote {\bibinfo {title} {{Quantum
  Simulation of Chemistry with Sublinear Scaling in Basis Size}},}\ }\href
  {https://www.nature.com/articles/s41534-019-0199-y} {\bibfield  {journal}
  {\bibinfo  {journal} {npj Quantum Information}\ }\textbf {\bibinfo {volume}
  {5}},\ \bibinfo {pages} {92} (\bibinfo {year} {2019})}\BibitemShut {NoStop}%
\bibitem [{\citenamefont {Low}\ and\ \citenamefont {Chuang}(2019)}]{Low2016}%
  \BibitemOpen
  \bibfield  {author} {\bibinfo {author} {\bibfnamefont {Guang~Hao}\
  \bibnamefont {Low}}\ and\ \bibinfo {author} {\bibfnamefont {Isaac~L}\
  \bibnamefont {Chuang}},\ }\bibfield  {title} {\enquote {\bibinfo {title}
  {{Hamiltonian Simulation by Qubitization}},}\ }\href
  {https://doi.org/10.22331/q-2019-07-12-163} {\bibfield  {journal} {\bibinfo
  {journal} {Quantum}\ }\textbf {\bibinfo {volume} {3}},\ \bibinfo {pages}
  {163} (\bibinfo {year} {2019})}\BibitemShut {NoStop}%
\bibitem [{\citenamefont {Low}\ and\ \citenamefont {Wiebe}(2018)}]{Low2018}%
  \BibitemOpen
  \bibfield  {author} {\bibinfo {author} {\bibfnamefont {Guang~Hao}\
  \bibnamefont {Low}}\ and\ \bibinfo {author} {\bibfnamefont {Nathan}\
  \bibnamefont {Wiebe}},\ }\bibfield  {title} {\enquote {\bibinfo {title}
  {{Hamiltonian Simulation in the Interaction Picture}},}\ }\href
  {http://arxiv.org/abs/1805.00675} {\bibfield  {journal} {\bibinfo  {journal}
  {arXiv:1805.00675}\ } (\bibinfo {year} {2018})}\BibitemShut {NoStop}%
\bibitem [{\citenamefont {Feynman}(1982)}]{Feynman1982}%
  \BibitemOpen
  \bibfield  {author} {\bibinfo {author} {\bibfnamefont {Richard~P}\
  \bibnamefont {Feynman}},\ }\bibfield  {title} {\enquote {\bibinfo {title}
  {{Simulating physics with computers}},}\ }\href {\doibase 10.1007/BF02650179}
  {\bibfield  {journal} {\bibinfo  {journal} {International Journal of
  Theoretical Physics}\ }\textbf {\bibinfo {volume} {21}},\ \bibinfo {pages}
  {467--488} (\bibinfo {year} {1982})}\BibitemShut {NoStop}%
\bibitem [{\citenamefont {Lloyd}(1996)}]{Lloyd1996}%
  \BibitemOpen
  \bibfield  {author} {\bibinfo {author} {\bibfnamefont {Seth}\ \bibnamefont
  {Lloyd}},\ }\bibfield  {title} {\enquote {\bibinfo {title} {{Universal
  Quantum Simulators}},}\ }\href {\doibase 10.1126/science.273.5278.1073}
  {\bibfield  {journal} {\bibinfo  {journal} {Science}\ }\textbf {\bibinfo
  {volume} {273}},\ \bibinfo {pages} {1073--1078} (\bibinfo {year}
  {1996})}\BibitemShut {NoStop}%
\bibitem [{\citenamefont {Aspuru-Guzik}\ \emph {et~al.}(2005)\citenamefont
  {Aspuru-Guzik}, \citenamefont {Dutoi}, \citenamefont {Love},\ and\
  \citenamefont {Head-Gordon}}]{Aspuru-Guzik2005}%
  \BibitemOpen
  \bibfield  {author} {\bibinfo {author} {\bibfnamefont {Alan}\ \bibnamefont
  {Aspuru-Guzik}}, \bibinfo {author} {\bibfnamefont {Anthony~D}\ \bibnamefont
  {Dutoi}}, \bibinfo {author} {\bibfnamefont {Peter~J}\ \bibnamefont {Love}}, \
  and\ \bibinfo {author} {\bibfnamefont {Martin}\ \bibnamefont {Head-Gordon}},\
  }\bibfield  {title} {\enquote {\bibinfo {title} {{Simulated Quantum
  Computation of Molecular Energies}},}\ }\href {\doibase
  10.1126/science.1113479} {\bibfield  {journal} {\bibinfo  {journal}
  {Science}\ }\textbf {\bibinfo {volume} {309}},\ \bibinfo {pages} {1704}
  (\bibinfo {year} {2005})}\BibitemShut {NoStop}%
\bibitem [{\citenamefont {Berry}\ \emph {et~al.}(2015)\citenamefont {Berry},
  \citenamefont {Childs}, \citenamefont {Cleve}, \citenamefont {Kothari},\ and\
  \citenamefont {Somma}}]{Berry2015}%
  \BibitemOpen
  \bibfield  {author} {\bibinfo {author} {\bibfnamefont {Dominic~W}\
  \bibnamefont {Berry}}, \bibinfo {author} {\bibfnamefont {Andrew~M}\
  \bibnamefont {Childs}}, \bibinfo {author} {\bibfnamefont {Richard}\
  \bibnamefont {Cleve}}, \bibinfo {author} {\bibfnamefont {Robin}\ \bibnamefont
  {Kothari}}, \ and\ \bibinfo {author} {\bibfnamefont {Rolando~D}\ \bibnamefont
  {Somma}},\ }\bibfield  {title} {\enquote {\bibinfo {title} {{Simulating
  Hamiltonian Dynamics with a Truncated Taylor Series}},}\ }\href {\doibase
  10.1103/PhysRevLett.114.090502} {\bibfield  {journal} {\bibinfo  {journal}
  {Physical Review Letters}\ }\textbf {\bibinfo {volume} {114}},\ \bibinfo
  {pages} {90502} (\bibinfo {year} {2015})}\BibitemShut {NoStop}%
\bibitem [{\citenamefont {Wiebe}\ and\ \citenamefont
  {Granade}(2015)}]{Wiebe2015b}%
  \BibitemOpen
  \bibfield  {author} {\bibinfo {author} {\bibfnamefont {Nathan}\ \bibnamefont
  {Wiebe}}\ and\ \bibinfo {author} {\bibfnamefont {Christopher~E}\ \bibnamefont
  {Granade}},\ }\bibfield  {title} {\enquote {\bibinfo {title} {{Efficient
  Bayesian Phase Estimation}},}\ }\href {http://arxiv.org/abs/1508.00869}
  {\bibfield  {journal} {\bibinfo  {journal} {arXiv:1508.00869}\ } (\bibinfo
  {year} {2015})}\BibitemShut {NoStop}%
\bibitem [{\citenamefont {Poulin}\ \emph {et~al.}(2017)\citenamefont {Poulin},
  \citenamefont {Kitaev}, \citenamefont {Steiger}, \citenamefont {Hastings},\
  and\ \citenamefont {Troyer}}]{Poulin2017}%
  \BibitemOpen
  \bibfield  {author} {\bibinfo {author} {\bibfnamefont {David}\ \bibnamefont
  {Poulin}}, \bibinfo {author} {\bibfnamefont {Alexei~Y}\ \bibnamefont
  {Kitaev}}, \bibinfo {author} {\bibfnamefont {Damian}\ \bibnamefont
  {Steiger}}, \bibinfo {author} {\bibfnamefont {Matthew}\ \bibnamefont
  {Hastings}}, \ and\ \bibinfo {author} {\bibfnamefont {Matthias}\ \bibnamefont
  {Troyer}},\ }\bibfield  {title} {\enquote {\bibinfo {title} {{Fast Quantum
  Algorithm for Spectral Properties}},}\ }\href
  {https://journals.aps.org/prl/abstract/10.1103/PhysRevLett.121.010501}
  {\bibfield  {journal} {\bibinfo  {journal} {Physical Review Letters}\
  }\textbf {\bibinfo {volume} {121}},\ \bibinfo {pages} {010501} (\bibinfo
  {year} {2017})}\BibitemShut {NoStop}%
\bibitem [{\citenamefont {Haah}\ \emph {et~al.}(2018)\citenamefont {Haah},
  \citenamefont {Hastings}, \citenamefont {Kothari},\ and\ \citenamefont
  {Low}}]{Haah2018b}%
  \BibitemOpen
  \bibfield  {author} {\bibinfo {author} {\bibfnamefont {Jeongwan}\
  \bibnamefont {Haah}}, \bibinfo {author} {\bibfnamefont {Matthew~B.}\
  \bibnamefont {Hastings}}, \bibinfo {author} {\bibfnamefont {Robin}\
  \bibnamefont {Kothari}}, \ and\ \bibinfo {author} {\bibfnamefont {Guang~Hao}\
  \bibnamefont {Low}},\ }\bibfield  {title} {\enquote {\bibinfo {title}
  {Quantum algorithm for simulating real time evolution of lattice
  {H}amiltonians},}\ }\href {\doibase 10.1137/18M1231511} {\bibfield  {journal}
  {\bibinfo  {journal} {SIAM Journal on Computing}\ }\textbf {\bibinfo {volume}
  {\!\!}},\ \bibinfo {pages} {FOCS18--250--FOCS18--284} (\bibinfo {year}
  {2018})}\BibitemShut {NoStop}%
\bibitem [{\citenamefont {Campbell}(2019)}]{CampbellPRL2018}%
  \BibitemOpen
  \bibfield  {author} {\bibinfo {author} {\bibfnamefont {Earl}\ \bibnamefont
  {Campbell}},\ }\bibfield  {title} {\enquote {\bibinfo {title} {{Random
  Compiler for Fast Hamiltonian Simulation}},}\ }\href {\doibase
  10.1103/PhysRevLett.123.070503} {\bibfield  {journal} {\bibinfo  {journal}
  {Physical Review Letters}\ }\textbf {\bibinfo {volume} {123}},\ \bibinfo
  {pages} {070503} (\bibinfo {year} {2019})}\BibitemShut {NoStop}%
\bibitem [{\citenamefont {Childs}\ \emph {et~al.}(2021)\citenamefont {Childs},
  \citenamefont {Su}, \citenamefont {Tran}, \citenamefont {Wiebe},\ and\
  \citenamefont {Zhu}}]{Childs2021}%
  \BibitemOpen
  \bibfield  {author} {\bibinfo {author} {\bibfnamefont {Andrew~M}\
  \bibnamefont {Childs}}, \bibinfo {author} {\bibfnamefont {Yuan}\ \bibnamefont
  {Su}}, \bibinfo {author} {\bibfnamefont {Minh~C}\ \bibnamefont {Tran}},
  \bibinfo {author} {\bibfnamefont {Nathan}\ \bibnamefont {Wiebe}}, \ and\
  \bibinfo {author} {\bibfnamefont {Shuchen}\ \bibnamefont {Zhu}},\ }\bibfield
  {title} {\enquote {\bibinfo {title} {{Theory of Trotter Error with Commutator
  Scaling}},}\ }\href {\doibase 10.1103/PhysRevX.11.011020} {\bibfield
  {journal} {\bibinfo  {journal} {Physical Review X}\ }\textbf {\bibinfo
  {volume} {11}},\ \bibinfo {pages} {011020} (\bibinfo {year}
  {2021})}\BibitemShut {NoStop}%
\bibitem [{\citenamefont {Motta}\ \emph
  {et~al.}(2020{\natexlab{a}})\citenamefont {Motta}, \citenamefont {Sun},
  \citenamefont {Tan}, \citenamefont {O'Rourke}, \citenamefont {Ye},
  \citenamefont {Minnich}, \citenamefont {Brandao},\ and\ \citenamefont
  {Chan}}]{motta2020determining}%
  \BibitemOpen
  \bibfield  {author} {\bibinfo {author} {\bibfnamefont {Mario}\ \bibnamefont
  {Motta}}, \bibinfo {author} {\bibfnamefont {Chong}\ \bibnamefont {Sun}},
  \bibinfo {author} {\bibfnamefont {Adrian T~K}\ \bibnamefont {Tan}}, \bibinfo
  {author} {\bibfnamefont {Matthew~J}\ \bibnamefont {O'Rourke}}, \bibinfo
  {author} {\bibfnamefont {Erika}\ \bibnamefont {Ye}}, \bibinfo {author}
  {\bibfnamefont {Austin~J}\ \bibnamefont {Minnich}}, \bibinfo {author}
  {\bibfnamefont {Fernando G S~L}\ \bibnamefont {Brandao}}, \ and\ \bibinfo
  {author} {\bibfnamefont {Garnet Kin-Lic}\ \bibnamefont {Chan}},\ }\bibfield
  {title} {\enquote {\bibinfo {title} {Determining eigenstates and thermal
  states on a quantum computer using quantum imaginary time evolution},}\
  }\href {https://doi.org/10.1038/s41567-019-0704-4} {\bibfield  {journal}
  {\bibinfo  {journal} {Nature Physics}\ }\textbf {\bibinfo {volume} {16}},\
  \bibinfo {pages} {205--210} (\bibinfo {year}
  {2020}{\natexlab{a}})}\BibitemShut {NoStop}%
\bibitem [{\citenamefont {{\c{S}}ahino\u{g}lu}\ and\ \citenamefont
  {Somma}(2020)}]{SS20}%
  \BibitemOpen
  \bibfield  {author} {\bibinfo {author} {\bibfnamefont {Burak}\ \bibnamefont
  {{\c{S}}ahino\u{g}lu}}\ and\ \bibinfo {author} {\bibfnamefont {Rolando~D}\
  \bibnamefont {Somma}},\ }\bibfield  {title} {\enquote {\bibinfo {title}
  {Hamiltonian simulation in the low energy subspace},}\ }\href
  {http://arxiv.org/abs/2006.02660} {\bibfield  {journal} {\bibinfo  {journal}
  {arXiv:2006.02660}\ } (\bibinfo {year} {2020})}\BibitemShut {NoStop}%
\bibitem [{\citenamefont {Tran}\ \emph {et~al.}(2021)\citenamefont {Tran},
  \citenamefont {Su}, \citenamefont {Carney},\ and\ \citenamefont
  {Taylor}}]{Tran2021}%
  \BibitemOpen
  \bibfield  {author} {\bibinfo {author} {\bibfnamefont {Minh~C}\ \bibnamefont
  {Tran}}, \bibinfo {author} {\bibfnamefont {Yuan}\ \bibnamefont {Su}},
  \bibinfo {author} {\bibfnamefont {Daniel}\ \bibnamefont {Carney}}, \ and\
  \bibinfo {author} {\bibfnamefont {Jacob~M}\ \bibnamefont {Taylor}},\
  }\bibfield  {title} {\enquote {\bibinfo {title} {Faster digital quantum
  simulation by symmetry protection},}\ }\href {\doibase
  10.1103/PRXQuantum.2.010323} {\bibfield  {journal} {\bibinfo  {journal} {PRX
  Quantum}\ }\textbf {\bibinfo {volume} {2}},\ \bibinfo {pages} {010323}
  (\bibinfo {year} {2021})}\BibitemShut {NoStop}%
\bibitem [{\citenamefont {Whitfield}\ \emph {et~al.}(2011)\citenamefont
  {Whitfield}, \citenamefont {Biamonte},\ and\ \citenamefont
  {Aspuru-Guzik}}]{Whitfield2010}%
  \BibitemOpen
  \bibfield  {author} {\bibinfo {author} {\bibfnamefont {James~D}\ \bibnamefont
  {Whitfield}}, \bibinfo {author} {\bibfnamefont {Jacob}\ \bibnamefont
  {Biamonte}}, \ and\ \bibinfo {author} {\bibfnamefont {Alan}\ \bibnamefont
  {Aspuru-Guzik}},\ }\bibfield  {title} {\enquote {\bibinfo {title}
  {{Simulation of electronic structure Hamiltonians using quantum
  computers}},}\ }\href {\doibase 10.1080/00268976.2011.552441} {\bibfield
  {journal} {\bibinfo  {journal} {Molecular Physics}\ }\textbf {\bibinfo
  {volume} {109}},\ \bibinfo {pages} {735--750} (\bibinfo {year}
  {2011})}\BibitemShut {NoStop}%
\bibitem [{\citenamefont {Wecker}\ \emph {et~al.}(2014)\citenamefont {Wecker},
  \citenamefont {Bauer}, \citenamefont {Clark}, \citenamefont {Hastings},\ and\
  \citenamefont {Troyer}}]{Wecker2014_b}%
  \BibitemOpen
  \bibfield  {author} {\bibinfo {author} {\bibfnamefont {David}\ \bibnamefont
  {Wecker}}, \bibinfo {author} {\bibfnamefont {Bela}\ \bibnamefont {Bauer}},
  \bibinfo {author} {\bibfnamefont {Bryan~K}\ \bibnamefont {Clark}}, \bibinfo
  {author} {\bibfnamefont {Matthew~B}\ \bibnamefont {Hastings}}, \ and\
  \bibinfo {author} {\bibfnamefont {Matthias}\ \bibnamefont {Troyer}},\
  }\bibfield  {title} {\enquote {\bibinfo {title} {{Gate-count estimates for
  performing quantum chemistry on small quantum computers}},}\ }\href {\doibase
  10.1103/PhysRevA.90.022305} {\bibfield  {journal} {\bibinfo  {journal}
  {Physical Review A}\ }\textbf {\bibinfo {volume} {90}},\ \bibinfo {pages}
  {022305} (\bibinfo {year} {2014})}\BibitemShut {NoStop}%
\bibitem [{\citenamefont {Hastings}\ \emph {et~al.}(2015)\citenamefont
  {Hastings}, \citenamefont {Wecker}, \citenamefont {Bauer},\ and\
  \citenamefont {Troyer}}]{Hastings2015}%
  \BibitemOpen
  \bibfield  {author} {\bibinfo {author} {\bibfnamefont {Matthew~B}\
  \bibnamefont {Hastings}}, \bibinfo {author} {\bibfnamefont {Dave}\
  \bibnamefont {Wecker}}, \bibinfo {author} {\bibfnamefont {Bela}\ \bibnamefont
  {Bauer}}, \ and\ \bibinfo {author} {\bibfnamefont {Matthias}\ \bibnamefont
  {Troyer}},\ }\bibfield  {title} {\enquote {\bibinfo {title} {{Improving
  Quantum Algorithms for Quantum Chemistry}},}\ }\href
  {http://arxiv.org/abs/1403.1539} {\bibfield  {journal} {\bibinfo  {journal}
  {Quantum Information {\&} Computation}\ }\textbf {\bibinfo {volume} {15}},\
  \bibinfo {pages} {1--21} (\bibinfo {year} {2015})}\BibitemShut {NoStop}%
\bibitem [{\citenamefont {Poulin}\ \emph {et~al.}(2015)\citenamefont {Poulin},
  \citenamefont {Hastings}, \citenamefont {Wecker}, \citenamefont {Wiebe},
  \citenamefont {Doherty},\ and\ \citenamefont {Troyer}}]{Poulin2014}%
  \BibitemOpen
  \bibfield  {author} {\bibinfo {author} {\bibfnamefont {David}\ \bibnamefont
  {Poulin}}, \bibinfo {author} {\bibfnamefont {M~B}\ \bibnamefont {Hastings}},
  \bibinfo {author} {\bibfnamefont {Dave}\ \bibnamefont {Wecker}}, \bibinfo
  {author} {\bibfnamefont {Nathan}\ \bibnamefont {Wiebe}}, \bibinfo {author}
  {\bibfnamefont {Andrew~C}\ \bibnamefont {Doherty}}, \ and\ \bibinfo {author}
  {\bibfnamefont {Matthias}\ \bibnamefont {Troyer}},\ }\bibfield  {title}
  {\enquote {\bibinfo {title} {{The Trotter Step Size Required for Accurate
  Quantum Simulation of Quantum Chemistry}},}\ }\href
  {http://arxiv.org/abs/1406.4920} {\bibfield  {journal} {\bibinfo  {journal}
  {Quantum Information {\&} Computation}\ }\textbf {\bibinfo {volume} {15}},\
  \bibinfo {pages} {361--384} (\bibinfo {year} {2015})}\BibitemShut {NoStop}%
\bibitem [{\citenamefont {Babbush}\ \emph {et~al.}(2016)\citenamefont
  {Babbush}, \citenamefont {Berry}, \citenamefont {Kivlichan}, \citenamefont
  {Wei}, \citenamefont {Love},\ and\ \citenamefont
  {Aspuru-Guzik}}]{BabbushSparse1}%
  \BibitemOpen
  \bibfield  {author} {\bibinfo {author} {\bibfnamefont {Ryan}\ \bibnamefont
  {Babbush}}, \bibinfo {author} {\bibfnamefont {Dominic~W}\ \bibnamefont
  {Berry}}, \bibinfo {author} {\bibfnamefont {Ian~D}\ \bibnamefont
  {Kivlichan}}, \bibinfo {author} {\bibfnamefont {Annie~Y}\ \bibnamefont
  {Wei}}, \bibinfo {author} {\bibfnamefont {Peter~J}\ \bibnamefont {Love}}, \
  and\ \bibinfo {author} {\bibfnamefont {Alan}\ \bibnamefont {Aspuru-Guzik}},\
  }\bibfield  {title} {\enquote {\bibinfo {title} {{Exponentially More Precise
  Quantum Simulation of Fermions in Second Quantization}},}\ }\href {\doibase
  10.1088/1367-2630/18/3/033032} {\bibfield  {journal} {\bibinfo  {journal}
  {New Journal of Physics}\ }\textbf {\bibinfo {volume} {18}},\ \bibinfo
  {pages} {33032} (\bibinfo {year} {2016})}\BibitemShut {NoStop}%
\bibitem [{\citenamefont {Babbush}\ \emph
  {et~al.}(2018{\natexlab{a}})\citenamefont {Babbush}, \citenamefont {Gidney},
  \citenamefont {Berry}, \citenamefont {Wiebe}, \citenamefont {McClean},
  \citenamefont {Paler}, \citenamefont {Fowler},\ and\ \citenamefont
  {Neven}}]{BGBWMPFN18}%
  \BibitemOpen
  \bibfield  {author} {\bibinfo {author} {\bibfnamefont {Ryan}\ \bibnamefont
  {Babbush}}, \bibinfo {author} {\bibfnamefont {Craig}\ \bibnamefont {Gidney}},
  \bibinfo {author} {\bibfnamefont {Dominic~W}\ \bibnamefont {Berry}}, \bibinfo
  {author} {\bibfnamefont {Nathan}\ \bibnamefont {Wiebe}}, \bibinfo {author}
  {\bibfnamefont {Jarrod}\ \bibnamefont {McClean}}, \bibinfo {author}
  {\bibfnamefont {Alexandru}\ \bibnamefont {Paler}}, \bibinfo {author}
  {\bibfnamefont {Austin}\ \bibnamefont {Fowler}}, \ and\ \bibinfo {author}
  {\bibfnamefont {Hartmut}\ \bibnamefont {Neven}},\ }\bibfield  {title}
  {\enquote {\bibinfo {title} {Encoding electronic spectra in quantum circuits
  with linear {T} complexity},}\ }\href {\doibase 10.1103/PhysRevX.8.041015}
  {\bibfield  {journal} {\bibinfo  {journal} {Physical Review X}\ }\textbf
  {\bibinfo {volume} {8}},\ \bibinfo {pages} {041015} (\bibinfo {year}
  {2018}{\natexlab{a}})}\BibitemShut {NoStop}%
\bibitem [{\citenamefont {Berry}\ \emph {et~al.}(2018)\citenamefont {Berry},
  \citenamefont {Kieferov{\'{a}}}, \citenamefont {Scherer}, \citenamefont
  {Sanders}, \citenamefont {Low}, \citenamefont {Wiebe}, \citenamefont
  {Gidney},\ and\ \citenamefont {Babbush}}]{Berry2018}%
  \BibitemOpen
  \bibfield  {author} {\bibinfo {author} {\bibfnamefont {Dominic~W}\
  \bibnamefont {Berry}}, \bibinfo {author} {\bibfnamefont {Maria}\ \bibnamefont
  {Kieferov{\'{a}}}}, \bibinfo {author} {\bibfnamefont {Artur}\ \bibnamefont
  {Scherer}}, \bibinfo {author} {\bibfnamefont {Yuval~R}\ \bibnamefont
  {Sanders}}, \bibinfo {author} {\bibfnamefont {Guang~Hao}\ \bibnamefont
  {Low}}, \bibinfo {author} {\bibfnamefont {Nathan}\ \bibnamefont {Wiebe}},
  \bibinfo {author} {\bibfnamefont {Craig}\ \bibnamefont {Gidney}}, \ and\
  \bibinfo {author} {\bibfnamefont {Ryan}\ \bibnamefont {Babbush}},\ }\bibfield
   {title} {\enquote {\bibinfo {title} {{Improved Techniques for Preparing
  Eigenstates of Fermionic Hamiltonians}},}\ }\href {\doibase
  10.1038/s41534-018-0071-5} {\bibfield  {journal} {\bibinfo  {journal} {npj
  Quantum Information}\ }\textbf {\bibinfo {volume} {4}},\ \bibinfo {pages}
  {22} (\bibinfo {year} {2018})}\BibitemShut {NoStop}%
\bibitem [{\citenamefont {Berry}\ \emph {et~al.}(2019)\citenamefont {Berry},
  \citenamefont {Gidney}, \citenamefont {Motta}, \citenamefont {McClean},\ and\
  \citenamefont {Babbush}}]{Berry2019b}%
  \BibitemOpen
  \bibfield  {author} {\bibinfo {author} {\bibfnamefont {Dominic~W}\
  \bibnamefont {Berry}}, \bibinfo {author} {\bibfnamefont {Craig}\ \bibnamefont
  {Gidney}}, \bibinfo {author} {\bibfnamefont {Mario}\ \bibnamefont {Motta}},
  \bibinfo {author} {\bibfnamefont {Jarrod}\ \bibnamefont {McClean}}, \ and\
  \bibinfo {author} {\bibfnamefont {Ryan}\ \bibnamefont {Babbush}},\ }\bibfield
   {title} {\enquote {\bibinfo {title} {{Qubitization of Arbitrary Basis
  Quantum Chemistry Leveraging Sparsity and Low Rank Factorization}},}\ }\href
  {https://quantum-journal.org/papers/q-2019-12-02-208/} {\bibfield  {journal}
  {\bibinfo  {journal} {Quantum}\ }\textbf {\bibinfo {volume} {3}},\ \bibinfo
  {pages} {208} (\bibinfo {year} {2019})}\BibitemShut {NoStop}%
\bibitem [{\citenamefont {Kivlichan}\ \emph {et~al.}(2018)\citenamefont
  {Kivlichan}, \citenamefont {McClean}, \citenamefont {Wiebe}, \citenamefont
  {Gidney}, \citenamefont {Aspuru-Guzik}, \citenamefont {Chan},\ and\
  \citenamefont {Babbush}}]{Kivlichan2017}%
  \BibitemOpen
  \bibfield  {author} {\bibinfo {author} {\bibfnamefont {Ian~D}\ \bibnamefont
  {Kivlichan}}, \bibinfo {author} {\bibfnamefont {Jarrod}\ \bibnamefont
  {McClean}}, \bibinfo {author} {\bibfnamefont {Nathan}\ \bibnamefont {Wiebe}},
  \bibinfo {author} {\bibfnamefont {Craig}\ \bibnamefont {Gidney}}, \bibinfo
  {author} {\bibfnamefont {Alan}\ \bibnamefont {Aspuru-Guzik}}, \bibinfo
  {author} {\bibfnamefont {Garnet Kin-Lic}\ \bibnamefont {Chan}}, \ and\
  \bibinfo {author} {\bibfnamefont {Ryan}\ \bibnamefont {Babbush}},\ }\bibfield
   {title} {\enquote {\bibinfo {title} {{Quantum Simulation of Electronic
  Structure with Linear Depth and Connectivity}},}\ }\href {\doibase
  10.1103/PhysRevLett.120.110501} {\bibfield  {journal} {\bibinfo  {journal}
  {Physical Review Letters}\ }\textbf {\bibinfo {volume} {120}},\ \bibinfo
  {pages} {110501} (\bibinfo {year} {2018})}\BibitemShut {NoStop}%
\bibitem [{\citenamefont {Kivlichan}\ \emph {et~al.}(2020)\citenamefont
  {Kivlichan}, \citenamefont {Gidney}, \citenamefont {Berry}, \citenamefont
  {Wiebe}, \citenamefont {McClean}, \citenamefont {Sun}, \citenamefont {Jiang},
  \citenamefont {Rubin}, \citenamefont {Fowler}, \citenamefont {Aspuru-Guzik},
  \citenamefont {Neven},\ and\ \citenamefont
  {Babbush}}]{Kivlichan2020improvedfault}%
  \BibitemOpen
  \bibfield  {author} {\bibinfo {author} {\bibfnamefont {Ian~D}\ \bibnamefont
  {Kivlichan}}, \bibinfo {author} {\bibfnamefont {Craig}\ \bibnamefont
  {Gidney}}, \bibinfo {author} {\bibfnamefont {Dominic~W}\ \bibnamefont
  {Berry}}, \bibinfo {author} {\bibfnamefont {Nathan}\ \bibnamefont {Wiebe}},
  \bibinfo {author} {\bibfnamefont {Jarrod}\ \bibnamefont {McClean}}, \bibinfo
  {author} {\bibfnamefont {Wei}\ \bibnamefont {Sun}}, \bibinfo {author}
  {\bibfnamefont {Zhang}\ \bibnamefont {Jiang}}, \bibinfo {author}
  {\bibfnamefont {Nicholas}\ \bibnamefont {Rubin}}, \bibinfo {author}
  {\bibfnamefont {Austin}\ \bibnamefont {Fowler}}, \bibinfo {author}
  {\bibfnamefont {Al{\'{a}}n}\ \bibnamefont {Aspuru-Guzik}}, \bibinfo {author}
  {\bibfnamefont {Hartmut}\ \bibnamefont {Neven}}, \ and\ \bibinfo {author}
  {\bibfnamefont {Ryan}\ \bibnamefont {Babbush}},\ }\bibfield  {title}
  {\enquote {\bibinfo {title} {Improved {F}ault-{T}olerant {Q}uantum
  {S}imulation of {C}ondensed-{P}hase {C}orrelated {E}lectrons via
  {T}rotterization},}\ }\href {\doibase 10.22331/q-2020-07-16-296} {\bibfield
  {journal} {\bibinfo  {journal} {{Quantum}}\ }\textbf {\bibinfo {volume}
  {4}},\ \bibinfo {pages} {296} (\bibinfo {year} {2020})}\BibitemShut {NoStop}%
\bibitem [{\citenamefont {McArdle}\ \emph {et~al.}(2020)\citenamefont
  {McArdle}, \citenamefont {Endo}, \citenamefont {Aspuru-Guzik}, \citenamefont
  {Benjamin},\ and\ \citenamefont {Yuan}}]{mcardle2020quantum}%
  \BibitemOpen
  \bibfield  {author} {\bibinfo {author} {\bibfnamefont {Sam}\ \bibnamefont
  {McArdle}}, \bibinfo {author} {\bibfnamefont {Suguru}\ \bibnamefont {Endo}},
  \bibinfo {author} {\bibfnamefont {Al{\'a}n}\ \bibnamefont {Aspuru-Guzik}},
  \bibinfo {author} {\bibfnamefont {Simon~C.}\ \bibnamefont {Benjamin}}, \ and\
  \bibinfo {author} {\bibfnamefont {Xiao}\ \bibnamefont {Yuan}},\ }\bibfield
  {title} {\enquote {\bibinfo {title} {Quantum computational chemistry},}\
  }\href {\doibase 10.1103/RevModPhys.92.015003} {\bibfield  {journal}
  {\bibinfo  {journal} {Reviews of Modern Physics}\ }\textbf {\bibinfo {volume}
  {92}},\ \bibinfo {pages} {015003} (\bibinfo {year} {2020})}\BibitemShut
  {NoStop}%
\bibitem [{\citenamefont {von Burg}\ \emph {et~al.}(2020)\citenamefont {von
  Burg}, \citenamefont {Low}, \citenamefont {H{\"{a}}ner}, \citenamefont
  {Steiger}, \citenamefont {Reiher}, \citenamefont {Roetteler},\ and\
  \citenamefont {Troyer}}]{vonBurg2020Catalysis}%
  \BibitemOpen
  \bibfield  {author} {\bibinfo {author} {\bibfnamefont {Vera}\ \bibnamefont
  {von Burg}}, \bibinfo {author} {\bibfnamefont {Guang~Hao}\ \bibnamefont
  {Low}}, \bibinfo {author} {\bibfnamefont {Thomas}\ \bibnamefont
  {H{\"{a}}ner}}, \bibinfo {author} {\bibfnamefont {Damian~S.}\ \bibnamefont
  {Steiger}}, \bibinfo {author} {\bibfnamefont {Markus}\ \bibnamefont
  {Reiher}}, \bibinfo {author} {\bibfnamefont {Martin}\ \bibnamefont
  {Roetteler}}, \ and\ \bibinfo {author} {\bibfnamefont {Matthias}\
  \bibnamefont {Troyer}},\ }\bibfield  {title} {\enquote {\bibinfo {title}
  {{Quantum computing enhanced computational catalysis}},}\ }\href
  {https://arxiv.org/abs/2007.14460} {\bibfield  {journal} {\bibinfo  {journal}
  {arXiv:2007.14460}\ } (\bibinfo {year} {2020})}\BibitemShut {NoStop}%
\bibitem [{\citenamefont {Lee}\ \emph {et~al.}(2021)\citenamefont {Lee},
  \citenamefont {Berry}, \citenamefont {Gidney}, \citenamefont {Huggins},
  \citenamefont {McClean}, \citenamefont {Wiebe},\ and\ \citenamefont
  {Babbush}}]{Lee2020}%
  \BibitemOpen
  \bibfield  {author} {\bibinfo {author} {\bibfnamefont {Joonho}\ \bibnamefont
  {Lee}}, \bibinfo {author} {\bibfnamefont {Dominic~W.}\ \bibnamefont {Berry}},
  \bibinfo {author} {\bibfnamefont {Craig}\ \bibnamefont {Gidney}}, \bibinfo
  {author} {\bibfnamefont {William~J.}\ \bibnamefont {Huggins}}, \bibinfo
  {author} {\bibfnamefont {Jarrod~R.}\ \bibnamefont {McClean}}, \bibinfo
  {author} {\bibfnamefont {Nathan}\ \bibnamefont {Wiebe}}, \ and\ \bibinfo
  {author} {\bibfnamefont {Ryan}\ \bibnamefont {Babbush}},\ }\bibfield  {title}
  {\enquote {\bibinfo {title} {{Even More Efficient Quantum Computations of
  Chemistry Through Tensor Hypercontraction}},}\ }\href {\doibase
  10.1103/PRXQuantum.2.030305} {\bibfield  {journal} {\bibinfo  {journal} {PRX
  Quantum}\ }\textbf {\bibinfo {volume} {2}},\ \bibinfo {pages} {030305}
  (\bibinfo {year} {2021})}\BibitemShut {NoStop}%
\bibitem [{\citenamefont {Su}\ \emph {et~al.}(2021)\citenamefont {Su},
  \citenamefont {Huang},\ and\ \citenamefont {Campbell}}]{Su2020_b}%
  \BibitemOpen
  \bibfield  {author} {\bibinfo {author} {\bibfnamefont {Yuan}\ \bibnamefont
  {Su}}, \bibinfo {author} {\bibfnamefont {Hsin-Yuan}\ \bibnamefont {Huang}}, \
  and\ \bibinfo {author} {\bibfnamefont {Earl~T}\ \bibnamefont {Campbell}},\
  }\bibfield  {title} {\enquote {\bibinfo {title} {{Nearly tight Trotterization
  of interacting electrons}},}\ }\href {\doibase 10.22331/q-2021-07-05-495}
  {\bibfield  {journal} {\bibinfo  {journal} {{Quantum}}\ }\textbf {\bibinfo
  {volume} {5}},\ \bibinfo {pages} {495} (\bibinfo {year} {2021})}\BibitemShut
  {NoStop}%
\bibitem [{\citenamefont {Babbush}\ \emph
  {et~al.}(2018{\natexlab{b}})\citenamefont {Babbush}, \citenamefont {Wiebe},
  \citenamefont {McClean}, \citenamefont {McClain}, \citenamefont {Neven},\
  and\ \citenamefont {Chan}}]{BabbushLow}%
  \BibitemOpen
  \bibfield  {author} {\bibinfo {author} {\bibfnamefont {Ryan}\ \bibnamefont
  {Babbush}}, \bibinfo {author} {\bibfnamefont {Nathan}\ \bibnamefont {Wiebe}},
  \bibinfo {author} {\bibfnamefont {Jarrod}\ \bibnamefont {McClean}}, \bibinfo
  {author} {\bibfnamefont {James}\ \bibnamefont {McClain}}, \bibinfo {author}
  {\bibfnamefont {Hartmut}\ \bibnamefont {Neven}}, \ and\ \bibinfo {author}
  {\bibfnamefont {Garnet Kin-Lic}\ \bibnamefont {Chan}},\ }\bibfield  {title}
  {\enquote {\bibinfo {title} {{Low-Depth Quantum Simulation of Materials}},}\
  }\href {https://journals.aps.org/prx/abstract/10.1103/PhysRevX.8.011044}
  {\bibfield  {journal} {\bibinfo  {journal} {Physical Review X}\ }\textbf
  {\bibinfo {volume} {8}},\ \bibinfo {pages} {011044} (\bibinfo {year}
  {2018}{\natexlab{b}})}\BibitemShut {NoStop}%
\bibitem [{\citenamefont {Seeley}\ \emph {et~al.}(2012)\citenamefont {Seeley},
  \citenamefont {Richard},\ and\ \citenamefont {Love}}]{Seeley2012}%
  \BibitemOpen
  \bibfield  {author} {\bibinfo {author} {\bibfnamefont {Jacob~T}\ \bibnamefont
  {Seeley}}, \bibinfo {author} {\bibfnamefont {Martin~J}\ \bibnamefont
  {Richard}}, \ and\ \bibinfo {author} {\bibfnamefont {Peter~J}\ \bibnamefont
  {Love}},\ }\bibfield  {title} {\enquote {\bibinfo {title} {{The Bravyi-Kitaev
  transformation for quantum computation of electronic structure}},}\ }\href
  {\doibase 10.1063/1.4768229} {\bibfield  {journal} {\bibinfo  {journal}
  {Journal of Chemical Physics}\ }\textbf {\bibinfo {volume} {137}},\ \bibinfo
  {pages} {224109} (\bibinfo {year} {2012})}\BibitemShut {NoStop}%
\bibitem [{\citenamefont {Bravyi}\ \emph {et~al.}(2017)\citenamefont {Bravyi},
  \citenamefont {Gambetta}, \citenamefont {Mezzacapo},\ and\ \citenamefont
  {Temme}}]{Bravyi2017}%
  \BibitemOpen
  \bibfield  {author} {\bibinfo {author} {\bibfnamefont {Sergey}\ \bibnamefont
  {Bravyi}}, \bibinfo {author} {\bibfnamefont {Jay~M}\ \bibnamefont
  {Gambetta}}, \bibinfo {author} {\bibfnamefont {Antonio}\ \bibnamefont
  {Mezzacapo}}, \ and\ \bibinfo {author} {\bibfnamefont {Kristan}\ \bibnamefont
  {Temme}},\ }\bibfield  {title} {\enquote {\bibinfo {title} {{Tapering off
  qubits to simulate fermionic Hamiltonians}},}\ }\href
  {http://arxiv.org/abs/1701.08213} {\bibfield  {journal} {\bibinfo  {journal}
  {arXiv:1701.08213}\ } (\bibinfo {year} {2017})}\BibitemShut {NoStop}%
\bibitem [{\citenamefont {Setia}\ and\ \citenamefont
  {Whitfield}(2018)}]{Setia2017}%
  \BibitemOpen
  \bibfield  {author} {\bibinfo {author} {\bibfnamefont {Kanav}\ \bibnamefont
  {Setia}}\ and\ \bibinfo {author} {\bibfnamefont {James~D}\ \bibnamefont
  {Whitfield}},\ }\bibfield  {title} {\enquote {\bibinfo {title}
  {{Bravyi-Kitaev Superfast simulation of fermions on a quantum computer}},}\
  }\href {\doibase 10.1063/1.5019371} {\bibfield  {journal} {\bibinfo
  {journal} {The Journal of Chemical Physics}\ }\textbf {\bibinfo {volume}
  {148}},\ \bibinfo {pages} {164104} (\bibinfo {year} {2018})}\BibitemShut
  {NoStop}%
\bibitem [{\citenamefont {Motta}\ \emph {et~al.}(2021)\citenamefont {Motta},
  \citenamefont {Ye}, \citenamefont {McClean}, \citenamefont {Li},
  \citenamefont {Minnich}, \citenamefont {Babbush},\ and\ \citenamefont
  {Chan}}]{Motta2018}%
  \BibitemOpen
  \bibfield  {author} {\bibinfo {author} {\bibfnamefont {Mario}\ \bibnamefont
  {Motta}}, \bibinfo {author} {\bibfnamefont {Erika}\ \bibnamefont {Ye}},
  \bibinfo {author} {\bibfnamefont {Jarrod~R.}\ \bibnamefont {McClean}},
  \bibinfo {author} {\bibfnamefont {Zhendong}\ \bibnamefont {Li}}, \bibinfo
  {author} {\bibfnamefont {Austin~J.}\ \bibnamefont {Minnich}}, \bibinfo
  {author} {\bibfnamefont {Ryan}\ \bibnamefont {Babbush}}, \ and\ \bibinfo
  {author} {\bibfnamefont {Garnet Kin-Lic}\ \bibnamefont {Chan}},\ }\bibfield
  {title} {\enquote {\bibinfo {title} {{Low rank representations for quantum
  simulation of electronic structure}},}\ }\href
  {https://www.nature.com/articles/s41534-021-00416-z} {\bibfield  {journal}
  {\bibinfo  {journal} {npj Quantum Information}\ }\textbf {\bibinfo {volume}
  {7}} (\bibinfo {year} {2021})}\BibitemShut {NoStop}%
\bibitem [{\citenamefont {Bauman}\ \emph {et~al.}(2019)\citenamefont {Bauman},
  \citenamefont {Bylaska}, \citenamefont {Krishnamoorthy}, \citenamefont {Low},
  \citenamefont {Wiebe}, \citenamefont {Granade}, \citenamefont {Roetteler},
  \citenamefont {Troyer},\ and\ \citenamefont
  {Kowalski}}]{bauman2019downfolding}%
  \BibitemOpen
  \bibfield  {author} {\bibinfo {author} {\bibfnamefont {Nicholas~P}\
  \bibnamefont {Bauman}}, \bibinfo {author} {\bibfnamefont {Eric~J}\
  \bibnamefont {Bylaska}}, \bibinfo {author} {\bibfnamefont {Sriram}\
  \bibnamefont {Krishnamoorthy}}, \bibinfo {author} {\bibfnamefont {Guang~Hao}\
  \bibnamefont {Low}}, \bibinfo {author} {\bibfnamefont {Nathan}\ \bibnamefont
  {Wiebe}}, \bibinfo {author} {\bibfnamefont {Christopher~E}\ \bibnamefont
  {Granade}}, \bibinfo {author} {\bibfnamefont {Martin}\ \bibnamefont
  {Roetteler}}, \bibinfo {author} {\bibfnamefont {Matthias}\ \bibnamefont
  {Troyer}}, \ and\ \bibinfo {author} {\bibfnamefont {Karol}\ \bibnamefont
  {Kowalski}},\ }\bibfield  {title} {\enquote {\bibinfo {title} {Downfolding of
  many-body hamiltonians using active-space models: Extension of the sub-system
  embedding sub-algebras approach to unitary coupled cluster formalisms},}\
  }\href {https://doi.org/10.1063/1.5094643} {\bibfield  {journal} {\bibinfo
  {journal} {The Journal of Chemical Physics}\ }\textbf {\bibinfo {volume}
  {151}},\ \bibinfo {pages} {014107} (\bibinfo {year} {2019})}\BibitemShut
  {NoStop}%
\bibitem [{\citenamefont {Takeshita}\ \emph {et~al.}(2020)\citenamefont
  {Takeshita}, \citenamefont {Rubin}, \citenamefont {Jiang}, \citenamefont
  {Lee}, \citenamefont {Babbush},\ and\ \citenamefont
  {McClean}}]{Takeshita2019b}%
  \BibitemOpen
  \bibfield  {author} {\bibinfo {author} {\bibfnamefont {Tyler}\ \bibnamefont
  {Takeshita}}, \bibinfo {author} {\bibfnamefont {Nicholas~C.}\ \bibnamefont
  {Rubin}}, \bibinfo {author} {\bibfnamefont {Zhang}\ \bibnamefont {Jiang}},
  \bibinfo {author} {\bibfnamefont {Eunseok}\ \bibnamefont {Lee}}, \bibinfo
  {author} {\bibfnamefont {Ryan}\ \bibnamefont {Babbush}}, \ and\ \bibinfo
  {author} {\bibfnamefont {Jarrod~R.}\ \bibnamefont {McClean}},\ }\bibfield
  {title} {\enquote {\bibinfo {title} {Increasing the representation accuracy
  of quantum simulations of chemistry without extra quantum resources},}\
  }\href {\doibase 10.1103/PhysRevX.10.011004} {\bibfield  {journal} {\bibinfo
  {journal} {Physical Review X}\ }\textbf {\bibinfo {volume} {10}},\ \bibinfo
  {pages} {011004} (\bibinfo {year} {2020})}\BibitemShut {NoStop}%
\bibitem [{\citenamefont {McClean}\ \emph {et~al.}(2020)\citenamefont
  {McClean}, \citenamefont {Faulstich}, \citenamefont {Zhu}, \citenamefont
  {O'Gorman}, \citenamefont {Qiu}, \citenamefont {White}, \citenamefont
  {Babbush},\ and\ \citenamefont {Lin}}]{McClean2020Galerkin_b}%
  \BibitemOpen
  \bibfield  {author} {\bibinfo {author} {\bibfnamefont {Jarrod~R}\
  \bibnamefont {McClean}}, \bibinfo {author} {\bibfnamefont {Fabian~M}\
  \bibnamefont {Faulstich}}, \bibinfo {author} {\bibfnamefont {Qinyi}\
  \bibnamefont {Zhu}}, \bibinfo {author} {\bibfnamefont {Bryan}\ \bibnamefont
  {O'Gorman}}, \bibinfo {author} {\bibfnamefont {Yiheng}\ \bibnamefont {Qiu}},
  \bibinfo {author} {\bibfnamefont {Steven~R}\ \bibnamefont {White}}, \bibinfo
  {author} {\bibfnamefont {Ryan}\ \bibnamefont {Babbush}}, \ and\ \bibinfo
  {author} {\bibfnamefont {Lin}\ \bibnamefont {Lin}},\ }\bibfield  {title}
  {\enquote {\bibinfo {title} {{Discontinuous Galerkin discretization for
  quantum simulation of chemistry}},}\ }\href {\doibase
  10.1088/1367-2630/ab9d9f} {\bibfield  {journal} {\bibinfo  {journal} {New
  Journal of Physics}\ }\textbf {\bibinfo {volume} {22}},\ \bibinfo {pages}
  {093015} (\bibinfo {year} {2020})}\BibitemShut {NoStop}%
\bibitem [{\citenamefont {Motta}\ \emph
  {et~al.}(2020{\natexlab{b}})\citenamefont {Motta}, \citenamefont {Gujarati},
  \citenamefont {Rice}, \citenamefont {Kumar}, \citenamefont {Masteran},
  \citenamefont {Latone}, \citenamefont {Lee}, \citenamefont {Valeev},\ and\
  \citenamefont {Takeshita}}]{Motta2020_b}%
  \BibitemOpen
  \bibfield  {author} {\bibinfo {author} {\bibfnamefont {Mario}\ \bibnamefont
  {Motta}}, \bibinfo {author} {\bibfnamefont {Tanvi~P}\ \bibnamefont
  {Gujarati}}, \bibinfo {author} {\bibfnamefont {Julia~E}\ \bibnamefont
  {Rice}}, \bibinfo {author} {\bibfnamefont {Ashutosh}\ \bibnamefont {Kumar}},
  \bibinfo {author} {\bibfnamefont {Conner}\ \bibnamefont {Masteran}}, \bibinfo
  {author} {\bibfnamefont {Joseph~A}\ \bibnamefont {Latone}}, \bibinfo {author}
  {\bibfnamefont {Eunseok}\ \bibnamefont {Lee}}, \bibinfo {author}
  {\bibfnamefont {Edward~F}\ \bibnamefont {Valeev}}, \ and\ \bibinfo {author}
  {\bibfnamefont {Tyler~Y}\ \bibnamefont {Takeshita}},\ }\bibfield  {title}
  {\enquote {\bibinfo {title} {{Quantum simulation of electronic structure with
  a transcorrelated Hamiltonian: improved accuracy with a smaller footprint on
  the quantum computer}},}\ }\href {\doibase 10.1039/D0CP04106H} {\bibfield
  {journal} {\bibinfo  {journal} {Physical Chemistry Chemical Physics}\
  }\textbf {\bibinfo {volume} {22}},\ \bibinfo {pages} {24270--24281} (\bibinfo
  {year} {2020}{\natexlab{b}})}\BibitemShut {NoStop}%
\bibitem [{\citenamefont {Reiher}\ \emph {et~al.}(2017)\citenamefont {Reiher},
  \citenamefont {Wiebe}, \citenamefont {Svore}, \citenamefont {Wecker},\ and\
  \citenamefont {Troyer}}]{Reiher2017}%
  \BibitemOpen
  \bibfield  {author} {\bibinfo {author} {\bibfnamefont {Markus}\ \bibnamefont
  {Reiher}}, \bibinfo {author} {\bibfnamefont {Nathan}\ \bibnamefont {Wiebe}},
  \bibinfo {author} {\bibfnamefont {Krysta~M}\ \bibnamefont {Svore}}, \bibinfo
  {author} {\bibfnamefont {Dave}\ \bibnamefont {Wecker}}, \ and\ \bibinfo
  {author} {\bibfnamefont {Matthias}\ \bibnamefont {Troyer}},\ }\bibfield
  {title} {\enquote {\bibinfo {title} {{Elucidating Reaction Mechanisms on
  Quantum Computers}},}\ }\href
  {http://www.pnas.org/content/114/29/7555.abstract} {\bibfield  {journal}
  {\bibinfo  {journal} {Proceedings of the National Academy of Sciences}\
  }\textbf {\bibinfo {volume} {114}},\ \bibinfo {pages} {7555--7560} (\bibinfo
  {year} {2017})}\BibitemShut {NoStop}%
\bibitem [{\citenamefont {Li}\ \emph {et~al.}(2019)\citenamefont {Li},
  \citenamefont {Li}, \citenamefont {Dattani}, \citenamefont {Umrigar},\ and\
  \citenamefont {Chan}}]{Li2019_b}%
  \BibitemOpen
  \bibfield  {author} {\bibinfo {author} {\bibfnamefont {Zhendong}\
  \bibnamefont {Li}}, \bibinfo {author} {\bibfnamefont {Junhao}\ \bibnamefont
  {Li}}, \bibinfo {author} {\bibfnamefont {Nikesh~S}\ \bibnamefont {Dattani}},
  \bibinfo {author} {\bibfnamefont {C~J}\ \bibnamefont {Umrigar}}, \ and\
  \bibinfo {author} {\bibfnamefont {Garnet Kin-Lic}\ \bibnamefont {Chan}},\
  }\bibfield  {title} {\enquote {\bibinfo {title} {{The electronic complexity
  of the ground-state of the FeMo cofactor of nitrogenase as relevant to
  quantum simulations}},}\ }\href {\doibase 10.1063/1.5063376} {\bibfield
  {journal} {\bibinfo  {journal} {The Journal of Chemical Physics}\ }\textbf
  {\bibinfo {volume} {150}},\ \bibinfo {pages} {024302} (\bibinfo {year}
  {2019})}\BibitemShut {NoStop}%
\bibitem [{\citenamefont {Kassal}\ \emph {et~al.}(2008)\citenamefont {Kassal},
  \citenamefont {Jordan}, \citenamefont {Love}, \citenamefont {Mohseni},\ and\
  \citenamefont {Aspuru-Guzik}}]{Kassal2008}%
  \BibitemOpen
  \bibfield  {author} {\bibinfo {author} {\bibfnamefont {Ivan}\ \bibnamefont
  {Kassal}}, \bibinfo {author} {\bibfnamefont {Stephen~P}\ \bibnamefont
  {Jordan}}, \bibinfo {author} {\bibfnamefont {Peter~J}\ \bibnamefont {Love}},
  \bibinfo {author} {\bibfnamefont {Masoud}\ \bibnamefont {Mohseni}}, \ and\
  \bibinfo {author} {\bibfnamefont {Alan}\ \bibnamefont {Aspuru-Guzik}},\
  }\bibfield  {title} {\enquote {\bibinfo {title} {{Polynomial-time quantum
  algorithm for the simulation of chemical dynamics}},}\ }\href
  {http://www.pnas.org/content/105/48/18681.abstract} {\bibfield  {journal}
  {\bibinfo  {journal} {Proceedings of the National Academy of Sciences}\
  }\textbf {\bibinfo {volume} {105}},\ \bibinfo {pages} {18681--18686}
  (\bibinfo {year} {2008})}\BibitemShut {NoStop}%
\bibitem [{\citenamefont {Cody~Jones}\ \emph {et~al.}(2012)\citenamefont
  {Cody~Jones}, \citenamefont {Whitfield}, \citenamefont {McMahon},
  \citenamefont {Yung}, \citenamefont {van Meter}, \citenamefont
  {Aspuru-Guzik},\ and\ \citenamefont {Yamamoto}}]{Jones2012}%
  \BibitemOpen
  \bibfield  {author} {\bibinfo {author} {\bibfnamefont {N}~\bibnamefont
  {Cody~Jones}}, \bibinfo {author} {\bibfnamefont {James~D}\ \bibnamefont
  {Whitfield}}, \bibinfo {author} {\bibfnamefont {Peter~L}\ \bibnamefont
  {McMahon}}, \bibinfo {author} {\bibfnamefont {Man-Hong}\ \bibnamefont
  {Yung}}, \bibinfo {author} {\bibfnamefont {Rodney}\ \bibnamefont {van
  Meter}}, \bibinfo {author} {\bibfnamefont {Alan}\ \bibnamefont
  {Aspuru-Guzik}}, \ and\ \bibinfo {author} {\bibfnamefont {Yoshihisa}\
  \bibnamefont {Yamamoto}},\ }\bibfield  {title} {\enquote {\bibinfo {title}
  {{Faster quantum chemistry simulation on fault-tolerant quantum
  computers}},}\ }\href {\doibase 10.1088/1367-2630/14/11/115023} {\bibfield
  {journal} {\bibinfo  {journal} {New Journal of Physics}\ }\textbf {\bibinfo
  {volume} {14}},\ \bibinfo {pages} {115023} (\bibinfo {year}
  {2012})}\BibitemShut {NoStop}%
\bibitem [{\citenamefont {Polanyi}(1972)}]{polanyi1972bconcepts}%
  \BibitemOpen
  \bibfield  {author} {\bibinfo {author} {\bibfnamefont {John~C}\ \bibnamefont
  {Polanyi}},\ }\bibfield  {title} {\enquote {\bibinfo {title} {Concepts in
  reaction dynamics},}\ }\href {\doibase 10.1021/ar50053a001} {\bibfield
  {journal} {\bibinfo  {journal} {Accounts of Chemical Research}\ }\textbf
  {\bibinfo {volume} {5}},\ \bibinfo {pages} {161--168} (\bibinfo {year}
  {1972})}\BibitemShut {NoStop}%
\bibitem [{\citenamefont {Born}\ and\ \citenamefont
  {Oppenheimer}(1927)}]{born1927quantentheorie}%
  \BibitemOpen
  \bibfield  {author} {\bibinfo {author} {\bibfnamefont {Max}\ \bibnamefont
  {Born}}\ and\ \bibinfo {author} {\bibfnamefont {Robert}\ \bibnamefont
  {Oppenheimer}},\ }\bibfield  {title} {\enquote {\bibinfo {title} {Zur
  quantentheorie der molekeln},}\ }\href
  {https://onlinelibrary.wiley.com/doi/abs/10.1002/andp.19273892002} {\bibfield
   {journal} {\bibinfo  {journal} {Annalen der Physik}\ }\textbf {\bibinfo
  {volume} {389}},\ \bibinfo {pages} {457--484} (\bibinfo {year}
  {1927})}\BibitemShut {NoStop}%
\bibitem [{\citenamefont {Yarkony}(2012)}]{yarkony2012nonadiabatic}%
  \BibitemOpen
  \bibfield  {author} {\bibinfo {author} {\bibfnamefont {David~R}\ \bibnamefont
  {Yarkony}},\ }\bibfield  {title} {\enquote {\bibinfo {title} {Nonadiabatic
  quantum chemistry past, present, and future},}\ }\href
  {https://pubs.acs.org/doi/full/10.1021/cr2001299} {\bibfield  {journal}
  {\bibinfo  {journal} {Chemical Reviews}\ }\textbf {\bibinfo {volume} {112}},\
  \bibinfo {pages} {481--498} (\bibinfo {year} {2012})}\BibitemShut {NoStop}%
\bibitem [{\citenamefont {Butler}(1998)}]{butler1998chemical}%
  \BibitemOpen
  \bibfield  {author} {\bibinfo {author} {\bibfnamefont {Laurie~J}\
  \bibnamefont {Butler}},\ }\bibfield  {title} {\enquote {\bibinfo {title}
  {Chemical reaction dynamics beyond the {Born-Oppenheimer} approximation},}\
  }\href
  {https://www.annualreviews.org/doi/pdf/10.1146/annurev.physchem.49.1.125}
  {\bibfield  {journal} {\bibinfo  {journal} {Annual Review of Physical
  Chemistry}\ }\textbf {\bibinfo {volume} {49}},\ \bibinfo {pages} {125--171}
  (\bibinfo {year} {1998})}\BibitemShut {NoStop}%
\bibitem [{\citenamefont {Curchod}\ and\ \citenamefont
  {Mart{\'\i}nez}(2018)}]{curchod2018ab}%
  \BibitemOpen
  \bibfield  {author} {\bibinfo {author} {\bibfnamefont {Basile F~E}\
  \bibnamefont {Curchod}}\ and\ \bibinfo {author} {\bibfnamefont {Todd~J}\
  \bibnamefont {Mart{\'\i}nez}},\ }\bibfield  {title} {\enquote {\bibinfo
  {title} {Ab initio nonadiabatic quantum molecular dynamics},}\ }\href
  {https://pubs.acs.org/doi/full/10.1021/acs.chemrev.7b00423} {\bibfield
  {journal} {\bibinfo  {journal} {Chemical Reviews}\ }\textbf {\bibinfo
  {volume} {118}},\ \bibinfo {pages} {3305--3336} (\bibinfo {year}
  {2018})}\BibitemShut {NoStop}%
\bibitem [{\citenamefont {Zhu}\ and\ \citenamefont
  {Yarkony}(2016)}]{zhu2016non}%
  \BibitemOpen
  \bibfield  {author} {\bibinfo {author} {\bibfnamefont {Xiaolei}\ \bibnamefont
  {Zhu}}\ and\ \bibinfo {author} {\bibfnamefont {David~R}\ \bibnamefont
  {Yarkony}},\ }\bibfield  {title} {\enquote {\bibinfo {title}
  {Non-adiabaticity: the importance of conical intersections},}\ }\href
  {https://www.tandfonline.com/doi/abs/10.1080/00268976.2016.1170218}
  {\bibfield  {journal} {\bibinfo  {journal} {Molecular Physics}\ }\textbf
  {\bibinfo {volume} {114}},\ \bibinfo {pages} {1983--2013} (\bibinfo {year}
  {2016})}\BibitemShut {NoStop}%
\bibitem [{\citenamefont {Car}\ and\ \citenamefont
  {Parrinello}(1985)}]{Car1985}%
  \BibitemOpen
  \bibfield  {author} {\bibinfo {author} {\bibfnamefont {R}~\bibnamefont
  {Car}}\ and\ \bibinfo {author} {\bibfnamefont {M}~\bibnamefont
  {Parrinello}},\ }\bibfield  {title} {\enquote {\bibinfo {title} {Unified
  approach for molecular dynamics and density-functional theory},}\ }\href
  {\doibase 10.1103/PhysRevLett.55.2471} {\bibfield  {journal} {\bibinfo
  {journal} {Physical Review Letters}\ }\textbf {\bibinfo {volume} {55}},\
  \bibinfo {pages} {2471--2474} (\bibinfo {year} {1985})}\BibitemShut {NoStop}%
\bibitem [{\citenamefont {Marx}\ and\ \citenamefont {Hutter}(2009)}]{marx2009}%
  \BibitemOpen
  \bibfield  {author} {\bibinfo {author} {\bibfnamefont {Dominik}\ \bibnamefont
  {Marx}}\ and\ \bibinfo {author} {\bibfnamefont {J{\"u}rg}\ \bibnamefont
  {Hutter}},\ }\href {\doibase 10.1017/CBO9780511609633} {\emph {\bibinfo
  {title} {{Ab Initio Molecular Dynamics: Basic Theory and Advanced
  Methods}}}}\ (\bibinfo  {publisher} {Cambridge University Press},\ \bibinfo
  {address} {Cambridge, U.K.},\ \bibinfo {year} {2009})\BibitemShut {NoStop}%
\bibitem [{\citenamefont {Hutter}(2012)}]{hutter2012car}%
  \BibitemOpen
  \bibfield  {author} {\bibinfo {author} {\bibfnamefont {J{\"u}rg}\
  \bibnamefont {Hutter}},\ }\bibfield  {title} {\enquote {\bibinfo {title}
  {Car--{P}arrinello molecular dynamics},}\ }\href
  {https://onlinelibrary.wiley.com/doi/abs/10.1002/wcms.90} {\bibfield
  {journal} {\bibinfo  {journal} {Wiley Interdisciplinary Reviews:
  Computational Molecular Science}\ }\textbf {\bibinfo {volume} {2}},\ \bibinfo
  {pages} {604--612} (\bibinfo {year} {2012})}\BibitemShut {NoStop}%
\bibitem [{\citenamefont {Gr{\"{u}}neis}\ \emph {et~al.}(2013)\citenamefont
  {Gr{\"{u}}neis}, \citenamefont {Shepherd}, \citenamefont {Alavi},
  \citenamefont {Tew},\ and\ \citenamefont {Booth}}]{gruneis2013explicitly}%
  \BibitemOpen
  \bibfield  {author} {\bibinfo {author} {\bibfnamefont {Andreas}\ \bibnamefont
  {Gr{\"{u}}neis}}, \bibinfo {author} {\bibfnamefont {James~J}\ \bibnamefont
  {Shepherd}}, \bibinfo {author} {\bibfnamefont {Ali}\ \bibnamefont {Alavi}},
  \bibinfo {author} {\bibfnamefont {David~P}\ \bibnamefont {Tew}}, \ and\
  \bibinfo {author} {\bibfnamefont {George~H}\ \bibnamefont {Booth}},\
  }\bibfield  {title} {\enquote {\bibinfo {title} {{Explicitly correlated plane
  waves: Accelerating convergence in periodic wavefunction expansions}},}\
  }\href {http://aip.scitation.org/doi/abs/10.1063/1.4818753} {\bibfield
  {journal} {\bibinfo  {journal} {The Journal of Chemical Physics}\ }\textbf
  {\bibinfo {volume} {139}},\ \bibinfo {pages} {84112} (\bibinfo {year}
  {2013})}\BibitemShut {NoStop}%
\bibitem [{\citenamefont {Hättig}\ \emph {et~al.}(2011)\citenamefont
  {Hättig}, \citenamefont {Klopper}, \citenamefont {Köhn},\ and\
  \citenamefont {Tew}}]{hattig2011explicitly}%
  \BibitemOpen
  \bibfield  {author} {\bibinfo {author} {\bibfnamefont {Christof}\
  \bibnamefont {Hättig}}, \bibinfo {author} {\bibfnamefont {Wim}\
  \bibnamefont {Klopper}}, \bibinfo {author} {\bibfnamefont {Andreas}\
  \bibnamefont {Köhn}}, \ and\ \bibinfo {author} {\bibfnamefont {David~P}\
  \bibnamefont {Tew}},\ }\bibfield  {title} {\enquote {\bibinfo {title}
  {{Explicitly correlated electrons in molecules}},}\ }\href
  {http://pubs.acs.org/doi/abs/10.1021/cr200168z} {\bibfield  {journal}
  {\bibinfo  {journal} {Chemical Reviews}\ }\textbf {\bibinfo {volume} {112}},\
  \bibinfo {pages} {4--74} (\bibinfo {year} {2011})}\BibitemShut {NoStop}%
\bibitem [{\citenamefont {Harl}\ and\ \citenamefont {Kresse}(2008)}]{Harl2008}%
  \BibitemOpen
  \bibfield  {author} {\bibinfo {author} {\bibfnamefont {Judith}\ \bibnamefont
  {Harl}}\ and\ \bibinfo {author} {\bibfnamefont {Georg}\ \bibnamefont
  {Kresse}},\ }\bibfield  {title} {\enquote {\bibinfo {title} {{Cohesive energy
  curves for noble gas solids calculated by adiabatic connection
  fluctuation-dissipation theory}},}\ }\href {\doibase
  10.1103/PhysRevB.77.045136} {\bibfield  {journal} {\bibinfo  {journal}
  {Physical Review B}\ }\textbf {\bibinfo {volume} {77}},\ \bibinfo {pages}
  {45136} (\bibinfo {year} {2008})}\BibitemShut {NoStop}%
\bibitem [{\citenamefont {Shepherd}\ \emph {et~al.}(2012)\citenamefont
  {Shepherd}, \citenamefont {Gr{\"{u}}neis}, \citenamefont {Booth},
  \citenamefont {Kresse},\ and\ \citenamefont
  {Alavi}}]{shepherd2012convergence}%
  \BibitemOpen
  \bibfield  {author} {\bibinfo {author} {\bibfnamefont {James~J}\ \bibnamefont
  {Shepherd}}, \bibinfo {author} {\bibfnamefont {Andreas}\ \bibnamefont
  {Gr{\"{u}}neis}}, \bibinfo {author} {\bibfnamefont {George~H}\ \bibnamefont
  {Booth}}, \bibinfo {author} {\bibfnamefont {Georg}\ \bibnamefont {Kresse}}, \
  and\ \bibinfo {author} {\bibfnamefont {Ali}\ \bibnamefont {Alavi}},\
  }\bibfield  {title} {\enquote {\bibinfo {title} {{Convergence of many-body
  wave-function expansions using a plane-wave basis: From homogeneous electron
  gas to solid state systems}},}\ }\href {\doibase 10.1103/PhysRevB.86.035111}
  {\bibfield  {journal} {\bibinfo  {journal} {Physical Review B}\ }\textbf
  {\bibinfo {volume} {86}},\ \bibinfo {pages} {35111} (\bibinfo {year}
  {2012})}\BibitemShut {NoStop}%
\bibitem [{\citenamefont {Helgaker}\ \emph {et~al.}(1998)\citenamefont
  {Helgaker}, \citenamefont {Klopper}, \citenamefont {Koch},\ and\
  \citenamefont {Noga}}]{helgaker1997basis_b}%
  \BibitemOpen
  \bibfield  {author} {\bibinfo {author} {\bibfnamefont {Trygve}\ \bibnamefont
  {Helgaker}}, \bibinfo {author} {\bibfnamefont {Wim}\ \bibnamefont {Klopper}},
  \bibinfo {author} {\bibfnamefont {Henrik}\ \bibnamefont {Koch}}, \ and\
  \bibinfo {author} {\bibfnamefont {Jozef}\ \bibnamefont {Noga}},\ }\bibfield
  {title} {\enquote {\bibinfo {title} {{Basis-set convergence of correlated
  calculations on water}},}\ }\href {\doibase 10.1063/1.473863} {\bibfield
  {journal} {\bibinfo  {journal} {Journal of Chemical Physics}\ }\textbf
  {\bibinfo {volume} {106}},\ \bibinfo {pages} {9639} (\bibinfo {year}
  {1998})}\BibitemShut {NoStop}%
\bibitem [{\citenamefont {Klopper}(1995)}]{klopper1995ab2}%
  \BibitemOpen
  \bibfield  {author} {\bibinfo {author} {\bibfnamefont {Wim}\ \bibnamefont
  {Klopper}},\ }\bibfield  {title} {\enquote {\bibinfo {title} {Limiting values
  for {M{\o}ller-Plesset} second-order correlation energies of polyatomic
  systems: A benchmark study on {Ne}, {HF}, {H\textsubscript{2}O},
  {N\textsubscript{2}}, and {He...He}},}\ }\href {\doibase 10.1063/1.469351}
  {\bibfield  {journal} {\bibinfo  {journal} {The Journal of Chemical Physics}\
  }\textbf {\bibinfo {volume} {102}},\ \bibinfo {pages} {6168--6179} (\bibinfo
  {year} {1995})}\BibitemShut {NoStop}%
\bibitem [{\citenamefont {Halkier}\ \emph {et~al.}(1998)\citenamefont
  {Halkier}, \citenamefont {Helgaker}, \citenamefont {J{\o}rgensen},
  \citenamefont {Klopper}, \citenamefont {Koch}, \citenamefont {Olsen},\ and\
  \citenamefont {Wilson}}]{Halkier1998b}%
  \BibitemOpen
  \bibfield  {author} {\bibinfo {author} {\bibfnamefont {Asger}\ \bibnamefont
  {Halkier}}, \bibinfo {author} {\bibfnamefont {Trygve}\ \bibnamefont
  {Helgaker}}, \bibinfo {author} {\bibfnamefont {Poul}\ \bibnamefont
  {J{\o}rgensen}}, \bibinfo {author} {\bibfnamefont {Wim}\ \bibnamefont
  {Klopper}}, \bibinfo {author} {\bibfnamefont {Henrik}\ \bibnamefont {Koch}},
  \bibinfo {author} {\bibfnamefont {Jeppe}\ \bibnamefont {Olsen}}, \ and\
  \bibinfo {author} {\bibfnamefont {Angela~K.}\ \bibnamefont {Wilson}},\
  }\bibfield  {title} {\enquote {\bibinfo {title} {Basis-set convergence in
  correlated calculations on {Ne}, {N\textsubscript{2}}, and
  {H\textsubscript{2}}},}\ }\href {\doibase
  https://doi.org/10.1016/S0009-2614(98)00111-0} {\bibfield  {journal}
  {\bibinfo  {journal} {Chemical Physics Letters}\ }\textbf {\bibinfo {volume}
  {286}},\ \bibinfo {pages} {243--252} (\bibinfo {year} {1998})}\BibitemShut
  {NoStop}%
\bibitem [{\citenamefont {Kato}(1957)}]{Kato1957}%
  \BibitemOpen
  \bibfield  {author} {\bibinfo {author} {\bibfnamefont {Tosio}\ \bibnamefont
  {Kato}},\ }\bibfield  {title} {\enquote {\bibinfo {title} {{On the
  eigenfunctions of many-particle systems in quantum mechanics}},}\ }\href
  {\doibase 10.1002/cpa.3160100201} {\bibfield  {journal} {\bibinfo  {journal}
  {Communications on Pure and Applied Mathematics}\ }\textbf {\bibinfo {volume}
  {10}},\ \bibinfo {pages} {151--177} (\bibinfo {year} {1957})}\BibitemShut
  {NoStop}%
\bibitem [{\citenamefont {Martin}(2004)}]{MartinES2004}%
  \BibitemOpen
  \bibfield  {author} {\bibinfo {author} {\bibfnamefont {Richard~M}\
  \bibnamefont {Martin}},\ }\href {\doibase 10.1017/CBO9780511805769} {\emph
  {\bibinfo {title} {{Electronic Structure}}}}\ (\bibinfo  {publisher}
  {Cambridge University Press},\ \bibinfo {address} {Cambridge, U.K.},\
  \bibinfo {year} {2004})\BibitemShut {NoStop}%
\bibitem [{\citenamefont {Helgaker}\ \emph {et~al.}(2000)\citenamefont
  {Helgaker}, \citenamefont {J{\o}rgensen},\ and\ \citenamefont
  {Olsen}}]{Helgaker2000}%
  \BibitemOpen
  \bibfield  {author} {\bibinfo {author} {\bibfnamefont {Trygve}\ \bibnamefont
  {Helgaker}}, \bibinfo {author} {\bibfnamefont {Poul}\ \bibnamefont
  {J{\o}rgensen}}, \ and\ \bibinfo {author} {\bibfnamefont {Jeppe}\
  \bibnamefont {Olsen}},\ }\href {\doibase 10.1002/9781119019572} {\emph
  {\bibinfo {title} {{Molecular Electronic-Structure Theory}}}}\ (\bibinfo
  {publisher} {John Wiley \& Sons, Ltd},\ \bibinfo {year} {2000})\BibitemShut
  {NoStop}%
\bibitem [{\citenamefont {McArdle}\ \emph {et~al.}(2021)\citenamefont
  {McArdle}, \citenamefont {Campbell},\ and\ \citenamefont {Su}}]{MCS21}%
  \BibitemOpen
  \bibfield  {author} {\bibinfo {author} {\bibfnamefont {Sam}\ \bibnamefont
  {McArdle}}, \bibinfo {author} {\bibfnamefont {Earl}\ \bibnamefont
  {Campbell}}, \ and\ \bibinfo {author} {\bibfnamefont {Yuan}\ \bibnamefont
  {Su}},\ }\bibfield  {title} {\enquote {\bibinfo {title} {Exploiting fermion
  number in factorized decompositions of the electronic structure
  {H}amiltonian},}\ }\href {http://arxiv.org/abs/2107.07238} {\bibfield
  {journal} {\bibinfo  {journal} {arXiv:2107.07238}\ } (\bibinfo {year}
  {2021})}\BibitemShut {NoStop}%
\bibitem [{\citenamefont {Babbush}\ \emph
  {et~al.}(2018{\natexlab{c}})\citenamefont {Babbush}, \citenamefont {Berry},
  \citenamefont {Sanders}, \citenamefont {Kivlichan}, \citenamefont {Scherer},
  \citenamefont {Wei}, \citenamefont {Love},\ and\ \citenamefont
  {Aspuru-Guzik}}]{BabbushSparse2}%
  \BibitemOpen
  \bibfield  {author} {\bibinfo {author} {\bibfnamefont {Ryan}\ \bibnamefont
  {Babbush}}, \bibinfo {author} {\bibfnamefont {Dominic~W}\ \bibnamefont
  {Berry}}, \bibinfo {author} {\bibfnamefont {Yuval~R}\ \bibnamefont
  {Sanders}}, \bibinfo {author} {\bibfnamefont {Ian~D}\ \bibnamefont
  {Kivlichan}}, \bibinfo {author} {\bibfnamefont {Artur}\ \bibnamefont
  {Scherer}}, \bibinfo {author} {\bibfnamefont {Annie~Y}\ \bibnamefont {Wei}},
  \bibinfo {author} {\bibfnamefont {Peter~J}\ \bibnamefont {Love}}, \ and\
  \bibinfo {author} {\bibfnamefont {Alan}\ \bibnamefont {Aspuru-Guzik}},\
  }\bibfield  {title} {\enquote {\bibinfo {title} {{Exponentially More Precise
  Quantum Simulation of Fermions in the Configuration Interaction
  Representation}},}\ }\href
  {http://iopscience.iop.org/article/10.1088/2058-9565/aa9463/meta} {\bibfield
  {journal} {\bibinfo  {journal} {Quantum Science and Technology}\ }\textbf
  {\bibinfo {volume} {3}},\ \bibinfo {pages} {015006} (\bibinfo {year}
  {2018}{\natexlab{c}})}\BibitemShut {NoStop}%
\bibitem [{\citenamefont {Kitaev}(1995)}]{KitaevARX1995}%
  \BibitemOpen
  \bibfield  {author} {\bibinfo {author} {\bibfnamefont {Alexei~Y}\
  \bibnamefont {Kitaev}},\ }\bibfield  {title} {\enquote {\bibinfo {title}
  {{Quantum measurements and the Abelian Stabilizer Problem}},}\ }\href
  {http://arxiv.org/abs/quant-ph/9511026} {\bibfield  {journal} {\bibinfo
  {journal} {arXiv:quant-ph/9511026}\ } (\bibinfo {year} {1995})}\BibitemShut
  {NoStop}%
\bibitem [{\citenamefont {Abrams}\ and\ \citenamefont
  {Lloyd}(1999)}]{Abrams1999}%
  \BibitemOpen
  \bibfield  {author} {\bibinfo {author} {\bibfnamefont {Daniel~S}\
  \bibnamefont {Abrams}}\ and\ \bibinfo {author} {\bibfnamefont {Seth}\
  \bibnamefont {Lloyd}},\ }\bibfield  {title} {\enquote {\bibinfo {title}
  {{Quantum Algorithm Providing Exponential Speed Increase for Finding
  Eigenvalues and Eigenvectors}},}\ }\href {\doibase
  10.1103/PhysRevLett.83.5162} {\bibfield  {journal} {\bibinfo  {journal}
  {Physical Review Letters}\ }\textbf {\bibinfo {volume} {83}},\ \bibinfo
  {pages} {5162--5165} (\bibinfo {year} {1999})}\BibitemShut {NoStop}%
\bibitem [{\citenamefont {Childs}\ and\ \citenamefont
  {Wiebe}(2012)}]{Childs2012}%
  \BibitemOpen
  \bibfield  {author} {\bibinfo {author} {\bibfnamefont {Andrew~M}\
  \bibnamefont {Childs}}\ and\ \bibinfo {author} {\bibfnamefont {Nathan}\
  \bibnamefont {Wiebe}},\ }\bibfield  {title} {\enquote {\bibinfo {title}
  {{Hamiltonian simulation using linear combinations of unitary operations}},}\
  }\href {https://dl.acm.org/citation.cfm?id=2481570} {\bibfield  {journal}
  {\bibinfo  {journal} {Quantum Information {\&} Computation}\ }\textbf
  {\bibinfo {volume} {12}},\ \bibinfo {pages} {901--924} (\bibinfo {year}
  {2012})}\BibitemShut {NoStop}%
\bibitem [{\citenamefont {Szegedy}(2004)}]{Szegedy2004}%
  \BibitemOpen
  \bibfield  {author} {\bibinfo {author} {\bibfnamefont {Mario}\ \bibnamefont
  {Szegedy}},\ }\bibfield  {title} {\enquote {\bibinfo {title} {{Quantum
  Speed-Up of Markov Chain Based Algorithms}},}\ }in\ \href {\doibase
  10.1109/FOCS.2004.53} {\emph {\bibinfo {booktitle} {45th Annual IEEE
  Symposium on Foundations of Computer Science}}}\ (\bibinfo  {publisher}
  {IEEE},\ \bibinfo {year} {2004})\ pp.\ \bibinfo {pages} {32--41}\BibitemShut
  {NoStop}%
\bibitem [{\citenamefont {Low}\ and\ \citenamefont {Chuang}(2017)}]{Low2017}%
  \BibitemOpen
  \bibfield  {author} {\bibinfo {author} {\bibfnamefont {Guang~Hao}\
  \bibnamefont {Low}}\ and\ \bibinfo {author} {\bibfnamefont {Isaac~L}\
  \bibnamefont {Chuang}},\ }\bibfield  {title} {\enquote {\bibinfo {title}
  {{Optimal Hamiltonian Simulation by Quantum Signal Processing}},}\ }\href
  {\doibase 10.1103/PhysRevLett.118.010501} {\bibfield  {journal} {\bibinfo
  {journal} {Physical Review Letters}\ }\textbf {\bibinfo {volume} {118}},\
  \bibinfo {pages} {010501} (\bibinfo {year} {2017})}\BibitemShut {NoStop}%
\bibitem [{\citenamefont {Kieferov\'{a}}\ \emph {et~al.}(2019)\citenamefont
  {Kieferov\'{a}}, \citenamefont {Scherer},\ and\ \citenamefont
  {Berry}}]{Kieferova18}%
  \BibitemOpen
  \bibfield  {author} {\bibinfo {author} {\bibfnamefont {M\'{a}ria}\
  \bibnamefont {Kieferov\'{a}}}, \bibinfo {author} {\bibfnamefont {Artur}\
  \bibnamefont {Scherer}}, \ and\ \bibinfo {author} {\bibfnamefont {Dominic~W}\
  \bibnamefont {Berry}},\ }\bibfield  {title} {\enquote {\bibinfo {title}
  {Simulating the dynamics of time-dependent {H}amiltonians with a truncated
  {D}yson series},}\ }\href {\doibase 10.1103/PhysRevA.99.042314} {\bibfield
  {journal} {\bibinfo  {journal} {Physical Review A}\ }\textbf {\bibinfo
  {volume} {99}},\ \bibinfo {pages} {042314} (\bibinfo {year}
  {2019})}\BibitemShut {NoStop}%
\bibitem [{\citenamefont {Fowler}\ \emph {et~al.}(2012)\citenamefont {Fowler},
  \citenamefont {Mariantoni}, \citenamefont {Martinis},\ and\ \citenamefont
  {Cleland}}]{Fowler2012}%
  \BibitemOpen
  \bibfield  {author} {\bibinfo {author} {\bibfnamefont {Austin~G}\
  \bibnamefont {Fowler}}, \bibinfo {author} {\bibfnamefont {Matteo}\
  \bibnamefont {Mariantoni}}, \bibinfo {author} {\bibfnamefont {John~M}\
  \bibnamefont {Martinis}}, \ and\ \bibinfo {author} {\bibfnamefont {Andrew~N}\
  \bibnamefont {Cleland}},\ }\bibfield  {title} {\enquote {\bibinfo {title}
  {{Surface codes: Towards practical large-scale quantum computation}},}\
  }\href {\doibase 10.1103/PhysRevA.86.032324} {\bibfield  {journal} {\bibinfo
  {journal} {Physical Review A}\ }\textbf {\bibinfo {volume} {86}},\ \bibinfo
  {pages} {32324} (\bibinfo {year} {2012})}\BibitemShut {NoStop}%
\bibitem [{\citenamefont {Low}\ \emph {et~al.}(2018)\citenamefont {Low},
  \citenamefont {Kliuchnikov},\ and\ \citenamefont {Schaeffer}}]{Low2018a}%
  \BibitemOpen
  \bibfield  {author} {\bibinfo {author} {\bibfnamefont {Guang~Hao}\
  \bibnamefont {Low}}, \bibinfo {author} {\bibfnamefont {Vadym}\ \bibnamefont
  {Kliuchnikov}}, \ and\ \bibinfo {author} {\bibfnamefont {Luke}\ \bibnamefont
  {Schaeffer}},\ }\bibfield  {title} {\enquote {\bibinfo {title} {{Trading
  T-gates for dirty qubits in state preparation and unitary synthesis}},}\
  }\href {http://arxiv.org/abs/1812.00954} {\bibfield  {journal} {\bibinfo
  {journal} {arXiv:1812.00954}\ } (\bibinfo {year} {2018})}\BibitemShut
  {NoStop}%
\bibitem [{\citenamefont {Sanders}\ \emph {et~al.}(2020)\citenamefont
  {Sanders}, \citenamefont {Berry}, \citenamefont {Costa}, \citenamefont
  {Tessler}, \citenamefont {Wiebe}, \citenamefont {Gidney}, \citenamefont
  {Neven},\ and\ \citenamefont {Babbush}}]{Sanders2020_b}%
  \BibitemOpen
  \bibfield  {author} {\bibinfo {author} {\bibfnamefont {Yuval~R}\ \bibnamefont
  {Sanders}}, \bibinfo {author} {\bibfnamefont {Dominic~W}\ \bibnamefont
  {Berry}}, \bibinfo {author} {\bibfnamefont {Pedro C~S}\ \bibnamefont
  {Costa}}, \bibinfo {author} {\bibfnamefont {Louis~W}\ \bibnamefont
  {Tessler}}, \bibinfo {author} {\bibfnamefont {Nathan}\ \bibnamefont {Wiebe}},
  \bibinfo {author} {\bibfnamefont {Craig}\ \bibnamefont {Gidney}}, \bibinfo
  {author} {\bibfnamefont {Hartmut}\ \bibnamefont {Neven}}, \ and\ \bibinfo
  {author} {\bibfnamefont {Ryan}\ \bibnamefont {Babbush}},\ }\bibfield  {title}
  {\enquote {\bibinfo {title} {{Compilation of Fault-Tolerant Quantum
  Heuristics for Combinatorial Optimization}},}\ }\href
  {https://journals.aps.org/prxquantum/abstract/10.1103/PRXQuantum.1.020312}
  {\bibfield  {journal} {\bibinfo  {journal} {PRX Quantum}\ }\textbf {\bibinfo
  {volume} {1}},\ \bibinfo {pages} {020312--020382} (\bibinfo {year}
  {2020})}\BibitemShut {NoStop}%
\bibitem [{\citenamefont {Sanders}\ \emph {et~al.}(2019)\citenamefont
  {Sanders}, \citenamefont {Low}, \citenamefont {Scherer},\ and\ \citenamefont
  {Berry}}]{SLSB19}%
  \BibitemOpen
  \bibfield  {author} {\bibinfo {author} {\bibfnamefont {Yuval~R}\ \bibnamefont
  {Sanders}}, \bibinfo {author} {\bibfnamefont {Guang~Hao}\ \bibnamefont
  {Low}}, \bibinfo {author} {\bibfnamefont {Artur}\ \bibnamefont {Scherer}}, \
  and\ \bibinfo {author} {\bibfnamefont {Dominic~W}\ \bibnamefont {Berry}},\
  }\bibfield  {title} {\enquote {\bibinfo {title} {Black-box quantum state
  preparation without arithmetic},}\ }\href {\doibase
  10.1103/PhysRevLett.122.020502} {\bibfield  {journal} {\bibinfo  {journal}
  {Physical Review Letters}\ }\textbf {\bibinfo {volume} {122}},\ \bibinfo
  {pages} {020502} (\bibinfo {year} {2019})}\BibitemShut {NoStop}%
\bibitem [{\citenamefont {Horn}\ and\ \citenamefont
  {Johnson}(2012)}]{horn2012matrix}%
  \BibitemOpen
  \bibfield  {author} {\bibinfo {author} {\bibfnamefont {Roger~A}\ \bibnamefont
  {Horn}}\ and\ \bibinfo {author} {\bibfnamefont {Charles~R}\ \bibnamefont
  {Johnson}},\ }\href@noop {} {\emph {\bibinfo {title} {Matrix analysis}}}\
  (\bibinfo  {publisher} {Cambridge university press},\ \bibinfo {year}
  {2012})\BibitemShut {NoStop}%
\bibitem [{\citenamefont {Berry}\ \emph
  {et~al.}(2014{\natexlab{a}})\citenamefont {Berry}, \citenamefont {Cleve},\
  and\ \citenamefont {Gharibian}}]{BerryQIC14}%
  \BibitemOpen
  \bibfield  {author} {\bibinfo {author} {\bibfnamefont {Dominic~W}\
  \bibnamefont {Berry}}, \bibinfo {author} {\bibfnamefont {Richard}\
  \bibnamefont {Cleve}}, \ and\ \bibinfo {author} {\bibfnamefont {Sevag}\
  \bibnamefont {Gharibian}},\ }\bibfield  {title} {\enquote {\bibinfo {title}
  {Gate-efficient discrete simulations of continuous-time quantum query
  algorithms},}\ }\href
  {http://www.rintonpress.com/xxqic14/qic-14-12/0001-0030.pdf} {\bibfield
  {journal} {\bibinfo  {journal} {Quantum Information and Computation}\
  }\textbf {\bibinfo {volume} {14}},\ \bibinfo {pages} {0001--0030} (\bibinfo
  {year} {2014}{\natexlab{a}})}\BibitemShut {NoStop}%
\bibitem [{\citenamefont {{Sheng-Tzong Cheng}}\ and\ \citenamefont {{Chun-Yen
  Wang}}(2006)}]{Qsort1}%
  \BibitemOpen
  \bibfield  {author} {\bibinfo {author} {\bibnamefont {{Sheng-Tzong Cheng}}}\
  and\ \bibinfo {author} {\bibnamefont {{Chun-Yen Wang}}},\ }\bibfield  {title}
  {\enquote {\bibinfo {title} {Quantum switching and quantum merge sorting},}\
  }\href {\doibase 10.1109/TCSI.2005.856669} {\bibfield  {journal} {\bibinfo
  {journal} {IEEE Transactions on Circuits and Systems I: Regular Papers}\
  }\textbf {\bibinfo {volume} {53}},\ \bibinfo {pages} {316--325} (\bibinfo
  {year} {2006})}\BibitemShut {NoStop}%
\bibitem [{\citenamefont {Beals}\ \emph {et~al.}(2013)\citenamefont {Beals},
  \citenamefont {Brierley}, \citenamefont {Gray}, \citenamefont {Harrow},
  \citenamefont {Kutin}, \citenamefont {Linden}, \citenamefont {Shepherd},\
  and\ \citenamefont {Stather}}]{Qsort2}%
  \BibitemOpen
  \bibfield  {author} {\bibinfo {author} {\bibfnamefont {Robert}\ \bibnamefont
  {Beals}}, \bibinfo {author} {\bibfnamefont {Stephen}\ \bibnamefont
  {Brierley}}, \bibinfo {author} {\bibfnamefont {Oliver}\ \bibnamefont {Gray}},
  \bibinfo {author} {\bibfnamefont {Aram~W}\ \bibnamefont {Harrow}}, \bibinfo
  {author} {\bibfnamefont {Samuel}\ \bibnamefont {Kutin}}, \bibinfo {author}
  {\bibfnamefont {Noah}\ \bibnamefont {Linden}}, \bibinfo {author}
  {\bibfnamefont {Dan}\ \bibnamefont {Shepherd}}, \ and\ \bibinfo {author}
  {\bibfnamefont {Mark}\ \bibnamefont {Stather}},\ }\bibfield  {title}
  {\enquote {\bibinfo {title} {Efficient distributed quantum computing},}\
  }\href {\doibase 10.1098/rspa.2012.0686} {\bibfield  {journal} {\bibinfo
  {journal} {Proceedings of the Royal Society A: Mathematical, Physical and
  Engineering Sciences}\ }\textbf {\bibinfo {volume} {469}},\ \bibinfo {pages}
  {20120686} (\bibinfo {year} {2013})}\BibitemShut {NoStop}%
\bibitem [{\citenamefont {Codish}\ \emph {et~al.}(2014)\citenamefont {Codish},
  \citenamefont {Cruz-Filipe}, \citenamefont {Frank},\ and\ \citenamefont
  {Schneider-Kamp}}]{SortingNetworks}%
  \BibitemOpen
  \bibfield  {author} {\bibinfo {author} {\bibfnamefont {Michael}\ \bibnamefont
  {Codish}}, \bibinfo {author} {\bibfnamefont {Luís}\ \bibnamefont
  {Cruz-Filipe}}, \bibinfo {author} {\bibfnamefont {Michael}\ \bibnamefont
  {Frank}}, \ and\ \bibinfo {author} {\bibfnamefont {Peter}\ \bibnamefont
  {Schneider-Kamp}},\ }\bibfield  {title} {\enquote {\bibinfo {title}
  {Twenty-five comparators is optimal when sorting nine inputs (and twenty-nine
  for ten)},}\ }in\ \href {\doibase 10.1109/ICTAI.2014.36} {\emph {\bibinfo
  {booktitle} {2014 IEEE 26th International Conference on Tools with Artificial
  Intelligence}}}\ (\bibinfo {year} {2014})\ pp.\ \bibinfo {pages}
  {186--193}\BibitemShut {NoStop}%
\bibitem [{\citenamefont {Rozzi}\ \emph {et~al.}(2006)\citenamefont {Rozzi},
  \citenamefont {Varsano}, \citenamefont {Marini}, \citenamefont {Gross},\ and\
  \citenamefont {Rubio}}]{rozzi2006exact}%
  \BibitemOpen
  \bibfield  {author} {\bibinfo {author} {\bibfnamefont {Carlo~A}\ \bibnamefont
  {Rozzi}}, \bibinfo {author} {\bibfnamefont {Daniele}\ \bibnamefont
  {Varsano}}, \bibinfo {author} {\bibfnamefont {Andrea}\ \bibnamefont
  {Marini}}, \bibinfo {author} {\bibfnamefont {Eberhard K~U}\ \bibnamefont
  {Gross}}, \ and\ \bibinfo {author} {\bibfnamefont {Angel}\ \bibnamefont
  {Rubio}},\ }\bibfield  {title} {\enquote {\bibinfo {title} {{Exact Coulomb
  cutoff technique for supercell calculations}},}\ }\href
  {https://journals.aps.org/prb/pdf/10.1103/PhysRevB.73.205119} {\bibfield
  {journal} {\bibinfo  {journal} {Physical Review B}\ }\textbf {\bibinfo
  {volume} {73}},\ \bibinfo {pages} {205119} (\bibinfo {year}
  {2006})}\BibitemShut {NoStop}%
\bibitem [{\citenamefont {Sundararaman}\ and\ \citenamefont
  {Arias}(2013)}]{sundararaman2013regularization}%
  \BibitemOpen
  \bibfield  {author} {\bibinfo {author} {\bibfnamefont {Ravishankar}\
  \bibnamefont {Sundararaman}}\ and\ \bibinfo {author} {\bibfnamefont {T~A}\
  \bibnamefont {Arias}},\ }\bibfield  {title} {\enquote {\bibinfo {title}
  {{Regularization of the Coulomb singularity in exact exchange by Wigner-Seitz
  truncated interactions: Towards chemical accuracy in nontrivial systems}},}\
  }\href {https://journals.aps.org/prb/pdf/10.1103/PhysRevB.87.165122}
  {\bibfield  {journal} {\bibinfo  {journal} {Physical Review B}\ }\textbf
  {\bibinfo {volume} {87}},\ \bibinfo {pages} {165122} (\bibinfo {year}
  {2013})}\BibitemShut {NoStop}%
\bibitem [{\citenamefont {Ismail-Beigi}(2006)}]{ismail2006truncation}%
  \BibitemOpen
  \bibfield  {author} {\bibinfo {author} {\bibfnamefont {Sohrab}\ \bibnamefont
  {Ismail-Beigi}},\ }\bibfield  {title} {\enquote {\bibinfo {title}
  {{Truncation of periodic image interactions for confined systems}},}\ }\href
  {https://journals.aps.org/prb/pdf/10.1103/PhysRevB.73.233103} {\bibfield
  {journal} {\bibinfo  {journal} {Physical Review B}\ }\textbf {\bibinfo
  {volume} {73}},\ \bibinfo {pages} {233103} (\bibinfo {year}
  {2006})}\BibitemShut {NoStop}%
\bibitem [{\citenamefont {Makov}\ and\ \citenamefont
  {Payne}(1995)}]{makov1995periodic2}%
  \BibitemOpen
  \bibfield  {author} {\bibinfo {author} {\bibfnamefont {G.}~\bibnamefont
  {Makov}}\ and\ \bibinfo {author} {\bibfnamefont {M.~C.}\ \bibnamefont
  {Payne}},\ }\bibfield  {title} {\enquote {\bibinfo {title} {Periodic boundary
  conditions in ab initio calculations},}\ }\href {\doibase
  10.1103/PhysRevB.51.4014} {\bibfield  {journal} {\bibinfo  {journal}
  {Physical Review B}\ }\textbf {\bibinfo {volume} {51}},\ \bibinfo {pages}
  {4014--4022} (\bibinfo {year} {1995})}\BibitemShut {NoStop}%
\bibitem [{\citenamefont {Martyna}\ and\ \citenamefont
  {Tuckerman}(1999)}]{Tuckerman1999}%
  \BibitemOpen
  \bibfield  {author} {\bibinfo {author} {\bibfnamefont {Glenn~J.}\
  \bibnamefont {Martyna}}\ and\ \bibinfo {author} {\bibfnamefont {Mark~E.}\
  \bibnamefont {Tuckerman}},\ }\bibfield  {title} {\enquote {\bibinfo {title}
  {{A reciprocal space based method for treating long range interactions in ab
  initio and force-field-based calculations in clusters}},}\ }\href {\doibase
  10.1063/1.477923} {\bibfield  {journal} {\bibinfo  {journal} {The Journal of
  Chemical Physics}\ }\textbf {\bibinfo {volume} {110}},\ \bibinfo {pages}
  {2810--2821} (\bibinfo {year} {1999})}\BibitemShut {NoStop}%
\bibitem [{\citenamefont {Preskill}(2012)}]{Preskill2012b}%
  \BibitemOpen
  \bibfield  {author} {\bibinfo {author} {\bibfnamefont {John}\ \bibnamefont
  {Preskill}},\ }\bibfield  {title} {\enquote {\bibinfo {title} {Quantum
  computing and the entanglement frontier},}\ }\href
  {http://arxiv.org/abs/1203.5813} {\bibfield  {journal} {\bibinfo  {journal}
  {arXiv:1203.5813}\ } (\bibinfo {year} {2012})}\BibitemShut {NoStop}%
\bibitem [{\citenamefont {Boixo}\ \emph {et~al.}(2018)\citenamefont {Boixo},
  \citenamefont {Isakov}, \citenamefont {Smelyanskiy}, \citenamefont {Babbush},
  \citenamefont {Ding}, \citenamefont {Jiang}, \citenamefont {Bremner},
  \citenamefont {Martinis},\ and\ \citenamefont {Neven}}]{Boixo2016}%
  \BibitemOpen
  \bibfield  {author} {\bibinfo {author} {\bibfnamefont {Sergio}\ \bibnamefont
  {Boixo}}, \bibinfo {author} {\bibfnamefont {Sergei~V}\ \bibnamefont
  {Isakov}}, \bibinfo {author} {\bibfnamefont {Vadim~N}\ \bibnamefont
  {Smelyanskiy}}, \bibinfo {author} {\bibfnamefont {Ryan}\ \bibnamefont
  {Babbush}}, \bibinfo {author} {\bibfnamefont {Nan}\ \bibnamefont {Ding}},
  \bibinfo {author} {\bibfnamefont {Zhang}\ \bibnamefont {Jiang}}, \bibinfo
  {author} {\bibfnamefont {Michael~J}\ \bibnamefont {Bremner}}, \bibinfo
  {author} {\bibfnamefont {John~M}\ \bibnamefont {Martinis}}, \ and\ \bibinfo
  {author} {\bibfnamefont {Hartmut}\ \bibnamefont {Neven}},\ }\bibfield
  {title} {\enquote {\bibinfo {title} {{Characterizing Quantum Supremacy in
  Near-Term Devices}},}\ }\href {\doibase 10.1038/s41567-018-0124-x} {\bibfield
   {journal} {\bibinfo  {journal} {Nature Physics}\ }\textbf {\bibinfo {volume}
  {14}},\ \bibinfo {pages} {595–600} (\bibinfo {year} {2018})}\BibitemShut
  {NoStop}%
\bibitem [{\citenamefont {Kim}\ \emph {et~al.}(2021)\citenamefont {Kim},
  \citenamefont {Lee}, \citenamefont {Liu}, \citenamefont {Pallister},
  \citenamefont {Pol},\ and\ \citenamefont {Roberts}}]{Kim2021b}%
  \BibitemOpen
  \bibfield  {author} {\bibinfo {author} {\bibfnamefont {Isaac~H.}\
  \bibnamefont {Kim}}, \bibinfo {author} {\bibfnamefont {Eunseok}\ \bibnamefont
  {Lee}}, \bibinfo {author} {\bibfnamefont {Ye-Hua}\ \bibnamefont {Liu}},
  \bibinfo {author} {\bibfnamefont {Sam}\ \bibnamefont {Pallister}}, \bibinfo
  {author} {\bibfnamefont {William}\ \bibnamefont {Pol}}, \ and\ \bibinfo
  {author} {\bibfnamefont {Sam}\ \bibnamefont {Roberts}},\ }\bibfield  {title}
  {\enquote {\bibinfo {title} {Fault-tolerant resource estimate for quantum
  chemical simulations: Case study on {L}i-ion battery electrolyte
  molecules},}\ }\href {http://arxiv.org/abs/2104.10653} {\bibfield  {journal}
  {\bibinfo  {journal} {arXiv:2104.10653}\ } (\bibinfo {year}
  {2021})}\BibitemShut {NoStop}%
\bibitem [{\citenamefont {Gidney}\ and\ \citenamefont
  {Fowler}(2019)}]{Gidney2019c}%
  \BibitemOpen
  \bibfield  {author} {\bibinfo {author} {\bibfnamefont {Craig}\ \bibnamefont
  {Gidney}}\ and\ \bibinfo {author} {\bibfnamefont {Austin~G.}\ \bibnamefont
  {Fowler}},\ }\bibfield  {title} {\enquote {\bibinfo {title} {Efficient magic
  state factories with a catalyzed {$|CCZ\rangle$} to {$2|T\rangle$}
  transformation},}\ }\href {\doibase 10.22331/q-2019-04-30-135} {\bibfield
  {journal} {\bibinfo  {journal} {{Quantum}}\ }\textbf {\bibinfo {volume}
  {3}},\ \bibinfo {pages} {135} (\bibinfo {year} {2019})}\BibitemShut {NoStop}%
\bibitem [{\citenamefont {Stein}\ and\ \citenamefont
  {Reiher}(2019)}]{stein2019autocas}%
  \BibitemOpen
  \bibfield  {author} {\bibinfo {author} {\bibfnamefont {Christopher~J}\
  \bibnamefont {Stein}}\ and\ \bibinfo {author} {\bibfnamefont {Markus}\
  \bibnamefont {Reiher}},\ }\bibfield  {title} {\enquote {\bibinfo {title}
  {\textsc{auto}{CAS}: {A} program for fully automated multiconfigurational
  calculations},}\ }\href
  {https://onlinelibrary.wiley.com/doi/abs/10.1002/jcc.25869} {\bibfield
  {journal} {\bibinfo  {journal} {Journal of Computational Chemistry}\ }\textbf
  {\bibinfo {volume} {40}},\ \bibinfo {pages} {2216--2226} (\bibinfo {year}
  {2019})}\BibitemShut {NoStop}%
\bibitem [{\citenamefont {Sayfutyarova}\ \emph {et~al.}(2017)\citenamefont
  {Sayfutyarova}, \citenamefont {Sun}, \citenamefont {Chan},\ and\
  \citenamefont {Knizia}}]{sayfutyarova2017automated}%
  \BibitemOpen
  \bibfield  {author} {\bibinfo {author} {\bibfnamefont {Elvira~R}\
  \bibnamefont {Sayfutyarova}}, \bibinfo {author} {\bibfnamefont {Qiming}\
  \bibnamefont {Sun}}, \bibinfo {author} {\bibfnamefont {Garnet Kin-Lic}\
  \bibnamefont {Chan}}, \ and\ \bibinfo {author} {\bibfnamefont {Gerald}\
  \bibnamefont {Knizia}},\ }\bibfield  {title} {\enquote {\bibinfo {title}
  {Automated construction of molecular active spaces from atomic valence
  orbitals},}\ }\href {https://pubs.acs.org/doi/abs/10.1021/acs.jctc.7b00128}
  {\bibfield  {journal} {\bibinfo  {journal} {Journal of Chemical Theory and
  Computation}\ }\textbf {\bibinfo {volume} {13}},\ \bibinfo {pages}
  {4063--4078} (\bibinfo {year} {2017})}\BibitemShut {NoStop}%
\bibitem [{\citenamefont {Keller}\ \emph {et~al.}(2015)\citenamefont {Keller},
  \citenamefont {Boguslawski}, \citenamefont {Janowski}, \citenamefont
  {Reiher},\ and\ \citenamefont {Pulay}}]{keller2015selection}%
  \BibitemOpen
  \bibfield  {author} {\bibinfo {author} {\bibfnamefont {Sebastian}\
  \bibnamefont {Keller}}, \bibinfo {author} {\bibfnamefont {Katharina}\
  \bibnamefont {Boguslawski}}, \bibinfo {author} {\bibfnamefont {Tomasz}\
  \bibnamefont {Janowski}}, \bibinfo {author} {\bibfnamefont {Markus}\
  \bibnamefont {Reiher}}, \ and\ \bibinfo {author} {\bibfnamefont {Peter}\
  \bibnamefont {Pulay}},\ }\bibfield  {title} {\enquote {\bibinfo {title}
  {Selection of active spaces for multiconfigurational wavefunctions},}\ }\href
  {https://aip.scitation.org/doi/abs/10.1063/1.4922352?crawler=true&mimetype=application%2Fpdf&journalCode=jcp}
  {\bibfield  {journal} {\bibinfo  {journal} {The Journal of Chemical Physics}\
  }\textbf {\bibinfo {volume} {142}},\ \bibinfo {pages} {244104} (\bibinfo
  {year} {2015})}\BibitemShut {NoStop}%
\bibitem [{\citenamefont {Bao}\ and\ \citenamefont
  {Truhlar}(2019)}]{bao2019automatic}%
  \BibitemOpen
  \bibfield  {author} {\bibinfo {author} {\bibfnamefont {Jie~J}\ \bibnamefont
  {Bao}}\ and\ \bibinfo {author} {\bibfnamefont {Donald~G}\ \bibnamefont
  {Truhlar}},\ }\bibfield  {title} {\enquote {\bibinfo {title} {Automatic
  active space selection for calculating electronic excitation energies based
  on high-spin unrestricted {H}artree--{F}ock orbitals},}\ }\href
  {https://doi.org/10.1021/acs.jctc.9b00535} {\bibfield  {journal} {\bibinfo
  {journal} {Journal of Chemical Theory and Computation}\ }\textbf {\bibinfo
  {volume} {15}},\ \bibinfo {pages} {5308--5318} (\bibinfo {year}
  {2019})}\BibitemShut {NoStop}%
\bibitem [{\citenamefont {Hermes}\ \emph {et~al.}(2020)\citenamefont {Hermes},
  \citenamefont {Pandharkar},\ and\ \citenamefont
  {Gagliardi}}]{hermes2020variational}%
  \BibitemOpen
  \bibfield  {author} {\bibinfo {author} {\bibfnamefont {Matthew~R}\
  \bibnamefont {Hermes}}, \bibinfo {author} {\bibfnamefont {Riddhish}\
  \bibnamefont {Pandharkar}}, \ and\ \bibinfo {author} {\bibfnamefont {Laura}\
  \bibnamefont {Gagliardi}},\ }\bibfield  {title} {\enquote {\bibinfo {title}
  {Variational localized active space self-consistent field method},}\ }\href
  {https://pubs.acs.org/doi/10.1021/acs.jctc.0c00222} {\bibfield  {journal}
  {\bibinfo  {journal} {Journal of Chemical Theory and Computation}\ }\textbf
  {\bibinfo {volume} {16}},\ \bibinfo {pages} {4923--4937} (\bibinfo {year}
  {2020})}\BibitemShut {NoStop}%
\bibitem [{\citenamefont {Tosoni}\ \emph {et~al.}(2007)\citenamefont {Tosoni},
  \citenamefont {Tuma}, \citenamefont {Sauer}, \citenamefont {Civalleri},\ and\
  \citenamefont {Ugliengo}}]{tosoni2007comparison}%
  \BibitemOpen
  \bibfield  {author} {\bibinfo {author} {\bibfnamefont {Sergio}\ \bibnamefont
  {Tosoni}}, \bibinfo {author} {\bibfnamefont {Christian}\ \bibnamefont
  {Tuma}}, \bibinfo {author} {\bibfnamefont {Joachim}\ \bibnamefont {Sauer}},
  \bibinfo {author} {\bibfnamefont {Bartolomeo}\ \bibnamefont {Civalleri}}, \
  and\ \bibinfo {author} {\bibfnamefont {Piero}\ \bibnamefont {Ugliengo}},\
  }\bibfield  {title} {\enquote {\bibinfo {title} {{A comparison between plane
  wave and Gaussian-type orbital basis sets for hydrogen bonded systems: Formic
  acid as a test case}},}\ }\href
  {http://aip.scitation.org/doi/abs/10.1063/1.2790019} {\bibfield  {journal}
  {\bibinfo  {journal} {The Journal of Chemical Physics}\ }\textbf {\bibinfo
  {volume} {127}},\ \bibinfo {pages} {154102} (\bibinfo {year}
  {2007})}\BibitemShut {NoStop}%
\bibitem [{\citenamefont {Booth}\ \emph {et~al.}(2016)\citenamefont {Booth},
  \citenamefont {Tsatsoulis}, \citenamefont {Chan},\ and\ \citenamefont
  {Gr{\"{u}}neis}}]{booth2016plane}%
  \BibitemOpen
  \bibfield  {author} {\bibinfo {author} {\bibfnamefont {George~H}\
  \bibnamefont {Booth}}, \bibinfo {author} {\bibfnamefont {Theodoros}\
  \bibnamefont {Tsatsoulis}}, \bibinfo {author} {\bibfnamefont {Garnet
  Kin-Lic}\ \bibnamefont {Chan}}, \ and\ \bibinfo {author} {\bibfnamefont
  {Andreas}\ \bibnamefont {Gr{\"{u}}neis}},\ }\bibfield  {title} {\enquote
  {\bibinfo {title} {{From plane waves to local Gaussians for the simulation of
  correlated periodic systems}},}\ }\href
  {http://aip.scitation.org/doi/abs/10.1063/1.4961301} {\bibfield  {journal}
  {\bibinfo  {journal} {The Journal of Chemical Physics}\ }\textbf {\bibinfo
  {volume} {145}},\ \bibinfo {pages} {84111} (\bibinfo {year}
  {2016})}\BibitemShut {NoStop}%
\bibitem [{\citenamefont {Kleinman}\ and\ \citenamefont
  {Bylander}(1982)}]{Kleinam1982}%
  \BibitemOpen
  \bibfield  {author} {\bibinfo {author} {\bibfnamefont {Leonard}\ \bibnamefont
  {Kleinman}}\ and\ \bibinfo {author} {\bibfnamefont {D~M}\ \bibnamefont
  {Bylander}},\ }\bibfield  {title} {\enquote {\bibinfo {title} {Efficacious
  form for model pseudopotentials},}\ }\href {\doibase
  10.1103/PhysRevLett.48.1425} {\bibfield  {journal} {\bibinfo  {journal}
  {Physical Review Letters}\ }\textbf {\bibinfo {volume} {48}},\ \bibinfo
  {pages} {1425--1428} (\bibinfo {year} {1982})}\BibitemShut {NoStop}%
\bibitem [{\citenamefont {Dal~Corso}(2014)}]{dal2014pseudopotentials}%
  \BibitemOpen
  \bibfield  {author} {\bibinfo {author} {\bibfnamefont {Andrea}\ \bibnamefont
  {Dal~Corso}},\ }\bibfield  {title} {\enquote {\bibinfo {title}
  {Pseudopotentials periodic table: From {H} to {P}u},}\ }\href
  {https://www.sciencedirect.com/science/article/abs/pii/S0927025614005187}
  {\bibfield  {journal} {\bibinfo  {journal} {Computational Materials Science}\
  }\textbf {\bibinfo {volume} {95}},\ \bibinfo {pages} {337--350} (\bibinfo
  {year} {2014})}\BibitemShut {NoStop}%
\bibitem [{\citenamefont {Elfving}\ \emph {et~al.}(2020)\citenamefont
  {Elfving}, \citenamefont {Broer}, \citenamefont {Webber}, \citenamefont
  {Gavartin}, \citenamefont {Halls}, \citenamefont {Lorton},\ and\
  \citenamefont {Bochevarov}}]{Elfving2020b}%
  \BibitemOpen
  \bibfield  {author} {\bibinfo {author} {\bibfnamefont {V~E}\ \bibnamefont
  {Elfving}}, \bibinfo {author} {\bibfnamefont {B~W}\ \bibnamefont {Broer}},
  \bibinfo {author} {\bibfnamefont {M}~\bibnamefont {Webber}}, \bibinfo
  {author} {\bibfnamefont {J}~\bibnamefont {Gavartin}}, \bibinfo {author}
  {\bibfnamefont {M~D}\ \bibnamefont {Halls}}, \bibinfo {author} {\bibfnamefont
  {K~P}\ \bibnamefont {Lorton}}, \ and\ \bibinfo {author} {\bibfnamefont
  {A}~\bibnamefont {Bochevarov}},\ }\bibfield  {title} {\enquote {\bibinfo
  {title} {How will quantum computers provide an industrially relevant
  computational advantage in quantum chemistry?}}\ }\href
  {http://arxiv.org/abs/2009.12472} {\bibfield  {journal} {\bibinfo  {journal}
  {arXiv:2009.12472}\ } (\bibinfo {year} {2020})}\BibitemShut {NoStop}%
\bibitem [{\citenamefont {F{\"{u}}sti-Molnar}\ and\ \citenamefont
  {Pulay}(2002)}]{fusti2002accurateb}%
  \BibitemOpen
  \bibfield  {author} {\bibinfo {author} {\bibfnamefont {L\'{a}szl\'{o}}\
  \bibnamefont {F{\"{u}}sti-Molnar}}\ and\ \bibinfo {author} {\bibfnamefont
  {Peter}\ \bibnamefont {Pulay}},\ }\bibfield  {title} {\enquote {\bibinfo
  {title} {{Accurate molecular integrals and energies using combined plane wave
  and {G}aussian basis sets in molecular electronic structure theory}},}\
  }\href {http://aip.scitation.org/doi/abs/10.1063/1.1467901} {\bibfield
  {journal} {\bibinfo  {journal} {The Journal of Chemical Physics}\ }\textbf
  {\bibinfo {volume} {116}},\ \bibinfo {pages} {7795--7805} (\bibinfo {year}
  {2002})}\BibitemShut {NoStop}%
\bibitem [{\citenamefont {Brassard}\ \emph {et~al.}(2002)\citenamefont
  {Brassard}, \citenamefont {H{\o}yer}, \citenamefont {Mosca},\ and\
  \citenamefont {Tapp}}]{Brassard2002}%
  \BibitemOpen
  \bibfield  {author} {\bibinfo {author} {\bibfnamefont {Gilles}\ \bibnamefont
  {Brassard}}, \bibinfo {author} {\bibfnamefont {Peter}\ \bibnamefont
  {H{\o}yer}}, \bibinfo {author} {\bibfnamefont {Michele}\ \bibnamefont
  {Mosca}}, \ and\ \bibinfo {author} {\bibfnamefont {Alain}\ \bibnamefont
  {Tapp}},\ }\bibfield  {title} {\enquote {\bibinfo {title} {{Quantum amplitude
  amplification and estimation}},}\ }in\ \href {\doibase
  10.1090/conm/305/05215} {\emph {\bibinfo {booktitle} {Quantum Computation and
  Information}}},\ \bibinfo {editor} {edited by\ \bibinfo {editor}
  {\bibnamefont {{Vitaly I Voloshin}}}, \bibinfo {editor} {\bibnamefont
  {{Samuel J. Lomonaco}}}, \ and\ \bibinfo {editor} {\bibnamefont {{Howard E.
  Brandt}}}}\ (\bibinfo  {publisher} {American Mathematical Society},\ \bibinfo
  {address} {Washington D.C.},\ \bibinfo {year} {2002})\ Chap.~\bibinfo
  {chapter} {3}, pp.\ \bibinfo {pages} {53--74}\BibitemShut {NoStop}%
\bibitem [{\citenamefont {Lidar}\ and\ \citenamefont {Wang}(1999)}]{Lidar1999}%
  \BibitemOpen
  \bibfield  {author} {\bibinfo {author} {\bibfnamefont {Daniel~A}\
  \bibnamefont {Lidar}}\ and\ \bibinfo {author} {\bibfnamefont {Haobin}\
  \bibnamefont {Wang}},\ }\bibfield  {title} {\enquote {\bibinfo {title}
  {{Calculating the thermal rate constant with exponential speedup on a quantum
  computer}},}\ }\href {\doibase 10.1103/PhysRevE.59.2429} {\bibfield
  {journal} {\bibinfo  {journal} {Physical Review E}\ }\textbf {\bibinfo
  {volume} {59}},\ \bibinfo {pages} {2429--2438} (\bibinfo {year}
  {1999})}\BibitemShut {NoStop}%
\bibitem [{\citenamefont {Babbush}\ \emph {et~al.}(2014)\citenamefont
  {Babbush}, \citenamefont {Love},\ and\ \citenamefont
  {Aspuru-Guzik}}]{BabbushSR14}%
  \BibitemOpen
  \bibfield  {author} {\bibinfo {author} {\bibfnamefont {Ryan}\ \bibnamefont
  {Babbush}}, \bibinfo {author} {\bibfnamefont {Peter~J}\ \bibnamefont {Love}},
  \ and\ \bibinfo {author} {\bibfnamefont {Alan}\ \bibnamefont
  {Aspuru-Guzik}},\ }\bibfield  {title} {\enquote {\bibinfo {title} {{Adiabatic
  Quantum Simulation of Quantum Chemistry}},}\ }\href {\doibase
  10.1038/srep06603} {\bibfield  {journal} {\bibinfo  {journal} {Scientific
  Reports}\ }\textbf {\bibinfo {volume} {4}},\ \bibinfo {pages} {6603}
  (\bibinfo {year} {2014})}\BibitemShut {NoStop}%
\bibitem [{\citenamefont {Somma}\ \emph {et~al.}(2002)\citenamefont {Somma},
  \citenamefont {Ortiz}, \citenamefont {Gubernatis}, \citenamefont {Knill},\
  and\ \citenamefont {Laflamme}}]{Somma2002}%
  \BibitemOpen
  \bibfield  {author} {\bibinfo {author} {\bibfnamefont {R~D}\ \bibnamefont
  {Somma}}, \bibinfo {author} {\bibfnamefont {G}~\bibnamefont {Ortiz}},
  \bibinfo {author} {\bibfnamefont {J~E}\ \bibnamefont {Gubernatis}}, \bibinfo
  {author} {\bibfnamefont {E}~\bibnamefont {Knill}}, \ and\ \bibinfo {author}
  {\bibfnamefont {R}~\bibnamefont {Laflamme}},\ }\bibfield  {title} {\enquote
  {\bibinfo {title} {{Simulating physical phenomena by quantum networks}},}\
  }\href {\doibase 10.1103/PhysRevA.65.042323} {\bibfield  {journal} {\bibinfo
  {journal} {Physical Review A}\ }\textbf {\bibinfo {volume} {65}},\ \bibinfo
  {pages} {17} (\bibinfo {year} {2002})}\BibitemShut {NoStop}%
\bibitem [{\citenamefont {Verstraete}\ and\ \citenamefont
  {Cirac}(2005)}]{Verstraete2005}%
  \BibitemOpen
  \bibfield  {author} {\bibinfo {author} {\bibfnamefont {F}~\bibnamefont
  {Verstraete}}\ and\ \bibinfo {author} {\bibfnamefont {J~I}\ \bibnamefont
  {Cirac}},\ }\bibfield  {title} {\enquote {\bibinfo {title} {{Mapping local
  Hamiltonians of fermions to local Hamiltonians of spins}},}\ }\href {\doibase
  10.1088/1742-5468/2005/09/P09012} {\bibfield  {journal} {\bibinfo  {journal}
  {Journal of Statistical Mechanics: Theory and Experiment}\ }\textbf {\bibinfo
  {volume} {2005}},\ \bibinfo {pages} {P09012--P09012} (\bibinfo {year}
  {2005})}\BibitemShut {NoStop}%
\bibitem [{\citenamefont {Jiang}\ \emph {et~al.}(2019)\citenamefont {Jiang},
  \citenamefont {McClean}, \citenamefont {Babbush},\ and\ \citenamefont
  {Neven}}]{Jiang2019}%
  \BibitemOpen
  \bibfield  {author} {\bibinfo {author} {\bibfnamefont {Zhang}\ \bibnamefont
  {Jiang}}, \bibinfo {author} {\bibfnamefont {Jarrod}\ \bibnamefont {McClean}},
  \bibinfo {author} {\bibfnamefont {Ryan}\ \bibnamefont {Babbush}}, \ and\
  \bibinfo {author} {\bibfnamefont {Hartmut}\ \bibnamefont {Neven}},\
  }\bibfield  {title} {\enquote {\bibinfo {title} {Majorana loop stabilizer
  codes for error mitigation in fermionic quantum simulations},}\ }\href
  {\doibase 10.1103/PhysRevApplied.12.064041} {\bibfield  {journal} {\bibinfo
  {journal} {Physical Review Applied}\ }\textbf {\bibinfo {volume} {12}},\
  \bibinfo {pages} {064041} (\bibinfo {year} {2019})}\BibitemShut {NoStop}%
\bibitem [{\citenamefont {Gidney}(2018)}]{Gidney18}%
  \BibitemOpen
  \bibfield  {author} {\bibinfo {author} {\bibfnamefont {Craig}\ \bibnamefont
  {Gidney}},\ }\bibfield  {title} {\enquote {\bibinfo {title} {Halving the cost
  of quantum addition},}\ }\href {\doibase 10.22331/q-2018-06-18-74} {\bibfield
   {journal} {\bibinfo  {journal} {{Quantum}}\ }\textbf {\bibinfo {volume}
  {2}},\ \bibinfo {pages} {74} (\bibinfo {year} {2018})}\BibitemShut {NoStop}%
\bibitem [{\citenamefont {Childs}\ \emph {et~al.}(2018)\citenamefont {Childs},
  \citenamefont {Maslov}, \citenamefont {Nam}, \citenamefont {Ross},\ and\
  \citenamefont {Su}}]{Childs2017}%
  \BibitemOpen
  \bibfield  {author} {\bibinfo {author} {\bibfnamefont {Andrew~M}\
  \bibnamefont {Childs}}, \bibinfo {author} {\bibfnamefont {Dmitri}\
  \bibnamefont {Maslov}}, \bibinfo {author} {\bibfnamefont {Yunseong}\
  \bibnamefont {Nam}}, \bibinfo {author} {\bibfnamefont {Neil~J}\ \bibnamefont
  {Ross}}, \ and\ \bibinfo {author} {\bibfnamefont {Yuan}\ \bibnamefont {Su}},\
  }\bibfield  {title} {\enquote {\bibinfo {title} {{Toward the first quantum
  simulation with quantum speedup}},}\ }\href {\doibase
  10.1073/pnas.1801723115} {\bibfield  {journal} {\bibinfo  {journal}
  {Proceedings of the National Academy of Sciences}\ }\textbf {\bibinfo
  {volume} {115}},\ \bibinfo {pages} {9456--9461} (\bibinfo {year}
  {2018})}\BibitemShut {NoStop}%
\bibitem [{\citenamefont {Kivlichan}\ \emph {et~al.}(2017)\citenamefont
  {Kivlichan}, \citenamefont {Wiebe}, \citenamefont {Babbush},\ and\
  \citenamefont {Aspuru-Guzik}}]{Kivlichan2016}%
  \BibitemOpen
  \bibfield  {author} {\bibinfo {author} {\bibfnamefont {Ian~D}\ \bibnamefont
  {Kivlichan}}, \bibinfo {author} {\bibfnamefont {Nathan}\ \bibnamefont
  {Wiebe}}, \bibinfo {author} {\bibfnamefont {Ryan}\ \bibnamefont {Babbush}}, \
  and\ \bibinfo {author} {\bibfnamefont {Alan}\ \bibnamefont {Aspuru-Guzik}},\
  }\bibfield  {title} {\enquote {\bibinfo {title} {{Bounding the costs of
  quantum simulation of many-body physics in real space}},}\ }\href
  {http://iopscience.iop.org/article/10.1088/1751-8121/aa77b8} {\bibfield
  {journal} {\bibinfo  {journal} {Journal of Physics A: Mathematical and
  Theoretical}\ }\textbf {\bibinfo {volume} {50}},\ \bibinfo {pages} {305301}
  (\bibinfo {year} {2017})}\BibitemShut {NoStop}%
\bibitem [{\citenamefont {Light}\ and\ \citenamefont
  {Carrington~Jr}(2000)}]{dvrreview}%
  \BibitemOpen
  \bibfield  {author} {\bibinfo {author} {\bibfnamefont {John~C}\ \bibnamefont
  {Light}}\ and\ \bibinfo {author} {\bibfnamefont {Tucker}\ \bibnamefont
  {Carrington~Jr}},\ }\bibfield  {title} {\enquote {\bibinfo {title}
  {{Discrete-variable representations and their utilization}},}\ }\href
  {http://onlinelibrary.wiley.com/doi/10.1002/9780470141731.ch4/summary}
  {\bibfield  {journal} {\bibinfo  {journal} {Advances in Chemical Physics}\
  }\textbf {\bibinfo {volume} {114}},\ \bibinfo {pages} {263--310} (\bibinfo
  {year} {2000})}\BibitemShut {NoStop}%
\bibitem [{\citenamefont {Cleve}\ and\ \citenamefont
  {Watrous}(2000)}]{cleve2000fast}%
  \BibitemOpen
  \bibfield  {author} {\bibinfo {author} {\bibfnamefont {Richard}\ \bibnamefont
  {Cleve}}\ and\ \bibinfo {author} {\bibfnamefont {John}\ \bibnamefont
  {Watrous}},\ }\bibfield  {title} {\enquote {\bibinfo {title} {Fast parallel
  circuits for the quantum {F}ourier transform},}\ }in\ \href {\doibase
  10.1109/SFCS.2000.892140} {\emph {\bibinfo {booktitle} {Proceedings 41st
  Annual Symposium on Foundations of Computer Science}}}\ (\bibinfo
  {organization} {IEEE},\ \bibinfo {year} {2000})\ pp.\ \bibinfo {pages}
  {526--536}\BibitemShut {NoStop}%
\bibitem [{\citenamefont {Berry}\ \emph
  {et~al.}(2014{\natexlab{b}})\citenamefont {Berry}, \citenamefont {Childs},
  \citenamefont {Cleve}, \citenamefont {Kothari},\ and\ \citenamefont
  {Somma}}]{Berry2013}%
  \BibitemOpen
  \bibfield  {author} {\bibinfo {author} {\bibfnamefont {Dominic~W}\
  \bibnamefont {Berry}}, \bibinfo {author} {\bibfnamefont {Andrew~M}\
  \bibnamefont {Childs}}, \bibinfo {author} {\bibfnamefont {Richard}\
  \bibnamefont {Cleve}}, \bibinfo {author} {\bibfnamefont {Robin}\ \bibnamefont
  {Kothari}}, \ and\ \bibinfo {author} {\bibfnamefont {Rolando~D}\ \bibnamefont
  {Somma}},\ }\bibfield  {title} {\enquote {\bibinfo {title} {{Exponential
  improvement in precision for simulating sparse Hamiltonians}},}\ }in\ \href
  {https://doi.org/10.1145/2591796.2591854} {\emph {\bibinfo {booktitle} {STOC
  '14 Proceedings of the 46th Annual ACM Symposium on Theory of Computing}}}\
  (\bibinfo {year} {2014})\ pp.\ \bibinfo {pages} {283--292}\BibitemShut
  {NoStop}%
\bibitem [{\citenamefont {Soeken}\ \emph {et~al.}(2017)\citenamefont {Soeken},
  \citenamefont {Roetteler}, \citenamefont {Wiebe},\ and\ \citenamefont
  {De~Micheli}}]{soeken2017hierarchical}%
  \BibitemOpen
  \bibfield  {author} {\bibinfo {author} {\bibfnamefont {Mathias}\ \bibnamefont
  {Soeken}}, \bibinfo {author} {\bibfnamefont {Martin}\ \bibnamefont
  {Roetteler}}, \bibinfo {author} {\bibfnamefont {Nathan}\ \bibnamefont
  {Wiebe}}, \ and\ \bibinfo {author} {\bibfnamefont {Giovanni}\ \bibnamefont
  {De~Micheli}},\ }\bibfield  {title} {\enquote {\bibinfo {title} {Hierarchical
  reversible logic synthesis using {LUT}s},}\ }in\ \href
  {https://doi.org/10.1145/3061639.3062261} {\emph {\bibinfo {booktitle} {2017
  54th ACM/EDAC/IEEE Design Automation Conference (DAC)}}}\ (\bibinfo
  {organization} {IEEE},\ \bibinfo {year} {2017})\ pp.\ \bibinfo {pages}
  {1--6}\BibitemShut {NoStop}%
\bibitem [{\citenamefont {H{\"{a}}ner}\ \emph {et~al.}(2018)\citenamefont
  {H{\"{a}}ner}, \citenamefont {Roetteler},\ and\ \citenamefont
  {Svore}}]{Haner2018}%
  \BibitemOpen
  \bibfield  {author} {\bibinfo {author} {\bibfnamefont {Thomas}\ \bibnamefont
  {H{\"{a}}ner}}, \bibinfo {author} {\bibfnamefont {Martin}\ \bibnamefont
  {Roetteler}}, \ and\ \bibinfo {author} {\bibfnamefont {Krysta~M.}\
  \bibnamefont {Svore}},\ }\bibfield  {title} {\enquote {\bibinfo {title}
  {{Optimizing Quantum Circuits for Arithmetic}},}\ }\href
  {https://arxiv.org/abs/1805.12445} {\bibfield  {journal} {\bibinfo  {journal}
  {arXiv:1805.12445}\ } (\bibinfo {year} {2018})}\BibitemShut {NoStop}%
\bibitem [{\citenamefont {Poirier}\ and\ \citenamefont
  {Jerke}(2021)}]{Poirier2021b}%
  \BibitemOpen
  \bibfield  {author} {\bibinfo {author} {\bibfnamefont {Bill}\ \bibnamefont
  {Poirier}}\ and\ \bibinfo {author} {\bibfnamefont {Jonathan}\ \bibnamefont
  {Jerke}},\ }\bibfield  {title} {\enquote {\bibinfo {title} {{Full-Dimensional
  Schr\"odinger Wavefunction Calculations using Tensors and Quantum Computers:
  the Cartesian component-separated approach}},}\ }\href
  {https://arxiv.org/abs/2105.03787} {\bibfield  {journal} {\bibinfo  {journal}
  {arXiv:2105.03787}\ } (\bibinfo {year} {2021})}\BibitemShut {NoStop}%
\end{thebibliography}%

	\appendix

\section{Variable and function names}
\label{app:names}

\begin{itemize}
    \item $N$ -- The number of reciprocal lattice vectors per electron.
    \item $\eta$ -- The number of electrons.
    \item $m$ -- The electron mass, and also used as an index of summation in other equations (particularly for discretized time). The difference is clear by the context.
    \item $T$ -- The kinetic energy operator.
    \item $U$ -- The potential energy between electrons and nuclei.
    \item $V$ -- Potential energy between electrons.
    \item $L$ -- The number of nuclei.
    \item $\zeta_\ell$ -- The nuclear charges.
    \item $R_\ell$ -- Nuclear positions.
    \item $\ell$ -- Usually indexes over the number of the nucleus.
    \item $\Omega$ -- The computational cell volume.
    \item $G$ -- The set of $N$ reciprocal lattice vectors defined in \eq{G}.
    \item $G_0$ -- The set of reciprocal lattice vectors defined below \eq{first_quant_ham}.
    \item $p,q$ -- Plane wave numbers, which are elements of $G$.
    \item $\nu$ -- Differences of plane wave numbers, which are elements of $G_0$.
    \item $k_p,k_\nu$ -- Reciprocal lattice vectors defined in \eq{G}.
    \item $i,j$ -- Used to index registers containing plane wave numbers.
    \item $\mi$ -- The complex number $\mi=\sqrt{-1}$ is used in upright font, and the difference from the index $i$ should be obvious from context.
    \item $\lambda$ -- The total sum of weights in the linear combination of unitaries.
    \item $\lambda_T$ -- The sum of the weights for $T$, see \eq{lambdas}.
    \item $\lambda_U$ -- The sum of weights for $U$, see \eq{lambdas}.
    \item $\lambda_V$ -- The sum of weights for $V$, see \eq{lambdas}.
    \item $\lambda_\zeta$ -- The sum of nuclear charges, $\sum_\ell \zeta_\ell$.
    \item $\lambda_\nu$ -- The sum over $\nu\in G_0$ of $1/\|\nu\|$.
    \item $n_p$ -- The number of bits used to store each component of the plane wave number.
    \item $n_\eta$ -- The number of bits needed to store the electron number, $\lceil\log\eta\rceil$.
    \item $n_{\mathcal{M}}$ -- The number of bits used in the inequality testing to prepare a state with amplitudes $1/\|\nu\|$.
    \item $n_t$ -- The number of bits used to store the time in each time register.
    \item $n_{\eta\zeta}$ -- The number of bits needed to store a number up to $\eta+2\lambda_\zeta$.
    \item $n_k$ -- The number of qubits used for the register that we use to prepare the superposition over $k$, given by $\lceil \log \sig(0)\rceil$.
    \item $n_T$ -- The number of bits used in the rotation on the qubit selecting between $T$ and $U+V$.
    \item $n_R$ -- The number of bits used to store each component of the nuclear position.
    \item $K$ -- The cutoff on the Dyson series order used in the interaction picture.
    \item $k$ -- The particular order used in the Dyson series. (Some other uses of $k$ are obvious from the context.)
    \item $\mathcal{M}$ -- The total number of alternative values used in the inequality testing for preparing amplitudes $1/\|\nu\|$, related to $n_{\mathcal{M}}$ as $\mathcal{M}=2^{n_{\mathcal{M}}}$.
    \item $M$ -- The number of time steps in the discretization of the integrals for the time in the Dyson series, related to $n_t$ as $M=2^{n_t}$.
    \item $\epsilon$ -- The allowable total RMS error in the estimation of the energy.
    \item $\epsilon_{\rm pha}$ -- The allowable RMS error in the energy due to phase estimation.
    \item $\epsilon_{\mathcal{M}}$ -- The allowable error in the estimated energy due to imprecise preparation of the state with amplitudes $1/\|\nu\|$.
    \item $\epsilon_R$ -- The allowable error in energies due to the finite precision of the nuclear positions.
    \item $\epsilon_t$ -- Allowable error in energies due to the finite precision in the times.
    \item $\epsilon_T$ -- Allowable error in energies due to the finite precision in the rotation of the qubit selecting between $T$ and $U+V$.
    \item $w$ -- Indexes over the $x$, $y$ and $z$ components of the plane wave numbers.
    \item $r,s$ -- Index over bits of a single component of the plane wave number.
    \item $\amam$ -- Is given as 1 without amplitude amplification or 3 for amplitude amplification (to account for the increase in cost).
    \item $\erase{n}$ -- The cost of erasing the QROM with $n$ values, defined in \eq{erse}.
    \item $\srt{n}$ -- The number of comparators to sort $n$ items, listed in \tab{sorttab}.
    \item $\sig(n)$ -- A number used in the preparation of the superposition over $k$, defined in \eq{sigdef}.
    \item $\eqprep{n,b_r}$ -- The probability of success for preparing an equal superposition state, given in \eq{eqprep}.
    \item $r_s$ -- The Wigner-Seitz radius, which satisfies $4\pi r_s^3/3=\Omega/\eta$.
    \item $\Delta$ -- The spacing of the inverse reciprocal lattice given by $\Omega^{1/3} / N^{1/3}$.
    \item $\reps$ -- The number of steps used in the phase estimation, proportional to $\lambda/\epsilon_{\rm pha}$.
    \item $\tau$ -- The length of time for a single step used in the phase estimation.
    \item $A$ -- The component of the Hamiltonian with larger norm considered in the interaction picture. We have $A=T$.
    \item $B$ -- The component of the Hamiltonian with smaller norm considered in the interaction picture.  We take $B=U+V$.
    \item $E_j$ -- The energy eigenvalues.
    \item $E_0$ -- An energy shift used to ensure that we can linearize the sine and arcsine in block encoding the Dyson series.
    \item $m_1,\ldots,m_k$ -- Numbers used to index the discretized times.
    \item $\psuc$ -- The probability of success of the state with the state with amplitudes $1/\|\nu\|$; see \eq{psuc}.
    \item $\psuc^{\rm amp}$ -- The probability of success when using amplitude amplification; see \eq{psucamp}.
    \item $\mu$ -- The numbering of the cube in the series of nested cubes used in preparing the state with amplitudes $1/\|\nu\|$.
    \item $b_r$ -- The number of bits used in rotating an ancilla used for preparing an equal superposition state.
    \item $b_{\rm grad}$ -- The number of bits used in the kinetic energy phasing according to \eq{bgrad}.
    \item $b_T$ -- The number of bits used for the implicit multiplication in the kinetic energy phasing.
    \item $\alp$ -- A factor used in the phasing by $T$ in the interaction picture.
\end{itemize}

\section{Background on Galerkin representations}
\label{app:background}

\subsection{Galerkin basis representations in first quantization}

The Galerkin discretization is derived from the weak formulation of the Schr\"{o}dinger equation in Hilbert space, given by finding $\ket{\phi}$ (spanned by the basis vectors $\{\ket{\phi_p}\}$) such that $\bra{\phi_p}H\ket{\phi} = E \braket{\phi_p}{\phi}$ for all $p$. This is contrasted with the strong formulation of the differential equation that requires the discretization hold at all points in space, as opposed to assessing error on the restricted subspace spanned by $\{\ket{\phi_p}\}$. Energy errors are bounded from below within a Galerkin representation, meaning that for a finite basis size $N$, the energy of an eigenstate is always over-estimated due to the discretization error. This property is beneficial when using variational methods to solve for eigenstates.

We now define both the Born-Oppenheimer and non-Born-Oppenheimer molecular Hamiltonians in a first quantized Galerkin discretization. In the Born-Oppenheimer case we will store our wavefunction by having a computational basis that encodes configurations of the electrons in $N$ basis functions such that a configuration is specified as $\ket{\phi_1 \phi_2 \cdots \phi_{\eta}}$ where each $\phi$ encodes the index of an occupied basis function. In the non-Born-Oppenheimer case we will represent both nuclear and electronic degrees of freedom explicitly such that a configuration is specified by $\ket{\phi_1 \phi_2 \cdots \phi_{\eta + L}}$ where the first $\eta$ registers encode the index of the basis function occupied by the $\eta$ electrons and the last $L$ registers encode the index of the basis function occupied by the $L$ nuclei. Then, we can write the BO and non-BO Hamiltonians as
\begin{align}
\label{eq:galerkin_app}
\begin{aligned}
    H_{\rm BO} & = T + U + V + \frac{1}{2}\sum_{\ell\neq\kappa=1}^{L} \frac{\zeta_\ell \zeta_\kappa}{\left\|R_\ell - R_\kappa\right\|} \\
	T &=  \sum_{i=1}^{\eta}\sum_{p,q=1}^{N} T_{pq}^{(1)} \ket{p}\!\bra{q}_i \\
	U &= -\sum_{i=1}^\eta \sum_{p,q=1}^N U_{pq} \ket{p}\!\bra{q}_i
	\\
	V & = \frac{1}{2}\sum_{i\neq j=1}^{\eta}\sum_{p,q,r,s=1}^N V_{pqrs}^{(1,1)} \ket{p}\!\bra{s}_i \ket{q}\!\bra{r}_j
	\end{aligned}\qquad\qquad
\begin{aligned}
    H_{\rm non-BO} &= T + T_{\rm nuc} + U_{\rm non-BO} + V + V_{\rm nuc}\\
	T_{\rm nuc} &=\sum_{\ell=\eta+1}^{L + \eta}\sum_{p,q=1}^{N} T_{pq}^{(m_\ell)}\ket{p}\!\bra{q}_\ell\\
	U_{\rm non-BO} &=-\sum_{i=1}^{\eta}\sum_{\ell=\eta+1}^{L + \eta}\sum_{p,q,r,s=1}^N V_{pqrs}^{(1,\zeta_\ell)} \ket{p}\!\bra{s}_i \ket{q}\!\bra{r}_\ell
	\\
	V_{\rm nuc} &= \frac{1}{2}\sum_{\ell \neq \kappa=\eta+1}^{L+\eta}\sum_{p,q,r,s=1}^N V_{pqrs}^{(\zeta_\ell, \zeta_\kappa)} \ket{p}\!\bra{s}_\ell \ket{q}\!\bra{r}_\kappa
	\end{aligned}
	\end{align}
	where the integrals are as defined in \eq{integrals}, $R_\ell$ are the positions of the $L$ nuclei (under the Born-Oppenheimer approximation), $m_\ell$ are the masses of those nuclei and $\zeta_\ell$ are their atomic numbers. As in the main text, the notation $\ket{p}\!\bra{q}_i$ is a shorthand for $I_1\otimes I_2\cdots\otimes\ket{p}\!\bra{q}_i \cdots\otimes I_{L +\eta}$. Finally, we note that it is not necessary to use the same basis functions for all distinguishable particles. For example, we could use a different basis (involving a fewer or greater number of functions) for the nuclei versus for the electrons. But here for simplicity, we have defined the Hamiltonians assuming the same basis for all particles.
	
	\subsection{Galerkin basis representations in second quantization}
	
	The methods and analysis of this paper focus on first quantization. However, the majority of past work in both classical quantum chemisty and quantum algorithms for chemistry focus on using second quantization. Thus, we will quickly review second quantization and its relation to first quantization for completeness and to place our work in context. In second quantization one works with computational basis states that represent symmetrized (or antisymmetrized) configurations. For example, when encoding fermions in second quantization the computational basis state $\ket{\phi_1 \phi_2 \cdots \phi_\eta}$ would typically encode a real space wavefunction corresponding to a Slater determinant of those basis functions; i.e.,
\begin{align}
&  \braket{r_1 r_2 \cdots r_{\eta}}{\phi_1 \phi_2 \cdots \phi_\eta} = \sqrt{\frac{1}{\eta!}}
\begin{vmatrix}
\phi_{p_1}\left(r_1\right) & \phi_{p_2}\left( r_1\right) & \cdots & \phi_{p_{\eta}} \left( r_1\right) \\
\phi_{p_1}\left(r_2\right) & \phi_{p_2}\left( r_2\right) & \cdots & \phi_{p_{\eta}} \left( r_2\right) \\
\vdots & \vdots & \ddots & \vdots\\
\phi_{p_1}\left(r_{\eta}\right) & \phi_{p_2}\left(r_{\eta}\right) & \cdots & \phi_{p_{\eta}} \left(r_{\eta}\right) \end{vmatrix}\, .
\end{align}
One can see that this function is antisymmetric by construction since swapping any two rows of a determinant incurs a negative sign. Of course, the actual computational basis states as defined on qubits do not have the same symmetries as the wavefunctions they encode. As a result, in second quantization one constructs operators so that they act on the wavefunction in a way that effectively enforces the proper symmetries.

The essential distinction between first and second quantization is that in first quantization the symmetries are encoded explicitly in the wavefunction whereas in second quantization the symmetries are encoded in the operators that act on the wavefunction. However, typically there is also a difference in how one indexes the occupied basis functions in first versus second quantization. Whereas in first quantization one tends to keep track of which basis function each particle occupies in a configuration, because the second quantized configurations encode a symmetrized representation, it is sufficient to merely keep track of which basis functions are occupied without associating a particular particle. For example, a typical way to store fermionic second quantized wavefunctions is to have a single qubit to represent each spin assigned basis function which is either $\ket{0}$ if the spin-orbital is unoccupied or $\ket{1}$ if the spin-orbital is occupied\footnote{While this style of encoding second quantized operators is by far the most common, it is not necessary. One can also perform second quantized simulations in a fashion that requires ${\cal O}((\eta + L)\log N)$ space rather than ${\cal O}((\eta + L)N)$ space. For example, the work of \cite{BabbushSparse2} performs a quantum simulation of the electronic structure Hamiltonian in a fixed particle-number manifold of the second quantized Hamiltonian in a fashion that requires ${\cal O}(\eta \log N)$ qubits rather than ${\cal O}(N)$ qubits.}. However, because one should only symmetrize identical particles this needs to be done separately for each distinguishable type of particle (in this context, we should consider ``spin-up'' particles to be distinct from ``spin-down'' particles, etc.). As a result, in second quantization if we have $J$ unique species of distinguishable nuclei then the number of qubits required will be proportional to $J$. With this notation, the Galerkin formulation of the second quantized molecular operators defining the Hamiltonians in \eq{galerkin_app} take the explicit forms
\begin{align}
\begin{aligned}
T &= \sum_{\alpha \in \{\uparrow, \downarrow\}} \sum_{p,q=1}^N  T_{pq}^{(1)} a^\dagger_{p,\alpha} a_{q,\alpha}\\
U &= -\sum_{\alpha \in \{\uparrow, \downarrow\}}\sum_{p,q=1}^N U_{pq} a^\dagger_{p,\alpha} a_{q,\alpha} \\
V & = \frac{1}{2}\sum_{\alpha,\beta \in \{\uparrow, \downarrow\}}\sum_{p,q,r,s=1}^N V_{pqrs}^{(1, 1)} a^\dagger_{p,\alpha} a^\dagger_{q,\beta} a_{r,\beta} a_{s,\alpha}
\end{aligned}
\begin{aligned}
T_{\rm nuc} &= \sum_{\ell=1}^J\sum_{p,q=1}^N  T_{pq}^{(m_\ell)} a^\dagger_{p + \ell N} a_{q + \ell N}\\
U_{\rm non-BO} & = \sum_{\alpha\,\beta \in \{\uparrow, \downarrow\}}\sum_{\ell=1}^{J} \sum_{p,q,r,s=1}^N V_{pqrs}^{(1, \zeta_\ell)} a^\dagger_{p,\alpha} a^\dagger_{q+\ell N,\beta} a_{r + \ell N,\beta} a_{s,\alpha} \\
V_{\rm nuc} & = \frac{1}{2}\sum_{\ell,\kappa=1}^{J} \sum_{p,q,r,s=1}^N V_{pqrs}^{(\zeta_\ell,\zeta_\kappa)} a^\dagger_{p+\ell N} a^\dagger_{q+\kappa N} a_{r + \kappa N} a_{s + \ell N} 
\end{aligned}
\end{align}
where the indices $\ell$ and $\kappa$ run over the $J$ different species of distinguishable nuclei (rather than index the $L$ different specific nuclei, as before).The nuclear orbitals do not have explicit spin assignments because we are assuming that the spin designation is summed over as one of the $J$ nuclei types.

The operators $a^\dagger_p$ and $a_q^\dagger$ are either fermionic or bosonic creation and annihilation operators, depending on whether the type of particle they correspond to are fermions or bosons. For fermions we require the anticommutation relation that $\{a_p, a^\dagger_q \} = \delta_{pq}$ whereas for bosons we require the commutation relation that $[a_p, a^\dagger_q] = \delta_{pq}$. For fermions there are many well known methods that one can use to map these operators to qubits \cite{Somma2002,Seeley2012,Verstraete2005,Jiang2019}. Such techniques yield a qubit Hamiltonian with ${\cal O}(J N)$ qubits where $J$ here is the number of distinguishable fermionic species. Mapping the bosonic algebra to qubits is considerably more onerous since each mode (basis function) could be occupied by a number of bosons that is between zero and the total number of bosons in the simulation. For example, if there are $b$ particles of a particular bosonic species then the number of qubits required to represent just that species is ${\cal O}(N \log b)$ in second quantization. This cost is multiplied by the number of different bosonic species.

\section{Qubit costings}
\label{app:qubcost}
\subsection{Qubit costings for qubitization approach}
    The complete step needed for the qubitization, to give the step that one would perform phase estimation on, needs a reflection on the control ancilla.
    The cost of this reflection corresponds to the numbers of qubits used in the state preparation.
    To give this cost we will first list the entire qubit cost, then describe the specific qubits that need to be reflected upon.
    \begin{enumerate}
        \item There are $\eta$ registers storing the momenta, each of which has three components with $n_p$ qubits, so the number of qubits needed is $3\eta n_p$.
        \item The control register for the phase estimation needs $\lceil\log\reps\rceil$ qubits, and there are $\lceil\log\reps\rceil-1$ qubits for the temporary registers, where
        \begin{equation}
            \reps=\left\lceil \frac{\pi\lambda}{2\epsilon_{\rm pha}} \right\rceil .
        \end{equation}
        \item The phase gradient state that is used for the phase rotations.
        There are $n_R+1$ bits used in the phasing, and $n_T$ bits used in the rotation of the qubit selecting $T$, so the number of qubits needed for the phase gradient state is the maximum of these.
        \item One qubit for the $\ket{T}$ state that is used catalytically for controlled Hadamards.
        \item The qubit that is rotated to select between the $T$ and $U+V$ components of the Hamiltonian.
        \item The $n_{\eta\zeta}+3$ qubits for the equal superposition for selecting between $U$ and $V$, including the $n_{\eta\zeta}$ qubits for the superposition itself, the extra rotated qubit, the qubit flagging success of the preparation, and the qubit flagging the result of the inequality test to select between $U$ and $V$.
        \item Three qubits for selecting whether each of $T$, $U$ and $V$ is performed.
        \item The $2n_\eta+5$ qubits from the preparation of the superpositions over $i$ and $j$.
        There are $\eta$ qubits for each of these registers in unary, then 2 qubits that are rotated for each of these preparations, 2 qubits that flag success of the two preparations, and 1 qubit that flags whether $i=j$.
        \item The preparing the superposition state over $\nu$ has the following subcosts.
        \begin{enumerate}
            \item Storing $\nu$ requires $3(n_p+1)$ qubits.
            \item $\mu$ needs $n_p$ qubits since it is given in unary.
            Note that $\mu$ takes values from 2 to $n_\mu=n_p+1$, so there are $n_p$ possible values to be encoded in unary.
            \item $n_{\mathcal{M}}$ qubits for the equal superposition state.
            \item $3n_p+2$ qubits for testing if we have negative zero, including the flag qubit.
            \item $2n_p+1$ qubits used in signalling whether $\nu$ is outside the box, including the flag qubit.
            There are 2 qubits needed as ancillae for each of the $n_p$ quadruple-controlled Toffolis, and the target qubit is the same each time.
            \item The qubits resulting from computing the sum of squares of components of $\nu$ and multiplying by $\mathcal{M}$ are kept and used for uncomputing without Toffolis, and there are $3n_p^2+n_p+1+4n_{\mathcal{M}}(n_p+1)$ (the same as the number of Toffolis).
            \item The qubit resulting from the inequality test.
            \item Two qubits, one flagging success of the inequality test, no negative zero and $\nu$ not outside the box, and the other an ancilla qubit used to produce the triply controlled Toffoli.
        \end{enumerate}
        \item For the preparation of the equal superposition state over 3 basis states, we need 4 qubits.
        There are 2 to store the state itself, one qubit that is rotated, and one flagging success of the preparation.
        \item The states $r$ and $s$ are prepared in unary, and need $n_p$ qubits each, for a total of $2n_p$.
        \item There are a number of temporary ancillae used in the arithmetic adding and subtracting $\nu$ into momenta.
        The cost here is given by items (a) and (c), giving a total of $5n_p+1$.
        \begin{enumerate}
        \item In implementing the $\SEL$ operations, we need to control a swap of a momentum register into an ancilla, which takes $3n_p$ qubits for the output.
        The $n_\eta-1$ temporary ancillae for the unary iteration on the $i$ or $j$ register can be ignored because they are fewer than the other temporary ancillae used later.
        \item We use $n_p+3$ temporary qubits to implement the block encoding of $T$, where we copy component $w$ of the momentum into an ancilla, copy out two bits of this component of the momentum, then perform a controlled phase with those two qubits as control as well as the qubit flagging that $T$ is to be performed.
        \item For the controlled addition or subtraction by $\nu$ in the $\SEL$ operations for $U$ and $V$, we use $n_p$ bits to copy a component of $\nu$ into an ancilla, then there are another $n_p+1$ temporary qubits used in the addition, for a total of $2n_p+1$ temporary qubits in this part.
        \item There are also temporary qubits used in converting the momentum back and forth between signed and two's complement, but these are fewer than those used in the previous step.
        \end{enumerate}
        \item There are 2 overflow qubits obtained every time we add or subtract a component of $\nu$ into a momentum.
        All these qubits must be kept, giving a total of $6$.
        \item For the preparation of the state with the amplitudes $\zeta_\ell$, we do not need to output $\ell$ directly, because we can just output $R_\ell$, which is all that is needed.
        The size of this output is $3n_R$ for the three components of the nuclear position.
        \item There are also temporary qubits used in the arithmetic to implement $e^{-\mi k_\nu \cdot R_\ell}$.
        The arithmetic requires a maximum of $2(n_R-2)$ qubits.
        Note that the $R_\ell$ can be output by the QROM, the phase factor applied, and the $R_\ell$ erased, after (or before) the arithmetic in item 12 is performed, so we need only take the maximum of the $5n_R-4$ qubits used in items 14 and 15 and the $5n_p+1$ temporary qubits used in 12.
        \item A single qubit is used to control between adding and subtracting $\nu$ in order to make $\SEL$ self-inverse.
    \end{enumerate}
    
   Out of these qubits, the ones that need to be reflected upon are as follows.
    \begin{itemize}
        \item Item 5, the qubit that is rotated to select between $T$ and $U+V$, needs to be reflected upon.
        \item From item 6, the $n_{\eta\zeta}$ qubits that the equal superposition is prepared on, as well as the rotated ancilla, all need to be reflected upon.
        That gives a total of $n_{\eta\zeta}+1$.
        \item From item 8, there are $2n_\eta$ qubits for $i$ and $j$.
        There are also 2 qubits that are rotated, for a total of $2n_\eta+2$ qubits (the flag qubits are rezeroed so do not need to be reflected upon).
        \item From item 9, we need to reflect upon the qubits listed in (a), (b) and (c), for a total of $4n_p+n_{\mathcal{M}}+3$.
        \item From item 10, there are two qubits for $w$, as well as 1 rotated qubit, for a total of 3.
        \item From item 11, there are $n_p$ for each of $r$ and $s$, for a total of $2n_p$.
        \item From item 13, there are $6$ qubits as overflow from the arithmetic that need to be reflected upon.
        \item The state with amplitudes $\sqrt{\zeta_\ell}$ is only prepared in an implicit way, via a QROM from the state generated in part 1.
        This QROM may be reversed erasing the extra ancillae, so there are no additional qubits from the state preparation that need to be reflected upon.
    \end{itemize}
    This gives a total of
    \begin{equation}
        n_{\eta\zeta} + 2n_\eta + 6n_p + n_{\mathcal{M}} + 16
    \end{equation}
    qubits to reflect on, with a corresponding Toffoli cost.

\subsection{Qubit costings for interaction picture approach}
Now we consider the complete qubit costing for the interaction picture algorithm.
In the interaction picture algorithm, we need to block-encode the potential energy $U+V$ many times ($K$) in succession.
In a naive implementation, one would keep all the ancilla qubits used in the block encoding of each $U+V$, and reflect upon them all at the end.
However, after each block encoding of $U+V$, one can use Toffolis to check that all the ancilla qubits are zero, putting the result in a flag qubit.
Then for the next block-encoding of $U+V$, one can reuse these ancilla qubits, because controlled on the flag qubit being zero these ancilla qubits have been rezeroed.
After block encoding the entire Dyson series with the product of $K$ block encodings of $U+V$ one can reflect upon these $K-1$ flag qubits as well.
The number of flag qubits is $K-1$ rather than $K$, because the checking is done in between each of the $K$ block encodings of $U+V$.
It should also be noted that when using ancilla for the purpose of arithmetic, these ancilla qubits can \emph{not} be used.
This is because in the arithmetic Toffoli gates are inverted using measurements, and that requires that the ancillae used are guaranteed to be rezeroed.

    For the computation of the time difference, multiplication of the time by the kinetic energy and phasing, at each step when a Toffoli is performed an extra ancilla qubit is introduced, so the number of ancillae is the same as the number of Toffolis for these three parts (the first three items in the list following \eq{bgrad} plus $n_t-1$ for computing the absolute values of the times),
    \begin{equation}\label{eq:phasinganc}
       2( n_t-1) + 2n_t(n_\eta+2n_p)-n_t + b_{\rm grad}-2 = 2n_t(n_\eta+2n_p) +n_t + b_{\rm grad}- 4.
    \end{equation}
    For updating the kinetic energy register, there are two main steps.
    \begin{itemize}
        \item Computing one component of $q\cdot\nu$ or $p\cdot\nu$ has cost $2n_p^2-3n_p$, and takes that many ancillae.
        The addition into the kinetic energy register uses $n_\eta+2n_p-1$ qubits. That gives a total
        \begin{equation}\label{eq:qnuanc}
            n_p(2n_p-1)+n_\eta-1.
        \end{equation}
        \item For controlled addition of $\|\nu\|$ we use $2n_p$ ancilla qubits and $n_\eta+2n_p$ qubits for the addition itself, which takes less ancilla qubits than the previous step provided $n_p\ge 3$.
    \end{itemize}
    The dominant ancilla cost in updating the kinetic energy is that in \eq{qnuanc}.

To quantify the total number of ancillae, we will consider the qubit costs specific to the interaction picture algorithm, then consider other qubit costs that are analogous to those for the qubitization approach.
The specific qubit costs are as follows.
\begin{enumerate}[label=I\arabic*.]
    \item In preparing the equal superposition state for the preparation for $k$, we need $n_k$ qubits to store the equal superposition, another qubit that is rotated, and a success flag qubit, for a total of $n_k+2$.
    \item In completing the preparation of the superposition over $k$, we are keeping temporary ancillae to reduce the cost of inverting the inequality tests.
    As discussed in \app{kstate}, the number of qubits including the unary representation of $k$ output is $n_k+1+\sum_{k=2}^{K-1} \lceil \log(\sig(k))\rceil$.
    \item The number of ancillae needed to store the $K$ times is $K n_t$.
    \item In the reversible sort, a qubit is generated for every comparator, so there are $\srt{K}$ qubits.
    \item We use an ancilla to store the total kinetic energy.
    This total kinetic energy has a maximum value of $3\eta (2^{n_p-1}-1)^2$, so the number of qubits needed is $n_\eta+2n_p$.
    \item In phasing by $T$, $2n_t(n_\eta+2n_p) + n_t +b_{\rm grad}- 4$ qubits are used according to \eq{phasinganc}, where $b_{\rm grad}$ is as in \eq{bgrad}.
    \item For updating the total kinetic energy the maximum number of ancillae used is $n_p(2n_p-1)+n_\eta-1$ according to \eq{qnuanc}.
    \item An extra phase gradient state with $b_{\rm grad}$ bits is used.
    \item There are $K-1$ flag qubits used to indicate success of each of the first $K+1$ block encodings of $U+V$.
    \item One qubit for controlling between forward and reverse time evolution for obtaining the sine.
\end{enumerate}
The following qubit costs are analogous to the qubitization, so we use the same numbering.
An exception is that we include I7 above as 13 (e) below, because it is needed at that step in the procedure.
    \begin{enumerate}
        \item There are $\eta$ registers storing the momenta, each of which has three components with $n_p$ qubits, so the number of qubits needed is $3\eta n_p$.
        \item The control register for the phase estimation needs $\lceil\log\reps\rceil$ qubits, and there are $\lceil\log\reps\rceil-1$ qubits for the temporary registers.
        \item The phase gradient state that is used for the phase rotations (except those for kinetic energy phasing).
        The maximum precision of the rotations is that in the phasing for $U$, and needs $n_R+1$ bits, so this is the number of qubits needed for the phase gradient state.
        \item One qubit for the $\ket{T}$ state.
        \item No qubit needs to be rotated to select between the $T$ and $U+V$ components of the Hamiltonian, since we are not block encoding $T$.
        \item For each of the $K$ block encodings of $U+V$, there are $n_{\eta\zeta}+3$ qubits for the equal superposition, including the $n_{\eta\zeta}$ qubits for the superposition itself, the extra rotated qubit, the qubit flagging success of the preparation, and the qubit flagging the result of the inequality test to select between $U$ and $V$.
        A feature here is that $n_{\eta\zeta}+1$ qubits need to be kept and checked to be equal to zero, whereas the two flag qubits are temporary and rezeroed.
        This means that the total qubit cost is $n_{\eta\zeta}+1$ plus 2 temporary ancillae.
        \item Because we are not block encoding $T$, we do not need an extra qubit for the result of the inequality here.
        \item For the superpositions over $i$ and $j$, there are $2n_\eta$ binary qubits.
        In addition, for each there are two qubits rotatated in the state preparation, giving a total of $2(n_\eta+1)$ ancillae to be checked.
        There are also 3 temporary ancillae flagging success of the two state preparations and $i\ne j$, for a total of $3$ temporary ancillae.
        \item For preparing the superposition state over $\nu$, the first three subcosts we list below are qubits that are imperfectly erased, so need to be checked to be equal to zero.
        For the rest of the subcosts they are temporary.
        \begin{enumerate}
            \item Storing $\nu$ requires $3(n_p+1)$ qubits.
            \item $\mu$ needs $n_p$ qubits.
            \item $n_{\mathcal{M}}$ qubits for the equal superposition state.
            \item $3n_p+2$ qubits for testing if we have negative zero, including a flag qubit.
            \item $2n_p+1$ qubits used in signalling whether $\nu$ is outside the box, including a flag qubit.
            There are 2 qubits needed as ancillae for each of the $n_p$ quadruple-controlled Toffolis, and the target qubit is the same each time.
            \item The qubits resulting from computing the sum of squares of components of $\nu$ and multiplying by $\mathcal{M}$ are kept and used for uncomputing without Toffolis, and there are $3n_p^2+n_p+1+4n_{\mathcal{M}}(n_p+1)$ (the same as the number of Toffolis).
            \item The qubit resulting from the inequality test.
            \item Two qubits, one flagging success of the inequality test, no negative zero and $\nu$ not outside the box, and the other an ancilla qubit used to produce the triply controlled Toffoli to set the first qubit.
        \end{enumerate}
        \item For the preparation of the state with the amplitudes $\zeta_\ell$, we do not need to output $\ell$ directly, because we can just output $R_\ell$, which is all that is needed.
        The size of this output is $3n_R$.
        \item We are no longer preparing the equal superposition state over 3 basis states for $w$.
        \item We are no longer preparing superposition states for $r$ and $s$.
        \item There are a number of temporary ancillae used in the arithmetic adding and subtracting $\nu$ into momenta.
        The cost here is given by items (a) and (e), giving a total of $2n_p^2+5n_p+n_\eta$.
        \begin{enumerate}
            \item In implementing the $\SEL$ operations, we need to control a swap of a momentum register into an ancilla, which takes $3n_p$ qubits.
            There are another $n_\eta-1$ temporary ancillae used in the unary iteration, but these may be ignored because they are not used at the same time as other temporary ancillae below.
            \item We are no longer block encoding $T$.
            \item For the controlled addition or subtraction by $\nu$ in the $\SEL$ operations for $U$ and $V$, we use $n_p$ bits to copy a component of $\nu$ into an ancilla, then there are another $n_p+1$ temporary qubits used in the addition, for a total of $2n_p+1$ temporary qubits in this part.
            \item There are also temporary qubits used in converting the momentum back and forth between signed and two's complement, but these are fewer than those used in the previous step.
            \item It is at this point that we also need to update the kinetic energy register, which takes $2n_p(n_p+1)+n_\eta$ qubits.
            This cost is needed at the same time as that in (a), but not those in (c) or (d).
        \end{enumerate}
        \item There are 2 overflow qubits obtained every time we add or subtract a component of $\nu$ into a momentum, for a total of 6.
        \item There are $2(n_R-2)$ temporary qubits used in the arithmetic to implement $e^{-\mi k_\nu \cdot R_\ell}$.
    \end{enumerate}
    
    Now we need to consider the number of qubits that need to be checked to be zero to implement the block encoding of $U+V$.
    We have costs from items 6, 8, 9(a) to (c), and 14.
    Adding these gives a total
    \begin{equation}\label{eq:checkqubits}
        n_{\eta\zeta}+1+2n_\eta+2+3n_p+3+n_p+n_{\mathcal{M}}+6 = n_{\eta\zeta}+2n_\eta+4n_p+n_{\mathcal{M}}+12.
    \end{equation}
    The cost of checking whether these qubits are all zero is one less Toffoli, and it is done $K-1$ times, for a total of
    \begin{equation}
        (K-1)( n_{\eta\zeta}+2n_\eta+4n_p+n_{\mathcal{M}}+11)
    \end{equation}
    Toffolis.
    
    Next, we consider how many qubits we need to reflect upon for the complete qubitization.
    We then have the qubits listed above, as well as $n_k+2$ from I1, $Kn_t$ from I3, the $K-1$ flag qubits from I9, and one qubit from I10.
    This gives a total of
    \begin{equation}
        n_{\eta\zeta}+2n_\eta+4n_p+n_{\mathcal{M}}+12 + n_k+2 + Kn_t+K-1 +1
    \end{equation}
    qubits to reflect upon, with a corresponding cost in Toffolis for the reflection.
    Again, we are including 2 Toffolis for the control on the qubits for phase estimation, as well as the iteration of that register.
    That gives a total number of Toffolis
    \begin{align}
        & (K-1)( n_{\eta\zeta}+2n_\eta+4n_p+n_{\mathcal{M}}+11)\nn 
        &+ n_{\eta\zeta}+2n_\eta+4n_p+n_{\mathcal{M}}+12 + n_k+2 + Kn_t+K-1 +1 \nn
         & =K ( n_{\eta\zeta}+2n_\eta+4n_p+n_{\mathcal{M}}+n_t + 12) + n_k+3.
    \end{align}
    
    The analysis of the total number of ancillae needed is complicated by the fact that there are temporary ancillae used in various parts of the procedure.
    The temporary qubits used in parts 13 and 15 are not needed at the same time, so we can take the maximum of those ancilla counts.
    But, the temporary ancillae used in part 9 are needed at the same time (because it is for preparing $\nu$), so we need to add that temporary ancilla cost to the maximum of those in 13 and 15.
    From part 8, 3 of the qubits are temporary, and also need to be added to the total temporary ancilla costs, as are 2 qubits from part 7.
    But, all those temporary ancilla costs are \emph{not} needed during the phasing by $T$.
    We also need ancilla qubits in order to check whether the number of qubits given in \eq{checkqubits} are all zero, but these are not needed at the same time as the other temporary ancillae.
    Therefore we take the maximum of these three temporary ancilla counts.
    
\section{Selection between \texorpdfstring{$T$}{T}, \texorpdfstring{$U$}{U}, and \texorpdfstring{$V$}{V}}
\label{app:selTUV}
Here we give detail for how the selection between $T$, $U$, and $V$ is performed, and show how to calculate the precision required for the qubit for selecting between $T$ and $U+V$.
We are using a qubit to select between $T$ and $U+V$, but are also applying $\SEL$ for $T$ in some cases where state preparations for $U+V$ fail.
To show this explicitly, we have a rotation on an ancilla selecting between $T$ and $U+V$, to give
\begin{equation}
\cos\theta \ket{0} + \sin\theta\ket{1},
\end{equation}
and a preparation on the register selecting between $V$ ($\ket{0}$) and $U$ ($\ket{1}$) of
\begin{equation}
\frac 1{\sqrt{\eta+2\lambda_\zeta}} \left( \sqrt{\eta} \ket{0} + \sqrt{2\lambda_\zeta} \ket{1} \right),
\end{equation}
where we have abstracted the state preparation as being just on a single qubit, rather than a superposition over $\eta+2\lambda_\zeta$ basis states as was proposed above.
The $\ket{0}$ and $\ket{1}$ here would be equivalent to a superposition of subsets of those basis states, but can be described in the way presented here for clarity.
The inequality test for $i\ne j$ has the effect of preparing the state
\begin{equation}
\frac 1{\sqrt{\eta}} \left( \sqrt{\eta-1}\ket{0} +  \ket{1} \right).
\end{equation}
Again we are using notation with a single qubit as a way of describing sets of basis states.
Here $\ket{0}$ would be equivalent to the subset of $i$ and $j$ where $i\ne j$.
The preparation of $1/\|\nu\|$ with some probability of failure can be written as
\begin{equation}
\ket{0} \sqrt{\psuc /\lambda_\nu} \sum_{\nu\in G_0} \frac 1{\|\nu\|} \ket{\nu} + \ket{1}\sqrt{1-\psuc}\ket{\nu^\perp},
\end{equation}
where $\lambda_\nu$ is needed for normalization for the state, $\psuc$ is the probability of success of the state preparation, and $\ket{\nu^\perp}$ is a state that is obtained in the case of failure.
In the case of amplitude amplification, we can simply replace $\psuc$ with $\psuc^{\rm amp}$.
We are adopting the convention that 0 corresponds to success and 1 flags failure.

Now there are two alternatives we can choose.
\begin{enumerate}
\item We can apply $T$ if the ancilla for selecting $T$ is 0 \emph{OR} we have failure of the state preparation or the inequality test for $i\ne j$ with $V$.
\item We can apply $T$ if the ancilla for selecting $T$ is 0 \emph{AND} we have failure of the state preparation or the inequality test for $i\ne j$ with $V$.
\end{enumerate}
There are four bits that need to be checked for either of these alternatives.
We will order them in the order they were listed above; the qubits for selecting between $T$ and $U+V$, the qubit for selecting between $U$ and $V$, the qubit flagging $i\ne j$, and the qubit flagging success of preparation of $1/\|\nu\|$. 
So, for the first case (OR), the $\SEL$ for $T$ is applied in the following cases.
Below we use ``$x$'' to indicate arbitrary values.
\begin{enumerate}
\item $\ket{0xxx}$ -- That is, we have $\ket{0}$ on the first qubit, and any values for the other qubits.
\item $\ket{1xx1}$ -- We have $\ket{1}$ on the first qubit, but the preparation of $1/\|\nu\|$ failed.
\item $\ket{1010}$ -- We have $\ket{1}$ on the first qubit, the preparation of $1/\|\nu\|$ succeeded, but the second qubit selected $V$ and the third qubit flagged $i=j$.
\end{enumerate}
Because we perform $\SEL_T$, then invert these state preparations, the total weight for $T$ is the sum of the squares of the amplitudes for these alternatives.
The total weight for $T$ is then
\begin{equation}
\cos^2\theta + \sin^2\theta \left( 1-\psuc + \frac \psuc{\eta+2\lambda_\zeta} \right)
\end{equation}
The case where $\SEL_V$ is applied is just $\ket{1000}$, where the first qubit flags $U+V$, the second qubit flags $V$, the third qubit flags $i\ne j$, and the fourth flags success of the preparation.
As a result, the weight on $V$ is
\begin{equation}
\psuc \sin^2\theta \frac {\eta-1}{\eta+2\lambda_\zeta}.
\end{equation}
The cases where $\SEL_U$ are applied are $\ket{11x0}$, where the first qubit flags $U+V$, the second qubit flags $U$, and the fourth flags success of the preparation.
As a result, the weight for $U$ is
\begin{equation}
\psuc \sin^2\theta \frac {2\lambda_\zeta}{\eta+2\lambda_\zeta}.
\end{equation}
These relative weights for $T$, $U$, and $V$ need to correspond to the relative weights between the $\lambda$-values.
First, note that the ratio of the weights between $U$ and $V$ is $2\lambda_\zeta/(\eta-1)$, which is exactly what we need for $\lambda_U/\lambda_V$.
This was the result of the inequality test $i\ne j$ adjusting the weights between $U$ and $V$.
We also need for the relative weight on $T$
\begin{equation}
\cos^2\theta + \sin^2\theta \left( 1-\psuc + \frac \psuc{\eta+2\lambda_\zeta} \right) = \frac{\lambda_T}{\lambda_T+\lambda_U+\lambda_V}.
\end{equation}
The value of $\theta$ needs to be chosen to make this expression hold.
If $\lambda_T$ is large, then it is always possible to choose $\theta$, because the left-hand side (LHS) can be arbitrarily close to 1.
On the other hand, the minimum value of the LHS is the expression in brackets.
If that is \emph{greater} than the RHS, then we need to perform the state preparation with the second alternative, where we use the AND.
Before we proceed to that alternative, note that because we always implement $\SEL$ for $T$, $U$, or $V$, the net $\lambda$-value is $\lambda_T+\lambda_U+\lambda_V$.

In the case of AND, we will implement $\SEL_T$ in the following cases.
\begin{enumerate}
\item $\ket{0xx1}$ -- We have $\ket{0}$ on the first qubit, and the preparation of $1/\|\nu\|$ failed.
\item $\ket{0010}$ -- We have $\ket{0}$ on the first qubit, and the preparation of $1/\|\nu\|$ succeeded, but the second qubit selected $V$ and the third qubit flagged $i=j$.
\end{enumerate}
As a result, the net weight on $T$ is
\begin{equation}
\cos^2\theta \left(1-\psuc + \frac \psuc{\eta+2\lambda_\zeta}\right) = \cos^2\theta - \psuc\cos^2\theta \left(1- \frac 1{\eta+2\lambda_\zeta}\right).
\end{equation}
The cases where $\SEL_U$ and $\SEL_V$ are applied are changed in that we no longer check the first qubit.
As a result, the weights for $V$ and $U$ are changed to
\begin{equation}
\psuc \frac {\eta-1}{\eta+2\lambda_\zeta}, \qquad \psuc \frac {2\lambda_\zeta}{\eta+2\lambda_\zeta},
\end{equation}
respectively.
We also need to apply the identity in the remaining cases.
These are as follows.
\begin{enumerate}
\item $\ket{1xx1}$ -- We have $\ket{1}$ on the first qubit, but the preparation of $1/\|\nu\|$ failed.
\item $\ket{1010}$ -- We have $\ket{1}$ on the first qubit, the preparation of $1/\|\nu\|$ succeeded, but the second qubit selected $V$ and the fourth qubit flagged $i=j$.
\end{enumerate}
Adding these weights together gives
\begin{equation}
\sin^2\theta \left( 1-\psuc + \frac \psuc{\eta+2\lambda_\zeta} \right) .
\end{equation}
It is easily verified that the weights add to 1.
However, the total $\lambda$ is larger than $\lambda_T+\lambda_U+\lambda_V$ now, because there is also $\lambda_I$ for performing the identity.
Adding the weights for $U$ and $V$ gives
\begin{equation}
\psuc \left( 1-\frac 1{\eta+2\lambda_\zeta} \right),
\end{equation}
so the total weight including $T$, $U$ and $V$ is
\begin{equation}
\cos^2\theta + \psuc\sin^2\theta \left( 1-\frac 1{\eta+2\lambda_\zeta} \right)
\end{equation}
We need the relative weight for $T$
\begin{equation}
\frac{\cos^2\theta \left(1-\psuc + \frac \psuc{\eta+2\lambda_\zeta}\right)}{\cos^2\theta + \psuc\sin^2\theta \left( 1-\frac 1{\eta+2\lambda_\zeta} \right)} = \frac{\lambda_T}{\lambda_T+\lambda_U+\lambda_V}.
\end{equation}
In this case the \emph{largest} value on the LHS is given by the expression in brackets in the numerator, which is the same as the smallest value before.
Therefore, one should choose one of the two alternatives for applying $T$ depending on the relative values of $1-\psuc+\psuc/(\eta+2\lambda_\zeta)$ and $\lambda_T/(\lambda_T+\lambda_U+\lambda_V)$, and in either case it is possible to choose $\theta$ to get the appropriate weightings.

To determine the total value of $\lambda$ in the second case (AND), we can use that the relative weights for $U$ and $V$ should satisfy
\begin{equation}
\psuc \left( 1-\frac 1{\eta+2\lambda_\zeta} \right) = \frac{\lambda_U+\lambda_V}{\lambda},
\end{equation}
so
\begin{equation}
\lambda = \frac{\lambda_U+\lambda_V}{\psuc \left( 1-\frac 1{\eta+2\lambda_\zeta} \right) }.
\end{equation}
Now we have $\lambda_V/\lambda_U=(\eta-1)/(2\lambda_\zeta)$, so
\begin{align}
\frac{\lambda_U+\lambda_V}{\lambda_U+\lambda_V/(1-1/\eta)} &= \frac{1+\lambda_V/\lambda_U}{1+\eta\lambda_V/\lambda_U/(\eta-1)} \nn
&=\frac{1+(\eta-1)/(2\lambda_\zeta)}{1+\eta/(2\lambda_\zeta)}\nn
&= 1-\frac 1{\eta+2\lambda_\zeta}.
\end{align}
Hence, the $\lambda$-value can be given as
\begin{equation}
\lambda = \frac 1{\psuc} \left(\lambda_U+\frac{\lambda_V}{1-1/\eta}\right) .
\end{equation}
This $\lambda$-value necessarily must be at least $\lambda_T+\lambda_U+\lambda_V$, because $\lambda_I$ is non-negative.

In the case where
\begin{equation}
1-\psuc+\frac {\psuc}{\eta+2\lambda_\zeta} < \frac{\lambda_T}{\lambda_T+\lambda_U+\lambda_V},
\end{equation}
then we need to apply the OR.
In that case we find that
\begin{equation}
1-\frac{\lambda_T}{\lambda_T+\lambda_U+\lambda_V} < \psuc \left( 1-\frac 1{\eta+2\lambda_\zeta} \right),
\end{equation}
so
\begin{align}
\frac{\lambda_U+\lambda_V}{\lambda_T+\lambda_U+\lambda_V} &< \psuc \left( 1-\frac 1{\eta+2\lambda_\zeta} \right)\nn
\frac{\lambda_U+\lambda_V}{\psuc \left( 1-\frac 1{\eta+2\lambda_\zeta} \right) } &< \lambda_T+\lambda_U+\lambda_V.
\end{align}
Therefore, in this case $\lambda_T+\lambda_U+\lambda_V$ is \emph{larger} than the expression we use for $\lambda$ in the AND case.
To account for the two cases, we can therefore give the effective $\lambda$-value as
\begin{equation}
\max \left( \lambda_T+\lambda_U+\lambda_V ~ , ~\frac 1{\psuc} \left(\lambda_U+\frac{\lambda_V}{1-1/\eta}\right) \right) .
\end{equation}

Next we consider the error due to the rotation $\theta$ being given to finite precision.
Let us define
\begin{equation}
\psuc' := \psuc \left( 1- \frac 1{\eta+2\lambda_\zeta} \right).
\end{equation}
In the case where OR is used in the implementation of $\SEL_T$, the amplitude on $U+V$ relative to the entire Hamiltonian is
\begin{equation}
\psuc'\sin^2\theta,
\end{equation}
and the weight on $T$ relative to the entire Hamiltonian is
\begin{equation}
1-\psuc'\sin^2\theta.
\end{equation}
When there is error $\Delta\theta$ in the value of $\theta$, the relative size in either can be changed by no more than
\begin{equation}
\psuc'\Delta\theta  < \Delta\theta .
\end{equation}
To bound the error in the overall energy, it can be no more than $\lambda_T\Delta\theta$ for $T$, plus $(\lambda_U + \lambda_V)\Delta\theta$ for $U+V$, for a total error no larger than $\lambda\Delta\theta$.

Similarly, in the case where AND is used in the implementation of $\SEL_T$, the relative weight for $T$ can be given as
\begin{equation}
\cos^2\theta (1-\psuc').
\end{equation}
The weight on $U+V$ is independent of $\theta$.
The error in the relative weight for $T$ is then no larger than $\Delta\theta (1-\psuc')<\Delta\theta$, so the error in the energy is no larger than $\lambda_T \Delta\theta$.
Taking into account both cases, the error in the energy is no larger than $\lambda\Delta\theta$.
Now say we have $n_T$ bits used in the rotation of the bit for selecting between $T$ and $U+V$.
The error in $\theta$ is then no greater than $\pi/2^{n_T}$.
We therefore require
\begin{equation}
\pi\lambda/2^{n_T} \le \epsilon_T,
\end{equation}
where $\epsilon_T$ is allowable error due to finite precision of this rotation.
We can then take
\begin{equation}
n_T \ge \log(\pi\lambda/\epsilon_T).
\end{equation}

For our block encodings, we replace $\lambda_T$ with the slightly larger $\lambda'_T$, which is due to the preparation of the registers for selecting the digits, and modify $\lambda_U$ and $\lambda_V$ to $\lambda^\alpha_U,\lambda^\alpha_V$ to account for the finite precision of preparation of the superposition of $\nu$.
These modified values can be used in the above reasoning unchanged.
Moreover, because $\lambda^\alpha_U,\lambda^\alpha_V$ are proportional to $\alpha$, we can select $\alpha$ such that the rotation required for $\theta$ uses a finite number of bits, and we can just consider imprecision in the preparation of the superposition over $\nu$ with that value of $\alpha$, rather than the finite number of bits for $\theta$.

To be more specific,
\begin{align}
	    \lambda'_T &= \frac{6\eta\pi^2}{\Omega^{2/3}} 2^{2n_p-2},\\
	    \lambda^\alpha_U+\lambda^\alpha_V &= \frac {\eta}{2\pi \Omega^{1/3}} (\eta-1+2\lambda_\zeta)\lambda^\alpha_\nu .
\end{align}
Therefore the ratio is
\begin{align}
\frac{\lambda^\alpha_U+\lambda^\alpha_V}{\lambda'_T} &= \frac{\Omega^{1/3}}{12\pi^3} \frac{(\eta-1+2\lambda_\zeta)}{2^{2n_p-2}}\lambda^\alpha_\nu \nn
&= \alpha \frac{4\Omega^{1/3}}{3\pi^3} \frac{(\eta-1+2\lambda_\zeta)}{2^{2n_p-2}}2^{n_p+2} p_\nu \nn
& = \alpha \frac{4\Omega^{1/3}}{3\pi^3} {(\eta-1+2\lambda_\zeta)}2^{-n_p} p_\nu.
\end{align}
In practice, we find the best values of $\alpha$ are around $1-3/2\mathcal{M}$ to $1-1/\mathcal{M}$, with only a few percent variation in the error in that range.
Therefore, if we have enough precision to give $\alpha$ in that range, then the error will be negligibly increased.
The precision required for $\alpha$ approximately corresponds to the number of bits in $\mathcal{M}$, so we can typically take approximately $n_{\mathcal{M}}$ bits for the rotation of $\alpha$.

\section{Detailed costing of preparation of \texorpdfstring{$k$}{k} state}
\label{app:kstate}

To explain the sequence of operations to prepare the superposition over $k$ in more detail, we first prepare an equal superposition state of the form
\begin{equation}
    \frac 1{\sqrt{\sig(0)}} \sum_{w=0}^{\sig(0)-1} \ket{w}.
\end{equation}
Introducing a zeroed ancilla and performing the inequality test with $\sig(2)$ gives the state
\begin{equation}
    \frac 1{\sqrt{\sig(0)}} \left( \sum_{w=\sig(2)}^{\sig(0)-1} \ket{w}\ket{0} + \sum_{w=0}^{\sig(2)-1} \ket{w}\ket{1} \right).
\end{equation}
Introducing a $\ket{0}$ state and performing a Hadamard gives $\ket{+}$.
Then if we have zero on this qubit and the result of the inequality test, then we set zero in another ancilla; that requires just one Toffoli.
This qubit is $\ket{0}$ for $k=0$, and $\ket{1}$ for all other values of $k$.
The state is now
\begin{equation}
    \frac 1{\sqrt{\sig(0)}} \left( \frac 1{\sqrt{2}}\sum_{w=\sig(2)}^{\sig(0)-1} \ket{w}\ket{0}_{2}(\ket{0}\ket{0}_{1}+\ket{1}\ket{1}_{1}) + \sum_{w=0}^{\sig(2)-1} \ket{w}\ket{1}_{2}\ket{+}\ket{1}_{1} \right),
\end{equation}
where the subscripts $1$ and $2$ denote the qubits with the unary encoding of $k$, which are $\ket{1}$ for $k\ge 1$ and $k\ge 2$, respectively.
Next, perform an inequality test with $\sig(3)$.
This gives
\begin{align}
    \frac 1{\sqrt{\sig(0)}} \left( \frac 1{\sqrt{2}}\sum_{w=\sig(2)}^{\sig(0)-1} \ket{w}\ket{0}_{2}(\ket{0}\ket{0}_{1}+\ket{1}\ket{1}_{1})\ket{0}_{3} + \sum_{w=\sig(3)}^{\sig(2)-1} \ket{w}\ket{1}_{2}\ket{+}\ket{1}_{1}\ket{0}_{3}+\sum_{w=0}^{\sig(3)-1} \ket{w}\ket{1}_{2}\ket{+}\ket{1}_{1}\ket{1}_{3} \right).
\end{align}
Now note that the inequality test here was performed on all the qubits storing $w$.
We only need to perform an inequality test on the qubits that are sufficient to store values up to $\sig(2)-1$, which gives
\begin{align}
    \frac 1{\sqrt{\sig(0)}} \left( \frac 1{\sqrt{2}}\sum_{w=\sig(2)}^{\sig(0)-1} \ket{w}\ket{0}_{2}(\ket{0}\ket{0}_{1}+\ket{1}\ket{1}_{1})\ket{?} + \sum_{w=\sig(3)}^{\sig(2)-1} \ket{w}\ket{1}_{2}\ket{+}\ket{1}_{1}\ket{0}+\sum_{w=0}^{\sig(3)-1} \ket{w}\ket{1}_{2}\ket{+}\ket{1}_{1}\ket{1} \right).
\end{align}
The question mark is used to indicate a value that may be 0 or 1 depending on $w$, but whose value we do not care about.
Now we set a new qubit to $\ket{1}$ if the result of that inequality test \emph{and} the subsystem with subscript 2 is in the state $\ket{1}$.
That takes one Toffoli, and gives the state
\begin{align}
    \frac 1{\sqrt{\sig(0)}} \!\!\left( \frac 1{\sqrt{2}}\!\!\sum_{w=\sig(2)}^{\sig(0)-1}\!\! \ket{w}\ket{0}_{2}(\ket{0}\ket{0}_{1}+\ket{1}\ket{1}_{1})\ket{?}\ket{0}_3 + \!\!\sum_{w=\sig(3)}^{\sig(2)-1}\!\! \ket{w}\ket{1}_{2}\ket{+}\ket{1}_{1}\ket{0}\ket{0}_3+\!\!\sum_{w=0}^{\sig(3)-1}\!\! \ket{w}\ket{1}_{2}\ket{+}\ket{1}_{1}\ket{1}\ket{1}_3 \right)\! .
\end{align}
That new qubit is now the third qubit in the unary representation of $k$.
In this way we generate the unary representation of $k$ directly, without needing to perform binary to unary conversion.
The final inequality test we perform is that with $\sig(K)$ to generate the unary qubit flagging $k=K$.

To explain how to perform the inequality test more explicitly, see Fig.~18 in \cite{Sanders2020_b}.
That shows how to implement addition with a (classically given) constant without needing to store that constant in qubits.
Simply by reversing the circuit, we can subtract a constant which can be used to perform an inequality test with that constant.
That is shown in \fig{subtract}.
In this figure, the inequality test would be on 4 bits, with $i_4=t_4=0$ and the result of the inequality test being given in the bottom register.
	\begin{figure*}[tbh]
		\centering
			\includegraphics[width=0.7\textwidth]{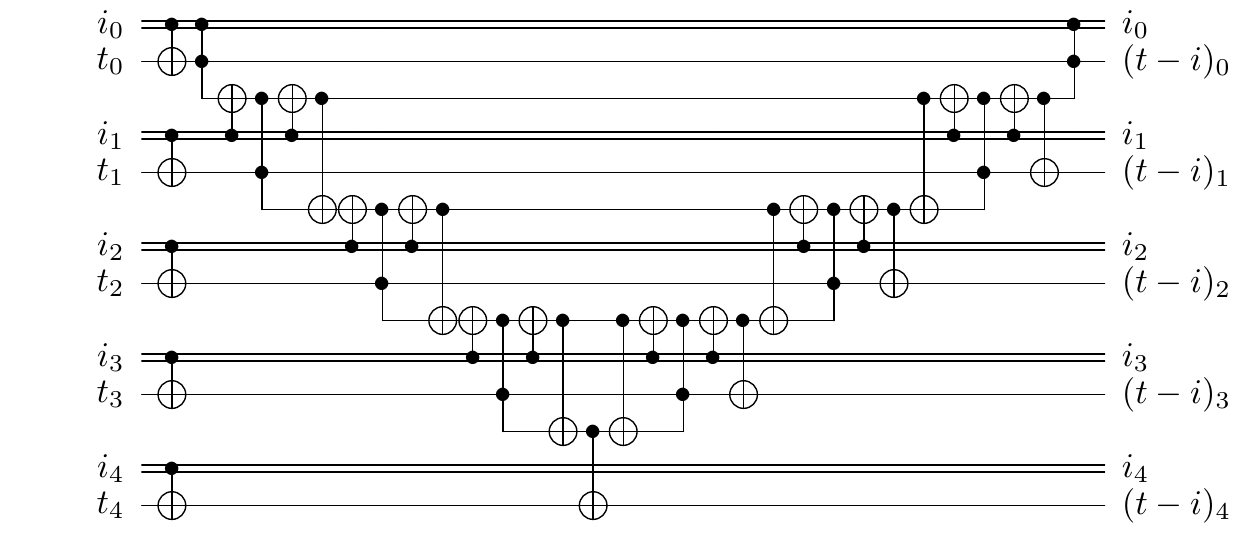}
		\caption{A circuit to subtract a constant from a quantum register without storing that constant in qubits.
		The double lines indicate classical registers.
		The qubit at the bottom provides the output of an inequality test.}
		\label{fig:subtract}
	\end{figure*}

In that case, the result of the inequality test is given in the temporary ancilla where the second CNOT on $t_4$ is performed.
That means that we can use the result of the inequality test immediately without needing to perform the rest of the quantum circuit.
We also want to be able to use the input qubits for other purposes, without them being altered.
The only operations that are performed with them as the target are the initial column of CNOTs.
Therefore, those CNOTs can be performed on those qubits again to return them to their initial values, giving the circuit on the left side of the dotted line in \fig{ineqtest}.

	\begin{figure*}[tbh]
		\centering
			\includegraphics[width=0.7\textwidth]{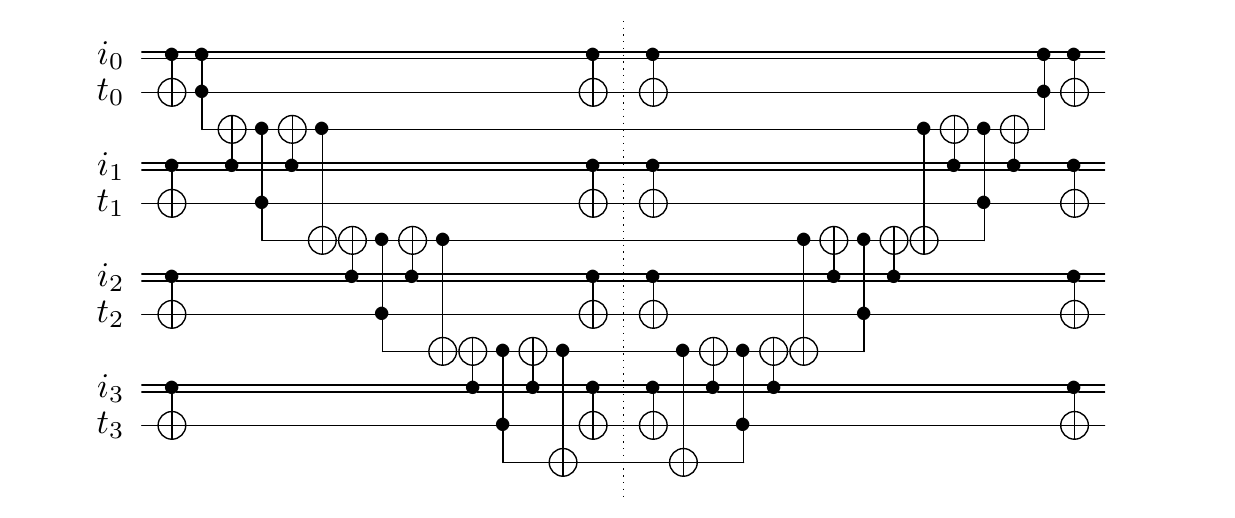}
		\caption{A circuit to perform an inequality test with a constant giving the result in the bottom register at the dotted line.
		At the position of the dotted line the input registers are returned to their original values.
		For application of this circuit, operations that use the result of the inequality test would be put in the position of the dotted line.
		For the state preparation over $k$, the part of the circuit to the left of the dotted line would be used for preparation, then the $\SEL$ operations would be performed, then the part of the circuit on the right would be used to invert the state preparation.}
		\label{fig:ineqtest}
	\end{figure*}

At the point of the dotted line, the result of the inequality test has been output, and the other qubits have been returned to their initial values.
Then one can simply reverse the circuit in order to erase the result of the inequality test.
In this example, there is an inequality test with 4 qubits, and only 3 Toffolis are needed.
Recall that the left and right elbows in these diagrams correspond to Toffolis, but the right elbows can be performed with measurements and Clifford gates.
The initial Toffoli has one control that is a classical register, so is a Clifford gate.
This circuit also indicates that 4 temporary ancillae are used, but the top one was produced with Clifford gates, so can be erased with Clifford gates and need not be kept through the rest of the calculation.
As a result, there is a Toffoli cost of one less than the number of qubits, and an ancilla cost of one less than the number of qubits.

For the number of qubits, first there are $n_k$ to store $w$.
Then for the first inequality test we have $n_k-1$ to perform the inequality test, including the output.
That result of the inequality test is the first unary qubit for $k$.
Then we have an additional $\ket{+}$ state, and generate the second unary qubit for $k$, with a total of $2n_k+1$ qubits so far.
Then for the other $K-2$ inequality tests, we have $\lceil\log\sig(k)\rceil-1$ qubits to make the inequality test reversible, then one qubit for the unary representation of $k$.
That is distinct from the result of the inequality test because we are not using all the qubits in the inequality test.
As a result, the total number of ancillae is $n_k+1+\sum_{k=2}^{K-1} \lceil\log\sig(k)\rceil$ including the unary representation of $k$, plus $n_k$ to store $w$.

Next we describe how to properly prepare the times based on $k$.
Let us first consider the case where the times are given in ascending order for implementing $e^{-iH\tau}$.
For each unary register for $k$, the controlled Hadamard on $n_t$ qubits gives an equal superposition.
For each register beyond the range of $k$, we set the $n_t$ qubits to 1.
Then we have $m_{k+1}=2^{n_t-1}-1=M-1$, and subtracting $m_k$ gives $M-1-m_k$.
We apply an extra phasing corresponding to $1/2$ at the end, yielding $M-1/2-m_k$.
Similarly, we apply a phasing of $1/2$ at the beginning, so we do not need to explicitly add $1/2$ to $m_1$.

For register $K$, which is storing $m_K$ if $k=K$ and $M-1$ otherwise, we flip all the bits.
This is equivalent to taking the number $M-1$ (all ones) and subtracting our number, so we would get $M-1-m_k$ if $k=K$, or $0$ otherwise.
In the case $k=K$ we have the value $M-1-m_k$ required.

Now we consider the case that we wish to control between forward and time evolutions.
This means the roles of the first and last time differences are reversed.
The only difference in the procedure we have used is that for the last time difference we have flipped $n_t$ bits.
Therefore we can control this part by simply using CNOTs, and there is no Toffoli cost.

For the remaining $K-1$ time differences, when we have the times in reverse order we obtain negative numbers in two's complement or zero.
To find the absolute value, we can perform bit flips and add 1.
To make this controlled, we can use CNOTs, and add the bit that flags that we are giving the times in reverse order.
The cost of addition on $n_t$ bits is $n_t-1$ Toffolis.
Combined with the $n_t-1$ Toffolis needed for each subtraction, we have a cost of $2(K-1)(n_t-1)$ Toffolis.

A further subtlety in this procedure is that in the case of failure of the equal superposition state, we should perform no operations on the target system.
However, we need only leave the qubit flagging failure of this state preparation flipped (not erasing it at the end).
In the block encoding, this part is automatically eliminated, since postselection on this qubit being zero eliminates the failure.

\section{The modified phase gradient state}
\label{app:mgrad}
    To account for multiplying factors in phase shifts, the phase gradient state can be modified to something of the form
    \begin{equation}
        \frac 1{\sqrt{2^{\bgr}}} \sum_{\ell=0}^{2^{b_{\rm grad}}-1} e^{-2m\pi i \ell/2^{\bgr}}\ket{\ell} ,
    \end{equation}
    for some integer $m$.
    As discussed in \sec{kinetic_exp}, the kinetic energy is given by $2\pi^2/\Omega^{2/3}$ times the sum of squares of the integers stored in the momentum registers, and this energy is multiplied by $\tau=1/(\lambda_U+\lambda_V)$ to give the phase required.
    Now say that
    \begin{equation}
        \frac {2\pi^2}{(\lambda_U+\lambda_V)\Omega^{2/3}} = \alp \frac {2\pi}{2^{\bgr}},
    \end{equation}
    so $\alp$ is the multiplying factor on $2\pi/2^{b_{\rm grad}}$ that we need.
    We choose $b_{\rm grad}$ so that $\alp$ is between $2^{b_T-1}$ and $2^{b_T}-1$ for some integer $b_T$, and we aim to round $\alp$ to an integer with $b_T$ bits.
    The worst case is when $\alp$ is between $2^{b_T-1}$ and $2^{b_T-1}+1$ (in the case of $b_T=10$, this is 512 and 513).
    Then the reduction in $\tau/\exp(\tau(\lambda_U+\lambda_V))$ in rounding to $2^{b_T-1}$ would be
    \begin{equation}
        \frac{e}{1/(\lambda_U+\lambda_V)} \frac{\alp/2^{b_T-1}/(\lambda_U+\lambda_V)} {\exp{(\alp/2^{b_T-1})}}= \frac{\alp/2^{b_T-1}}{\exp{(\alp/2^{b_T-1}-1)}}.
    \end{equation}
    The reduction in $\tau/\exp(\tau(\lambda_U+\lambda_V))$ in rounding to $2^{b_T-1}+1$ would be
    \begin{equation}
        \frac{e}{1/(\lambda_U+\lambda_V)} \frac{\alp/(2^{b_T-1}+1)/(\lambda_U+\lambda_V)} {\exp{(\alp/(2^{b_T-1}+1))}}= \frac{\alp/(2^{b_T-1}+1)}{\exp{(\alp/(2^{b_T-1}+1)-1)}}.
    \end{equation}
    Because there is freedom in which way $\alp$ is rounded, we should consider the maximum of these two.
    The worst case (the smallest maximum) is when they are equal, which gives
    \begin{align}
        \frac{1/2^{b_T-1}}{\exp{(\alp/2^{b_T-1}-1)}} &= \frac{1/(2^{b_T-1}+1)}{\exp{(\alp/(2^{b_T-1}+1)-1)}} \nn
       {1+1/2^{b_T-1}} &= \frac{\exp{(\alp/2^{b_T-1}-1)}}{\exp{(\alp/(2^{b_T-1}+1)-1)}}\nn
       {1+1/2^{b_T-1}} &= \exp{(\alp/(2^{b_T-1}+2^{2b_T-2}))} \nn
       \alp &= (2^{b_T-1}+2^{2b_T-2}) \ln (1+1/2^{b_T-1}).
    \end{align}
    Substituting that value of $\alp$ gives
    \begin{align}
        \frac{\alp/(2^{b_T-1}+1)}{\exp{(\alp/(2^{b_T-1}+1)-1)}} &= \frac{e 2^{b_T-1} \ln (1+1/2^{b_T-1})}{\exp[2^{b_T-1} \ln (1+1/2^{b_T-1})]}\nn
        &\ge \frac 1{1 + 1/2^{2b_T+1}}.
    \end{align}
    
    Following this reasoning, in evaluating the complexity we can compute the values of $\lambda_U$ and $\lambda_V$, choose a value of $b_{\rm grad}$ and find $\alp$ as
    \begin{equation}
        \alp = \frac{\pi 2^{b_{\rm grad}}}{(\lambda_U+\lambda_V)\Omega^{2/3}},
    \end{equation}
    then take
    \begin{equation}
        b_T = \lceil \log \alp \rceil =  b_{\rm grad} + \left\lceil \log \left( \frac{\pi}{(\lambda_U+\lambda_V)\Omega^{2/3}} \right) \right\rceil ,
    \end{equation}
    where $b_T$ is the number of bits being used for the approximation of $\alp$.
    Then the increase in the effective $\lambda$ from adjusting $\tau$ is no more than a factor of $1+1/2^{2b_T+1}$.
    The choice of $b_{\rm grad}$ can alternatively be made by choosing $b_T$ first, then taking
    \begin{equation}
        b_{\rm grad} = b_T - \left\lceil \log \left( \frac{\pi}{(\lambda_U+\lambda_V)\Omega^{2/3}} \right) \right\rceil .
    \end{equation}
    The cost of the addition into the phase gradient register is $b_{\rm grad}-2$.
    Rather than using the upper bound $1+1/2^{2b_T+1}$ on the increase in the cost, it is also possible to just compute the effective time as
    \begin{equation}
        \frac{[\alp]/\alp}{(\lambda_U+\lambda_V)\exp([\alp]/\alp)},
    \end{equation}
    where $[\alp]$ is being used to indicate the rounded value of $\alp$.

\section{Complexity of binary squaring}
\label{app:squaring}

In this appendix we provide a sequence of results on the Toffoli complexity of computing squares of numbers.
The first result uses a relatively simple technique based on repeated addition, but taking account of the fact that some bits are repeated.

	\begin{lemma}[Binary squaring]
		\label{lem:squaring}
		An $n$-bit binary number can be squared using $n^2-2$ $\Toffoli$ gates.
	\end{lemma}
	\begin{proof}
		Consider a binary number $b_{n-1}\cdots b_{0}$ stored in the input register. Our goal is to compute its square in a result register of length $2n$. In the initial step, we copy $b_0=b_0\cdot b_0$ to the result register using one $\CNOT$ gate. We also create $b_0\cdot b_{n-1},\ldots,b_0\cdot b_{1}$ using $n-1$ $\Toffoli$s at positions $3,\ldots,n+1$ from the right. We then uncompute all the intermediate results.
		
		In the second step, we copy $b_1=b_1\cdot b_1$ to an ancilla register of length $n$ using one $\CNOT$ gate. We also create $b_1\cdot b_{n-1},\ldots,b_1\cdot b_{2}$ in the ancilla register using $n-2$ $\Toffoli$s at positions $3,\ldots,n$. We now add this number to the corresponding positions of the result register and we need an ($n+1$)-bit adder (including the carry), which by \cite{Gidney18} can be implemented using $n$ $\Toffoli$s. We uncompute all the intermediate results.
		
		In the next step, we copy $b_2=b_2\cdot b_2$ to an ancilla register of length $n-1$ using one $\CNOT$ gate. We also create $b_2\cdot b_{n-1},\ldots,b_2\cdot b_{3}$ using $n-3$ $\Toffoli$s at positions $3,\ldots,n-1$. We now add this number to the result register and we need an $n$-bit adder (including the carry), which can be implemented using $n-1$ $\Toffoli$s. We uncompute all the intermediate results.
		
		This procedure can be recursed. In the penultimate step, we copy $b_{n-2}=b_{n-2}\cdot b_{n-2}$ to an ancilla register of length $3$ using one $\CNOT$ gate. We also create $b_{n-2}\cdot b_{n-1}$ using $1$ $\Toffoli$s at the left-most position. We now add this number to the result register and we need a $4$-bit adder (including the carry), which can be implemented using $3$ $\Toffoli$s. We uncompute all the intermediate results.
		
		In the last step, we copy $b_{n-1}=b_{n-1}\cdot b_{n-1}$ to a single-qubit ancilla register. We use a $2$-bit adder to add this value back to the result register. This costs $1$ $\Toffoli$ gate.
		
		Overall, we have used a total number 
		\begin{equation}
		\frac{[1+(n-1)](n-1)}{2}+\frac{(3+n)(n-2)}{2}+1=n^2-2
		\end{equation}
		of $\Toffoli$ gates.
	\end{proof}

However, it is possible to improve the complexity if we apply the approach for computing Hamming sums from \cite{Kivlichan2020improvedfault}.
We also used this method for improved arithmetic in \cite{Sanders2020_b}.
The general principle is that bits are grouped into those at different levels, where there are those multiplied by 1, those multiplied by 2, those multiplied by $2^2$ and so forth.
Three (or two) bits can be summed with a single Toffoli, giving a carry bit at the next level.

Therefore, one should sum bits at the first level, giving carry bits at the next level, then sum bits at the next level giving carry bits at the level after that, and so forth.
In computing Hamming weights as in \cite{Kivlichan2020improvedfault} there were only carry bits, but when computing products or sums there are bits at each level as well as carry bits.
The bits obtained at each level will correspond to products of bits from the original numbers that can be computed with Toffolis.

In the following we will refer to the complexity of ``summing'' bits at a level to mean the complexity of simply summing bits and giving carry bits at the next level, and refer to ``completely summing'' bits to mean that we sum the carry bits, then the carry bits from that sum, and so forth, as in the scheme for computing Hamming sums from \cite{Kivlichan2020improvedfault}.
Then the complexity of summing $m$ bits at a level is equal to $\lfloor m/2\rfloor$, so for even numbers it is just $m/2$ whereas for odd numbers there is rounding down.
The number of carry bits is \emph{equal} to the number of Toffolis used.
The complexity of \emph{completely} summing $m$ bits (the Hamming sum) is no more than $m-1$, though it is often less in specific instances \cite{Kivlichan2020improvedfault}.

First we consider the simpler case of computing a square, then proceed to sums of squares.

	\begin{lemma}[Binary squaring improved]
		\label{lem:squaring_improved}
		An $n$-bit binary number can be squared using $n(n-1)$ $\Toffoli$ gates.
	\end{lemma}
	\begin{proof}
Writing an $n$-bit number $p$ in terms of its bits
\begin{equation}
p = \sum_{j=0}^{n-1} p_j 2^j,
\end{equation}
we have
\begin{align}
p^2 &= \sum_{j=0}^{n-1}\sum_{k=0}^{n-1} p_j p_k 2^{j+k}\nn
&=  \sum_{j=0}^{n-1} p_j 2^{2j} + 2\sum_{j=0}^{n-1}\sum_{k=0}^{j-1} p_j p_k 2^{j+k}\nn
&= \sum_{j=0}^{n-1} p_j 2^{2j} +
\sum_{\ell=0}^{2n-3} 2^{\ell+1} \sum_{j=\max(0,\ell-n+1)}^{\lfloor (\ell-1)/2\rfloor} p_{\ell-j} p_j ,
\end{align}
where $\lfloor (\ell-1)/2\rfloor$ is to ensure that $\ell-j>j$.
Breaking this up into odd and even gives
\begin{equation}
p^2 = p_0 + \sum_{\ell=1}^{n-2} 2^{2\ell+1}  \sum_{j=\max(0,2\ell-n+1)}^{\ell-1} p_{2\ell-j} p_j +
\sum_{\ell=1}^{n-1} 2^{2\ell} \left( p_\ell+\sum_{j=\max(0,2\ell-n)}^{\ell-1} p_{2\ell-1-j} p_j\right) .
\end{equation}

Now say $n$ is even, so that for $\ell < n/2$ we have $\max(0,2\ell-n+1)=0$ and $\max(0,2\ell-n)=0$ and for $\ell\ge n/2$ we have $\max(0,2\ell-n+1)=2\ell-n+1$ and $\max(0,2\ell-n)=2\ell-n$.
Then we get
\begin{align}\label{eq:even}
p^2 &= p_0 + \sum_{\ell=1}^{n/2-1} 2^{2\ell} \left( p_\ell+\sum_{j=0}^{\ell-1} p_{2\ell-1-j} p_j\right)+
\sum_{\ell=1}^{n/2-1} 2^{2\ell+1}  \sum_{j=0}^{\ell-1} p_{2\ell-j} p_j \nn & \quad +
\sum_{\ell=n/2}^{n-1} 2^{2\ell} \left( p_\ell+\sum_{j=2\ell-n}^{\ell-1} p_{2\ell-1-j} p_j\right)+  \sum_{\ell=n/2}^{n-2} 2^{2\ell+1}  \sum_{j=2\ell-n+1}^{\ell-1} p_{2\ell-j} p_j .
\end{align}

Alternatively, if $n$ is odd,
for $\ell<(n+1)/2$ we have $\max(0,2\ell-n+1)=0$ and $\max(0,2\ell-n)=0$ and for $\ell\ge (n+1)/2$ we have $\max(0,2\ell-n+1)=2\ell-n+1$ and $\max(0,2\ell-n)=2\ell-n$.
Then we get
\begin{align}\label{eq:odd}
p^2 &= p_0 + \sum_{\ell=1}^{(n-1)/2} 2^{2\ell} \left( p_\ell+\sum_{j=0}^{\ell-1} p_{2\ell-1-j} p_j\right)+
\sum_{\ell=1}^{(n-1)/2} 2^{2\ell+1}  \sum_{j=0}^{\ell-1} p_{2\ell-j} p_j \nn & \quad +
\sum_{\ell=(n+1)/2}^{n-1} 2^{2\ell} \left( p_\ell+\sum_{j=2\ell-n}^{\ell-1} p_{2\ell-1-j} p_j\right)+  \sum_{\ell=(n+1)/2}^{n-2} 2^{2\ell+1}  \sum_{j=2\ell-n+1}^{\ell-1} p_{2\ell-j} p_j .
\end{align}

To write both odd and even cases together we can write
\begin{align}\label{eq:evenodd}
p^2 &= p_0 + \Sigma_1+\Sigma_2+\Sigma_3+\Sigma_4
\end{align}
where
\begin{align}
\Sigma_1 &= \sum_{\ell=1}^{\lceil n/2\rceil -1} 2^{2\ell} \left( p_\ell+\sum_{j=0}^{\ell-1} p_{2\ell-1-j} p_j\right) \\
\Sigma_2 &= \sum_{\ell=1}^{\lceil n/2\rceil -1} 2^{2\ell+1}  \sum_{j=0}^{\ell-1} p_{2\ell-j} p_j \\
\Sigma_3 &= \sum_{\ell=\lceil n/2\rceil}^{n-1} 2^{2\ell} \left( p_\ell+\sum_{j=2\ell-n}^{\ell-1} p_{2\ell-1-j} p_j\right) \\
\Sigma_4 &= \sum_{\ell=\lceil n/2\rceil}^{n-2} 2^{2\ell+1}  \sum_{j=2\ell-n+1}^{\ell-1} p_{2\ell-j} p_j .
\end{align}
In the following we will refer to $\Sigma_j$ by these symbols, and use ``sums'' to refer to sums of bits from individual terms of $\Sigma_j$.

The terms in each $\Sigma_j$ provide bits at various levels that need to be summed.
There is initially a bit from $p_0$, which is on its own and does not need to be summed.
Then the next level with bits to be summed corresponds to $\ell=1$ in $\Sigma_1$, then $\ell=1$ in $\Sigma_2$, then $\ell=2$ in $\Sigma_1$, and so forth.
Summing each set of bits leads to carry bits at the next level.
This proceeds up to bits at a level corresponding to $\ell=\lfloor n/2\rfloor -1$ for $\Sigma_2$, then we go to $\ell=\lfloor n/2\rfloor$ for $\Sigma_3$.
We now alternate between $\Sigma_3$ and $\Sigma_4$, until the last level is that with $\ell=n-1$ in $\Sigma_3$ (because $\Sigma_4$ ends at $\ell=n-2$).

Now, for each of the $\Sigma_j$ we have the following numbers of bits to sum given $\ell$.
\begin{enumerate}
\item $\Sigma_1$: There are $\ell+1$ bits to sum.
\item $\Sigma_2$: There are $\ell$ bits to sum.
\item $\Sigma_3$: There are $n-\ell+1$ bits to sum.
\item $\Sigma_4$: There are $n-\ell-1$ bits to sum.
\end{enumerate}
Therefore, for $\ell=1$ with $\Sigma_1$, we have $2$ bits to sum, giving one carry bit to $\ell=1$ with $\Sigma_2$, so there are again 2 bits to sum, giving one carry bit to $\ell=2$ with $\Sigma_1$.
There are then 4 bits to sum including the carry bit, with complexity 2 and 2 carry bits to $\ell=2$ with $\Sigma_2$.

In general, there are $2\ell$ bits to sum.
We have seen this at the starting values of $\ell$.
To prove by induction, we assume it is true for some $\ell$ for $\Sigma_1$.
Then there are $\ell$ carry bits from $\Sigma_1$ to $\Sigma_2$ at the same value of $\ell$.
That gives a total number of bits to sum $2\ell$ as required.
Then there are $(\ell+1)+1+\ell$ bits to sum for $\Sigma_1$ with $\ell'=\ell+1$.
That is equal to $2\ell'$, so the formula is correct for the next value of $\ell$, this proving the formula by induction.

Passing from $\Sigma_2$ to $\Sigma_3$ with $\ell=\lceil n/2\rceil$, there are $\lceil n/2\rceil-1$ carry bits, giving a total number of bits to sum $n-\lceil n/2\rceil +1 +\lceil n/2\rceil -1 = n$.
Now we will consider the cases of odd and even $n$ separately.
For even $n$, there are $n/2$ bits carried to $\Sigma_4$, and then a total number of bits to sum $n-n/2-1+n/2=n-1$.
Then there are $n/2-1$ bits carried to $\Sigma_3$ at $\ell=n/2+1$, giving a total number of bits to sum $n-(n/2+1)+1+n/2-1=n-1$.
Then there are $n/2-1$ bits carried to $\Sigma_4$ at $\ell=n/2+1$, giving a total number of bits to sum $n-(n/2+1)-1+n/2-1=n-3$.
These bits can be summed with $n/2-2$ Toffolis, giving the same number of carry bits.

In general, the number of bits to sum is $2(n-\ell)+1$ for $\Sigma_3$ and $2(n-\ell)-1$ for $\Sigma_4$, except for $2(n-\ell)$ with $\ell=n/2$ for $\Sigma_3$.
We have seen this for the starting values, and for the iteration there are $n-\ell$ carry bits from $\Sigma_3$ to $\Sigma_4$, giving a total number of bits to sum $2(n-\ell)-1$.
Then there are $n-\ell-1$ carry bits from $\Sigma_4$ to $\Sigma_3$ at $\ell'=\ell+1$, so there is then $n-(\ell+1)+1+n-\ell-1=2(n-\ell')+1$ bits to sum.
This gives the iteration required.

Then at the end we have $2(n-(n-1))+1=3$ bits to sum with $\ell=n-1$ for $\Sigma_3$, and there are no further bit sums needed.
Therefore, for the even $n$ case the total number of Toffolis is
\begin{align}
    \sum_{\ell=1}^{n/2-1} \ell + \sum_{\ell=1}^{n/2-1} \ell + \sum_{\ell=n/2}^{n-1} (n-\ell) + \sum_{\ell=n/2}^{n-2} (n-\ell-1) =n(n-1)/2.
\end{align}
The four sums here correspond to $\Sigma_1$ to $\Sigma_4$, and the special case $\ell=n/2$ for $\Sigma_3$ can be combined with the others because the cost of the bit sums is still $n/2$.
There are also $n(n-1)/2$ Toffolis needed to compute the products of bits that are to be summed, giving a total complexity $n(n-1)$ as required.

Then for the odd $n$ case, there are $(n-1)/2$ carry bits to $\Sigma_4$ at $\ell=(n+1)/2$, giving a total number of bits to sum $n-(n+1)/2-1+(n-1)/2=n-2$.
That is equal to $2(n-\ell)-1$, and so we are finding that the same formula as before holds.
Similarly we have $2(n-\ell)+1$ bits to sum for $\Sigma_3$, and in fact we find that this formula holds for the initial $\ell(n+1)/2$ for $\Sigma_3$ as well, so it does not need to be treated as a special case.
The iteration to prove the formula is exactly the same as for the case of even $n$, so will not be repeated.
Then the total Toffoli complexity is
\begin{align}
    \sum_{\ell=1}^{(n-1)/2} \ell + \sum_{\ell=1}^{(n-1)/2} \ell + \sum_{\ell=(n+1)/2}^{n-1} (n-\ell) + \sum_{\ell=(n+1)/2}^{n-2} (n-\ell-1) =n(n-1)/2.
\end{align}
Again there is the complexity $n(n-1)/2$ to compute the products of bits, giving a total complexity $n(n-1)$ as required.
	\end{proof}
	
Next we consider the case where we need to sum three squares.
This is a task that we need repeatedly, in order to calculate the norms of vectors.
Summing the bits all together is more efficient than simply separately computing the sums then adding them together.
The result is as follows.
	
	\begin{lemma}[Summing triple squares]
	\label{lem:triple}
	The sum of three $n$-bit binary numbers can be computed using $3n^2-n-1$ Toffolis.
	\end{lemma}
	\begin{proof}
We use the same expression for the squares as before from \eq{evenodd}, again using $\Sigma_j$.
In this case there are three bits at the first level, corresponding to $p_0$ for the 3 value of $p$ to square.
These may be summed with complexity 1, giving 1 carry bit at the next level (bits multiplied by 2).
However, the next level with bits to sum is that for bits multiplied by $4$, corresponding to $\ell=1$ in $\Sigma_1$.
Therefore the carry bits are not included at that level, and
there are $3(\ell+1)=6$ bits to sum, with complexity 3.
That gives three carry bits to the next level, which corresponds to $\Sigma_2$ with $\ell=1$.
There are again $3\ell+3=6$ bits to sum, with complexity 3.
Then these 3 carry bits are taken back to $\Sigma_1$ with $\ell=2$, and then we have $3(\ell+1)=12$ bits to sum, with complexity 6.
We have 6 carry bits for the next level, $\Sigma_2$ with $\ell=2$, which means there are $3\ell+6=12$ bits to sum, again with complexity 6.

In general, for $\Sigma_1$ and $\Sigma_2$ there are $6\ell$ bits to sum, which has complexity $3\ell$.
Note that we have already shown this for $\ell=1$ and $2$.
For induction, if it is true for some $\ell$ with $\Sigma_1$, then there are $3\ell$ carry bits being taken to term $\ell$ in $\Sigma_2$, so the number of bits to sum is $6\ell$ again.
Then $3\ell$ bits are being taken to term $\ell+1$ in $\Sigma_1$, and there are then $3(\ell+2)+3\ell = 6(\ell+1)$ bits to sum.
This demonstrates that our formula is correct via induction.

Next we consider the special case of even $n$, where
we have the following special starting cases for $\ell\ge n/2$.
\begin{enumerate}
\item For $\Sigma_3$ with $\ell=n/2$, we have $3n/2-3$ carry bits from the previous level ($\Sigma_2$ with $\ell=n/2-1$).
The number of bits to sum is therefore $3(n-\ell+1)+3n/2-3=3n$, which can be summed with $3n/2$ Toffolis.
\item There are $3n/2$ carry bits for $\Sigma_4$ with $\ell=n/2$, so there are $3(n-\ell-1)+3n/2=3n-3$ bits to sum.
These are summed with complexity $3n/2-2$.
\item There are $3n/2-2$ carry bits for $\Sigma_3$ at $\ell=n/2+1$.
That gives $3(n-\ell+1)+3n/2-2=3n-2$ bits to sum, which is done with complexity $3n/2-1$.
\item There are $3n/2-1$ carry bits for $\Sigma_4$ with $\ell=n/2+1$, so there are $3(n-\ell-1)+3n/2-1=3n-7$ bits to sum.
These are summed with complexity $3n/2-4$.
\item There are $3n/2-4$ carry bits for $\Sigma_3$ at $\ell=n/2+2$.
That gives $3(n-\ell+1)+3n/2-4=3n-7$ bits to sum, which is done with complexity $3n/2-4$.
\end{enumerate}
From this point on, we find that there are $6(n-\ell)-1$ bits to sum for $\Sigma_4$ with $\ell\ge n/2+1$, and $6(n-\ell)+5$ bits to sum for $\Sigma_3$ with $\ell\ge n/2+2$.
The last two items above we use as the starting point for the induction, where we see that the formula is correct.

To prove the iteration for induction, say there are $6(n-\ell)+5$ bits to sum for $\Sigma_3$ with some $\ell$.
Then there are $3(n-\ell)+2$ carry bits to $\Sigma_4$.
That gives $3(n-\ell-1)+3(n-\ell)+2=6(n-\ell)-1$ bits for $\Sigma_4$ as required.
Then those can be summed with $3(n-\ell)-1$ Toffolis, giving this number of carry bits to $\Sigma_3$ at $\ell+1$.
Then the number of bits to sum is $3(n-\ell)+3(n-\ell)-1=6(n-\ell)-1 = 6(n-(\ell+1))+5$.
This is the formula for $\Sigma_3$ for $\ell+1$ as required.

Finally, for $\ell=n-1$, there are $11$ bits to sum for $\Sigma_3$.
These give 5 bits at the next level, but there is no matching term from $\Sigma_4$, so these can be completely summed with just another 3 Toffolis (there are two Toffolis and carry bits to the next level, then the two bits at the next level can be summed with one Toffoli).
As a result, the overall Toffoli complexity is
\begin{align}
&1+\sum_{\ell=1}^{n/2-1} 3\ell +\sum_{\ell=1}^{n/2-1} 3\ell +3n/2+(3n/2-2)+(3n/2-1)+\sum_{\ell=n/2+2}^{n-1} [3(n-\ell) +2] + \sum_{\ell=n/2+1}^{n-2} [3(n-\ell)-1] + 3 \nn
&=n(3n+1)/2 -1 .
\end{align}
Now, we also note that to compute the $n(n-1)/2$ bit products $p_jp_k$ for each of the three components of $p$ requires $3n(n-1)/2$ Toffolis.
Adding that to the Toffoli count above gives $3n^2-n-1$, as required.

Next we consider the case of odd $n$ with $\Sigma_3$ and $\Sigma_4$.
Then for $\ell\ge (n+1)/2$, we have the following starting cases.
\begin{enumerate}
\item For $\Sigma_3$ with $\ell=(n+1)/2$ we have $3(n-1)/2$ carry bits from $\Sigma_2$.
There is then a total of $3(n-\ell+1)+3(n-1)/2=3n$, which can be summed with $(3n-1)/2$ Toffolis.
\item For $\Sigma_4$ with $\ell=(n+1)/2$ we have $3(n-1)/2$ carry bits, giving a total $3(n-\ell-1)+(3n-1)/2=3n-5$.
This can be summed with $(3n-5)/2$ Toffolis.
\item For $\Sigma_3$ with $\ell=(n+3)/2$ we have $(3n-5)/2$ carry bits from $\Sigma_4$.
The total number of bits to sum is $3(n-\ell+1)+(3n-5)/2=3n-4$, which can be summed with $(3n-5)/2$ Toffolis.
\item For $\Sigma_4$ with $\ell=(n+3)/2$ we have $3(n-5)/2$ carry bits, giving a total $3(n-\ell-1)+(3n-5)/2=3n-10$.
This can be summed with $(3n-11)/2$ Toffolis.
\item For $\Sigma_3$ with $\ell=(n+5)/2$ we have $(3n-11)/2$ carry bits from $\Sigma_4$.
The total number of bits to sum is $3(n-\ell+1)+(3n-5)/2=3n-10$, which can be summed with $(3n-11)/2$ Toffolis.
\end{enumerate}
From these last two cases on, we have the same formulae for the numbers of bits to sum as before, and again they can be used as the starting point for the induction.
The induction proceeds in exactly the same way as before, so will not be repeated.

Lastly, for $\ell=n-1$, there are again $11$ bits to sum for $\Sigma_3$, and the 5 carry bits can be completely summed with 3 Toffolis.
As a result, the overall Toffoli complexity is
\begin{align}
&1+\sum_{\ell=1}^{(n-1)/2} 3\ell +\sum_{\ell=1}^{(n-1)/2} 3\ell +(3n-1)/2+(3n-5)+\sum_{\ell=(n+5)/2}^{n-1} [3(n-\ell) +2] + \sum_{\ell=(n+3)/2}^{n-2} [3(n-\ell)-1] + 3 \nn
&=n(3n+1)/2 -1.
\end{align}
Again, adding the $3n(n-1)/2$ Toffolis needed for computing the bit products gives the total of $3n^2-n-1$, as required.
	\end{proof}
	
Next we consider a more general case where we are summing $k$ squares.
In this case the bound we prove an upper bound on the complexity that is more efficient than computing the squares then summing, but numerical testing indicates that the actual complexity will be slightly lower.

\begin{lemma}[Binary sums of squares]
\label{lem:sumsqa}
The sum of the squares of $k$ binary numbers, each of $n$ bits, can be computed using $kn^2$ Toffoli gates.
\end{lemma}
\begin{proof}
Again we can use \eq{evenodd} with $\Sigma_j$ for a square of a single integer.
Similar to the previous cases, we need to consider the sum of bits that are multiplied by 1, giving carry bits to those multiplied by 2, then sum the bits that are multiplied by 2, and so forth.
In the above case where we summed the squares of three numbers, there were three bits to sum that were multiplied by 1 (the least significant bit for each of the three numbers).
That gave a single carry bit that was multiplied by 2.
That was the only bit that was multiplied by 2, then the next level was bits that were multiplied by 4, corresponding to those for $\ell=1$ in $\Sigma_1$.
Therefore no sum was needed for the single bit, and one could simply consider the sum of the bits for $\ell=1$ for $\Sigma_1$.

In the case where we are summing the squares for $k$ numbers where $k$ may be larger than 3, there may be more than one carry bit multiplied by 2.
Those then need to be summed, which will introduce carry bits that are multiplied by 4, and will need to be summed together with the other bits from $\ell=1$ in $\Sigma_1$.
The complexity of the sum of $m$ bits (giving carry bits at the next level) is generally $\lfloor m/2\rfloor$ Toffolis, so if $m$ is even it is just $m/2$, but if $m$ is odd then the complexity is $m/2$ rounded down to the nearest integer.
This feature can be used to bound the complexity due to summing bits in terms in $\Sigma_j$ together with possible carry bits that were obtained from the initial sum of $k$ bits (the least significant bits from each of the $k$ numbers).
If we have $m$ bits to sum and another $c$ bits to sum together with those, then the number of Toffolis needed is no larger than $(m+c)/2$, regardless of whether $m+c$ is odd or even.

What this means is that we can upper bound the complexity by considering the sums of the bits in $\Sigma_j$ separately from the carry bits from the initial sum of $k$ bits.
If we simply upper bound the complexity by dividing by 2, then the complexity due to summing both together can be upper bounded by upper bounding the complexities independently then adding the two upper bounds.

With that in mind, let us consider the numbers of bits to sum for $\Sigma_1$ and $\Sigma_2$ disregarding any carry bits from the initial sum over $k$ bits.
Here, the number of bits to sum (including carry bits from $\Sigma_1$ and $\Sigma_2$) will be shown to be $2k\ell$.
Recall that we sum bits with $\ell=1$ in $\Sigma_1$, then with $\ell=1$ for $\Sigma_2$, then for $\ell=2$ for $\Sigma_1$, and so forth.
To show this expression by induction, we first consider $\Sigma_1$ and $\ell=1$, where we have $k(\ell+1)=2k=2k\ell$ bits to sum.
Therefore the formula is true for $\ell=1$.
Then assuming the formula is true for some $\ell$ for $\Sigma_1$, we have $k\ell$ carry bits to $\Sigma_2$.
The total number of bits to sum is then $k\ell+k\ell=2k\ell$ as required.
Then, given the formula for $\Sigma_2$, there are $k\ell$ bits to carry over to $\Sigma_1$ at $\ell'=\ell+1$, giving $k\ell+k(\ell+2)=2k\ell'$ bits to sum, which is the correct formula with $\ell'=\ell+1$, thereby establishing the iteration for induction.
That proves that the formula $2k\ell$ is correct.

Now, considering the case of even $n$, we need to consider $\Sigma_3$ with $\ell=n/2$.
Then we have $k(n/2-1)$ carry bits and $k(n/2+1)$ bits from $\Sigma_3$, giving a total number of bits to sum $kn$.
We will write this number of bits as $2k(n-\ell+1-2^{n-2\ell})$ for reasons that will become evident shortly.
Then, given this formula, the number of carry bits to $\Sigma_4$ is no more than $k(n-\ell+1-2^{n-2\ell})$.
Here we are now ignoring any rounding down for odd cases so we can later introduce the complexity of summing the initial $k$ bits.

The total number of bits to sum at $\ell$ for $\Sigma_4$ is then $k(n-\ell-1)+k(n-\ell+1-2^{n-2\ell})=2k(n-\ell-2^{n-2\ell-1})$.
Then there are $k(n-\ell-2^{n-2\ell-1})$ carry bits to $\Sigma_3$ with $\ell'=\ell+1$.
The total number of bits to sum is then $k(n-\ell)+k(n-\ell-2^{n-2\ell-1})=2k(n-\ell-2^{n-2\ell-2})=2k(n-\ell'+1-2^{n-2\ell'})$.
This iteration means that for $\ell>n/2$ the number of bits to sum may be upper bounded as $2k(n-\ell+1-2^{n-2\ell})$ and $2k(n-\ell-2^{n-2\ell-1})$ for $\Sigma_3$ and $\Sigma_4$, respectively.

The highest level in $\Sigma_3$ and $\Sigma_4$ is $\ell=n-1$ for $\Sigma_3$, which corresponds to bits multiplied by $2^{2n-2}$.
Using the formula we just proved, we have an upper bound on the number of bits to sum $2k(2-2^{2-n})$, and
the number of carried bits is $k(2-2^{2-n})$.
To completely sum those bits the complexity would usually be given as the number of bits minus 1.
However, that is using rounding down, and in order to account for the initial $k$ bits summed we are not rounding down.
Then the complexity of summing this number of bits is just a number of Toffolis equal to the number of bits.
Similarly, we give the complexity of summing the $k$ bits from the first level as $k$.
This gives the overall complexity as
\begin{align}
&k+\sum_{\ell=1}^{n/2-1} k\ell +\sum_{\ell=1}^{n/2-1} k\ell +\sum_{\ell=n/2}^{n-1} k(n-\ell+1-2^{n-2\ell}) + \sum_{\ell=n/2}^{n-2} k(n-\ell-2^{n-2\ell-1}) + k(2-2^{2-n}) \nn
&=kn(n+1)/2.
\end{align}
In this expression the initial term with $k$ is the upper bound on the contribution to the complexity from summing the initial $k$ bits.
Terms 2 to 5 are then complexities of the sums for $\Sigma_1$ to $\Sigma_4$, respectively.
The final term is the upper bound on the complexity for summing the final $k(2-2^{2-n})$ bits.
There are $kn(n-1)/2$ Toffolis needed to compute the bit products of the form $p_jp_k$, so the total Toffoli complexity is upper bounded by $kn^2$.

Then for the case of odd $n$, for $\Sigma_3$ at $\ell=(n+1)/2$, we have $k(n-1)/2$ carry bits and $k(n-\ell+1)=k(n+1)/2$ bits from $\Sigma_3$.
The total number of bits to sum is $kn$ again.
We can again write this as $2k(n-\ell+1-2^{n-2\ell})$, and we obtain the formula $2k(n-\ell-2^{n-2\ell-1})$ for the number of bits to sum at each level for $\Sigma_4$ again, by exactly the same process of logic as for the even $n$ case.
In the sum for the upper bound on the number of Toffolis needed, only the bounds on the sums are changed, and we obtain
\begin{align}
&k+\sum_{\ell=1}^{(n-1)/2} k\ell +\sum_{\ell=1}^{(n-1)/2} k\ell +\sum_{\ell=(n+1)/2}^{n-1} k(n-\ell+1-2^{n-2\ell}) + \sum_{\ell=(n+1)/2}^{n-2} k(n-\ell-2^{n-2\ell-1}) + k(2-2^{2 - n}) \nn
&=kn(n+1)/2.
\end{align}
Again adding the $kn(n-1)/2$ Toffolis for computing bit products gives a total of $kn^2$.
\end{proof}
Despite this result being quite neat, it is not as tight as what we find when computing the Toffoli costs for particular examples.
Numerically we find that the number of Toffolis needed is $kn^2-n$, which is $n$ lower than was proven above.
In the case where $k$ is a multiple of 3, numerical testing indicates the even tighter bound $kn^2-n-1$.

Lastly we give the complexity for products.
The result is as follows.
\begin{lemma}[Binary products]
\label{lem:products}
The cost of multiplying two numbers, one with $n$ bits and the other with $m$ bits, is $2nm-n$ Toffolis for $n\ge m$.
\end{lemma}
\begin{proof}
Without loss of generality we take $p$ to be of length $n$ and $q$ to be of length $m$ bits.
To consider the cost of products, we can write out the product as
\begin{equation}
\sum_{j=0}^{n-1} \sum_{k=0}^{m-1} p_j q_k 2^{j+k} =
\sum_{\ell=0}^{m-1} \sum_{k=0}^\ell p_{\ell-k} q_k 2^{\ell} +
\sum_{\ell=m}^{n-1}  \sum_{k=0}^{m-1} p_{\ell-k} q_k 2^{\ell} +
\sum_{\ell=n}^{n+m-2}  \sum_{k={\ell-n+1}}^{m-1} p_{\ell-k} q_k 2^{\ell} .
\end{equation}
Then for the first few levels we have the following.
\begin{enumerate}
\item At $\ell=0$, there is just a single bit and no sum needed.
\item At $\ell=1$ there are two bits with a cost of 1 and one carry bit.
\item At $\ell=2$ there are three bits and one carry bit to sum giving cost 2 and 2 carry bits.
\end{enumerate}
The general formula is that at $\ell$ the cost of summing is $\ell$.
To see that, if there are $\ell$ carry bits to $\ell'=\ell+1$, the number of bits to sum is then $2\ell+2$, which has cost $\ell+1=\ell'$, so the iteration needed for induction holds.

Then for $\ell=m$ to $n-1$ there are always $m$ bits to sum and the carry bits.
So, at $\ell=m$ there are $m-1$ carry bits and $m$ bits giving cost $m-1$.
After that there are always $m-1$ carry bits and $m$ bits to sum, giving cost $m-1$.

After that, for $\ell\ge n$, there are $n+m-1-\ell$ bits to sum plus the carry bits.
At $\ell=n$ there are $m-1$ carry bits and another $m-1$, giving cost $m-1$.
Then at $\ell=n+1$ there are $m-1$ carry bits and another $m-2$, giving cost $m-2$.
In general, the cost of summing at $\ell$ is $n+m-1-\ell$.
The total cost is then
\begin{equation}
\sum_{\ell=0}^{m-1} \ell +
\sum_{\ell=m}^{n-1} (m-1)+
\sum_{\ell=n}^{n+m-2} (n+m-1-\ell) =nm-n.
\end{equation}
The cost of computing the bit products is then $nm$ Toffolis, giving a total cost of $2nm-n$ as required.
\end{proof}

\section{Bound on time-discretization error}
\label{app:timedisc}

For the time evolution under the Dyson series, we have the time evolution operator that evolves an initial state as
\begin{equation}
\ket{\psi(t)}=\mathcal{T} \exp \left( -\mi \int_0^t H(s) \, ds \right)\ket{\psi(0)}.
\end{equation}
Here 
\begin{equation}
    \partial_t \mathcal{T} \exp \left( -\mi \int_0^t H(s) \, ds \right) = -\mi H(t) \mathcal{T} \exp \left( -\mi \int_0^t H(s) \, ds \right).
\end{equation}
We now consider the error due to the discretization of this integral.
In each short time interval, we have
\begin{equation}
\int_{\tau}^{\tau+\delta} H(s) \, ds \approx H(\tau) \delta \, .
\end{equation}
We can approximate $H$ within the interval by
\begin{equation}
H(s) \approx H(\tau) + (s-\tau) H'(\tau).\label{eq:Hlin}
\end{equation}
The error is upper bounded by
\begin{equation}
\left\|  H(\tau) -H(s) \right\| \le \delta \max_{v\in [\tau,\tau+\delta]} \|H'(v)\|.
\end{equation}
Then we can use the triangle inequality on unitary operators in the form
\begin{align}
\| ABC - A' B' C' \| &= \| ABC -A'BC + A'BC - A'B'C + A'B'C - A' B' C' \| \nn
&\le \|(A-A')BC\| + \|A'(B-B')C\| + \|A'B'(C-C')\| \nn
&= \|A-A'\| + \|B-B'\| + \|C-C'\|.
\end{align}
That means, for the time evolution, which is a product of evolutions for infinitesimal time periods, we can bound the difference by the integral of the differences between the Hamiltonians.
So, if we have Hamiltonians $H_1$ and $H_2$, then
\begin{align}
&\left\| \mathcal{T} \exp \left( -\mi \int_0^t H_1(s) \, ds \right) - \mathcal{T} \exp \left( -\mi \int_0^t H_2(s) \, ds \right) \right\| \nonumber\\
&= \left\| \int_0^t \partial_\tau \mathcal{T} \exp \left( -\mi \int_0^\tau H_1(s) \, ds \right)\mathrm{d}\tau - \int_0^t \partial_\tau \mathcal{T} \exp \left( -\mi \int_0^\tau H_2(s) \, ds \right)\mathrm{d}\tau \right\| \nonumber\\
&= \left\| \int_0^t (-\mi H_1(\tau)) \mathcal{T} \exp \left( -\mi \int_0^\tau H_1(s) \, ds \right)\mathrm{d}\tau - \int_0^t (-\mi H_2(\tau)) \mathcal{T} \exp \left( -\mi \int_0^\tau H_2(s) \, ds \right)\mathrm{d}\tau \right\| \nonumber\\
&\le \int_0^t \| H_1(s) - H_2(s)\| \, ds.
\end{align}
In the case of the interaction picture, we have
\begin{equation}
H(t) = e^{\mi tA} B e^{-\mi tA},
\end{equation}
and thus the product rule implies
\begin{equation}
H'(t) = \mi \, e^{\mi tA} [A,B] e^{-\mi tA}.
\end{equation}
As a result, using the triangle inequality and sub-multiplicative property of the spectral norm we can bound $\|H'(t)\|$ as
\begin{equation}
\|H'(t)\| \le 2\|A\| \, \|B\|.
\end{equation}
For a single time step of length $\delta$ we have
\begin{align}
 \int_{\tau}^{\tau+\delta} \left\| H(\tau)-H(s) \right\| ds  &\le \int_{\tau}^{\tau+\delta}|s-\tau| 2\|A\|\, \|B\| \, ds \nn
&= \delta^2 \|A\| \, \|B\|.
\end{align}
Therefore, if the integral is approximated by $M$ time steps of size $\delta=\tau/M$, then the error in the integral is bounded by
\begin{equation}
M (\tau/M)^2 \|A\| \, \|B\| = \frac{\tau^2}{M} \|A\| \, \|B\|.
\end{equation}

As we are working in the interaction frame of the kinetic operator, $\|A\|$ can be upper bounded by $\lambda_T$, and $\|B\|$ upper bounded by $\lambda_U+\lambda_V$.
We are block encoding the Dyson series for the time-evolution operator over a time interval $\tau = 1/(\lambda_U+\lambda_V)$, and the number of time steps is (by the definition of $n_t$) $M=2^{n_t}$.
As a result, the error in the block encoded operator due to time discretization is no larger than
\begin{equation}
\frac{\tau^2}{M} \|A\| \, \|B\| \le \frac{\lambda_T}{2^{n_t}(\lambda_U+\lambda_V)}.
\end{equation}
This error translates to an error in the estimated eigenvalue of $H/(\lambda_U+\lambda_V)$.
Therefore, if $\epsilon_t$ is the allowable error due to the time discretization, then we need
\begin{equation}
    \lambda_T 2^{-n_t} \le \epsilon_t.
\end{equation}
Therefore, $n_t$ should be chosen as
\begin{equation}
    n_t = \lceil \log(\lambda_T/ \epsilon_t) \rceil.
\end{equation}
Note that, in the above, the approximation of the time-dependent Hamiltonian is made at one side of the interval.
A more accurate approximation can be made by using the middle of the interval, though that only reduces the number of bits for the time by 1 in this analysis.
A more careful analysis can be made where the second-order correction to the Hamiltonian is considered.
That can be used to show that the error in the approximation of the integral over the Hamiltonian is higher order.

%%%%%%%%%%%%%%%%%%%%%%

Consider $U_1(\delta)$ to be the time evolution operator over the interval $[\tau-\delta/2,\tau+\delta/2)$ for the Hamiltonian $H$, and $U_0(\delta)$ to be the time evolution operator for the Hamiltonian in the center of the interval, $H(\tau)$.  We can find an expression for the derivative of $U_1(\delta)$ by integrating the differential equation for the time-evolution operator.  As appropriate, we will also use the notation that the time evolution operator from times $a$ to $b$ is $U_1(b,a)$.
Then we obtain
\begin{align}
    \partial_\delta U_1(\delta) &= \partial_\delta (U(\tau+\delta/2,\tau)U(\tau,\tau-\delta/2))\nonumber\\
    &= -\frac{\mi}{2}H(\tau+\delta/2)U_1(\delta) + U(\tau+\delta/2,\tau) (\partial_\delta U(\tau,\tau-\delta/2))\nonumber\\
    &=-\frac{\mi}{2}H(\tau+\delta/2)U_1(\delta) + U(\tau+\delta/2,\tau) (\partial_\delta U^\dagger(\tau-\delta/2,\tau))\nonumber\\
    &=-\frac{\mi}{2}H(\tau+\delta/2)U_1(\delta) + U(\tau+\delta/2,\tau) (\partial_\delta U(\tau-\delta/2,\tau))^\dagger\nonumber\\
    &=-\frac{\mi}{2}[H(\tau+\delta/2)U_1(\delta) + U_1(\delta)H(\tau-\delta/2)].
\end{align}
Next, it follows immediately  that the derivative of $U_0$ is
\begin{align}
    \frac {\partial}{\partial \delta} U_0(\delta) &= -\frac {\mi}2 \left[H(\tau)U_0(\delta)+U_0(\delta)H(\tau) \right].
\end{align}
We can then upper bound the rate of change of the norm of the difference of $U_0$ and $U_1$ by
\begin{align}
    \lim_{h\rightarrow 0} &\frac{1}{h}\big| \|U_1(\delta +h) - U_0(\delta +h)\| - \|U_1(\delta) - U_0(\delta)\| \big|\nn
     &\le\left\| \lim_{h\rightarrow 0} \frac{1}{h} \left[U_1(\delta +h) - U_0(\delta +h) - U_1(\delta) - U_0(\delta) \right]\right\|\nn
    &= \left\|\frac{\partial}{\partial \delta} [U_1(\delta) - U_0(\delta)]\right\| \nn
    &= \frac 12 \| U_1(\delta)H(\tau-\delta/2)+H(\tau+\delta/2) U_1(\delta)
    -U_0(\delta)H(\tau)-H(\tau) U_0(\delta)  \| \nn
    &= \frac 12 \Biggr\| U_1(\delta)H(\tau-\delta/2)+U_1(\delta)H(\tau+\delta/2) -2 U_1(\delta) H(\tau)\nn
& \quad - U_1(\delta)H(\tau+\delta/2)  + U_1(\delta) H(\tau) 
+ H(\tau+\delta/2) U_1(\delta)-H(\tau) U_1(\delta)\nn
& \quad  + U_1(\delta) H(\tau)-U_0(\delta)H(\tau)+H(\tau) U_1(\delta)-H(\tau) U_0(\delta) \Biggr\| \nn
&\le \frac 12 \Biggr( \| U_1(\delta) [H(\tau-\delta/2)+H(\tau+\delta/2)-2 H(\tau)] \|\nn
& \quad + \| [H(\tau+\delta/2)-H(\tau),U_1(\delta) ] \| \nn
& \quad +\|(U_1(\delta)-U_0(\delta))H(\tau)\| + \|H(\tau)(U_1(\delta)-U_0(\delta))\|\Biggr) .
\end{align}

Using an upper bound $\Lambda$ on $\|H\|$, we have that $\|U(\delta)-I\|\le \delta \Lambda$.
We can therefore say that
\begin{align}
\| [H(\tau+\delta/2)- H(\tau),U_1(\delta)] \| & = \| [H(\tau+\delta/2)- H(\tau) , U_1(\delta)-I] \| \nn
&\le 2\|U_1(\delta)-I\| \, \|H(\tau+\delta/2)-H(\tau)\| \nn
&\le 2\delta \Lambda \|H(\tau+\delta/2)-H(\tau)\|.
\end{align}
If there is an upper limit $\Gamma_1$ on the norm of the derivative of $H$, then
\begin{equation}
\| [H(\tau+\delta/2)- H(\tau),U_1(\delta)] \| \le \delta^2 \Lambda \Gamma_1.
\end{equation}
If there is a constant upper bound $\Gamma_2$ on the second derivative of $H$, we get
\begin{equation}
\| U_1(\delta) [ H(\tau+\delta/2) + H(\tau-\delta/2) - 2H(\tau)]\| \le \Gamma_2 \delta^2/4 .
\end{equation}
These bounds give us
\begin{equation}
\left\|U_1(\delta)-U_0(\delta) \right\| \le \int_0^\delta \left\{ \frac 12 (\Lambda \Gamma_1+\Gamma_2/4) s^2 + \Lambda \left\|U_1(s)-U_0(s) \right\| \right\}\mathrm{d}s .
\end{equation}
A solution to the above integral equation is
\begin{equation}
\left\|U_1(\delta)-U_0(\delta) \right\| \le \frac{4\Lambda\Gamma_1+\Gamma_2}{(2\Lambda)^3}\left[ 2\left( e^{\Lambda\delta}-1 \right) - \Lambda\delta (2+\Lambda\delta) \right] .
\end{equation}
For the interaction picture, we have 
\begin{align}
H(t) &= e^{\mi tA} B e^{-\mi tA}, \\
H'(t) &= \mi e^{\mi tA} [A,B] e^{-\mi tA}, \\
H''(t) &= - e^{\mi tA} [A,[A,B]] e^{-\mi tA}.
\end{align}
Therefore we can take $\Lambda=\|B\|$, $\Gamma_1=2\|A\| \, \|B\|$ and $\Gamma_2=4\|A\|^2  \|B\|$, so we get
\begin{equation}
\left\|U_1(\delta)-U_0(\delta) \right\| \le \left(\frac{\|A\|}{\|B\|} + \frac{\|A\|^2}{2\|B\|^2}\right) \left[ 2\left( e^{\|B\|\delta}-1 \right) - \|B\|\delta (2+\|B\|\delta) \right] .
\end{equation}
For our application, we are approximating the integral by $2^{n_t}$ time steps of size $\delta=\tau/2^{n_t}$.
Then the error in the time-ordered exponential over time step $\delta$ is bounded by
\begin{equation}\label{eq:timediscer}
2^{-n_t}\left(\frac{2\|A\|}{\|B\|} + \frac{\|A\|^2}{\|B\|^2}\right) \left[ 2^{n_t} \left( e^{\|B\|\tau/2^{n_t}}-1 \right) - \|B\|\tau (1+\|B\|\tau/2^{n_t+1}) \right] .
\end{equation}
Multiplying by $2^{n_t}$ gives the bound on the error for the time-ordered exponential over the entire region
\begin{equation}
\left(\frac{2\|A\|}{\|B\|} + \frac{\|A\|^2}{\|B\|^2}\right) \left[ 2^{n_t} \left( e^{\|B\|\tau/2^{n_t}}-1 \right) - \|B\|\tau (1+\|B\|\tau/2^{n_t+1}) \right] .
\end{equation}

This gives us the error in approximating the time-evolution operator with the centered difference formula and this leads us to the following lemma.
\begin{lemma}
Let $H(t) = e^{\mi tA} B e^{-\mi tA}$ for bounded Hermitian operators $A$ and $B$ and $t\in\mathbb{R}$.  We then have that if $n_t$ is a positive integer and $\delta = \tau/2^{n_t}$ 
then 
\begin{align*}
&\|\mathcal{T}e^{-\mi\int_{t_0 -\delta/2}^{t_0 +\delta/2} H(t) \mathrm{d}t} - e^{-\mi H({t_0})\delta}\|\nonumber\\
&\qquad\le 2^{-n_t}\left(\frac{2\|A\|}{\|B\|} + \frac{\|A\|^2}{\|B\|^2}\right) \left[ 2^{n_t} \left( e^{\|B\|\tau/2^{n_t}}-1 \right) - \|B\|\tau (1+\|B\|\tau/2^{n_t+1}) \right].
\end{align*}
\end{lemma}

In practice, for the interaction picture one should have $\|A\|\ll \|B\|$, so that it is the second term given in the round brackets above that should be dominant.
Moreover, expanding the exponential gives
\begin{equation}
    2\left( e^{\Lambda\delta}-1 \right) - \Lambda\delta (2+\Lambda\delta) = \frac 13 (\Lambda\delta)^3 + \mathcal{O} ((\Lambda\delta)^4).
\end{equation}
Therefore, we expect that the error should be approximately
\begin{equation}
    \left\|U_1(\delta)-U_0(\delta) \right\| \lesssim \frac {\delta^3 \|A\|^2 \|B\|}{6}
\end{equation}
for a time step, or
\begin{equation}
    \frac {\tau^3 \|A\|^2 \|B\|}{6\times 2^{2n_t}},
\end{equation}
for the entire time interval.

A similar result can be obtained for the approximation of just the integral by the discretized integral.
\begin{equation}
\left\|  H(\tau) + (s-\tau) H'(\tau)  -H(s) \right\| \le \frac {(s-\tau)^2} 2 \max_{v\in [\tau-\delta/2,\tau+\delta/2]} \|H''(v)\| .
\end{equation}
The bound on the difference between the integral over time $\delta$ and the value of $H$ is therefore
\begin{align}
\left\| H(\tau) \delta - \int_{\tau-\delta/2}^{\tau+\delta/2} H(s) \, ds \right\| &\le \int_{\tau-\delta/2}^{\tau+\delta/2}\frac {(s-\tau)^2} 2 4\|A\|^2 \|B\| \, ds \nn
&= \frac{\delta^3}{24} 4\|A\|^2 \|B\| = \frac{\delta^3}{6} \|A\|^2 \|B\|.
\end{align}
In taking the time-ordered exponential of the integral, we cannot simply use this bound because the term $(s-\tau) H'(\tau)$ does not average out, and there are contributions from higher-order terms in the exponential.
The above derivation is able to take account of that.

%%%%%%%%%%%%%%%%%%%%%%%
\section{Interaction picture algorithm for simulating generic Hamiltonians}
\label{app:interaction_picture}

We have given an implementation of the interaction picture algorithm in \sec{interaction_picture} for simulating first quantized quantum chemistry. In this appendix, we describe a related implementation for simulating a generic Hamiltonian. 

As in \sec{interaction_picture}, we assume that the target Hamiltonian takes the form $H=A+B$, where $\norm{B}$ is much smaller than $\norm{A}$. Then, the evolution under $H$ for a short time $\tau$ can be represented by the Dyson series as
\begin{align}
e^{-\mi(A+B)\tau} &= e^{-\mi\tau A} \mathcal{T} \exp \left( -\mi \int_0^\tau ds \, e^{isA} B e^{-\mi sA}\right) \\
&= \sum_{k=0}^\infty \frac{(-\mi)^k}{k!} \int_{0}^\tau d\tau_1 \int_{0}^{\tau} d\tau_2 \cdots \int_{0}^\tau d\tau_k \,  e^{-\mi(\tau-\tau'_k)A} B e^{-\mi(\tau'_k-\tau'_{k-1})A} B \ldots B e^{-\mi(\tau'_2-\tau'_1)A} B e^{-\mi\tau_1 A} \nn
&= \lim_{\substack{K\rightarrow\infty\\ M\rightarrow \infty}}\sum_{k=0}^K \frac{(-\mi\tau)^k}{M^k k!} \sum_{m_1=0}^{M-1} \sum_{m_2=0}^{M-1} \cdots \sum_{m_k=0}^{M-1} e^{-\mi\tau(M-1/2-m'_k)A/M} B e^{-\mi\tau(m'_k-m'_{k-1})A/M} B \ldots \nn& \quad \ldots B e^{-\mi\tau(m'_2-m'_1)A/M} B e^{-\mi\tau(m'_1+1/2) A/M} \nonumber
\end{align}
where $\tau'_1,\ldots,\tau'_k$ are sorted times from $\tau_1,\ldots,\tau_k$, and $m'_1,\ldots,m'_k$ are sorted integers from $m_1,\ldots,m_k$. Here, we assume that the operator $B$ is given by the block encoding
\begin{equation}
    \bra{0}\PREP_B^\dagger\cdot \SEL_B\cdot\PREP_B\ket{0}=B/\lambda_B,
\end{equation}
and $A$ can be directly exponentiated on a quantum computer. This gives the linear combination of unitaries
\begin{align}
&e^{-\mi(A+B)\tau} \approx
\left(\bra{0}\PREP_B^\dagger\right)^{\otimes K}\sum_{k=0}^K \frac{(\lambda_B\tau)^k}{M^k k!} \sum_{m_1,...,m_k=0}^{M-1} \bigg(e^{-\mi\tau(M-1/2-m'_k)A/M} \left(-\mi\SEL_B\right) e^{-\mi\tau(m'_k-m'_{k-1})A/M}  \nn& \quad \times \left(-\mi\SEL_B\right)\ldots\ldots \left(-\mi\SEL_B\right) e^{-\mi\tau(m'_2-m'_1)A/M} \left(-\mi\SEL_B\right) e^{-\mi\tau(m'_1+1/2) A/M}\bigg)\left(\PREP_B\ket{0}\right)^{\otimes K},
\end{align}
where $\PREP_B$ prepares initial states in $K$ different ancilla registers and, for each fixed value of $k$, the $k$ selection operators $\SEL_B$ act on $k$ different registers and the target system register.

To implement this on a quantum computer, we first prepare the state that encodes the coefficients from the Dyson-series expansion
\begin{equation}
    \frac{1}{\sqrt{\beta}}\sum_{k=0}^{K}\sqrt{\frac{(\lambda_B \tau)^k}{k!}}\ket{k},
\end{equation}
where $\beta$ is the normalization factor and $\ket{k}$ is represented in unary as $\ket{1^k0^{K-k}}$. We would like to prepare this state in a similar way as in \sec{dyscost}. To this end, we assume that the time interval $\tau$ can be taken as a rational multiple of $1/\lambda_B$:
\begin{equation}
\label{eq:requirement_reps}
    \tau=\frac{t}{\reps}=\frac{c}{d\lambda_B},
\end{equation}
where $\reps$ is the number of simulation steps and $c$ and $d$ are positive integers that are relatively prime. 
We can see that the squared amplitudes will be integers if we multiply by $\Sigma(0)$, where we redefine $\Sigma(k)$ as
\begin{equation}
     \Sigma(k)=\sum_{\ell=k}^{K}\frac{K!c^\ell d^{K-\ell}}{\ell!}.
\end{equation}
This means that we can perform state preparation using the same inequality testing procedure as in \sec{dyscost} except replacing the definition of $\Sigma(k)$.
One minor change we need to make is that the first two amplitudes ($k=0$ and 1) are no longer equal.
This means we need to perform the inequality test with the number $\Sigma(k+1)$ for $0\le k < K$, with complexity $\ceil{\log\left(\Sigma(k)\right)}-1$.
For $0<k<K$ we also need a Toffoli for the conditioning on the result of the previous inequality test.
Combining this with the cost in \eq{nkprep}, the total Toffoli complexity for this state preparation (including inverse preparation) is
\begin{equation}\label{eq:kprcost}
    2\left(3n_k+2b_r-9\right)
    -1+\sum_{k=0}^{K-1}\ceil{\log\left(\Sigma(k)\right)}.
\end{equation}
where $n_k=\ceil{\log\left(\Sigma(0)\right)}$ is the number of qubits for the equal superposition state.

Next, we consider the preparation of the register:
\begin{align}
    &\frac{1}{\sqrt{\beta}}\sum_{k=0}^{K}\sqrt{\frac{(\lambda_B\tau)^k}{M^kk!}}\ket{k}\sum_{m_1,...,m_k=0}^{M-1}\ket{m_1,...,m_k,M-1,...,M-1}\\
    \mapsto&\frac{1}{\sqrt{\beta}}\sum_{k=0}^{K}\sqrt{\frac{(\lambda_B\tau)^k}{M^kk!}}\ket{k}\sum_{m_1,...,m_k=0}^{M-1}\ket{m_1',...,m_k',M-1,...,M-1}\nonumber\\
    \mapsto&\frac{1}{\sqrt{\beta}}\sum_{k=0}^{K}\sqrt{\frac{(\lambda_B\tau)^k}{M^kk!}}\ket{k}\sum_{m_1,...,m_k=0}^{M-1}\ket{m_1',m_2'-m_1',...,m_k'-m_{k-1}',M-1-m_k',0,...,0}.\nonumber
\end{align}
Specifically, we prepare $k$ equal superposition states for each order $k$ of the Dyson series using Hadamard gates controlled by the $\ket{k}$ register. We initialize the remaining $K-k$ ancilla registers to state $\ket{M-1}$ corresponding to the maximal discrete time. This state preparation has Toffoli complexity $Kn_t$ and will be inverted with the same cost. We then quantumly sort the $K$ registers with cost $2n_t\srt{K}$ and compute the time differences with cost $(K-1)(n_t-1)$. The same complexity will be needed to invert the sorting while no Toffoli gate is required to uncompute the time differences.

At the beginning and the end we also need to apply evolution by time $\tau/2M$.
There are $(K+1)n_t$ bits controlling evolution under $A$.
For simplicity we will combine these bits with the two uncontrolled evolutions, giving a total of $(K+1)n_t+2$.
Therefore we implement exponentials of $A$ controlled by the time registers and perform $\PREP_B$ and $\SEL_B$ to block-encode $B$ with cost
\begin{equation}\label{eq:contrcost}
    [(K+1)n_t+2]\ \textrm{controlled-}e^{isA}+2K\ \PREP_B+K\ \SEL_B.
\end{equation}
Note that this only represents a upper bound on the cost and can often be improved by using specific properties of the target Hamiltonian; see \sec{interaction_picture} for improvements for performing the first quantized quantum chemistry simulation.
This implements the interaction picture algorithm for time $\tau$. 

The success amplitude of this procedure can be made larger than $1/2$. To make the amplitude equal to $1/2$, we introduce an ancilla register with state
\begin{equation}
\label{eq:theta_state}
    \cos\theta\ket{0}+\sin\theta\ket{1},
\end{equation}
so that the success amplitude is now multiplied by $\cos\theta$. Letting $n_\theta$ be the number of bits used to represent $\theta$, we can then prepare this state with $n_\theta-3$ Toffoli gates.

To boost the success amplitude close to unity, we need to perform a single step of oblivious amplitude amplification, which triples the implementation of the truncated Dyson series and introduces two additional reflections. Specifically,
\begin{enumerate}
    \item We implement the pre-amplified truncated Dyson series and prepare the $\theta$ state with a complexity as described above.
    \item We perform a reflection on $n_k + K \left( n_t + n_B \right)+2$ qubits.
    \item We invert the truncated Dyson series and unprepare the $\theta$ state.
    \item We perform a reflection again on $n_k + K \left( n_t + n_B \right)+2$ qubits.
    \item We implement the truncated Dyson series without preparing the $\theta$ state.
\end{enumerate}
The $+2$ accounts for the qubit rotated by $\theta$ here, as well as a qubit rotated for the preparation of the equal superposition state for preparing $k$.
Each reflection can be implemented with Toffoli complexity
\begin{equation}
\label{eq:ref_cost}
        n_k + K \left( n_t + n_B \right),
\end{equation}
where $n_B$ is the number of ancilla qubits for block-encoding operator $B$.

We now analyze the normalization factor introduced by the above steps. Such a factor will come from three sources: (i) the preparation of the Dyson-expansion coefficients; (ii) the preparation of the equal superposition state on the $\ket{k}$ register; and (iii) the implementation of the block encoding of operator $B$. Same as in \sec{dyscost}, we have that the success probability of Item (ii) is $p=\eqprep{\Sigma(0),b_r}$, whereas Items (i) and (iii) together contribute a normalization factor of $\beta=\sum_{k=0}^{K}c^k/(d^kk!)\leq\exp(c/d)$. To boost this close to unity with a single step of oblivious amplitude amplification, we choose coprime integers $c$ and $d$ so that
\begin{equation}
\label{eq:requirement_cd}
    \frac{\eqprep{\Sigma(0),b_r}}{\exp(c/d)}\geq\frac{1}{2},
\end{equation}
which implies
\begin{equation}\label{eq:intreps2}
    \reps =\frac{\lambda_B td}{c}.
\end{equation}

The error of our circuit implementation is due to the truncation of the Dyson-series expansion ($\epsilon_K$) and the finite binary representation of the time steps ($\epsilon_t$), as well as the preparation of $\theta$ state ($\epsilon_\theta$). With $n_\theta$ digits to represent the angle $\theta$, the error in the amplitude $\cos\theta$ is upper bounded by
\begin{equation}\label{eq:intphaer}
    \epsilon_\theta = \frac{\pi}{2^{n_\theta}},
\end{equation}
which causes an overall error of at most $\epsilon_K+\epsilon_\theta$ in the pre-amplified truncated Dyson series. After a single step of oblivious amplitude amplification, the error is bounded as \cite[Lemma G.2]{Childs2017}
\begin{equation}
    (\epsilon_K+\epsilon_\theta)\frac{(\epsilon_K+\epsilon_\theta)^2+3(\epsilon_K+\epsilon_\theta)+4}{2}.
\end{equation}
To ensure that the entire simulation has accuracy $\epsilon$, we require that
\begin{equation}\label{eq:totsimac}
    \reps(\epsilon_K+\epsilon_\theta)\frac{(\epsilon_K+\epsilon_\theta)^2+3(\epsilon_K+\epsilon_\theta)+4}{2}+\reps\epsilon_t\leq\epsilon.
\end{equation}
We summarize the cost of the circuit implementation in the following theorem.

\begin{theorem}
Using the interaction picture algorithm, the time evolution under the Hamiltonian $H=A+B$ can be simulated for time $t$ that is a rational multiple of $1/\lambda_B$ to within error $\epsilon$ by choosing positive integers $K$, $n_t$, $n_\theta$, $b_r$, coprime $c$ and $d$, and
using $3[\reps(K+1)n_t+2]\ \textrm{controlled-}e^{isA}$, $6\reps K\ \PREP_B$, $3\reps K\ \SEL_B$, and a number of Toffoli gates
\begin{align}\label{eq:intpiccost}
    &6\reps(3n_k +2b_r - 9)+3\reps\left(-1+ \sum_{k=0}^{K-1} \lceil \log( \sig (k) ) \rceil\right)+6\reps Kn_t+12\reps(n_t-1)\srt{K}\nonumber\\
    &+6\reps (K-1)(n_t-1)+2\reps(n_\theta-3)+2\reps\left[n_k + K (n_t + n_B )\right],
\end{align}
where $\srt{K}$ is the cost of sorting $K$ items, $\tau = c/(d\lambda_B)$, $\reps=t/\tau$ is an integer, and
\begin{align}\label{eq:intpierr1}
    \frac{\eqprep{\Sigma(0),b_r}}{\exp(c/d)}\geq\frac{1}{2},& \\
    \label{eq:intpierr2}
    \reps(\epsilon_K+\epsilon_\theta)\frac{(\epsilon_K+\epsilon_\theta)^2+3(\epsilon_K+\epsilon_\theta)+4}{2}+\reps\epsilon_t \leq\epsilon,&\\
    \label{eq:intpierr3}
    \epsilon_K=\exp(c/d)-\sum_{k=0}^{K}\frac{c^k}{d^kk!} ,\qquad\qquad
    \epsilon_\theta = \frac{\pi}{2^{n_\theta}}&,\\
    \label{eq:intpierr4}
    \epsilon_t = \left(\frac{2\|A\|}{\lambda_B} + \frac{\|A\|^2}{\lambda_B^2}\right)  \left[ 2^{n_t} \left( e^{\lambda_B\tau/2^{n_t}}-1 \right) - \lambda_B\tau (1+\lambda_B\tau/2^{n_t+1}) \right].&
\end{align}
\end{theorem}

\begin{proof}
The complete interaction picture simulation is indicated in \fig{interact}, with the interaction picture operator to be amplified shown in \fig{preamp} and the clock preparation shown in \fig{clock}.
The total costs for each of the parts of the simulation are shown in \tab{intcosts_gen}.
The sum of the costs in entries 1 to 7 gives the cost in \eq{intpiccost}.
The costs in terms of the evolution under $A$ and the block encoding of $B$ are given in entries 8 and 9 of \tab{intcosts_gen}.

\begin{table}
\begin{tabular}{| m{11cm} | m{5cm} |}
\hline
\centering Procedure & \hspace{2cm} Cost \\
 \hline \hline
The preparation of the equal superposition over $\sig (0)$ numbers; see \eq{nkprep}. & $6\reps(3n_k +2b_r - 9)$ \\
\gline
Inequality tests used for the preparation of the superposition over $k$; see \eq{kprcost}. & $3\reps\left(-1+ \sum_{k=0}^{K-1} \lceil \log( \sig (k) ) \rceil\right) $ \\
\gline
Controlled Hadamards for preparing the superpositions over times; see \eq{ctrlhadcosts}. & $6\reps Kn_t$ \\
\gline
The sorting of the time registers; see \eq{srtcosts}. & $12\reps n_t\srt{K}$ \\
\gline
The differences of the times; see \eq{timesubcost}. & $6\reps (K-1)(n_t-1)$ \\
\gline
The preparation of $\theta$ state; see \eq{theta_state}. & $2\reps(n_\theta-3)$\\
\gline
The reflection cost for oblivious amplitude amplification; see \eq{ref_cost}. & $2\reps\left[n_k + K (n_t + n_B )\right]$ \\
\gline
The exponential of $A$ needs to be performed $K+1$ times; see \eq{contrcost}. & $3[\reps(K+1)n_t+2]\ \textrm{controlled-}e^{isA}$ \\
\gline
Perform the controlled block encoding of operator $B$ $K$ times; see \eq{contrcost}. & $6\reps K\ \PREP_B+3\reps K\ \SEL_B$ \\
\hline
\end{tabular}
\caption{The costs of the interaction picture algorithm for time evolution of a general Hamiltonian.
The costs are Toffoli costs unless otherwise specified (the last two entries).}
\label{tab:intcosts_gen}
\end{table}

Provided $\lambda_Bt$ is rational, for any integer $\reps$, $\lambda_Bt/\reps$ must be rational, and we can choose $c$ and $d$ to satisfy $c/d=\lambda_B\tau$ for $\tau=t/\reps$.
Then \eq{intpierr1} comes from \eq{requirement_cd}, and constrains the value of $c/d$.
Since $c/d$ is set by $\reps$, we need a way to select $\reps$ so that $c/d$ satisfies \eq{intpierr1}. 
Choosing some value of $b_r$ we can find a lower bound $P_{\min}$ on $\eqprep{\Sigma(0),b_r}$ (for example as given in \app{sucbnd} below).
Then we can set $\reps = \lceil\lambda_Bt/\ln (2P_{\min})\rceil$, and $\lambda_B\tau=\lambda_Bt/\reps$ is strictly less than $\ln (2P_{\min})$.
Then $c/d=\lambda_B\tau<\ln (2P_{\min})$ implies that \eq{intpierr1} is satisfied.

The error bound in \eq{intpierr2} is from \eq{totsimac}.
The quantities $\epsilon_K$, $\epsilon_\theta$ and $\epsilon_t$ are bounds on the error due to truncating the Dyson series, the rotation on the ancilla qubit for the amplitude amplification, and the time discretisation, respectively.
The norm of each successive term in the Dyson series is bounded as $(\tau\lambda_B)^k/k!$, and $\lambda_B\tau=c/d$.
Therefore the error in truncating the Dyson series to order $K$ can be bounded by the error in truncating the exponential of $c/d$, giving the expression for $\epsilon_K$ in \eq{intpierr3}.
The expression for $\epsilon_\theta$ in \eq{intpierr3} is from \eq{intphaer}.
The bound in \eq{intpierr4} is from \eq{timediscer}, where we have used $\lambda_B$ as the upper bound on $\|B\|$.
\end{proof}

The algorithm will be most efficient when $c/d$ is a close rational approximation of $\ln 2$, so $\reps$ is minimized.
This is because $c/d$ is upper bounded by $\ln[2\eqprep{\Sigma(0),b_r}]$ according to \eq{intpierr1}, $\eqprep{\Sigma(0),b_r}$ is close to 1, and $\reps=\lambda_B t d/c$.
It would be expected that for a real application one can choose the exact time to be simulated, so that $c$ and $d$ are relatively small (for example $9/13$ is less than but close to $\ln 2$).
If evolution for a specific $t$ is needed where $\lambda_Bt$ is not rational, we prepare an ancilla register in the state
\begin{equation}
    \cos\theta_B\ket{0}+\sin\theta_B\ket{1}
\end{equation}
and use the amplitude $\cos\theta_B$ to artificially adjust the normalization factor $\lambda_B$. This completes our costing of the interaction picture algorithm for generic Hamiltonians.

\begin{figure}[!htbp]
    \centering
      \begin{tabular}{c @{\qquad} c }
\includegraphics[width=0.422\textwidth]{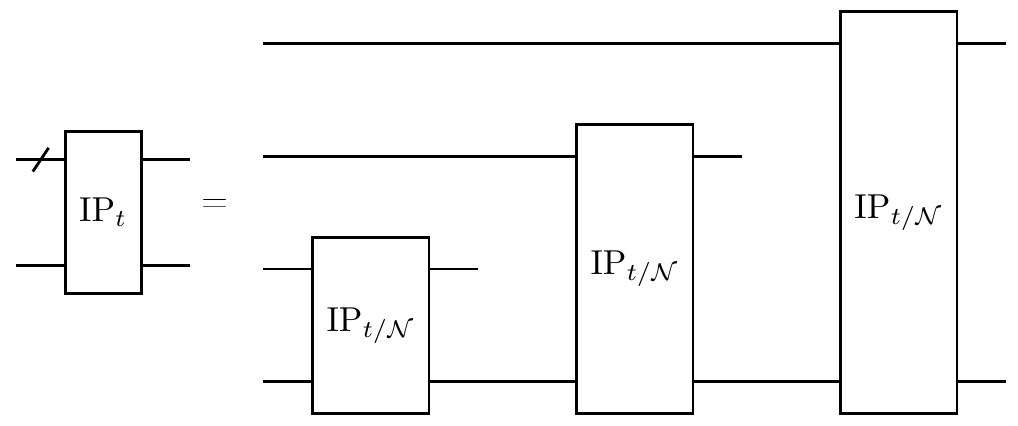} 
    &
\includegraphics[width=0.538\textwidth]{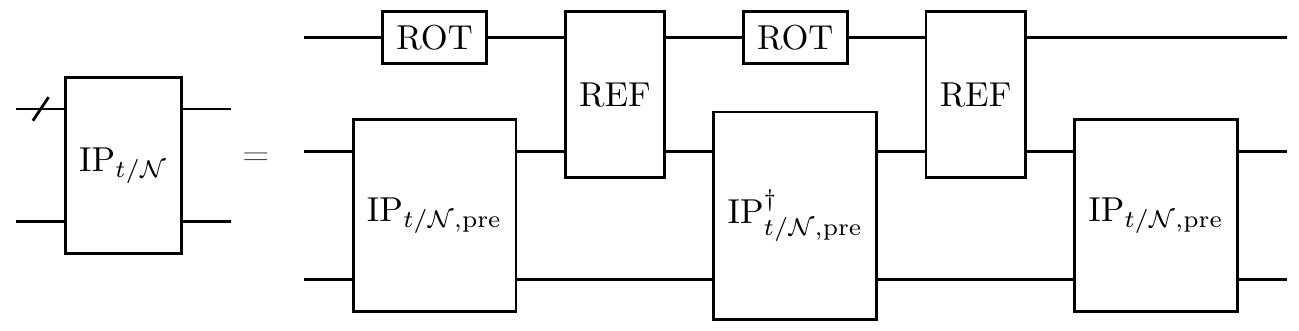} 
        \\
        \small (a) Interaction picture algorithm with 3 steps. & \small (b) Oblivious amplitude amplification.
      \end{tabular}
    \caption{Quantum circuits for the interaction picture algorithm for simulating generic Hamiltonians.}
    \label{fig:interact}
\end{figure}

\begin{figure}[!htbp]
    \centering
    \resizebox{.95\textwidth}{!}{%
        \includegraphics[width=\textwidth]{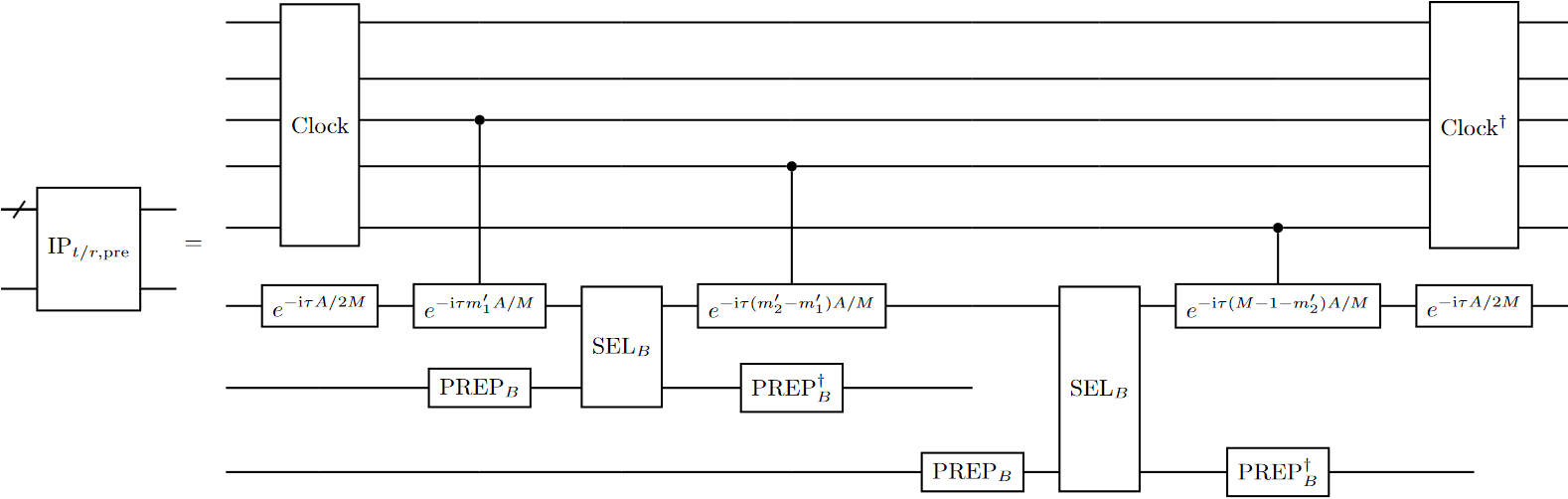}
        }
    \caption{Circuit implementation of the pre-amplified operator ($K=2$).}
    \label{fig:preamp}
\end{figure}

\begin{figure}[!htbp]
    \centering
    \resizebox{.95\textwidth}{!}{%
    \includegraphics[width=\textwidth]{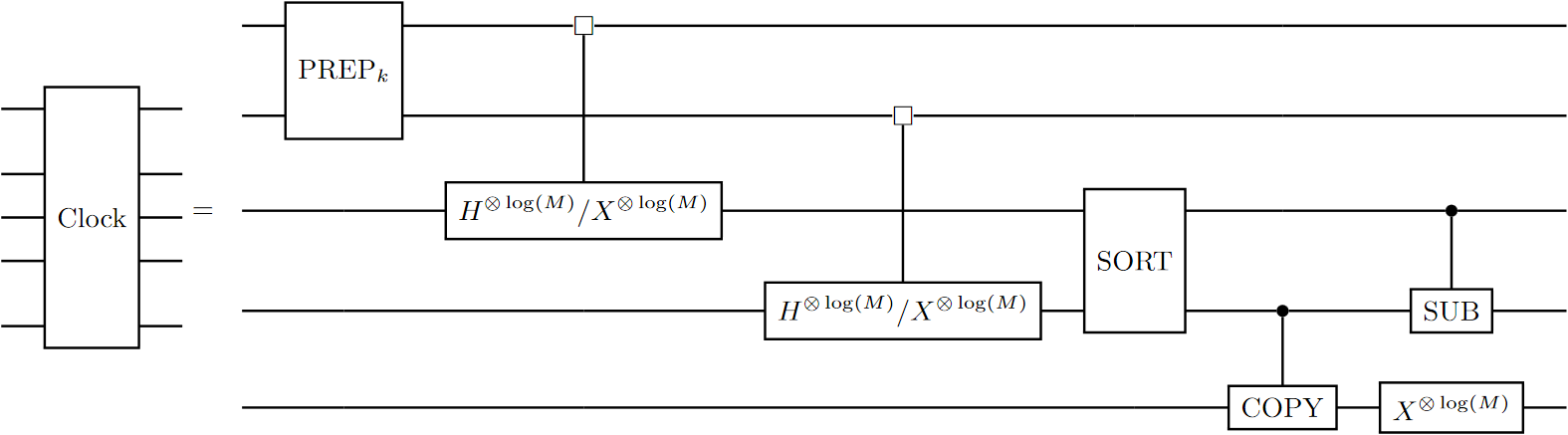}
%    \begin{quantikz}
%\qw & \gate[wires=5]{\mathrm{Clock}} & \qw \\
%\qw &                                & \qw \\
%\qw &                                & \qw \\
%\qw &                                & \qw \\
%\qw &                                & \qw
%    \end{quantikz}
%    =
%    \begin{quantikz}
%\qw & \gate[wires=2]{\mathrm{PREP}_k} & \square\qw\vqw{2}                                             & \qw                                                           & \qw                           & \qw                 & \qw                 & \qw \\
%\qw &                                 & \qw                                                           & \square\qw\vqw{2}                                             & \qw                           & \qw                 & \qw                 & \qw \\
%\qw & \qw                             & \gate{H^{\otimes \log(M)} / X^{\otimes\log(M)}} & \qw                                                           & \gate[wires=2]{\mathrm{SORT}} & \qw                 & \ctrl{1}            & \qw \\
%\qw & \qw                             & \qw                                                           & \gate{H^{\otimes \log(M)} / X^{\otimes\log(M)}} &                               & \ctrl{1}            & \gate{\mathrm{SUB}} & \qw \\
%\qw & \qw                             & \qw                                                           & \qw                              & \qw                              & \gate{\mathrm{COPY}} & \gate{X^{\otimes\log(M)}}                 & \qw
%    \end{quantikz}%
}
    \caption{Quantum circuit for preparing the clock state ($K=2$). Note that the last time difference $M-1-m_2'$ can be computed with a copying operation followed by flipping the qubits.}
    \label{fig:clock}
\end{figure}

\section{Bound on the success probability for preparing equal superposition}
\label{app:sucbnd}

Here we consider the probability of success for preparation of an equal superposition state and place a lower bound on it for given $b_r$.
The probability for success was given in \cite{Sanders2020_b}, and is
    \begin{equation}
\eqprep{n,b_r} = \frac{n}{2^{\lceil\log n\rceil}} \left[ \left( 1+ \left(2-\frac{4n}{2^{\lceil\log n\rceil}}\right)\sin^2\theta(n,b_r)\right)^2 + \sin^2(2\theta(n,b_r)) \right],
    \end{equation}
    where $\theta(n,b_r)$ is the angle of rotation on an ancilla qubit.
It turns out that there is a simple lower bound for the success probability, though the expression for $\theta$ is adjusted slightly from the formula used before. The result is as in the following theorem.
\begin{theorem}
For all $n,b_r\in \mathbb{N}$, taking
\begin{align}
x&= \frac n {2^{\lceil\log n\rceil}}, \\
\theta_0 &= \arcsin\left(\frac 1{2\sqrt{x}}\right),\\
\theta(n,b_r) &= \frac {2\pi}{2^{b_r}} {\rm round} \left[ \frac {2\pi}{2^{b_r}} \left( \theta_0-\frac{(\pi/2^{b_r})^2 (2 x-1)}{\sqrt{4 x-1}} \right)\right]
\end{align}
gives the success probability for preparation of an equal superposition state lower bounded as
\begin{equation}\label{eq:psucbnd}
\eqprep{n,b_r} \ge 1-\left( \frac 32 \right)^2 \left( \frac \pi {2^{b_r}} \right)^2 .
\end{equation}
\end{theorem}

\begin{proof}
Now with $b_r=1$ or $b_r=2$ the result is trivial, because the right-hand side is negative.
We therefore aim to prove the result for $b_r \ge 3$.
Because the expression for $\eqprep{n,b_r}$ depends only on the ratio $x=n/2^{\lceil \log n \rceil}$, we can prove the lower bound for all values of $x$ in the range $x\in(1/2,1)$, and that will give us the result for all $n \in \mathbb{N}$.
Then the probability of success can be written as
\begin{equation}
\eqprep{n,b_r} = x \left[ \left( 1+ \left(2-4x\right)\sin^2\theta\right)^2 + \sin^2(2\theta) \right].
\end{equation}

Let us take $\Delta\theta = \theta-\theta_0$.
Using elementary trigonometric identities, the probability of success can be rewritten as
\begin{align}
\eqprep{n,b_r} &=1-4(1-x)(4x-1)\sin^2(\Delta\theta) \nn & \quad 
+ 4(1-x)\left[ 2(4x-1-2x^2)\sin^2(\Delta\theta) - (2x-1)\sqrt{4x-1}\sin(2\Delta\theta) \right] \sin^2(\Delta\theta).
\end{align}
To prove lower bounds for $x\in(1/2,1)$, we use
\begin{align}
(2x-1)(3-x) &\ge 0 \nn
(2x-1)(3-x)+2 &> 0 \nn
4(1-x)(2x-1)[(2x-1)(3-x)+2] &\ge 0 \nn
4(1-x)(2x-1)[(2x-1)(3-x)+2]+1 &> 0 \nn
\left[ 2(4x-1-2x^2) \right]^2 - \left[ (2x-1)\sqrt{4x-1} \right]^2 &> 0.
\end{align}
Next, $4x-1-2x^2=(2 x - 1) (3/2 - x) + 1/2$ is positive for $x\in(1/2,1)$, and $2x-1$ and $4x-1$ are non-negative.
As a result, we have
\begin{equation}
2(4x-1-2x^2) > (2x-1)\sqrt{4x-1}.
\end{equation}

We also find that
\begin{equation}
\sin^2\Delta\theta - \left[ \sin(2\Delta\theta)- 2\Delta\theta \right]= (\Delta\theta)^2 + \mathcal{O}((\Delta\theta)^3),
\end{equation}
so it is non-negative for small $\Delta\theta$.
Moreover this expression has zeros of its derivative for $n\pi$ and $n\pi-\arctan(1/2)$.
This means the only zeros of its derivative for $\Delta\theta\in[-0.7,0.7]$ are at 0 and $-\arctan(1/2)\approx -0.46$.
This means that the function $\sin^2\Delta\theta - \left[ \sin(2\Delta\theta)- 2\Delta\theta \right]$ can turn around towards negative values for $\Delta\theta<-0.46$, but computing it at $\Delta\theta=0.7$ gives about $0.000466$, which is still positive.
This means that for $\Delta\theta\in[-0.7,0.7]$ we have
\begin{equation}
\sin^2\Delta\theta - \left[ \sin(2\Delta\theta)- 2\Delta\theta \right]\ge 0,
\end{equation}
because negative values would require an additional turning point.
That implies
\begin{equation}
2(4x-1-2x^2)\sin^2(\Delta\theta) \ge (2x-1)\sqrt{4x-1}\left[ \sin(2\Delta\theta)- 2\Delta\theta \right],
\end{equation}
which in turn implies
\begin{equation}
2(4x-1-2x^2)\sin^2(\Delta\theta) - (2x-1)\sqrt{4x-1}\sin(2\Delta\theta) \ge - 2(2x-1)\sqrt{4x-1} \Delta\theta 
\end{equation}
for $x\in(1/2,1)$ and $\Delta\theta\in[-0.7,0.7]$.
This means that the probability is lower bounded by
\begin{align}
\eqprep{n,b_r} &\ge 1-4(1-x)(4x-1)\sin^2(\Delta\theta)
- 8(1-x)(2x-1)\sqrt{4x-1} \Delta\theta \, \sin^2(\Delta\theta) \nn
&= 1-4(1-x)\left[ (4x-1) +2(2x-1)\sqrt{4x-1} \Delta\theta \right]\sin^2(\Delta\theta).
\end{align}

Next, we have for $x\in(1/2,1)$ and $\Delta\theta\in[-0.7,0.7]$
\begin{align}
4(1-x)(3x-1) &\ge 0 \nn
4x-1-3(2x-1)^2&\ge 0 \nn
\sqrt{4x-1} &\ge \sqrt{3}(2x-1) \nn
4x-1 &\ge \sqrt{3}(2x-1)\sqrt{4x-1} \nn
4x-1 &\ge -2(2x-1)\sqrt{4x-1}\Delta\theta .
\end{align}
As a result, we can lower bound the success probability by
\begin{equation}
\eqprep{n,b_r} \ge 1-4(1-x)\left[ (4x-1) +2(2x-1)\sqrt{4x-1} \Delta\theta \right](\Delta\theta)^2.
\end{equation}
Note also that the derivative of this expression with respect to $\Delta\theta$ is
\begin{equation}
-8(1-x)\left[ (4x-1) +3(2x-1)\sqrt{4x-1} \Delta\theta \right]\Delta\theta.
\end{equation}
Now $4x-1 \ge \sqrt{3}(2x-1)\sqrt{4x-1}$ implies the expression in square brackets must be non-negative for $\Delta\theta\in[-1/\sqrt 3,1/\sqrt 3]$.
This implies that, provided we are considering $\Delta\theta\in[-1/\sqrt 3,1/\sqrt 3]$, we can lower bound the probability of success in any region by the probabilities of success at the bounds of the region.

Now, using the formula in the theorem for $\theta$, we have
\begin{equation}
\Delta\theta + \frac{(\pi/2^{b_r})^2 (2 x-1)}{\sqrt{4 x-1}} \in [-\pi/2^{b_r}, \pi/2^{b_r}].
\end{equation}
Note that $\sqrt{4x-1} \ge \sqrt{3}(2x-1)$ implies
\begin{equation}
\frac{2x-1}{\sqrt{4x-1}} \le \frac 1{\sqrt 3}.
\end{equation}
Therefore the maximal value of $\Delta\theta$ with ${b_r}\ge 3$ is
\begin{equation}
\frac{(\pi/8)^2}{\sqrt{3}}+ \pi/2^8 \approx 0.481734 < \frac 1{\sqrt{3}}.
\end{equation}
Because $\Delta\theta\in[-1/\sqrt 3,1/\sqrt 3]$, we can lower bound the probability by the lower bounds at the extremal points in the range of $\Delta\theta$, which give
\begin{align}\label{eq:sixbnd}
\eqprep{n,b_r} &\ge
1-4(1-x)(4x-1)\left(\frac\pi{2^{b_r}}\right)^2 \nn
&\quad + (2x-1)^2(1-x)\left(\frac\pi{2^{b_r}}\right)^4
 \pm \frac {24(2x-1)^3(1-x)}{\sqrt{4x-1}} \left(\frac\pi{2^{b_r}}\right)^5
 + \frac{8(2x-1)^4(1-x)}{4x-1} \left(\frac\pi{2^{b_r}}\right)^6 .
\end{align}
Next, we have that $\sqrt{4x-1} \ge \sqrt{3}(2x-1)$ implies
\begin{align}
20(2x-1)^2(1-x)\sqrt{4x-1} & \ge 20\sqrt{3}(2x-1)^3(1-x) \nn
20(2x-1)^2(1-x) & \ge \frac{24(2x-1)^3(1-x)}{\sqrt{4x-1}} \nn
(2x-1)^2(1-x)\left(\frac\pi{2^{b_r}}\right)^4
&\ge \frac {24(2x-1)^3(1-x)}{\sqrt{4x-1}} \left(\frac\pi{2^{b_r}}\right)^5
\end{align}
because $\pi/2^{b_r}<1$.
Hence we have
\begin{equation}
(2x-1)^2(1-x)\left(\frac\pi{2^{b_r}}\right)^4
 \pm \frac {24(2x-1)^3(1-x)}{\sqrt{4x-1}} \left(\frac\pi{2^{b_r}}\right)^5 \ge 0,
\end{equation}
and the sixth-order term in \eq{sixbnd} is non-negative, so we obtain the lower bound
\begin{equation}
\eqprep{n,b_r} \ge 1-4(1-x)(4x-1)\left(\frac\pi{2^{b_r}}\right)^2.
\end{equation}
Now $(1 - x) (4 x - 1)$ has its maximal value of $9/16$ for $x=5/8$, so the lower bound becomes
\begin{equation}
\eqprep{n,b_r} \ge 1-\left(\frac 32 \right)^2\left(\frac\pi{2^{b_r}}\right)^2
\end{equation}
as required.
\end{proof}

If we take $b_r=8$ for example, then the probability of success must be at least $0.999661$.
That is useful for the general interaction picture simulations described in \app{interaction_picture}, because then $\ln[2\eqprep{\Sigma(0),8}]> 0.6928$.
That is larger than $9/13$, so one can use $c/d=9/13$ as a good approximation of $\ln 2$ and still satisfy \eq{intpierr1}.

\section{Simulating first quantized real space grid representations}
\label{app:real_space}

Whether the goal is to perform non-Born-Oppenheimer dynamics or to solve the electronic structure problem, in order to simulate molecular or material systems on a quantum computer one must discretize the wavefunction and map it to qubits. Perhaps the most obvious representation involves projecting the wavefunction onto a real space grid and encoding the positions of each particle explicitly in a quantum register. In this case, the computational basis states that span the wavefunction of interest might be stored in the quantum computer as $\ket{r_1 r_2 \cdots r_\eta R_1 R_2 \cdots R_L}$ where the $r_i$ and $R_\ell$ index the coordinates of the $\eta$ electrons and $L$ nuclei, respectively, using binary registers with $N$ grid points (each requiring $\log N$ qubits). Thus, this encoding would require $(\eta + L)\log N$ qubits.

Representing the potential part of the Hamiltonian is straightforward when using a real space grid (the operators $U$, $V$ and $V_{\rm nuc}$ would take the same diagonal forms as they do in \eq{non-bo}). However, representing the kinetic operators will require some thought. One can see that if basis functions are not overlapping in space (e.g., delta functions or finite elements with disjoint support) the integrals from the Galerkin discretization would suggest that the kinetic operator is also diagonal. To avoid this, one idea would be to take a finite-difference approach and approximate the Laplacian using a $k$-point stencil; this is exactly the approach taken in first quantization in \cite{Kivlichan2016} and described in second quantization in Appendix A of \cite{BabbushLow}. Another idea, similar to using what is referred to as a ``discrete value representation'' (DVR) \cite{dvrreview}, would be to exactly define the kinetic operator in its eigenbasis and approximate a transformation that relates the grid representation to the eigenbasis of the kinetic operator. For example, if using a regularly spaced grid then one might make the approximation that the quantum Fourier transform (QFT) diagonalizes the kinetic operator. After all, the kinetic operator is diagonal in the momentum basis (plane waves) and plane waves are obtained as the continuous Fourier transform of delta functions. However, this is an approximation when using a finite grid; the QFT is a discrete Fourier transform (which is linear) and a linear combination of $N$ plane waves cannot sum to a delta function. This sort of representation is employed by the first paper to consider quantum simulations of molecular systems using a grid representation by Kassal \emph{et al.}~\cite{Kassal2008}. Components of that algorithm were assessed for fault-tolerance in the work of Jones \emph{et al.}~\cite{Jones2012} but that work stopped short of estimating all of the constant factors associated with the algorithm realization.

In addition to the necessity of the approximations already mentioned there are other unfavorable characteristics of real space grid representations for use in quantum simulation. For example, basis errors are no longer variationally bounded when deviating from a Galerkin discretization and extrapolations to the continuum limit might be less reliable. Furthermore, ways of softening the potential and removing core electrons (akin to pseudopotentials for plane wave methods) are less developed for grids. Thus, the majority of work on simulating molecular systems focuses on Galerkin discretizations involving basis functions rather than grids. Nonetheless, in this section we will briefly describe how one can block encode a real space representation of the Hamiltonian with the same asymptotic complexity as the main approaches pursued in this paper. Since these grid representations are not periodic, it is possible that they might converge faster than plane waves when simulating non-periodic systems. Also, since the block encoding algorithms are entirely different, it is possible that this approach might have lower constant factors.

We can express the aforementioned first quantized real space Hamiltonian as
\begin{align}
\label{eq:real_space}
H_{\rm BO} & = T + U + V + \frac{1}{2}\sum_{\ell \neq \kappa=1}^{L} \frac{\zeta_\ell \zeta_\kappa}{\left\|R_\ell - R_\kappa\right\|} 
        \qquad \qquad
    H_{\rm non-BO}  = T + T_{\rm nuc} + U_{\rm non-BO} + V + V_{\rm nuc}\\
T & = \sum_{i=1}^{\eta} {\rm QFT}_i \left(   \sum_{p\in G} \frac{\left \| k_p\right\|^2}{2} \ket{p}\!\bra{p}_{i} \right) {\rm QFT}_i^\dagger
\qquad 
T_{\rm nuc}  =  \sum_{\ell=\eta+1}^{L+\eta} {\rm QFT}_\ell \left(  \sum_{p\in G} \frac{\left \| k_p\right\|^2}{2\, m_\ell} \ket{p}\!\bra{p}_{\ell} \right) {\rm QFT}_\ell^\dagger\\\label{eq:KineticAppendix}
U & = -\sum_{i=1}^\eta\sum_{\ell =1}^{L}  \sum_{p\in G}\frac{\zeta_\ell}{\left\|R_\ell - r_p\right\|} \ket{p}\!\bra{p}_{i}
\qquad \quad
U_{\rm non-BO}  = -\sum_{i=1}^\eta \sum_{\ell=\eta+1}^{L+\eta} \sum_{p,q\in G}\frac{\zeta_\ell}{2\left\|r_p - r_q\right\|} \ket{p}\!\bra{p}_{i} \ket{q}\!\bra{q}_{\ell}\\
V & = \sum_{i\neq j=1}^\eta \sum_{p,q\in G}\frac{1}{2\left\|r_p - r_q\right\|} \ket{p}\!\bra{p}_{i} \ket{q}\!\bra{q}_{j}
\qquad 
V_{\rm nuc} = \sum_{\ell\neq \kappa=\eta+1}^{L+\eta} \sum_{p,q\in G}\frac{\zeta_\ell \zeta_\kappa}{2\left\|r_p - r_q\right\|} \ket{p}\!\bra{p}_{\ell} \ket{q}\!\bra{q}_{\kappa}
\end{align}
where ${\rm QFT}_i$ is the usual quantum Fourier transform applied to register $i$\footnote{Again, we emphasize that $T$ is only approximately given by the expression involving the QFT. This relation is exact in the continuum limit. For finite sized grids, it cannot be the case that the QFT completely diagonalizes the momentum operator.}. As in the main text, we have used atomic units where $\hbar$ as well as the mass and charge of the electron are unity and $\| \cdot \|$ denotes a function taking the 2-norm of its argument. In the above expression $\ell$ and $\kappa$ index nuclear degrees of freedom; thus, $R_\ell$ represents the positions of nuclei and $\zeta_\ell$ the atomic numbers of nuclei. Furthermore, we have the following definition of grid points and their frequencies in the dual space defined by the QFT:
\begin{equation}
\label{eq:r_p}
r_p = \frac{p \, \Omega^{1/3}}{N^{1/3}} \qquad \qquad k_p = \frac{2 \pi p}{\Omega^{1/3}} \qquad \qquad
p \in G \qquad \qquad G = \left[-\frac{N^{1/3}-1}{2},\frac{N^{1/3}-1}{2}\right]^3 \subset \mathbb{Z}^3 \, ,
\end{equation}
where $\Omega$ is the volume of our simulation cell and $N$ is the number of grid points in the cell. Again, this is the same representation used in the first work on quantum simulating chemistry in first quantization, by Kassal \emph{et al.}~\cite{Kassal2008}, well over a decade ago. We note that like with the plane waves, it would be possible to use different grids for the different particles but we avoid doing that here for simplicity.

We now show that one can simulate the Hamiltonian of \eq{real_space} with the same complexity as the approaches of this paper using either qubitization, or interaction picture simulation frameworks. Here we will discuss only the Born-Oppenheimer case for simplicity, but the techniques readily generalize to the non-Born-Oppenheimer Hamiltonian. Focusing first on the qubitization approach, the goal is to block encode the operators $T$, $U$ and $V$. The block encoding of $T$ can be accomplished using almost the same techniques used for plane waves. The only difference comes in the implementation of the $\SEL$ operator. Whenever the ancilla register flags that a term from $T$ should be applied, one must apply the QFT, then phase the register by the appropriate amount, then apply the inverse QFT.
The QFT needs to be performed independently on $\eta$ registers each of $\log N$ qubits, so has gate complexity $\widetilde{\mathcal{O}}(\eta \log N)$, where here the $\widetilde{\mathcal{O}}$ indicates terms double-logarithmic in $N$ are omitted~\cite{cleve2000fast}.
That is complexity in terms of one and two-qubit gates, and requires rotations up to $\mathcal{O}(\log N)$ bits of precision. The complexity in terms of Toffoli gates is $\widetilde{\mathcal{O}}(\eta \log^2 N)$.

Here the block encoding of $U$ and $V$ will deviate from the other strategies of this paper. Our approach will be to break up the terms as a sum of many terms with equal magnitude coefficients that can be computed on-the-fly and applied as phases, mirroring a strategy first introduced in \cite{Berry2013}.
First, we begin by re-expressing the $V$ operator as
\begin{align}
V & = \frac{v_{\rm max}}{M}\sum_{i\neq j=1}^\eta  \sum_{p,q\in G} \sum_{m=1}^M v\left(m,r_p,r_q\right) \ket{p}\!\bra{p}_{i} \ket{q}\!\bra{q}_{j}\\
v\left(m,r_p,r_q\right) & = \begin{cases}
1 & \qquad \frac{m\, v_{\rm max}}{M} < \frac{1}{2\left\|r_p - r_q\right\|}\\
\label{eq:inequality_V}
0 & \qquad {\rm otherwise} \end{cases}\\
 v_{\rm max} & = \textrm{max}_{p\neq q}\left( \frac{1}{2\left\|r_p - r_q\right\|}\right) = \frac{N^{1/3}}{2 \,\Omega^{1/3}} \, ,
\end{align}
where in the inequality above we take advantage of the fact that the repulsive Coulomb kernel is always greater than zero. From this definition we can see that
\begin{equation}
\left | \frac{1}{2\left\|r_p - r_q\right\|} - \frac{v_{\rm max}}{M} \sum_{m=1}^M  v\left(m,r_p,r_q\right) \right | \leq \frac{v_{\rm max}}{M} \, .
\end{equation}
Hence, we can make sure that the error in the Hamiltonian coefficients is less than $\epsilon$ if $M > v_{\rm max} / \epsilon$.
This error can be suppressed exponentially since the cost of the comparison test is quadratic in the number of qubits used to represent the number (because we need to compuare squares)~\cite{Sanders2020_b}, which is logarithmic in $M$, and thus logarithmic $1/\epsilon$.

We can now block encode $V$ in the following way. We will start by applying $\PREP$ as a series of Hadamard gates which realize the equal superposition state over the ancilla registers $\ket{i,j}\ket{m}$ (essentially all of the complexity of our approach will enter through $\SEL$). Then, controlled on register $\ket{i}$ we will copy the momentum held in the $i^{\rm th}$ system register to an ancilla register and translate it into $r_p$ using the relation in \eq{r_p}. We will obtain $r_q$ in the same way, but this time controlled on the register $\ket{j}$. Then, we will compute $v(m,r_p,r_q)$ and record its value in an ancilla register. Because this ancilla register is initialized to $\ket{0}$ by $\PREP$, if it is changed to the $\ket{1}$ state by $\SEL$ then it will be removed from the linear combination of unitaries. Indeed, we wish to remove the term from the linear combination of unitaries whenever the inequality in \eq{inequality_V} is false. Therefore, the final step of our procedure should be to apply a bit flip operator to the ancilla holding $v(m,r_p,r_q)$ and to uncompute the $r_p$ and $r_q$ ancillae. Thus, $\SEL$ will act with the following steps (with $p_i$ and $p_j$ equivalent to $p$ and $q$, respectively):
\begin{align}
\SEL \ket{i,j}\ket{m}\ket{0}^{\otimes {\cal O}(\log \left(N/\epsilon\right)} \underbrace{\ket{p_0 \cdots p_i \cdots p_j \cdots p_\eta}}_{\ket{\psi}} & \mapsto \ket{i,j}\ket{m}\ket{r_{p_i},r_{p_j}}\ket{0} \ket{p_0 \cdots p_i \cdots p_j \cdots p_\eta}\\
& \mapsto \ket{i,j}\ket{m}\ket{r_{p_i},r_{p_j}}\ket{v\left(m,r_{p_i},r_{p_j}\right)} \ket{p_0 \cdots p_i \cdots p_j \cdots p_\eta}\nonumber\\
& \mapsto \ket{i,j}\ket{m}\ket{0}^{\otimes {\cal O}(\log \left(N/\epsilon\right))} \ket{\neg v\left(m,r_{p_i},r_{p_j}\right)} \ket{p_0 \cdots p_i \cdots p_j \cdots p_\eta}\, \nonumber
\end{align}
where $\neg$ implies the negation of the following bit.

One might be concerned that we would need to compute a reciprocal and a square root as part of the 2-norm in the inequality for $v(m,r_p,r_q)$. However, we can much more efficiently compute \eq{inequality_V} by multiplying the norm through to the other side and squaring it; in other words we can evaluate $v(m,r_p,r_q)$ by testing
\begin{equation}
4\, m^2 v_{\rm max}^2\left\|r_p - r_q\right\|^2 < M^2 \, .
\end{equation}
The cost of evaluating this expression is also $\mathcal{O}(\log^2 M) = \mathcal{O}(\log^2(1/\epsilon))$ gates using the result of~\lem{sumsqa}, which again only contributes a polylogarithmic contribution to the gate complexity. Apart from simplicity, this form of a comparison test is much easier (will have lower constant factors) to implement because the squared 2-norm is equivalent to just computing the sum of squares; whereas the reciprocal square root can be a costly function to evaluate directly~\cite{soeken2017hierarchical,Haner2018}. We note that another approach to avoiding the computation of reciprocal square roots in the context of quantum simulating first quantized real space Hamiltonians is discussed in \cite{Poirier2021b}, where it suggested that one take advantage of structure arising from the inverse Laplace transformed version of the Coulomb kernel represented in a Cartesian component-separated approach. While that approach may have some advantages for product formula based methods, it is unclear whether it might also be useful for further reducing the cost of block encoding the Coulomb operator.

The block encoding of $U$ follows a similar pattern. In this case, we will re-express the $U$ operator as
\begin{align}
\label{eq:U_m}
U & = -\frac{u_{\rm max}}{M} \sum_{i=1}^\eta  \sum_{\ell=1}^L \zeta_\ell \sum_{p\in G} \sum_{m=1}^M u\left(m,r_p,R_\ell\right) \ket{p}\!\bra{p}_{i}\\
u\left(m,r_p,R_\ell\right) & = \begin{cases}
1 & \qquad {\rm if} \qquad \frac{m\,u_{\rm max}}{M} < \frac{1}{\left\|r_p - R_\ell\right\|}\\
0 & \qquad {\rm otherwise} \end{cases}\\
 u_{\rm max} & = \textrm{max}_{p,\ell}\left(  \frac{1}{\left\|r_p - R_\ell\right\|}\right) = {\cal O}\left(\frac{N^{1/3}}{\Omega^{1/3}}\right) \, .
\end{align}
where we will require that the $R_\ell$ do not overlap with any of the grid points and we once again take advantage of the fact that the right-hand side of the inequality in $u(m,r_p,R_\ell)$ is always greater than zero. Similar to the strategy employed for $V$, here the $\SEL$ operation will compute $u(m,r_p,R_\ell)$ into an ancilla that is then negated (thus, excluding the cases where $u(m,r_p,R_\ell)=0$ from the linear combination of unitaries). But there are two main differences: (1) the value $R_\ell$ will need to be loaded using a QROM and (2) we have not incorporated the $\zeta_\ell$ into the definition of the inequality function that is computed as part of $\SEL$. 

 Instead, we will include the $\zeta_\ell$ through the $\PREP$ operation. At the same time we will directly load the $\ket{R_\ell}$ values into superposition (rather than computing it from $\ell$). Thus, controlled on a flag ancilla specifying that the $U$ term is to be applied, our $\PREP$ circuit will act as
\begin{equation}
    \PREP \ket{0}^{\otimes {\cal O}\left(\log \left(N/\epsilon\right) \right)} \mapsto \sum_{\ell=1}^L \sum_{i=1}^{\eta}\sum_{m=1}^M \sqrt{\frac{\zeta_\ell}{\eta M}} \ket{R_\ell}\ket{i}\ket{m}\ket{0} \, .
\end{equation}
Using the coherent aliasing strategies developed in \cite{BGBWMPFN18}, the weighted superposition over $R_\ell$ can be prepared with no more than $\epsilon$ error in the amplitudes at Toffoli complexity $L + {\cal O}(\log(1/\epsilon))$. That will be a negligible additive cost compared to the rest of the $\SEL$ operation which will have Toffoli complexity ${\cal O}(\eta)$ for charge neutral systems. Thus, up to exponentially small errors in $\PREP$ we can see that
\begin{equation}
\left |\sum_{\ell=1}^L \frac{\zeta_\ell}{\left\|r_p - R_\ell\right\|} - \sum_{\ell=1}^L\frac{u_{\rm max} \zeta_\ell}{M} \sum_{m=1}^M  u\left(m,r_p, R_\ell\right) \right | \leq \frac{u_{\rm max}\lambda_\zeta}{M} \,
\end{equation}
with $\lambda_\zeta=\sum_\ell \zeta_\ell$,
and again we can make sure that the error in the Hamiltonian coefficients is less than $\epsilon$ if $M > u_{\rm max} \lambda_\zeta/ \epsilon$. Once again, this $\epsilon$ will only enter the gate complexity inside of logarithms.

Therefore, we see that we can implement the block encoding of $(T + U + V) / \lambda$ with gate complexity $\widetilde{\cal O}(\eta)$. Given our scheme, the $\lambda_T$ value is exactly the same as in \eq{lambdas}, and 
\begin{equation}
\lambda_V = \eta^2 v_{\rm max} = {\cal O}\left(\frac{\eta^2 N^{1/3}}{\Omega^{1/3}} \right)
\qquad \qquad
\lambda_U = \eta \left(\sum_{\ell=1}^L \zeta_\ell\right) u_{\rm max} = {\cal O}\left(\frac{\eta^2 N^{1/3}}{\Omega^{1/3}}\right)
\end{equation}
where these values of $\lambda_V$ and $\lambda_U$ involve slightly different constant factors compared to the forms in \eq{lambdas}, but the same asymptotic scaling. Accordingly, we see that the overall asymptotic complexity of the qubitization approach using this Hamiltonian would be
\begin{equation}
        \widetilde{\cal O} \left(\frac{\left(\lambda_T + \lambda_U + \lambda_V\right) \eta}{\epsilon}\right) 
        = \widetilde{\cal O}\left(\frac{\eta^2 N^{2/3}}{\epsilon \, \Omega^{2/3}} + \frac{ \eta^{3}N^{1/3}}{\epsilon \, \Omega^{1/3}}\right)
        = \widetilde{\cal O}\left(\frac{\eta^{4/3} N^{2/3} + \eta^{8/3}N^{1/3}}{\epsilon}\right)\,
\end{equation}
for $\eta\propto\Omega$,
which is the same as in \eq{asymptotic_qubitization}.
It is likely that the QFT will be the dominant cost in this method.
The complexity for the plane waves approach tends to be dominated by the cost of controlled selection of the momentum registers, which has complexity  $\mathcal{O}(\eta \log N)$ and is the only part that is polynomial in $\eta$.
The QFT has scaling as $\widetilde{\mathcal{O}}(\eta \log^2 N)$, which is larger and so should dominate the complexity.

Finally, we discuss how the interaction picture scheme can be combined with the real space first quantized representation in order to suppress the complexity term proportional to $N^{2/3}$. To do this, we follow a similar approach as in the main paper with the block encodings of $U$ and $V$ as described in this appendix. The main difference is that before evolving under the kinetic operator we must apply the QFT. The gate complexity of performing the QFT in each time step is $\widetilde{\cal O}(\eta \log^2 N)$ Toffoli gates.
A second difference is that it is no longer possible to keep the total kinetic energy in an ancilla register and update it when we block-encode $U$ and $V$.
At each step the kinetic energy will need to be recomputed, which has complexity $\mathcal{O}(\eta \log^2 N)$.
That is almost as large as the complexity scaling for the QFT. As a result there is the same asymptotic scaling as the interaction picture algorithm in the main paper up to a logarithmic factor,
\begin{equation}
\widetilde{\cal O}\left(\frac{\left(\lambda_U + \lambda_V\right)\eta }{\epsilon}\right)
\widetilde{\cal O}\left(\frac{\eta^{3} N^{1/3}}{\epsilon \, \Omega^{1/3} }\right) = \widetilde{\cal O}\left(\frac{\eta^{8/3} N^{1/3}}{\epsilon}\right) \, .
\end{equation}
Again, the complexity will be increased over the plane wave approach primarily due to the QFT.

\end{document}